\documentclass[runningheads,envcountsame]{llncs}

\usepackage[T1]{fontenc}
\usepackage[utf8]{inputenc}
\usepackage{graphicx}
\usepackage{amsmath}
\usepackage{amssymb}

\makeatletter
\let\claim\@undefined  \let\endclaim\@undefined
\makeatother

\spnewtheorem{claim}[theorem]{Claim}{\bfseries}{\itshape}
\spnewtheorem{assumption}[theorem]{Assumption}{\bfseries}{\itshape}

\makeatletter
\AtBeginDocument{\let\ll@endproof\endproof
  \def\endproof{\qed\ll@endproof}}
\makeatother

\usepackage{booktabs}
\usepackage{array}
\newcolumntype{R}[1]{>{\raggedright\arraybackslash}p{#1}}
\usepackage{url}
\usepackage{xcolor}

\newcount\Comments  % 0 suppresses notes to selves in text
\definecolor{darkgreen}{rgb}{0,0.6,0}
\newcommand{\kibitz}[2]{\ifnum\Comments=1{\textcolor{#1}{{#2}}}\fi}
\newcommand{\rmr}[1]{\kibitz{red}{[RM:#1]}}
\newcommand{\guy}[1]{\kibitz{blue}{[GUY:#1]}}

\newcommand{\Guy}[1]{\guy{#1}}

\newcommand{\figfile}[1]{%
  \IfFileExists{figures/#1.tex}{\input{figures/#1}}{\input{../figures/#1}}}

\usepackage{pgfplots}
\pgfplotsset{compat=1.17}
\usepgfplotslibrary{groupplots,colormaps}

\newcount\Draft \Draft=1
\newcommand{\plan}[1]{%
  \ifnum\Draft=1
    \par\medskip\noindent
    \begingroup\small\color{black!80}%
    \leftskip=1.1em
    \textcolor{blue!55!black}{\rule{2pt}{0.9em}}\hspace{0.45em}#1\par
    \endgroup\medskip
  \fi}
\newcommand{\tg}[1]{\ifnum\Draft=1{\footnotesize\color{red!60!black}\textsf{#1}}\fi}

\usepackage{environ}
\usepackage{etoolbox}
\newcount\DeferProofs \DeferProofs=1
\makeatletter
\newcommand{\DeferredProofList}{}
\newcommand{\dp@item}[2]{%
  \par\medskip\noindent{\itshape Proof of #1.}\enskip #2%
  \qed\medskip}
\NewEnviron{deferredproof}[1]{%
  \ifnum\DeferProofs=1
    \edef\dp@tmp{\noexpand\dp@item{\unexpanded{#1}}{\expandonce\BODY}}%
    \expandafter\gappto\expandafter\DeferredProofList\expandafter{\dp@tmp}%
    \ifnum\DeferProofNote=1
      \par\noindent{\small\itshape Proof in \Cref{app:proofs}.}\par\medskip
    \fi
  \else
    \dp@item{#1}{\BODY}%
  \fi}
\newcount\DeferProofNote \DeferProofNote=1
\newcommand{\printdeferredproofs}{\DeferredProofList}
\makeatother

\providecommand{\stat}[1]{\textup{\small[\emph{#1}]}}
\newcommand{\proofin}[1]{\par\noindent{\small\itshape Proof in \Cref{#1}.}\par\medskip}

\newcommand{\citet}[1]{\cite{#1}}
\newcommand{\citep}[1]{\cite{#1}}

\usepackage[hidelinks]{hyperref}
\usepackage{cleveref}
\crefname{conjecture}{Conjecture}{Conjectures}
\Crefname{conjecture}{Conjecture}{Conjectures}
\crefname{theorem}{Theorem}{Theorems}          \Crefname{theorem}{Theorem}{Theorems}
\crefname{claim}{Claim}{Claims}                \Crefname{claim}{Claim}{Claims}
\crefname{corollary}{Corollary}{Corollaries}   \Crefname{corollary}{Corollary}{Corollaries}
\crefname{proposition}{Proposition}{Propositions}\Crefname{proposition}{Proposition}{Propositions}
\crefname{lemma}{Lemma}{Lemmas}                \Crefname{lemma}{Lemma}{Lemmas}
\crefname{assumption}{Assumption}{Assumptions} \Crefname{assumption}{Assumption}{Assumptions}
\crefname{definition}{Definition}{Definitions} \Crefname{definition}{Definition}{Definitions}
\crefname{remark}{Remark}{Remarks}             \Crefname{remark}{Remark}{Remarks}
\crefname{example}{Example}{Examples}          \Crefname{example}{Example}{Examples}
\crefname{figure}{Fig.}{Figs.}                 \Crefname{figure}{Fig.}{Figs.}
\crefname{equation}{Eq.}{Eqs.}                 \Crefname{equation}{Eq.}{Eqs.}

\newcommand{\zin}[1]{x^{\mathrm{in}}_{#1}}
\newcommand{\zout}[1]{x^{\mathrm{out}}_{#1}}
\crefname{section}{Section}{Sections}          \Crefname{section}{Section}{Sections}
\crefname{appendix}{Appendix}{Appendices}      \Crefname{appendix}{Appendix}{Appendices}
\def\eps{\varepsilon}
\begin{document}

\title{Costly Voting in the Hotelling--Downs Model}
\titlerunning{Costly Voting in the Hotelling--Downs Model}
\author{Guy Wolf and Reshef Meir}
\authorrunning{Guy Wolf and Reshef Meir}
\institute{Technion---Israel Institute of Technology}

\maketitle

\begin{abstract} 
We study a partial-participation variation of the Hotelling-Downs model. Voters each have a cost to vote, and only vote when the comparative gain from their preferred candidate exceeds the cost. Under this model the median voter theorem breaks, and we study the extent of polarization under equilibria in different voters and cost distributions. We find that the main predictor of polarization  is the reverse-hazard-rate of the cost distribution, indicating that the driver of polarization under our model is the willingness of voters to respond to changes in positions of candidates. We then extend the model by adding parameters governing alienation and candidate competitiveness, showing that our results are robust even when taking into account other realistic factors.
\end{abstract}

%%%%%%%%%%%%%%%%%%%%%%%%%%%%%%%%%%%%%%%%%%%%%%%%%%%%%%%%%%%%%%%%%%%%%%%%%%%%%%
%%%%%%%%%%%%%%%%%%%%%%%%%%%%%%%%%%%%%%%%%%%%%%%%%%%%%%%%%%%%%%%%%%%%%%%%%%%%%%
%% Introduction: hook, model sketch, contributions, then related work.
%% Every headline claim here is stated in the form its theorem proves --
%% uniform-voter hypotheses and the alienation caveat included.
%%%%%%%%%%%%%%%%%%%%%%%%%%%%%%%%%%%%%%%%%%%%%%%%%%%%%%%%%%%%%%%%%%%%%%%%%%%%%%
\section{Introduction}\label{sec:intro}
Two developments mark democratic politics in recent decades: in the United
States the two parties' elected representatives have moved
apart~\cite{mccarty2006polarized} and partisans with
them~\cite{pew2017partisan}, while across much of Europe partisans are
sharply hostile to one another~\cite{reiljan2020fear}; and in many
democracies fewer citizens turn out to
vote~\cite{hooghe2017tipping,franklin2004voter}. The canonical account of
two-candidate competition predicts neither.

That account is the \emph{Hotelling--Downs model}
\cite{hotelling1929,downs1957economic}, in which two office-seeking candidates
choose policy platforms to maximize vote share. Under full voter participation
and single-peaked preferences it yields convergence to the median voter, whose
ideal policy defeats any alternative in pairwise majority voting---the
\emph{median voter theorem}. Yet divergence, not convergence, is the recurrent
empirical finding. In the United States,
Ansolabehere et al.~\cite{ansolabehere2001candidate} and Lee et
al.~\cite{lee2004voters} document systematic departures from median
convergence in U.S. House elections, while Levitt~\cite{levitt1996senators}
finds similar patterns in the U.S. Senate, and Bafumi and
Herron~\cite{bafumi2010leapfrog} show that members of Congress tend to adopt
positions more ideologically extreme than those of their constituencies.

One key difference between the benchmark Hotelling--Downs model and real-world
elections is that the benchmark typically assumes \emph{everyone votes}, whereas
in practice, voters may abstain, which makes it increasingly important to
incorporate non-participation into models of electoral
competition~\cite{meir2025tyranny}. That the two trends above are connected is
not a new conjecture. Calibrating a spatial model with voter abstention to
survey data from the U.S. presidential elections of 1980, 1984, 1988, 1996 and
2000, Adams, Merrill and Grofman report that in each of them the major-party
candidates had electoral incentives to diverge under their estimated turnout
model, with equilibrium divergence increasing in the policy distance between
the two parties' partisans~\cite{adams2005abstention,adams2003turnout}. What
has been missing is a characterization of \emph{which} feature of the
participation margin---the voters sitting at the edge of turning out---does
the work.

In this paper, we extend the Hotelling--Downs framework by introducing costly
voting. Inspired by the classic \emph{Calculus of Voting} (see below), voters participate only when the relative utility gain from supporting
their preferred candidate exceeds their individual cost of voting; otherwise,
they abstain. Turnout therefore depends endogenously on the
perceived policy difference between the candidates, and candidates must
account not only for ideological proximity but also for the turnout incentives
induced by their platform choices. Our model accommodates arbitrary
distributions of voter ideal points (density $v$) and voting costs (CDF $K$,
density $k$), and adds two
one-parameter ``dials'' to the participation side and the candidate side of
the game: voter \emph{alienation} $\alpha\ge0$, under which participation also
declines with a voter's distance from her own candidate, and candidate
\emph{competitiveness} $\beta\in[0,1]$, which interpolates between maximizing
one's vote share ($\beta=0$) and maximizing one's margin of victory ($\beta=1$). Costly voting is not
itself a dial: it is the always-on baseline, carried by the entire cost
distribution rather than by a single parameter, and both dials are read
against it.

\subsection{Our Contribution}

We first derive necessary conditions that every pure Nash equilibrium must
satisfy. A single \emph{balance principle} asserts that at every equilibrium the
marginal voters the candidates gain toward the center and those they lose on
their flanks are traded off at a constant rate, one fixed by the candidates'
objective alone and not by the electorate or the cost distribution
(\Cref{lem:balance,lem:balance-gen}). Equilibrium takes a particularly simple
form when either voter ideal points or voting costs are uniformly distributed:
under uniform costs every equilibrium traps a share of the electorate between
the platforms that again depends only on that objective, never on the
electorate itself (\Cref{thm_uniform_costs,prop:massgap}). These conditions are
not sufficient in general, and we show by example that profiles solving them
may fail to be equilibria. We therefore complement them with sufficiency
results: for uniformly distributed voters, a single monotonicity condition on
the cost distribution, satisfied by every log-concave cost density, under which
the equilibrium exists, is unique and is completely characterized, for
\emph{every} competitiveness $\beta<1$ and alienation $\alpha\ge0$; and for
uniformly distributed costs, density-ratio bounds under which the profiles at
fixed quantiles of a symmetric electorate (the quartiles in the baseline) are
equilibria (\Cref{cor_quarter,cor:quantile}).

Our central message concerns polarization (\Cref{sec:base:polar}). Costly
voting alone breaks median convergence: under uniform costs the
\emph{quartile} profile of a symmetric electorate (density ratio at most four)
is an equilibrium, so \emph{mass polarization} (of the electorate) transmits to
\emph{elite polarization} (here, of the candidates) one for one. Fixing the
electorate at uniform, we prove a comparative-statics theorem in the
reverse-hazard-rate stochastic order: cost distributions that are higher in
this order induce weakly more polarized equilibria, under the monotonicity
condition above. We further show this is exactly the right
order: the equilibrium correspondence is invariant to rescaling $K$, hence
depends on the cost distribution \emph{only} through its reverse hazard rate
$k/K$ (the proportional responsiveness of turnout to the stakes), and in
particular the \emph{level} of turnout has no causal role: distributions
exist with stochastically higher costs, lower turnout, and strictly less
polarization. What polarizes candidates is an electorate whose participation
is \emph{responsive} to the stakes, not one whose turnout is merely low.

We then turn each dial. Under alienation (\Cref{sec:alien}) we prove that
stronger alienation strictly polarizes for a uniform electorate and every
log-concave cost density, with the equilibrium separation $\Delta^*$ increasing
in $\alpha$ and tending to full separation, and that this holds at an
unchanged level of turnout, so the effect is not a by-product of lower
participation (\Cref{cor:alien-turnout}); that the cost channel survives
alienation, so that costlier voting in the reverse-hazard-rate order still
polarizes at every $\alpha$; and that turnout
becomes non-monotone in the platform gap, so the reverse causal channel
(extreme candidates depressing participation) reappears. Under
competitiveness (\Cref{sec:comp}) we prove that competition for the margin is
the model's one de-polarizing force, with a closed-form equilibrium path
$\Delta^*(\alpha,\beta)$ in the doubly-uniform benchmark, but that under
alienation the margin limit is a knife-edge: for a uniform electorate the
$\beta\uparrow1$ limit selects the \emph{most} polarized point of the margin
game's equilibrium band rather than the median; under alienation and uniform
costs the median is a strict local equilibrium of the margin game (no small
deviation pays) when the electorate is centrally peaked ($v''(\tfrac12)<0$) and fails when it is hollow
($v''(\tfrac12)>0$), a dichotomy settled at fourth order because the
second-order term vanishes; and for hollow electorates pure
equilibria can fail to exist at all. A computer-assisted certificate proves
non-existence at $\alpha=1$ with uniform costs and an arcsine electorate
($v=\mathrm{Beta}(\tfrac12,\tfrac12)$, the extreme hollow case), for
every $\beta\ge0.9168$, while at $\alpha=0$, under uniform costs, we prove existence for every
$\beta<1$ for symmetric electorates satisfying explicit density bounds, the
arcsine electorate among them (\Cref{m:thm:E1}). \Cref{tab:drivers}
collects the resulting map from primitives to polarization.

Methodologically, the non-existence results rest on a certificate scheme for
continuous-strategy spatial games: a Lipschitz bound on the payoff, exact
quadrature at the grid points, and adaptive refinement of the cells that
remain undecided, which together turn a finite computation into a statement
about the continuum. The scheme is not specific to this model and may be of
use wherever pure equilibria of a continuous game must be ruled out.

A released verification suite recomputes the paper's numerical claims from the
model's definitions, including both computer-assisted theorems
(\Cref{n:thm:C1,m:thm:E2}): the grid certificates are reproduced to the digits
reported, and the exact payoff evaluation they rest on is cross-checked against
independent adaptive quadrature.\footnote{Available at
\url{https://github.com/ANONYMIZED/costly-voting-hd}; see its \texttt{README}
for what is and is not covered.} Each figure states in its source how its data
were generated.

\subsection{Related Work}

\paragraph{Hotelling--Downs and variants}
Due to the simplicity and interpretability of the Hotelling--Downs spatial
model, a large literature has developed extensions and variations, often with
the aim of getting more `natural' equilibria where candidates diverge;
Osborne~\cite{osborne1995spatial} surveys it, and Grofman~\cite{grofman2004downs}
catalogues the Downsian assumptions---full participation among them---whose
relaxation produces divergence. Most
extensions retain the assumption that voters prefer the closer candidate but
vary the details of the model, e.g.,\ to consider voters' uncertainty about
candidates~\cite{kamada2014voter}, policy-motivated
candidates~\cite{wittman1983candidate,calvert1985robustness}, candidates'
uncertainty about voters~\cite{roemer1994theory,coughlin2015probabilistic},
relaxing equilibrium requirements~\cite{van2023rationalizable}, introduce
asymmetries between candidate favorability~\cite{aragones2002mixed} and so
on. A separate axis,
less explored in this line of work, is the candidates' \emph{objective}: vote
share, margin of victory, and probability of winning coincide under full
participation, so the classical model cannot distinguish them
\cite{aranson1974election,duggan2006candidate}, whereas
with abstention they come apart. That is not incidental: the equivalence
result of \citet{patty2002equivalence} holds precisely when abstention is
ruled out, and variable participation is the setting in which
Hinich and Ordeshook~\cite{hinich1970plurality} separated plurality from vote maximization.
\Cref{sec:comp} shows how consequential the choice is.

\paragraph{Abstention}
Voting costs and abstention have received considerable attention since the
outset of formal political economy, beginning with the Calculus of Voting
framework \cite{riker1968theory} and the pivotal-voter tradition it launched
\cite{palfrey1983strategic,palfrey1985participation,feddersen2004paradox}, continued
in the costly-voting models of B\"orgers~\cite{borgers2004costly},
Krasa and Polborn~\cite{krasa2009mandatory}, Krishna and
Morgan~\cite{krishna2012voluntary} and Herrera, Morelli and
Palfrey~\cite{herrera2014turnout}, in which, as here, each voter draws a private cost
from a distribution, but the alternatives on the ballot are fixed and the
questions are turnout, welfare and information aggregation, as well as later
``lazy bias'' models \cite{desmedt2010equilibria,elkind2015equilibria}. In all of these, the
fundamental assumption is that voters are more likely to participate when \emph{the outcome matters more} and when
 \emph{they are more pivotal}, and the analysis focuses on the
implications for equilibrium outcomes. This assumption is relaxed in
\cite{meir2025condorcet}, where voters consider various heuristics for their
pivotality. Other papers, such as \cite{cohensius2017proxy,meir2020sybil},
examine the effects of arbitrary or random participation on electoral outcomes
under the Median rule.\footnote{In some simulations in \cite{cohensius2017proxy}
the authors sampled more voters closer to the edges of the interval (i.e.,\
assuming `extreme' voters are more active), but without any underlying utility
model or formal justification.} However, none of the above models endogenize candidate responses to
abstention; instead, candidate positions are taken as fixed. That the
median result is fragile once participation is endogenous is an old theme:
Hinich, Ledyard and Ordeshook~\cite{hinich1972nonvoting} study existence of majority-rule equilibrium when
voters may abstain, and Hinich~\cite{hinich1977artifact} shows that in the
probabilistic-voting version the median is an equilibrium only under a
knife-edge symmetry: the convergence prediction is, in his phrase, an
artifact. The sharpest foil is Ledyard~\cite{ledyard1984pure}, in which platforms and
participation are determined together and voters, like ours, carry a private
voting cost; but they weigh the platform difference by their probability of
being pivotal, and the outcome is convergence. Both candidates adopt the same
welfare-maximizing platform and, the platforms being identical, nobody votes.
Our voters act as if decisive (\Cref{sec:participation}), and \Cref{sec:base}
shows that convergence then fails, beginning with \Cref{lem:interior}, which
rules out tied equilibria. A classical and
distinct driver of abstention is \emph{alienation}---non-participation out of
disaffection with even one's preferred option---introduced into spatial models
by Hinich and Ordeshook~\cite{hinich1969abstentions}, and since distinguished
from abstention out of \emph{indifference} between the candidates
\cite{brody1973indifference,anderson1992alienation,adams2006consequences};
our parameter $\alpha$ embeds this channel alongside voting costs, and \Cref{sec:alien} shows it is the
model's strongest polarizing force. That alienation-driven abstention can break
convergence is known: Callander and Wilson~\cite{callander2007turnout} obtain divergent positions
in a spatial entry model with alienation, and Oprea, Martin and Brennan~\cite{oprea2024moving} argue
that compulsory voting de-polarizes precisely by removing extreme voters'
ability to threaten abstention. What is new here is that the threat is priced
rather than assumed: the cost distribution $K$ is a primitive, and
\Cref{sec:alien} identifies which of its features the equilibrium responds to.

Closest to the present paper is Jones, Sirianni and Fu~\cite{jones2022polarization}, who examine how
a relative cost of voting that deters indifferent voters, third-party entry,
and a bimodal electorate each undermine the median voter theorem. Their
treatment of the first channel is numerical (candidate positions follow
gradient dynamics on a two-Gaussian or empirically estimated electorate, and
the relative cost of voting is a single scalar common to all voters), whereas
ours is analytic and holds for arbitrary $v$ and $K$, which is what lets
\Cref{sec:base:polar} isolate the reverse hazard rate as the only feature of
the cost distribution the equilibrium responds to.

Some variations of the Hotelling--Downs model explicitly incorporate the
possibility of voter abstention. In \cite{feldman2016variations,shen2016hotelling} voters abstain whenever all
candidates lie outside a fixed \emph{attraction interval}, so participation
does not decline gradually as candidates become less attractive;
Cohen and Peleg~\cite{cohen2019tolerance} instead draw each voter's tolerance range at random
from a density, so that the share of voters who still turn out falls gradually
with the distance to the nearest candidate, the nearest-candidate analogue of
the cost draw used here.
Coughlin~\cite{coughlin2015probabilistic} considers probabilistic attraction
models with a finite set of voters, in which each voter may abstain according
to an independent probability function. The attraction-interval models condition participation on the distance to the
nearest candidate alone, which is natural in \emph{economic competition}, where
a consumer's decision to buy depends only on the best offer available, but not
in \emph{political competition}, where a voter's stake is the difference between
the platforms. Models in which participation responds to that difference do
exist: the indifference-based abstention models cited above
\cite{anderson1992alienation,adams2006consequences}, the pivotal model of
Ledyard~\cite{ledyard1984pure}, the numerical study of Jones et al.~\cite{jones2022polarization},
and a recent workshop paper whose voters are prospect-theoretic and whose focus
is the computational complexity of optimizing a candidate's position
\cite{clevelandprospect26}. Two papers in political economy come closer still.
In Llavador~\cite{llavador2006platforms} voters weigh both platforms against a cost of
voting that grows with the distance to the platform they support (the
alienation channel of \Cref{sec:alien}), but his parties are policy-motivated,
so there the alienation cost is a moderating force that keeps them from
radicalizing, whereas here it acts on vote-share maximizers who would
otherwise converge, and pushes them apart. In Herrera, Levine and Martinelli~\cite{herrera2008platforms}
turnout is likewise endogenous to the platforms, but through the parties'
mobilization spending rather than through any voter's own cost--benefit
comparison. What distinguishes the present model is the combination:
participation responds to the platform difference; the threshold is drawn from
a full cost distribution $K$ that enters the analysis as a primitive; voters
weigh the platform difference directly against their cost, with no
pivot-probability discount (\Cref{sec:participation}); and the analysis is
carried out for arbitrary $v$ and $K$ rather than a fixed electorate.

\paragraph{Candidate polarization}
The distinction between \emph{mass} and \emph{elite} polarization drawn above
is a live one in political science; one prominent debate is on which type of
polarization precedes, or even causes, the other
\cite{hetherington2001resurgent,zingher2018high,abramowitz2008polarization,fiorina2008political,farina2015congressional}.

Neither trend with which this paper opened is universal: four-decade trends
across OECD democracies are mixed, with increases in some countries and
decreases in others~\cite{boxell2024cross}. Two recent studies speak directly
to the
channel studied here. Across eleven Western European countries (1977--2016),
Dreyer and Bauer~\cite{dreyer2019does} find that parties adopt more extreme positions as their
electorates polarize, and that the response runs through abstention: it is
strongest where turnout is low and fades as turnout rises.
Al Yussef~\cite{alyussef2026turnout} finds, across Dutch municipalities, that the vote
is more polarized where turnout is lower, consistent with centrist voters
abstaining disproportionately. Both studies measure the abstention margin by
the turnout rate; \Cref{sec:base:polar} shows that the model separates the
rate from the responsiveness of participation to the stakes, and that only
the latter moves the equilibrium.

Prior theoretical work in computational social choice has modeled polarization
in a variety of settings, such as committee selection~\cite{dong2025selecting},
which sometimes require complex definitions to measure. In the context of
2-candidate spatial competition, measuring polarization is trivial, and so we
can isolate and test the effect of each component of the model: the voters'
distribution, the voting costs, alienation, and the candidates' objective.

\paragraph{Core versus swing}
The distributive-politics literature asks when a party should reward its core
and when it should woo the center, with transfers rather than platforms as the
instrument
\cite{cox1986electoral,lindbeck1987balanced,dixit1996determinants,golden2013distributive};
there a group is ``core'' because a party can deliver to it more effectively,
a driver absent from our model, whose candidates hold one public,
non-discriminating instrument. Within this literature the nearest to our
question is Cox~\cite{cox2010swing}, who adds turnout to the targeting problem:
the most valuable voter supports you whenever she votes, and her participation
responds to what you offer. Both features describe the marginal voters on a
candidate's outer flank, and his comparative static is the crossing of
\Cref{cor:crossing}: the more mobilization and the less persuasion a party
believes possible, the more it concentrates on its core. Our balance condition
prices the trade-off between the two flanks at a rate that rises with
competitiveness and falls with alienation (\Cref{lem:balance-gen}), a price
no full-participation model can carry (\Cref{rem:beta}).

\paragraph{The two dials}
Each extension answers one of the gaps above: the $\alpha$-dial lets
participation fall with the distance to a voter's \emph{own} candidate, the
difference between economic and political competition noted above, and the
$\beta$-dial makes the candidates' objective a parameter rather than a
convention. Neither ingredient is new on its own; what the present model adds
is that both are read against an arbitrary cost distribution, so the effect of
each can be separated from the effect of costly voting itself.

\paragraph{Road map}
\Cref{sec:model} sets out the model; \Cref{sec:base,sec:alien,sec:comp}
analyze the baseline and the two extensions, each section moving from
characterizations to their consequences for polarization, which in
\Cref{sec:alien,sec:comp} come in two steps: first at a fixed value of the
parameter, then as it varies. \Cref{sec:comp} closes with the margin game
$\beta=1$, a special case; \Cref{sec:discussion} concludes. Proofs are deferred to
\Cref{app:proofs} and the subsequent appendices. A reader with limited time
can take \Cref{sec:base:polar,sec:alien:polar,sec:comp:polar}, the
subsections on what each force does to polarization, together with
\Cref{tab:drivers}; the subsections preceding each supply what they rest on.

%%%%%%%%%%%%%%%%%%%%%%%%%%%%%%%%%%%%%%%%%%%%%%%%%%%%%%%%%%%%%%%%%%%%%%%%%%%%%%
%%%%%%%%%%%%%%%%%%%%%%%%%%%%%%%%%%%%%%%%%%%%%%%%%%%%%%%%%%%%%%%%%%%%%%%%%%%%%%
%% Section 2 (Model) -- rewritten from scratch.
%%
%% REQUIRED PREAMBLE ADDITIONS (in addition to the current preamble):
%%   \usepackage{pgfplots}
%%   \pgfplotsset{compat=1.17}
%%   \usepgfplotslibrary{groupplots,colormaps}
%%
%% All figures are drawn in pgfplots and are PARAMETRIC: the model parameters
%% (c1, c2, alpha, kappa) are set as macros at the top of each tikzpicture,
%% and every coordinate is computed from them.  To change a figure, edit the
%% macros; no external tool is needed.
%%
%% fig:benefit (Sec. 2.2, the extension-free model) and fig:benefit-alien
%% (Sec. 2.5, alienation) share one panel macro that takes alpha as its
%% argument: fig:benefit instantiates it at alpha=0, fig:benefit-alien at
%% alpha=0 and alpha=1.  Keep the two macro blocks in sync.
%%%%%%%%%%%%%%%%%%%%%%%%%%%%%%%%%%%%%%%%%%%%%%%%%%%%%%%%%%%%%%%%%%%%%%%%%%%%%%

\section{Model}\label{sec:model}

\Cref{sec:hd} states the classical Hotelling--Downs game and the median
prediction it delivers, and isolates the assumption our model gives up.
\Cref{sec:participation} gives it up: participation becomes a cost--benefit
decision taken voter by voter, which is the paper's one substantive departure
and is in force everywhere in it. \Cref{sec:objectives,sec:eq-example} fix the
candidates' objective, the equilibrium notions and a running example; the game
assembled by then is the one analyzed in \Cref{sec:base}.
\Cref{sec:model:ext} adds the two extensions---alienation among voters,
competitiveness among candidates---each governed by its own parameter, each
switched off in the model presented first, and each taken up in a later
section (\Cref{sec:alien,sec:comp}).
\Cref{tab:notation} summarizes the notation.

\begin{table}[!t]
\centering
\caption{Notation of the model. The last block belongs to the extensions of
\Cref{sec:model:ext}; the model of \Cref{sec:participation,sec:objectives}
sets $\alpha=\beta=0$. Shorthands introduced later, in the analysis, are
collected separately in \Cref{tab:notation-analysis}.}
\label{tab:notation}
\setlength{\tabcolsep}{4pt}\small
\begin{tabular}{@{}l R{7.4cm}@{}}
\toprule
\multicolumn{2}{@{}l}{\emph{Primitives and parameters}}\\
$V,v$;\ $K,k$ & CDF/density of voter positions; of voting costs $\kappa$ (on $[0,1]$)\\
$\rho_v$ & density ratio $\sup v/\inf v$ (electorate evenness)\\
\midrule
\multicolumn{2}{@{}l}{\emph{Profiles and induced behavior}}\\
$c_i$;\ $d_i(x):=|x-c_i|$ & platform of candidate $i\in\{1,2\}$; distance to voter $x$\\
$c_1\le c_2$;\ $\Delta$;\ $m$ & labeling; separation $c_2-c_1$; midpoint $\tfrac{c_1+c_2}{2}$\\
$B_i(x)$;\ $p_i(x)$ & benefit of voting for $i$: $d_{-i}-d_i$, \Cref{eq:benefit}; share of the voters at $x$ who vote for $i$, \Cref{eq:shares}\\
$u_i$;\ $T$ & vote share $\int_0^1 p_i v\,dx$; turnout $u_1+u_2$\\
$\bar O_i$;\ $\bar I_i$ & votes from $i$'s flank; from her half of the contested middle, \Cref{eq:arms}\\
$\Delta^*$ & equilibrium separation\\
\midrule
\multicolumn{2}{@{}l}{\emph{Decorations}}\\
$\,{}^{*}$;\ $\widetilde{\,\cdot\,}$ & equilibrium value; the comparison primitive of a stochastic order\\
\midrule
\multicolumn{2}{@{}l}{\emph{Extension parameters (\Cref{sec:model:ext})}}\\
$\alpha\ge 0$ & alienation: benefit becomes $d_{-i}-(1+\alpha)d_i$, \Cref{eq:benefit-gen}\\
$\beta\in[0,1]$;\ $w_i$ & competitiveness; margin $u_i-u_{-i}$\\
$U_i=u_i-\beta\,u_{-i}$ & candidate $i$'s objective, \Cref{eq:objective}\\
$\Delta^*(\alpha)$;\ $\Delta^*(\alpha,\beta)$ & equilibrium separation, when the dependence matters\\
\bottomrule
\end{tabular}
\end{table}

\subsection{The Hotelling--Downs Benchmark}\label{sec:hd}

In the classical model \cite{hotelling1929,downs1957economic}, every voter
votes for the candidate nearest to her ideal point, and each candidate chooses
a platform to maximize the number of votes she receives; if both candidates
choose the same platform, they split the electorate evenly. With voters
uniformly distributed on $[0,1]$, the unique equilibrium famously places both
candidates at the median, $\tfrac12$. This is a special case of a more general
statement. Let $V$ be the CDF of an arbitrary continuous distribution of voter
ideal points on $[0,1]$, and let candidate $i\in\{1,2\}$ choose a
\emph{platform} $c_i\in[0,1]$, so that a \emph{profile} is a pair $(c_1,c_2)$.

\begin{claim}[See \cite{black1948rationale,downs1957economic}]\label{claim:hd}
The only equilibria of the Hotelling--Downs game are the profiles $(c_1,c_2)$
with $V(c_1)=V(c_2)=\tfrac12$, i.e., both candidates at a median of $V$.
\end{claim}

\Cref{claim:hd} is robust to how a voter's preference over the two platforms is
modeled, so long as she votes. Let her payoff depend on \emph{both} platforms
through the gain $d_{-i}-d_i$ introduced in \Cref{sec:participation} below,
rather than on the nearer one alone: with full participation she still votes
for the candidate nearer to her, the induced vote shares are the classical
ones, and the equilibrium set is unchanged. What breaks the claim is therefore
not the utility model but the decision of whether to vote at all. The remainder
of this section replaces the assumption of full participation with an explicit
cost--benefit decision by each voter.

\subsection{Voters, Costs, and Participation}\label{sec:participation}

\paragraph{Positions and costs}
There is a continuum of voters with positions $x\in[0,1]$, distributed
according to a continuous distribution with CDF $V$ and density $v$. In
addition to her position, each voter draws an idiosyncratic voting cost
$\kappa\in[0,1]$, i.i.d.\ across voters according to a continuous distribution
with CDF $K$ and density $k$, independently of her position.\footnote{%
Independence reflects the interpretation of $\kappa$ as the logistics of
voting (registration, travel, forgone time), which are plausibly unrelated to
ideology. Restricting the support to $[0,1]$ is essentially without loss of
generality: all distances lie in $[0,1]$, so the benefit from voting defined
below never exceeds $1$, voters with $\kappa>1$ never participate, and only the
restriction of $K$ to $[0,1]$ matters.}
The following mild conditions are assumed throughout the paper; we note
explicitly whenever one of them can be relaxed.\footnote{We reserve the term
\emph{regularity} for its usual meaning, monotonicity of the virtual cost
\eqref{eq:psi}; that is a hypothesis of \Cref{thm_uniform_votes}(b), not a
standing assumption.}

\begin{assumption}[Standing assumptions]\label{ass:reg}
The densities $v$ and $k$ are continuous and strictly positive on $(0,1)$,
$K(0)=0$, and $K(1)\in(0,1]$. A defect $1-K(1)>0$ is the mass of voters whose
cost exceeds every attainable benefit (the never-voters of
\Cref{prop:scale}), so $K$ need not be a proper distribution on $[0,1]$.
\end{assumption}

\paragraph{The participation rule}
Given a profile of platforms $(c_1,c_2)$, write $d_i(x):=|x-c_i|$ for the
distance between a voter at $x$ and candidate $i$. The voter's \emph{benefit}
from voting for candidate $i$ is her preferred candidate's distance advantage,
how much better off she is if that candidate wins rather than the opponent:
\begin{equation}\label{eq:benefit}
B_i(x)\;:=\;d_{-i}(x)\;-\;d_i(x),
\end{equation}
and she votes for candidate $i$ if $B_i(x)\ge\kappa$, abstaining if neither
candidate clears her threshold. Since $B_i(x)>0$ requires $x$ to be strictly
closer to $c_i$ than to $c_{-i}$, at most one candidate can offer a positive
benefit, so the rule assigns (almost) every voter at most one vote.

We emphasize that the benefit \eqref{eq:benefit} compares the platforms
directly, rather than weighting the comparison by the probability of being
pivotal as in the \emph{Calculus of Voting} \cite{riker1968theory} and the
pivotal-voter models that formalize it \cite{palfrey1985participation}. Voters
thus behave as if decisive, an assumption in the spirit of expressive
\cite{brennan1993democracy} or boundedly rational voting, appropriate when voters do not know the positions
of other voters or cannot assess their own pivotality \cite{meir2025condorcet}.

The departure is narrower than it looks. A common pivot probability
$\varpi$ that does not depend on the platforms is absorbed into the cost
distribution: voting when $\varpi\,B_i(x)\ge\kappa$ is voting when
$B_i(x)\ge\kappa/\varpi$, which is the same game with costs
$t\mapsto K(\varpi t)$. \Cref{prop:mass} shows what that
reparametrization does: nothing at all when $K$ is a power law, uniform costs
included, and otherwise, for a uniform electorate, a move whose direction is
set by the elasticity of $K$. What
the model genuinely assumes away is a pivot probability that responds to the
platforms themselves.

\paragraph{Participation shares}
Aggregating over costs, the share of voters at position $x$ who vote for
candidate $i$ is
\begin{equation}\label{eq:shares}
p_i(x)\;=\;
\begin{cases}
K\big(B_i(x)\big), & \text{if } B_i(x)\ge 0,\\[2pt]
0, & \text{otherwise.}
\end{cases}
\end{equation}
\Cref{fig:benefit} illustrates the mechanics. On candidate $1$'s side of the
midpoint, $B_1$ is a piecewise-linear ``tent'': its \emph{outer arm}, facing
away from the opponent, is flat at the full gain $\Delta:=c_2-c_1$---every
voter on the far side of $c_1$ enjoys the same advantage---and its
\emph{inner arm}, facing the opponent, falls at slope $2$ from the peak at
$c_1$ to the midpoint $m:=\tfrac{c_1+c_2}{2}$, where it vanishes; the mirror
statement holds for candidate $2$. Letting voting costs tend to zero recovers
the vote shares of the classical model for any profile with $c_1\ne c_2$.

\paragraph{Flank and contested middle}
The arms are pieces of the benefit function; the voters standing under them
are worth naming separately, because every first-order condition in the paper
is a statement about those two groups. Since $B_1<0$ beyond the midpoint,
candidate $1$ draws her votes entirely from $[0,m]$, which her own platform
splits in two: her \emph{flank} $[0,c_1]$, the voters standing behind her,
and her half $[c_1,m]$ of the \emph{contested middle} $[c_1,c_2]$, the voters
lying between the two platforms. Write
\begin{equation}\label{eq:arms}
\bar O_1\;:=\;\int_0^{c_1}\!K\big(B_1(x)\big)v(x)\,dx\;=\;K(\Delta)\,V(c_1),
\qquad
\bar I_1\;:=\;\int_{c_1}^{m}\!K\big(B_1(x)\big)v(x)\,dx
\end{equation}
for the votes each group yields, so that $u_1=\bar O_1+\bar I_1$; the closed
form for $\bar O_1$ is available because the outer arm is flat, so every voter
on the flank faces the same benefit $\Delta$. The mirror definitions serve
candidate $2$. Note that $B_1$ is continuous at $c_1$, taking the value
$\Delta$ from both sides: nothing distinguishes the outermost voter of the
middle from the innermost voter of the flank. That is what makes the split
pure bookkeeping, and it is what \Cref{lem:balance} turns on.

The tie-breaking convention, however, differs from the classical one: at
$c_1=c_2$ every voter's benefit is $0$, so by $K(0)=0$ (\Cref{ass:reg})
(almost) every voter abstains and both candidates receive zero votes, whereas
the classical model splits the electorate evenly. This discrepancy is immaterial for our
analysis: as we show below (\Cref{lem:interior}), tied profiles are never
equilibria of the model just described. The one exception arises in the margin
limit of the second extension, where they survive as zero-turnout,
zero-margin equilibria (\Cref{prop:margin_uniform}).

% fig:benefit -- extracted from model_section.tex.
% The float is self-contained: every macro it uses is defined inside it.
% Inputs from the active draft as \input{figures/fig_benefit}.
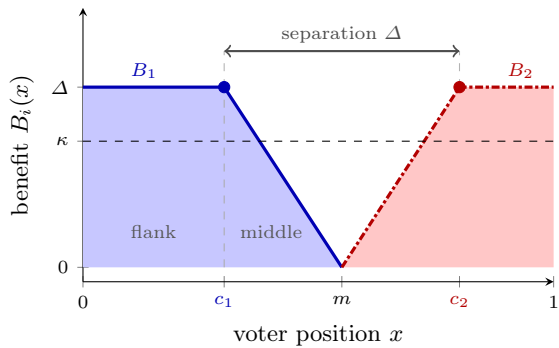
\begin{figure}[t]
\centering
%% ---------------------------------------------------------------------------
%% Fig: the benefit of voting.  PARAMETRIC.
%% Edit \cA (c1), \cB (c2), \kap (cost threshold shown).  All coordinates are
%% computed from these macros.  \benefitpanel is written as a function of
%% alpha so that the same code draws fig:benefit-alien in Sec. 2.5; here it is
%% instantiated at alpha=0, where the outer arm is flat and no clipping
%% caveat applies.
%% The two platforms are marked twice over: a dot on each apex, a coloured
%% tick label, and the span that measures the separation between them.
%% \sepY is the height of that span; keep it above the apexes (= \DD) and
%% below ymax.
%% ---------------------------------------------------------------------------
\def\cA{0.3}   % c1
\def\cB{0.8}   % c2
\def\kap{0.35} % individual cost level shown as dashed line
\pgfmathsetmacro\DD{\cB-\cA}       % Delta
\pgfmathsetmacro\mm{(\cA+\cB)/2}   % midpoint m
\pgfmathsetmacro\sepY{\DD+0.10}    % height of the separation span
% Draws one panel for a given alpha (#1).
\newcommand{\benefitpanel}[1]{%
  % shaded positive parts (= p_i under uniform costs; areas = u_i under uniform voters)
  \addplot[draw=none,fill=blue!22,forget plot] coordinates
    {(0,{\DD-#1*\cA}) (\cA,\DD) ({\cA+\DD/(2+#1)},0)} \closedcycle;
  \addplot[draw=none,fill=red!22,forget plot] coordinates
    {({\cB-\DD/(2+#1)},0) (\cB,\DD) (1,{\DD-#1*(1-\cB)})} \closedcycle;
  % benefit curves
  \addplot[blue!70!black,very thick] coordinates
    {(0,{\DD-#1*\cA}) (\cA,\DD) (\mm,{-#1*\DD/2})};
  \addplot[red!70!black,very thick,densely dashdotted] coordinates
    {(\mm,{-#1*\DD/2}) (\cB,\DD) (1,{\DD-#1*(1-\cB)})};
  % individual cost threshold
  \addplot[black,dashed,forget plot] coordinates {(0,\kap) (1,\kap)};
}
\begin{tikzpicture}
\begin{axis}[
  width=0.64\linewidth, height=5.2cm,
  xmin=0, xmax=1, ymin=-0.04, ymax=0.72,
  xtick={0,\cA,\mm,\cB,1},
  xticklabels={$0$,\textcolor{blue!70!black}{$c_1$},$m$,\textcolor{red!70!black}{$c_2$},$1$},
  ytick={0,\kap,\DD},
  yticklabels={$0$,$\kappa$,$\Delta$},
  xlabel={voter position $x$},
  ylabel={benefit $B_i(x)$},
  tick label style={font=\scriptsize}, label style={font=\small},
  axis lines=left,
  clip=false,
]
\benefitpanel{0}
% the platforms, marked on their own apexes
\addplot[only marks,mark=*,mark size=2.1pt,blue!70!black,forget plot]
  coordinates {(\cA,\DD)};
\addplot[only marks,mark=*,mark size=2.1pt,red!70!black,forget plot]
  coordinates {(\cB,\DD)};
% the separation, as the span between the two platforms
\addplot[gray!60,dashed,thin,forget plot] coordinates {(\cA,\DD) (\cA,\sepY)};
\addplot[gray!60,dashed,thin,forget plot] coordinates {(\cB,\DD) (\cB,\sepY)};
\draw[<->,black!70,thick] (axis cs:\cA,\sepY) -- (axis cs:\cB,\sepY);
\node[above,font=\scriptsize,black!70] at (axis cs:\mm,\sepY) {separation $\Delta$};
% The partition of eq:arms, drawn on candidate 1's side only: her own platform
% splits the voters she draws from into the flank and her half of the contested
% middle.  Candidate 2 mirrors it; drawing both doubles the labels for nothing.
% The rule runs axis-to-apex.  The labels sit low, at y = 0.10: the inner arm
% descends from (\cA,\DD) to (\mm,0), so a "middle" label placed centrally at
% the height of the cost line is cut by the curve.  At y = 0.10 and centred on
% x = 0.40 its right edge still clears the arm by a comfortable margin.
\addplot[gray!65,dashed,thin,forget plot] coordinates {(\cA,0) (\cA,\DD)};
\node[font=\scriptsize,black!65] at (axis cs:0.15,0.10) {flank};
\node[font=\scriptsize,black!65] at (axis cs:0.40,0.10) {middle};
% curves labelled where they run, instead of in a legend
\node[anchor=south,font=\scriptsize,blue!70!black] at (axis cs:0.13,\DD) {$B_1$};
\node[anchor=south,font=\scriptsize,red!70!black] at (axis cs:0.93,\DD) {$B_2$};
\end{axis}
\end{tikzpicture}
\caption{The benefit of voting at the profile $(c_1,c_2)=(0.3,0.8)$. A voter at
$x$ whose cost is $\kappa$ votes for candidate $i$ when $B_i(x)$ rises above the
dashed line; aggregating over costs, a share $p_i(x)=K\big(B_i(x)^+\big)$ of the
voters at $x$ vote for $i$. The separation $\Delta=c_2-c_1$ appears twice over, as
the gap between the platforms and as the height of each tent, because a voter
standing behind her own candidate gains exactly $\Delta$ by voting. The dashed
vertical marks the partition of \eqref{eq:arms}: candidate $1$'s own platform
splits the voters she draws from into her flank and her half of the contested
middle, and candidate $2$ mirrors it. Under uniform costs $p_i$ is the positive
part of $B_i$, so with uniform voters the shaded areas are the vote shares,
here $(u_1,u_2)=(0.2125,0.1625)$, and the two labelled pieces of the blue area
are $\bar O_1$ and $\bar I_1$.}
\label{fig:benefit}
\end{figure}

\subsection{Candidates and Objectives}\label{sec:objectives}

Two candidates, indexed by $i\in\{1,2\}$, simultaneously choose platforms
$c_i\in[0,1]$ in a one-shot game. Candidates know the distributions $V$ and
$K$ (but not the realized cost of any individual voter). Candidate $i$'s
\emph{vote share} and the resulting \emph{turnout} are
\[
u_i(c_1,c_2)\;=\;\int_0^1 p_i(x)\,v(x)\,dx ,
\qquad
T\;=\;u_1+u_2\;\le\;1 .
\]

Each candidate maximizes her vote share $u_i$. This is the natural objective
when raw support matters in itself (think of public funding, perceived
legitimacy, or proportional systems), and it is also the objective of the
classical candidate of \Cref{sec:hd}, so the only thing that has changed so
far is that voters may now stay home.

\subsection{Equilibrium and a Running Example}\label{sec:eq-example}

We study pure Nash equilibria of the simultaneous-move game with payoffs
$(u_1,u_2)$. The game is symmetric, so equilibria come in mirror pairs; when
describing a profile we label the candidates so that $c_1\le c_2$, and write
$\Delta=c_2-c_1$ for their \emph{separation} (our measure of candidate
polarization) and $m=\tfrac{c_1+c_2}{2}$ for the midpoint.

Three nested notions of a profile recur throughout, and the gaps between them
carry real content, so we fix them here.

\begin{definition}[Stationary, local, and Nash]\label{def:stationary}
Fix a profile with $0<c_1<c_2<1$ and write $\pi_i$ for candidate $i$'s payoff,
which is $u_i$ here and $U_i$ once \Cref{sec:model:ext} is in play. The
profile is \emph{stationary} if $\partial\pi_i/\partial c_i=0$ for both
candidates; a \emph{local equilibrium} if each $c_i$ is a local maximum of
$\pi_i(\cdot,c_{-i})$, and a \emph{strict} one if that maximum is strict; and
a \emph{Nash equilibrium} if each $c_i$ is a global maximum of
$\pi_i(\cdot,c_{-i})$ on $[0,1]$.
\end{definition}

Every interior Nash equilibrium is a local equilibrium, and, the payoffs being
differentiable at interior profiles (\Cref{n:lem:master}), every interior local
equilibrium is stationary. Neither converse holds, and the paper exhibits both
failures: a stationary profile can be a local \emph{minimum} of the deviator's
payoff (\Cref{n:rem:chi-criterion}), and a local equilibrium can lose to a
distant deviation to a different local maximum (\Cref{app:counterexample}).
That is why every characterization in this paper comes in two parts:
conditions every equilibrium must satisfy, and distributional conditions under
which a profile meeting them is an equilibrium.

As a running illustration we use the \emph{doubly-uniform environment}, in
which voter positions and voting costs are both uniform on $[0,1]$:
$x,\kappa\sim U(0,1)$, that is $v\equiv1$ and $K(t)=t$. It returns in every
section that follows, where we also call it the doubly-uniform
\emph{benchmark}. \Cref{fig:benefit} depicts the profile
$(c_1,c_2)=(0.3,0.8)$, under which $u_1=0.2125$ and $u_2=0.1625$. Fixing $c_2=0.8$, the left panel of
\Cref{fig:example} plots candidate $1$'s vote share as a function of her own
platform: her best response $c_1=c_2/3\simeq0.267$ trades off a larger flank of
nearby voters (pushing her toward the opponent) against the higher turnout
generated by a wider separation (pushing her away). The right panel shows the
full strategy space: the two best-response curves intersect exactly twice, at
$\big(\tfrac14,\tfrac34\big)$ and $\big(\tfrac34,\tfrac14\big)$, which we
later show (\Cref{cor_uu}) to be the only equilibria in this environment.
Intuitively, when the candidates are too close to each other, many voters
abstain, driving the candidates apart, and when they are too far the mass of
voters between them becomes worth chasing.

% fig:example -- extracted from model_section.tex.
% The float is self-contained: every macro it uses is defined inside it.
% Inputs from the active draft as \input{figures/fig_example}.
\begin{figure}[t]
\centering
%% ---------------------------------------------------------------------------
%% Fig: the candidates' problem in the doubly-uniform baseline.  PARAMETRIC.
%% Left panel: u_1(c_1; c2=\cOpp).  Closed forms (alpha=0, V,K ~ U(0,1)):
%%   u_1 = c_1*(c_2-c_1) + (c_2-c_1)^2/4          for c_1 <= c_2,
%%   u_1 = (1-c_1)*(c_1-c_2) + (c_1-c_2)^2/4      for c_1 >= c_2;
%%   best response c_1 = c_2/3 with value c_2^2/3 (for c_2 >= 1/2).
%% Right panel: heatmap of u_1(c_1,c_2) with both best-response curves
%%   BR_1(c_2) = c_2/3 if c_2>=1/2, (2+c_2)/3 if c_2<=1/2  (solid white)
%%   BR_2(c_1) = (2+c_1)/3 if c_1<=1/2, c_1/3 if c_1>=1/2  (dashed white)
%% and the two equilibria (1/4,3/4), (3/4,1/4) as stars.
%% ---------------------------------------------------------------------------
\def\cOpp{0.8} % c2 held fixed in the left panel
\begin{minipage}[b]{0.455\linewidth}
\centering
\begin{tikzpicture}
\begin{axis}[
  width=1.05\linewidth, height=5.0cm,
  xmin=0, xmax=1, ymin=0, ymax=0.25,
  xlabel={$c_1$}, ylabel={$u_1(c_1,c_2)$, \ $c_2=0.8$},
  xtick={0,{\cOpp/3},0.5,\cOpp,1},
  xticklabels={$0$,$\frac{c_2}{3}$,$\frac12$,$c_2$,$1$},
  axis lines=left,
]
\addplot[blue!70!black, very thick, domain=0:\cOpp, samples=120]
  {x*(\cOpp-x)+(\cOpp-x)^2/4};
\addplot[blue!70!black, very thick, domain=\cOpp:1, samples=60]
  {(1-x)*(x-\cOpp)+(x-\cOpp)^2/4};
\addplot[gray, dashed, forget plot] coordinates
  {({\cOpp/3},0) ({\cOpp/3},{\cOpp^2/3}) (0,{\cOpp^2/3})};
\addplot[only marks, mark=*, mark size=1.8pt, black] coordinates
  {({\cOpp/3},{\cOpp^2/3})};
\node[anchor=south west, font=\scriptsize] at (axis cs:{\cOpp/3},{\cOpp^2/3})
  {best response};
\end{axis}
\end{tikzpicture}
\end{minipage}\hfill
\begin{minipage}[b]{0.50\linewidth}
\centering
\begin{tikzpicture}
\begin{axis}[
  width=0.78\linewidth, height=5.6cm,
  view={0}{90}, axis on top,
  xmin=0, xmax=1, ymin=0, ymax=1,
  xlabel={$c_1$}, ylabel={$c_2$},
  colormap/viridis, colorbar,
  colorbar style={width=2.5mm, ytick={0,0.1,0.2,0.3},
    tick label style={font=\scriptsize},
    ylabel={vote share $u_1$}, ylabel style={font=\scriptsize}},
  tick label style={font=\scriptsize},
]
% u_1(c_1,c_2); the two branches meet continuously (both -> 0 on the diagonal)
\addplot3[surf, shader=interp, samples=45, samples y=45, domain=0:1, y domain=0:1]
  {ifthenelse(x<y, x*(y-x)+(y-x)^2/4, (1-x)*(x-y)+(x-y)^2/4)};
% best responses (drawn at z=1, i.e., on top of the surface)
\addplot3[white, very thick, domain=0.5:1, samples=2] ({x/3},{x},{1});
\addplot3[white, very thick, domain=0:0.5, samples=2] ({(2+x)/3},{x},{1});
\addplot3[white, very thick, dashed, domain=0:0.5, samples=2] ({x},{(2+x)/3},{1});
\addplot3[white, very thick, dashed, domain=0.5:1, samples=2] ({x},{x/3},{1});
\addplot3[only marks, mark=star, mark size=4.5pt, white, line width=1.6pt] coordinates
  {(0.25,0.75,1) (0.75,0.25,1)};
\addplot3[only marks, mark=star, mark size=3pt, black, thick] coordinates
  {(0.25,0.75,1) (0.75,0.25,1)};
\node[white,font=\scriptsize,anchor=west] at (axis cs:0.06,0.30,1) {$c_1=BR_1(c_2)$};
\node[white,font=\scriptsize,anchor=east] at (axis cs:0.94,0.70,1) {$c_2=BR_2(c_1)$};
\end{axis}
\end{tikzpicture}
\end{minipage}
\caption{The candidates' problem in the doubly-uniform environment
($x,\kappa\sim U(0,1)$). \emph{Left:} candidate $1$'s vote share
against $c_2=0.8$; the best response $c_1=c_2/3$ balances proximity to more
voters against the turnout gained from separation (the kink is at
$c_1=c_2$). \emph{Right:} heatmap of $u_1$ over all profiles, with the
best-response curves of candidate $1$ (solid) and candidate $2$ (dashed);
their intersections (stars) are the two equilibria
$\big(\tfrac14,\tfrac34\big)$ and $\big(\tfrac34,\tfrac14\big)$.}
\label{fig:example}
\end{figure}
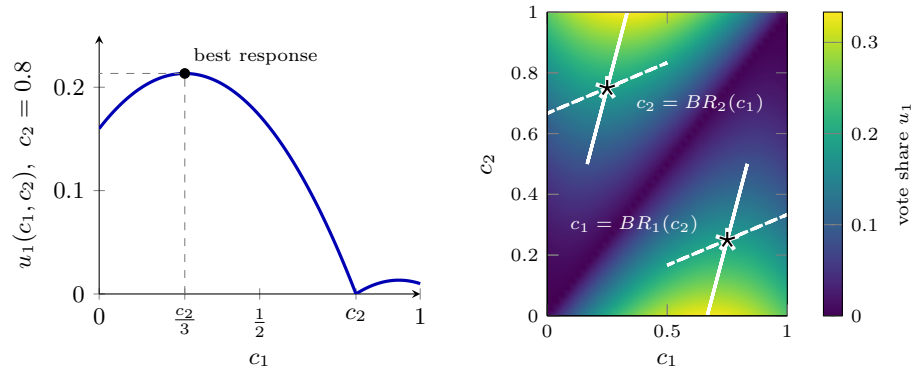

\subsection{Extensions: Alienation and Competitiveness}\label{sec:model:ext}

The model above---costly voting, and nothing else---is the model analyzed in
\Cref{sec:base}. We now add two further ingredients. \emph{Alienation},
governed by a parameter $\alpha\ge0$, changes the voters: participation may
decline not only with a small benefit of voting, but also with the voter's
distance from her \emph{own} candidate; \Cref{sec:alien} switches it on.
\emph{Competitiveness}, governed by a parameter $\beta\in[0,1]$, changes the
candidates: they value their own votes but may also discount votes cast for
the opponent, interpolating between pure vote-share maximization ($\beta=0$)
and pure margin maximization ($\beta=1$); \Cref{sec:comp} switches it on.
Setting $(\alpha,\beta)=(0,0)$ returns the model of
\Cref{sec:participation,sec:objectives} verbatim.

\paragraph{Alienation}
Alienation is a classic driver of abstention in spatial models
\cite{hinich1969abstentions}: voters who are ideologically distant from all
candidates are less inclined to turn out, even when their gain is unchanged.
We capture it by penalizing the distance to a voter's own candidate inside the
benefit, so that \cref{eq:benefit} becomes
\begin{equation}\label{eq:benefit-gen}
B_i(x)\;:=\;d_{-i}(x)\;-\;(1+\alpha)\,d_i(x),
\qquad \alpha\ge 0 ,
\end{equation}
\rmr{regarding our discussion on presenting $\alpha$.  suggested $\alpha (\underline C-d_i) + (1-\alpha)(d_{-i}-d_i)$ whereas the current form can still be presented as a different convex combination: $\alpha (-d_i) + (1-\alpha)(d_{-i}-d_i)$. The difference is just a constant ($C\alpha$) but it might matter conceptually. In the second form, which we currently use, benefit is monotonically decreasing with $\alpha$, and you show that (under conditions) higher $\alpha$ induces more polarization. This is a clean message, but a possible objection is that this actually happens mainly because participation drops. There are several ways to handle this:
\begin{itemize}
    \item ignore it and hope no one will raise this objection :)
    \item show a simple example (e.g. doubly-uniform) where polarization increases despite same or higher participation (by scaling $K$)
    \item in fact maybe this can be generalized? we know scaling $K$ does not change the equilibrium so we can always scale $K$ so that participation is same for higher $\alpha$. 
    \item Prove same results for $C=1$ or $C=0.5$ (note that $C-d_i$ still positive)
\end{itemize}}\guy{Claude: your third bullet, and it generalizes cleanly---\Cref{cor:alien-turnout}. Two remarks on it. First, the objection is already answered in principle by \Cref{rem:fosd}: \Cref{prop:scale} lets the \emph{level} of participation be set anywhere without moving the equilibrium, so ``polarization rose because turnout fell'' is not a mechanism the model has, and that remark was sitting in \Cref{sec:base} un-invoked. The corollary is the constructive witness. Second, a detail in the phrasing: scaling $K$ so that the \emph{alienated} game matches the baseline needs $z\ge1$, hence $T^*(\alpha)\le T^*(0)$, which the paper does not prove. Scaling whichever side has more turnout \emph{down} keeps $z\le1$ and needs no such hypothesis. On the first bullet: not advisable, it is the obvious objection to the headline. On the second: the doubly-uniform case is now the illustration, with the explicit $z(\alpha)=16(3+\alpha)/3(4+\alpha)^2$. On the fourth: I would close it. Reproving for $C=1$ or $C=\tfrac12$ would move \eqref{eq:benefit-gen}, the tent geometry, $\lambda=2+\alpha$, the zeros, every first-order condition in \Cref{sec:alien,sec:comp}, both supplements and the figures, and the $\zeta$ micro-foundation would need re-deriving for the new constant---to buy the same separation of $\alpha$ from the participation level that the corollary buys in three lines. Say the word if you want it reopened.} \rmr{so to summarize: add a remark that this can be presented as the second convex combination, and when presenting the results explain that the effect on polarization is not since $\alpha$ reduces participation as scaling has no effect on equilibrium.}\guy{Done: the convex combination now follows the definition, as a mix of the relative and the absolute view of the outcome. The weight stays $\zeta$, since a convex weight must lie in $[0,1)$ and $\alpha$ does not. For (b), the point was already at \Cref{cor:alien-turnout} and in the opening of \Cref{sec:alien}; it is now in the introduction too.}
which is \cref{eq:benefit} at $\alpha=0$.\rmr{cite equations using Eq.~\eqref{eq:benefit}}\guy{Claude: done, but through cleveref rather than by hand. The paper had three styles at once---157 bare \texttt{\char92 eqref}, 8 \texttt{\char92 Cref\{eq:\ldots\}} that were printing ``Equation~(3)'', and ``Eq.''\ nowhere. Equations now have a \texttt{\char92 crefname} like every other reference class (Fig., Section, Appendix), so the prefix is automatic and cannot drift again. Bare $(3)$ is kept only where the number is an appositive---``the partition~$(1)$''---since ``Eq.''\ would be ungrammatical there, and ranges now print ``Eqs.~(3)--(5)''.} The definition
is a convex combination of the two natural ways to value the outcome of an
election. One is \emph{relative}: a voter cares how much better one
candidate is than the other, $d_{-i}-d_i$, the gain of \cref{eq:benefit},
which makes her abstain when the candidates are too alike for the choice to
matter. The other is \emph{absolute}: she cares about the platform she
actually helps elect, $-d_i$, which makes distance from her own candidate a
reason to stay home---the alienation of \cite{hinich1969abstentions}. A
voter who puts weight $1-\zeta$ on the first view and $\zeta\in[0,1)$ on the
second has
\[
(1-\zeta)\big(d_{-i}-d_i\big)+\zeta\big({-d_i}\big)
=(1-\zeta)\Big(d_{-i}-\tfrac{1}{1-\zeta}\,d_i\Big)
=(1-\zeta)\,B_i
\]
for $\alpha=\zeta/(1-\zeta)$, and the positive factor $1-\zeta$ is absorbed into the cost distribution
exactly as the pivot probability is in \Cref{sec:participation}. So $\alpha$
is not an ad hoc penalty. It indexes the mix between the two views, from the
purely relative voter of \Cref{sec:base} at $\alpha=0$ toward the purely
absolute one as $\alpha\to\infty$; the weight is written $\zeta$ because a
convex weight lies in $[0,1)$ while $\alpha$ ranges over $[0,\infty)$. The
endpoint $\zeta=1$ is excluded, and necessarily so: there $B_i=-d_i\le0$, so
with non-negative costs nobody ever votes and the model is empty. Every
$\alpha\ge0$ therefore keeps some weight on the comparison between the
candidates.

Everything else in \Cref{sec:participation} is stated in terms of $B_i$ and
therefore holds verbatim for every $\alpha\ge0$: the voting rule
$B_i(x)\ge\kappa$, the argument that at most one candidate can offer a
positive benefit, and the participation shares \eqref{eq:shares}. The
definition also has a compact reading: $B_i(x)\ge\kappa$ iff
$d_{-i}(x)-d_i(x)\ge\kappa+\alpha\,d_i(x)$, so alienation inflates the
effective cost of voting by a fraction of the distance to one's own
candidate.\rmr{the first was already specified above}\guy{Claude: right---it had been given three times before being announced as new, at the extensions overview, at the head of this paragraph, and again here. Only the micro-foundation is new, so only it is now presented as a second reading.}

Geometrically, alienation tilts the benefit tent of \Cref{fig:benefit}. On
candidate $1$'s side of the midpoint, $B_1$ now rises at slope $\alpha$ toward
the platform, still peaks at height $\Delta$ at $c_1$, and falls at slope
$2+\alpha$ to the midpoint $m$, where it equals $-\alpha\Delta/2$; the mirror
statement holds for candidate $2$. The outer arm is therefore no longer flat,
so participation falls fastest among voters far from their own candidate, and
a ``dead zone'' of full abstention opens around the midpoint.
\Cref{fig:benefit-alien} places the two cases side by side. The tie-breaking
convention of \Cref{sec:participation} is unaffected: at $c_1=c_2$ the benefit
is $-\alpha\,d_i(x)\le0$, so (almost) every voter abstains a fortiori.

% fig:benefit-alien -- extracted from model_section.tex.
% The float is self-contained: every macro it uses is defined inside it.
% Inputs from the active draft as \input{figures/fig_benefit_alien}.
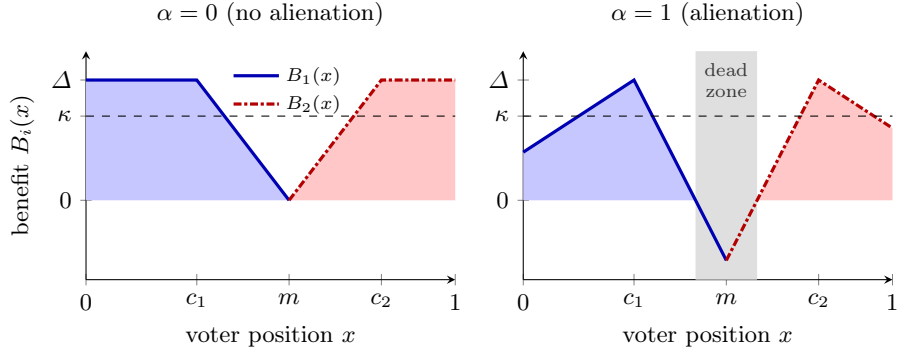
\begin{figure}[t]
\centering
%% ---------------------------------------------------------------------------
%% Fig: the benefit of voting, without and with alienation.  PARAMETRIC.
%% Same macros and same panel code as fig:benefit in Sec. 2.2; only the alpha
%% argument differs.  Edit \cA (c1), \cB (c2), \kap (cost threshold shown).
%% NOTE: coordinates assume no outer dead zone, i.e. alpha*c1 <= Delta and
%% alpha*(1-c2) <= Delta; otherwise clip the outer arm at its zero
%% x0 = c1 - Delta/alpha (resp. c2 + Delta/alpha).
%% ---------------------------------------------------------------------------
\def\cA{0.3}   % c1
\def\cB{0.8}   % c2
\def\kap{0.35} % individual cost level shown as dashed line
\pgfmathsetmacro\DD{\cB-\cA}       % Delta
\pgfmathsetmacro\mm{(\cA+\cB)/2}   % midpoint m
% Draws one panel for a given alpha (#1).  Identical to \benefitpanel in
% fig:benefit; renamed so the two figures are independent of each other.
\newcommand{\benefitpanelalien}[1]{%
  % shaded positive parts (= p_i under uniform costs; areas = u_i under uniform voters)
  \addplot[draw=none,fill=blue!22,forget plot] coordinates
    {(0,{\DD-#1*\cA}) (\cA,\DD) ({\cA+\DD/(2+#1)},0)} \closedcycle;
  \addplot[draw=none,fill=red!22,forget plot] coordinates
    {({\cB-\DD/(2+#1)},0) (\cB,\DD) (1,{\DD-#1*(1-\cB)})} \closedcycle;
  % benefit curves
  \addplot[blue!70!black,very thick] coordinates
    {(0,{\DD-#1*\cA}) (\cA,\DD) (\mm,{-#1*\DD/2})};
  \addplot[red!70!black,very thick,densely dashdotted] coordinates
    {(\mm,{-#1*\DD/2}) (\cB,\DD) (1,{\DD-#1*(1-\cB)})};
  % individual cost threshold
  \addplot[black,dashed,forget plot] coordinates {(0,\kap) (1,\kap)};
}
\begin{tikzpicture}
\begin{groupplot}[
  group style={group size=2 by 1, horizontal sep=9mm, ylabels at=edge left},
  width=0.53\linewidth, height=4.6cm,
  xmin=0, xmax=1, ymin=-0.33, ymax=0.62,
  xtick={0,\cA,\mm,\cB,1},
  xticklabels={$0$,$c_1$,$m$,$c_2$,$1$},
  ytick={0,\kap,\DD},
  yticklabels={$0$,$\kappa$,$\Delta$},
  xlabel={voter position $x$},
  every axis title/.append style={font=\small},
  axis lines=left,
  clip=false,
  legend style={font=\scriptsize, draw=none, fill=none, at={(0.55,0.98)}, anchor=north},
]
\nextgroupplot[title={$\alpha=0$ (no alienation)}, ylabel={benefit $B_i(x)$}]
\benefitpanelalien{0}
\legend{$B_1(x)$,$B_2(x)$}
\nextgroupplot[title={$\alpha=1$ (alienation)}]
% the dead zone (\zin{1},\zin{2}), drawn first so the curves sit on top of it:
% \zin{1} = c_1 + Delta/(2+alpha), \zin{2} = c_2 - Delta/(2+alpha), at alpha = 1.
\pgfmathsetmacro\xone{\cA+\DD/3}
\pgfmathsetmacro\yone{\cB-\DD/3}
\addplot[draw=none,fill=gray!25,forget plot] coordinates
  {(\xone,-0.33) (\yone,-0.33) (\yone,0.62) (\xone,0.62)} \closedcycle;
\node[font=\scriptsize,gray!45!black,align=center] at (axis cs:\mm,0.50)
  {dead\\zone};
\benefitpanelalien{1}
\end{groupplot}
\end{tikzpicture}
\caption{Alienation reshapes the benefit of voting; the profile is
$(c_1,c_2)=(0.3,0.8)$, as in \Cref{fig:benefit}. \emph{Left:} $\alpha=0$, the
model of \Cref{sec:participation}, reproduced here for comparison.
\emph{Right:} $\alpha=1$. Alienation tilts the outer arms (slope $\alpha$
instead of $0$), steepens the inner arms (slope $2+\alpha$ instead of $2$),
and opens a dead zone of abstention around the midpoint $m$. Under uniform
costs and uniform voters the shaded areas are the vote shares:
$(u_1,u_2)=(0.2125,0.1625)$ on the left and $(0.147,0.122)$ on the right.}
\label{fig:benefit-alien}
\end{figure}

\begin{remark}\label{rem:alpha}
Under full participation the parameter $\alpha$ would be irrelevant. A voter
prefers candidate $i$ exactly when $B_i>B_{-i}$, that is when
$d_{-i}-(1+\alpha)d_i>d_i-(1+\alpha)d_{-i}$, i.e.\ when
$(2+\alpha)\,d_{-i}>(2+\alpha)\,d_i$: the common factor cancels and she votes
for the nearer candidate at every $\alpha\ge0$. If $T\equiv1$ the induced
vote shares are therefore the classical ones and \Cref{claim:hd} is untouched,
median equilibrium included. Alienation is in this respect exactly like
competitiveness (\Cref{rem:beta}): it has no bite of its own, and acts only
through the participation decision. Everything in \Cref{sec:alien} is
accordingly attributable to abstention.
\end{remark}

\paragraph{Competitiveness}
Vote share is not the only objective a candidate might have: the \emph{margin}
$w_i:=u_i-u_{-i}$ is the natural objective when only victory and the size of
the mandate matter. We span both with a single competitiveness parameter
$\beta\in[0,1]$: candidate $i$ values her own votes at $1$ and the opponent's
at $-\beta$,
\begin{equation}\label{eq:objective}
U_i\;=\;u_i-\beta\,u_{-i}
\;=\;(1-\beta)\,u_i+\beta\,w_i
\;=\;\frac{1-\beta}{2}\,T\;+\;\frac{1+\beta}{2}\,w_i .
\end{equation}
The middle expression is the most direct reading of the parameter: $U_i$ is the
convex combination of vote share and margin that puts weight $\beta$ on the
margin. The one on the right shows that $U_i$ is equally a weighted combination of
turnout, which the candidates share, and margin, over which they are in
pure conflict. At $\beta=0$ we have $U_i=u_i$ and we are back in the model of
\Cref{sec:objectives}; at $\beta=1$ the common turnout term drops out, the
game is zero-sum, and candidates are margin maximizers. With either extension
in play, the equilibrium notion of \Cref{sec:eq-example} is applied to the game
with payoffs $(U_1,U_2)$, which we call the \emph{$\beta$-game}; its two
endpoints are the \emph{vote-share game} ($\beta=0$) and the \emph{margin
game} ($\beta=1$). We write $\Delta^*(\alpha)$ or $\Delta^*(\alpha,\beta)$
when the dependence of the equilibrium separation on $\alpha$ and $\beta$
matters.

The reader can see how different parameters affect the turnout along the voter axis using the artifact: \url{https://claude.ai/artifact/2Ggrf3obhnSywKHsyEs2BE}

\begin{remark}\label{rem:beta}
Under full participation the parameter $\beta$ would be irrelevant: if
$T\equiv 1$ then $U_i=(1+\beta)\,u_i-\beta$ is a positive affine
transformation of $u_i$, so all $\beta\in[0,1]$ induce the same game. In
particular, vote-share and margin maximization coincide in the classical
model. It is precisely abstention, which makes turnout $T(c_1,c_2)$ respond to
the platforms, that separates the objectives; how much they diverge is the
subject of \Cref{sec:comp}.
\end{remark}

%%%%%%%%%%%%%%%%%%%%%%%%%%%%%%%%%%%%%%%%%%%%%%%%%%%%%%%%%%%%%%%%%%%%%%%%%%%%%%

%%%%%%%%%%%%%%%%%%%%%%%%%%%%%%%%%%%%%%%%%%%%%%%%%%%%%%%%%%%%%%%%%%%%%%%%%%%%%%
%% Results sections (final prose).  Long proofs are typeset in
%% \Cref{app:proofs} via the deferredproof environment; proofs carried over
%% from the previous manuscript are not reproduced (see the note there);
%% proofs of n:/m:-labelled supplementary results are in the two supplementary appendix batches.
%%%%%%%%%%%%%%%%%%%%%%%%%%%%%%%%%%%%%%%%%%%%%%%%%%%%%%%%%%%%%%%%%%%%%%%%%%%%%%

\section{Costly Voting}\label{sec:base}

We begin with the baseline model: voters weigh only the gain from voting
against its cost, and candidates maximize vote share. Both extensions of
\Cref{sec:model:ext} are off throughout this section ($\alpha=\beta=0$, so
that $B_i=d_{-i}-d_i$ and $U_i=u_i$), and every statement below is to be read
under that convention, without further mention.
This section establishes the pattern that the two extensions will repeat.
\Cref{sec:base:char} characterizes equilibria, starting from a single balance
principle, stated here in its baseline form and extended in turn by each
parameter, and specializing along a ladder of increasingly explicit cases:
uniform
costs, uniform voters, and both at once, which is the doubly-uniform
environment of \Cref{sec:eq-example} that runs through the paper as its
simplest worked example.
\Cref{sec:base:polar} then extracts the paper's central message: costly voting
alone breaks median convergence, and it does so through two separate
channels---the shape of the electorate and the shape of the cost
distribution---while the \emph{level} of turnout plays no role at all.
\Cref{tab:notation-analysis} collects the shorthands used from here on.

\begin{table}[!t]
\centering
\caption{Shorthands introduced in the analysis, collected for reference; each is
defined where it first appears. The second block is used only from
\Cref{sec:alien} on. The symbols of the model itself are in
\Cref{tab:notation}.}
\label{tab:notation-analysis}
\setlength{\tabcolsep}{4pt}\small
\begin{tabular}{@{}l R{7.4cm}@{}}
\toprule
\multicolumn{2}{@{}l}{\emph{Analysis shorthands}}\\
$r:=k/K$;\ $G(t):=\int_0^tK$ & reverse hazard rate of $K$ (\Cref{cor:crossing}); integrated CDF\\
$\psi$;\ $\psi_\beta$ & virtual cost $t+K/k$, \Cref{eq:psi}; its $\beta$-weighting $t+\tfrac{1+\beta}{1-\beta}K/k$, \Cref{thm:beta-suff}\\
$I_i,O_i$ & marginal voter mass on $i$'s inner/outer arm, \Cref{lem:balance}\\
$\sigma$;\ $\eta(u)=u\,k(u)/K(u)$ & electorate spread; elasticity of $K$, \Cref{prop:mass}\\
$\bar\Delta(\alpha)$ & turnout-maximizing separation, \Cref{prop:turnout}\\
$c_{\mathrm{med}}$;\ $w(c,c')$ & median $V^{-1}(\tfrac12)$; margin of a candidate at $c$ against one at $c'$, \Cref{thm:margin_median,prop:nonexistence}\\
\midrule
\multicolumn{2}{@{}l}{\emph{Geometry of the benefit tent under alienation (\Cref{fig:anatomy})}}\\
$\lambda:=2+\alpha$ & inner-arm slope (outer arms have slope $\alpha$) \rmr{I think $\lambda$ s overloaded. check}\guy{Claude: checked, and it is not, any more. The clash was $\lambda$ as the mixing weight in the alienation micro-foundation, where $\alpha=\lambda/(1-\lambda)$; that weight is now $\zeta$ (\Cref{sec:model:ext}) and the model section contains no $\lambda$ at all. Every occurrence in the paper now means $2+\alpha$. I audited the rest of the symbol table while I was there: \Cref{claim:hd} writes profiles as $(c_1,c_2)$, so $\zin{1}$ has lost its third meaning too. The one live clash left was voter position $x$ against the tent zeros, your separate note on \Cref{prop:master}; that is now fixed too, and the zeros carry a candidate subscript like everything else.}\\
$\zin{i}$;\ $\zout{i}$ & inner and outer zero of candidate $i$'s tent, where her benefit reaches $0$\\
$b,b_1,b_2$ & boundary benefits: $B_i$ at the nearer edge of $[0,1]$, \Cref{eq:b-delta}\\
$\Delta_0:=\alpha/\lambda$ & maximal separation inside the margin band $E(\alpha)$, \Cref{prop:margin_uniform}\\
$E(\alpha)$ & margin band: the set of pure equilibria at $\beta=1$, \Cref{prop:margin_uniform}\\
\bottomrule
\end{tabular}
\end{table}

\subsection{Equilibrium Characterization}\label{sec:base:char}

We first record that equilibria can involve neither ties nor boundary
positions.

\begin{lemma}[Interiority]\label{lem:interior}
Under \Cref{ass:reg}, every pure Nash equilibrium
$(c_1,c_2)$ satisfies
$c_1\neq c_2$ and, labeling the candidates so that $c_1<c_2$, we have
$0<c_1<c_2<1$.
\end{lemma}
\proofin{app:omitted-proofs}

\paragraph{One balance principle}
All our characterizations flow from a single first-order condition, which we
state here in its baseline form and sharpen once in each of the two sections
that follow. It is what the partition \eqref{eq:arms} gives on
differentiation. Define the \emph{marginal masses} $I_i$ and $O_i$ on the two
arms: for candidate $1$, \rmr{are those definitions?}\guy{Claude: they are, and they now say so. The paper's definitional marker is $:=$---it already carries $r$, $\psi$, $\Psi$, $\lambda$, $\rho_v$ and twenty more---so these were stragglers. Fixed the other four I found while checking: $d_i(x)$ in \Cref{sec:participation}, $\Delta$ and $m$ at the head of \Cref{app:omitted-proofs}, and $I,O$ in \Cref{m:sec:D}.}
\[
I_1:=\int_{c_1}^{m} k\big(B_1(x)\big)\,v(x)\,dx,
\qquad
O_1:=\int_{0}^{c_1} k\big(B_1(x)\big)\,v(x)\,dx\;=\;k(\Delta)V(c_1) ,
\]
with the mirror-image definitions for candidate $2$, and $I:=I_1+I_2$,
$O:=O_1+O_2$. The voter at $x$ contributes $k(B_i(x))v(x)$: the density of
voters at $x$ who are exactly indifferent between voting and abstaining, and
whom a marginal platform move therefore gains or loses.

Moving $c_1$ acts on $\bar O_1$ and $\bar I_1$ in two ways at once. The
benefits inside each region change: on the flank $B_1\equiv\Delta$ falls at
unit rate, costing $O_1$, while in the middle $B_1=c_1+c_2-2x$ rises at unit
rate, earning $I_1$. And the boundary between the regions moves with her, so
voters just behind $c_1$ pass from the middle into the flank. That second
effect is the one that cancels. Each transferred voter carries the same
benefit $\Delta$ on either side of the boundary, so she adds
$K(\Delta)v(c_1)$ to the flank and removes exactly as much from the middle:
\[
\frac{\partial\bar O_1}{\partial c_1}=K(\Delta)v(c_1)-O_1,
\qquad
\frac{\partial\bar I_1}{\partial c_1}=I_1-K(\Delta)v(c_1),
\qquad\text{so}\qquad
\frac{\partial u_1}{\partial c_1}=I_1-O_1 .
\]
Neither $I_1$ nor $O_1$ is by itself the derivative of the region it names;
each is the part of that derivative due to benefits changing, and the parts
due to the partition shifting are equal and opposite. Setting the sum to zero
gives the lemma. \rmr{ok good}

\begin{lemma}[Balance conditions]\label{lem:balance}
Let $v,k$ satisfy \Cref{ass:reg}. Every interior stationary profile
$(c_1,c_2)$ (\Cref{def:stationary}) satisfies, for each candidate
$i$,
\[
I_i\;=\;O_i ,
\]
and in aggregate $I=O$.
\end{lemma}
\noindent{\small\itshape Proof: the baseline case of \Cref{m:prop:ratio}.}

\medskip
In aggregate, exactly half of the total marginal mass sits on the inner arms,
for every $v$ and $k$.
Neither extension changes what is being balanced, only the rate at which the two
arms exchange: alienation reshapes the arms themselves (\Cref{prop:master})
and competitiveness re-prices them, and once both are in play the two effects
combine into a single statement, \Cref{lem:balance-gen}. The rest of this
subsection specializes this balance one distribution at a time: first uniform
costs against an arbitrary electorate, then a uniform electorate against
arbitrary costs, then both at once.\footnote{Each specialization gives
conditions that are necessary but not sufficient, in the sense of
\Cref{def:stationary}; \Cref{app:exist-uniq,app:counterexample} supply the two
examples promised there.}

\begin{remark}[Existence]\label{rem:exist}
Payoffs are continuous on $[0,1]^2$ (including at ties, where vote shares
vanish), so a possibly-mixed equilibrium always exists by Glicksberg's
fixed-point theorem \cite{glicksberg1952further}. Pure existence is subtler. It holds in
the benchmark settings of this section, and in each extension:
\begin{itemize}
\item uniform voters, with $K$ regular (\Cref{thm_uniform_votes});
\item uniform costs, with a symmetric electorate of density ratio at most four
  (\Cref{cor_quarter});
\item a non-increasing cost density with $\rho_k:=k(0^+)/k(1^-)<\infty$ and a
  symmetric electorate whose density ratio $\rho_v$ (\Cref{tab:notation})
  satisfies $\rho_v\rho_k\le4$ (\Cref{n:prop:A4});
\item under alienation (\Cref{thm:alien-unique,n:cor:A3exist}) and under
  competitiveness (\Cref{thm:beta-suff,m:thm:E1}), on the analogous classes.
\end{itemize}

\noindent Between the two lies a third statement, weaker than pure
existence and available far more widely: a \emph{stationary} profile---one
solving the first-order conditions---exists by an intermediate-value
argument along the symmetric diagonal, under no sufficiency hypothesis at
all. For uniform voters and every $K$ this is the boundary behaviour of
$\Psi$ in the proof of \Cref{thm:alien-unique}, which uses no
super-regularity, so \cref{eq:sym-foc} always has a root. For uniform costs
and a symmetric electorate, \Cref{m:prop:roots} gives the same conclusion
for every $\alpha>0$ and every $\beta<1$: the own-side derivative on the
diagonal is $(1+\alpha+\beta)V(1/\lambda)>0$ at $c_1=0^+$ and negative just
left of $c_1=\tfrac12$. What these arguments do not deliver is that the root
is a \emph{best response}, and closing that gap is exactly what the
density-ratio umbrellas above do.
\rmr{Didn't we have an argument at least for existence of a stationary profile, by showing the best-response lines intersect?}\guy{Claude: we did, twice, and this remark cited neither---it went straight from Glicksberg to the list of settings with pure equilibria. Both are now in the paragraph above. The one you are thinking of is almost certainly \Cref{m:prop:roots}, whose proof is the sign change you describe. One caveat on the phrasing: intersecting \emph{best-response} curves would give a Nash equilibrium, not merely a stationary profile, and we do not have that---best responses need not be single-valued, which is why pure existence is the subtle part. What the diagonal argument gives is the first-order conditions having a common root.}
We conjecture that a pure equilibrium always exists in the baseline
($\alpha=\beta=0$) under \Cref{ass:reg}. The case the list leaves
conspicuously open is uniform costs with $\rho_v>4$, where we have neither a
proof nor a counterexample: the four-cluster electorate of
\Cref{app:counterexample} has $\rho_v=120$ and does possess equilibria, even
though not every profile solving the conditions is one. We expect existence to
survive alienation as well, at $\beta=0$, but on thinner evidence, since the
range of density ratios covered by \Cref{n:cor:A3exist} shrinks to nothing as
$\alpha$ grows.
\Cref{sec:comp} shows that pure existence can fail once candidates maximize
the margin instead.
\end{remark}

\paragraph{Uniform costs}
Assume first that costs are uniform while voters are arbitrary. With
$k\equiv1$, the balance $I_i=O_i$ turns into a statement about voter mass
alone.

\begin{theorem}[Uniform costs: every equilibrium traps half the electorate]\label{thm_uniform_costs}Let $\kappa\sim U(0,1)$ and let $v$ satisfy \Cref{ass:reg}. Then every Nash
equilibrium $(c_1,c_2)$ satisfies
\[
V(c_1)\;=\;V(c_2)-\tfrac12\;=\;\tfrac12\,V(m):
\]
the candidates trap exactly half of the electorate between them, and each
candidate's flank holds exactly half the mass on her side of the midpoint. \Guy{add in AAMAS}
\end{theorem}
\proofin{app:omitted-proofs}

\noindent The one half is specific to vote-share maximization, but the law is
not. Under the competitive objective $U_i=u_i-\beta\,u_{-i}$ of
\cref{eq:objective}, with alienation still off, the same computation traps
exactly $\tfrac{1-\beta}2$ of the electorate, again whatever the electorate
(\Cref{prop:massgap}).

These conditions are not sufficient in general, but they are whenever the
electorate is not too uneven. Let
$\rho_v:=\sup_{(0,1)}v\,/\,\inf_{(0,1)}v$ denote the density ratio.

\begin{proposition}[Sufficiency under uniform costs]\label{prop:fourflat}
Let $\kappa\sim U(0,1)$ and $\rho_v\le4$. Then each candidate's payoff is concave
in her own position on each side of the opponent, so a profile
$0<c_1<c_2<1$ is a Nash equilibrium if and only if it satisfies
\Cref{thm_uniform_costs} and neither candidate gains by a \emph{cross-over}
deviation, a move to the far side of her opponent. The cross-over condition is
one inequality per candidate, comparing her equilibrium payoff with the best
she could get on the far side. \Guy{cut in AMASS}
\end{proposition}
\proofin{app:omitted-proofs}

\begin{corollary}[Quartile equilibria]\label{cor_quarter}
Let $\kappa\sim U(0,1)$ and $\rho_v\le4$. Any profile with
$V(c_1)=\tfrac14$, $V(c_2)=\tfrac34$ and $V(m)=\tfrac12$ is a Nash
equilibrium. If $v$ is symmetric about $\tfrac12$, the \emph{quartile profile}
$\big(V^{-1}(\tfrac14),V^{-1}(\tfrac34)\big)$ satisfies all three conditions
automatically (so an equilibrium exists) and is the only profile
satisfying them. \guy{add in AMASS in short as a followup of \ref{thm_uniform_costs}}
\end{corollary}
\proofin{app:omitted-proofs}

\noindent Competitiveness generalizes the quartiles to quantiles:
\Cref{cor:quantile} puts the candidates at $V^{-1}\big(\tfrac{1+\beta}4\big)$
and $V^{-1}\big(\tfrac{3-\beta}4\big)$, which are the quartiles at $\beta=0$
and collapse to the median as $\beta\uparrow1$.

\rmr{for AAMAS I'd just leave the result on symmetric $V$. For asymmetric it seems like the condition in the conditions on the Quartile eq. almost never hold. so other than examples (and a conjecture that eq. exists for low ratio?) we don't have much to say about asymmetric voters}
For skewed electorates the equilibrium detaches from any fixed quantiles: for
$v=\mathrm{Beta}(1.5,4)$ the conditions of \Cref{thm_uniform_costs} pin down
$(c_1,c_2)\approx(0.149,0.413)$, an asymmetric profile with
$V(c_1)\approx0.29$ and $V(c_2)\approx0.79$, which we verified numerically to
be an equilibrium; a reminder, too, that the bound $\rho_v\le4$ is sufficient
but far from necessary (this density violates it). Nor need the solution be
unique: \Cref{app:multi-eq} exhibits a \emph{symmetric} electorate with three
distinct equilibria, two of them asymmetric mirror images, with $\rho_v=3$:
inside the sufficiency range $\rho_v\le4$ of \Cref{prop:fourflat} the quartile
profile is an
equilibrium, but not necessarily the only one.

\paragraph{Uniform voters}
Now let voters be uniform while costs are arbitrary. This is our main
result, and the characterization here is complete: a single monotonicity
condition on the cost distribution delivers existence \emph{and} uniqueness,
and pins the equilibrium separation to the root of one scalar equation. That
condition is monotonicity of the \emph{virtual cost}
\begin{equation}\label{eq:psi}
\psi(t)\;:=\;t+\frac{K(t)}{k(t)},\qquad t\in(0,1),
\end{equation}
the cost-side counterpart of Myerson's virtual value
\cite{myerson1981optimal}. 

\begin{definition}[Regularity]
A  distribution whose virtual cost is
strictly increasing is called \emph{regular}.    
\end{definition}
Regularity of cost distribution is a standard hypothesis in economic
settings~\cite{LIU20163049}. Several results below ask for more: that the
second term of Eq.~\eqref{eq:psi} not decrease on its own. Call such a $K$
\emph{super-regular}. Since the first term increases, super-regularity
implies regularity. Moreover, super-regularity is implied by basic properties shared by most commonly used distributions~\cite{bagnoli2005logconcave}:
\[
k\ \text{non-increasing, or } k\ \text{log-concave}
\quad\Longrightarrow\quad K\ \text{super-regular}
\quad\Longrightarrow\quad K\ \text{regular},
\]
where neither implication
reverses. 
In particular, Power costs $K(t)=t^p$ with $p<1$ are super-regular although $k$ is
log-\emph{convex} (\Cref{ex:power}), and Footnote~\ref{fn:psi} exhibits a regular $K$ that is not super-regular.
\Guy{Reshef: adopted the name and the citation, and renamed Assumption~2
accordingly. I kept $\psi$ itself---see the paragraph after
\Cref{cor:crossing} for why: regularity is strictly weaker than a
non-increasing $r$, and it is the weaker hypothesis that
\Cref{thm_uniform_votes}(b) actually needs.}
\begin{theorem}[Uniform voters: every equilibrium, and when it is unique]\label{thm_uniform_votes}
Let $x\sim U(0,1)$ and let $k$ satisfy \Cref{ass:reg}.
\begin{enumerate}
\item[(a)] Every Nash equilibrium is symmetric, $c_1=1-c_2$, with gap
$\Delta\in(0,1)$ satisfying
\begin{equation}\label{eq:uv_foc}
K(\Delta)=(1-\Delta)\,k(\Delta),
\qquad\text{equivalently}\qquad \psi(\Delta)=1 .
\end{equation}
\item[(b)] If $K$ is regular,\rmr{the comment on log concave etc. should appear not here but after the definition of regularity}\guy{Claude: moved up to \cref{eq:psi}, where \emph{regular} is defined. Took it as a rule and checked the rest, which turned up the real problem: log-concavity was glossed at nine sites for \emph{two} different conditions the paper never separated---three for regularity and six for the stronger one you later named super-regularity. The six could not be consolidated while that condition was nameless; with your note on \Cref{cor:crossing} adopted they all now read \emph{super-regular}, and only the definition at \eqref{eq:psi} glosses log-concavity at all.} then \cref{eq:uv_foc} has a unique
root $\Delta^*$, and $\big(\tfrac{1-\Delta^*}2,\tfrac{1+\Delta^*}2\big)$ is
the unique Nash equilibrium up to relabeling. \guy{add in AMASS}
\end{enumerate}
\end{theorem}
\proofin{app:omitted-proofs}

Without regularity, \cref{eq:uv_foc} may have several roots and,
\emph{not all of them are equilibria}: \Cref{app:exist-uniq} gives
a cost density with three roots of which only the outer two are equilibria. Part (a) therefore cannot be reversed without an extra condition.
A root of \cref{eq:uv_foc} exists for every $k$ satisfying \Cref{ass:reg},
as does a solution of the conditions of \Cref{thm_uniform_costs} for every
$v$ (\Cref{app:exist-uniq}); what is open (\Cref{rem:exist}) is only whether
some solution is an equilibrium.

\begin{corollary}[The doubly-uniform benchmark]\label{cor_uu}
If $x,\kappa\sim U(0,1)$ then $\psi(t)=2t$, $\Delta^*=\tfrac12$, and
$\big(\tfrac14,\tfrac34\big)$ is the unique Nash equilibrium up to
relabeling; equilibrium turnout is $\tfrac38$. \guy{potentially use as example in AMASS}
\end{corollary}
\begin{proof}
With $K(t)=t$ the cost density $k\equiv1$ satisfies \Cref{ass:reg} and
$\psi(t)=t+K(t)/k(t)=2t$ is strictly increasing, so \Cref{thm_uniform_votes}(b)
applies; \cref{eq:uv_foc} reads $\Delta=1-\Delta$, giving $\Delta^*=\tfrac12$
and the profile $(\tfrac14,\tfrac34)$. There, by \cref{eq:shares}, candidate
$1$'s share is the area under the positive part of her benefit tent
(\Cref{fig:benefit}): the flank $[0,\tfrac14]$ votes with probability
$K(\Delta^*)=\tfrac12$, contributing $\tfrac18$, and the inner arm contributes
$\int_{1/4}^{1/2}(1-2x)\,dx=\tfrac1{16}$. Hence $u_1=u_2=\tfrac3{16}$ and
$T=\tfrac38$.
\end{proof}

\Cref{cor_uu} is the benchmark every extension will be measured against, and
it is the equilibrium depicted in \Cref{fig:example}.

\paragraph{Beyond the two benchmarks}
Sufficiency does not stop at the uniform cases. A stationary profile is an
equilibrium once each candidate's payoff is
single-peaked on her own side of the opponent, so that her stationary point is
the best position available there, and neither candidate gains by crossing to
the far side. Both hold when the two distributions are jointly not too uneven:
for a non-increasing cost density,
\[
\rho_v\,\rho_k\;\le\;4
\]
suffices, and stationarity together with the two cross-over checks then
characterizes equilibrium exactly as it does under uniform costs
(\Cref{n:prop:A4}). The two ratios multiply, so the criterion degrades
gracefully: it is $\rho_v\le4$ when costs are uniform (\Cref{prop:fourflat}),
and for a uniform electorate it asks only that $k$ be non-increasing, one of
the cases in which $K$ is regular (\Cref{thm_uniform_votes}). For a symmetric
electorate the cross-over checks come free: own-side optimality alone makes a
mirror profile an equilibrium, whatever the cost distribution
(\Cref{n:lem:reflection}). What is not covered is a cost density that rises and
falls, where the second derivative picks up terms in $k'$ of indefinite sign
and no bound in $\rho_v$ and $\rho_k$ alone can absorb them
(\Cref{m:rem:D1k}); for a uniform electorate that case is settled anyway, by
\Cref{thm_uniform_votes}.

\subsection{Polarization}\label{sec:base:polar}

Costly voting opens two routes from primitives to candidate polarization: who
the voters are, and what voting costs look like. We take them in turn.

\paragraph{Mass polarization drives elite polarization}
Under uniform costs, \Cref{cor_quarter} makes the quartile profile of any
symmetric electorate with $\rho_v\le4$ an equilibrium, so at that equilibrium
the separation \emph{equals the interquartile range} of the voter
distribution. Mass polarization therefore transmits to the candidates one for
one: a \emph{hollow} electorate, with more mass near the extremes than at the
center, has a wider interquartile range and a correspondingly wider
equilibrium, while a \emph{peaked} one contracts both (\Cref{fig:quartiles}). \guy{Which statistic ``hollowness'' refers to is fixed by
the result at hand: here the interquartile range, a global quantity; in
\Cref{sec:comp:margin}, the sign of the curvature $v''(\tfrac12)$ at the
center, a local one.} Under \Cref{claim:hd} the shape of the electorate is irrelevant and both
candidates sit at the median, so the whole dependence is an artifact of
costly participation. Dreyer and Bauer~\cite{dreyer2019does} document this
direction across eleven Western European party systems (1977--2016):
parties' positions grow more extreme as their electorates polarize, and the
effect runs through the abstention margin.

% fig:quartiles -- extracted from results_sections.tex.
% The float is self-contained: every macro it uses is defined inside it.
% Inputs from the active draft as \input{figures/fig_quartiles}.
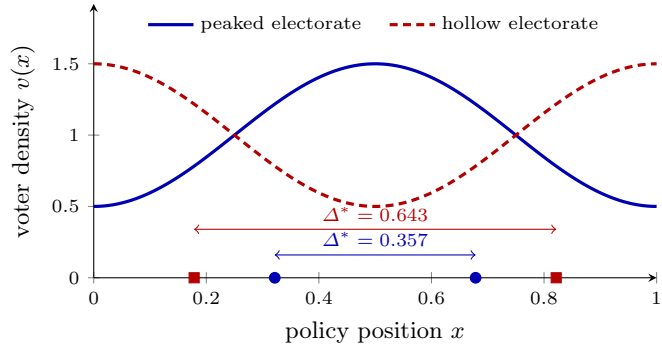
\begin{figure}[t]
\centering
%% -------------------------------------------------------------------------
%% Mass polarization drives elite polarization (uniform costs).
%% Two symmetric electorates with the same density ratio rho_v = 3:
%%   peaked  v(x) = 1 - 0.5 cos(2 pi x)   quartiles 0.3216 / 0.6784
%%   hollow  v(x) = 1 + 0.5 cos(2 pi x)   quartiles 0.1784 / 0.8216
%% Under uniform costs the equilibrium platforms sit exactly at the
%% quartiles (Cor. cor_quarter), so the hollow electorate's candidates are
%% nearly twice as far apart.  Fully parametric: edit the +-0.5 amplitude
%% and recompute the quartiles (V(x) = x +- sin(2 pi x)/(4 pi) = 1/4).
%% -------------------------------------------------------------------------
\begin{tikzpicture}
\begin{axis}[
  width=0.74\linewidth, height=5.2cm,
  xmin=0, xmax=1, ymin=0, ymax=1.92,
  xlabel={policy position $x$}, ylabel={voter density $v(x)$},
  legend style={font=\scriptsize, at={(0.5,0.995)}, anchor=north, draw=none,
    fill=none, legend columns=2, /tikz/every even column/.append style={column sep=8pt}},
  tick label style={font=\scriptsize}, label style={font=\small},
  axis lines=left, clip=false, domain=0:1, samples=120,
]
\addplot[blue!70!black, very thick] {1-0.5*cos(360*x)};
\addlegendentry{peaked electorate}
\addplot[red!70!black, very thick, densely dashed] {1+0.5*cos(360*x)};
\addlegendentry{hollow electorate}
% equilibrium platforms = quartiles
\addplot[only marks, mark=*, mark size=2pt, blue!70!black, forget plot]
  coordinates {(0.3216,0) (0.6784,0)};
\addplot[only marks, mark=square*, mark size=1.9pt, red!70!black, forget plot]
  coordinates {(0.1784,0) (0.8216,0)};
\draw[<->, blue!70!black]  (axis cs:0.3216,0.16) -- (axis cs:0.6784,0.16)
  node[midway, above, font=\scriptsize] {$\Delta^*=0.357$};
\draw[<->, red!70!black] (axis cs:0.1784,0.34) -- (axis cs:0.8216,0.34)
  node[midway, above, font=\scriptsize] {$\Delta^*=0.643$};
\end{axis}
\end{tikzpicture}
\caption{Mass polarization drives elite polarization under uniform costs
(\Cref{cor_quarter}). Two symmetric electorates with the same density ratio
$\rho_v=3$: $v(x)=1\mp\tfrac12\cos(2\pi x)$. The quartile profile (marks) is
an equilibrium (\Cref{cor_quarter}), so the hollow electorate---more mass near the extremes,
wider interquartile range---supports an equilibrium nearly twice as polarized,
even though nothing else in the environment changed.}
\label{fig:quartiles}
\end{figure}

\paragraph{Elite polarization without mass polarization}
The second channel is the one absent from the classical model altogether: fix
the electorate (uniform, so entirely unpolarized, with its median at
$\tfrac12$) and vary only the cost distribution.
Before comparing cost distributions it is worth asking which features of $K$
the model can respond to at all. The answer is sharp, and it is not the level
of turnout. Throughout write $r(t):=k(t)/K(t)$ for the \emph{reverse hazard
rate} of $K$ and $G(t):=\int_0^t K$, so that $\psi=t+1/r$ and
super-regularity says exactly that $r$ is non-increasing.\footnote{Throughout, the \emph{reverse}
hazard rate means $k/K$, not the hazard rate $k/(1-K)$ of the
monotone-hazard-rate condition familiar from mechanism design; every
monotonicity hypothesis below concerns $k/K$ (equivalently $K/k$).}

\begin{proposition}[Scale invariance]\label{prop:scale}
Fix any voter distribution $V$, and let $z>0$ be
such that $zK$ is again an admissible cost c.d.f.\ on $[0,1]$ (i.e.\
$zK(1)\le1$; the residual mass sits above $1$ and consists of voters who never
participate). Then the games with cost distributions $K$ and $zK$ have
identical best-response correspondences, hence identical sets of pure and
mixed equilibria. All vote shares---and therefore turnout---are multiplied by
$z$. \rmr{for AAMAS you can probably shorten this to a brief comment}
\end{proposition}

\noindent The argument uses neither $\alpha$ nor $\beta$, so the statement
survives both extensions unchanged; \Cref{prop:scale-gen} records it in that
generality.

\begin{proof}
By \cref{eq:shares}, $u_i(c_1,c_2)=\int_0^1 K\big(B_i(x)^+\big)v(x)\,dx$ is
linear in $K$, so replacing $K$ by $zK$ multiplies $u_1,u_2$, and hence
each candidate's payoff, by the constant $z>0$. Positive
affine transformations of payoffs leave $\arg\max$ unchanged.\Guy{Reshef: QED symbols are in. llncs does define \texttt{\char92 qed}, it just never calls it from the proof environment; three lines in the preamble hook it on, and the deferred proofs now use the same glyph instead of a different square. I did not switch to ACM: it would move theorem numbering, figure widths, the bibliography style, the front matter and the hyperref/cleveref stack, and the arXiv version is not venue-bound. If the AAMAS cut is really the target, that version needs ACM format anyway and should be its own document---tell me if that is what you meant.} \rmr{I don't mind either way}
\end{proof}

\begin{corollary}\label{cor:rhr-canonical}
The equilibrium set depends on $K$ only through its reverse hazard rate $r$.
Indeed $\log K(t)=\log K(1)-\int_t^1 r$, so two cost distributions with the
same reverse hazard rate are positive multiples of one another, and by
\Cref{prop:scale} they induce the same equilibria. By \Cref{prop:scale-gen}
the same is true at every $(\alpha,\beta)$.
\end{corollary}

\Cref{cor:rhr-canonical} explains \emph{why} the reverse-hazard-rate order
\cite{shaked2007stochastic} is the right order for the comparative static of
\Cref{prop:rh} below (we call it the rhr order): it is precisely the natural
order on the equivalence classes of cost distributions that the model can
distinguish. It is the stricter of the two orders one might reach for: if
$\widetilde k/\widetilde K\ge k/K$ pointwise and $\widetilde K(1)\le K(1)$,
which holds in particular when both are proper distributions on $[0,1]$, then
$\widetilde K\le K$ pointwise, so rhr-higher costs are stochastically higher
too. The converse fails, and \Cref{rem:fosd} is the counterexample. It also fixes the form any polarization claim here can take: a
true one must be a statement about the shape of $r$, never about the level of
costs or of turnout, since both of those can be moved without moving the
equilibrium at all.

\paragraph{What the reverse hazard rate measures}
When the benefit at stake is $t$, a fraction $K(t)$ of the electorate already
votes, and
\[
r(t)\;=\;\frac{k(t)}{K(t)}\;=\;\frac{\mathrm d}{\mathrm dt}\log K(t)
\]
is the share of those voters who are roughly indifferent: raise the stake by
$\eps$ and the number who turn out rises by a factor of about $r(t)\eps$. So
$r$ is the \emph{proportional responsiveness of turnout to the stakes}. A
large $r$ means an electorate poised at the margin; a small $r$ a committed
one, whose participation barely notices what is at stake.

This is exactly the quantity a candidate can trade on. Take uniform voters and
consider candidate~1 widening the gap by moving out. She excites her own
flank: every voter behind her faces the same benefit $\Delta$, so the mass
$\tfrac{1-\Delta}2$ behind her yields extra turnout at rate $k(\Delta)$. She
also cedes the \emph{contested middle} $[c_1,c_2]$ (the voters lying between
the candidates), losing turnout $K(\Delta)/2$. Balancing the two gives
\cref{eq:uv_foc}, which rearranges into a transparent form.

\begin{corollary}[Equilibrium as a crossing]\label{cor:crossing}
Let $x\sim U(0,1)$ and let $k$ satisfy \Cref{ass:reg}. A gap $\Delta$ solves
\cref{eq:uv_foc}, equivalently $\psi(\Delta)=1$, if and only if
\begin{equation}\label{eq:crossing}
r(\Delta)\;=\;\frac1{1-\Delta}.
\end{equation}
If $K$ is super-regular the left side falls while the right side
rises, so this is the unique equilibrium.
\end{corollary}

The two sides of \cref{eq:crossing} move in opposite directions, and that is
what makes both the crossing and the comparative static easy to see. As the gap
widens, the flank that stands to gain shrinks to mass $\tfrac{1-\Delta}2$, so
the right side rises: each further step out has to be paid for by a larger
proportional response. The left side is what the cost distribution actually
delivers at that gap. An electorate whose $r$ lies above another's meets the
rising requirement further out, so its crossing sits further out too
(\Cref{fig:rhr}), which is the next proposition.

\noindent That hypothesis is super-regularity, and it is strictly stronger
than what \Cref{thm_uniform_votes}(b) needs: regularity permits $K/k$ to fall,
provided it falls more slowly than $t$ rises---equivalently $r'<r^2$, so $r$
may rise as long as it rises slowly.\footnote{For an explicit witness, prescribe
$\psi(t)=2t+\tfrac{3}{4\pi}\sin(2\pi t)$, the doubly-uniform $\psi(t)=2t$
perturbed, and recover $K$ from $r=1/(\psi(t)-t)$. Then
$\psi'=2+\tfrac32\cos(2\pi t)\ge\tfrac12$, so $K$ is regular and satisfies
\Cref{ass:reg}, while $r$ increases wherever $\cos(2\pi t)<-\tfrac23$. Here
$\psi(\tfrac12)=1$, so the equilibrium gap is exactly $\tfrac12$.\label{fn:psi}} This is why
we keep $\psi$ alongside $r$: the picture needs a falling left-hand side, but
\Cref{thm_uniform_votes}(b) does not, and the sharp hypothesis is the one
stated in terms of $\psi$.

\begin{proposition}[More responsive participation polarizes]\label{prop:rh}Let $x\sim U(0,1)$ and let $k,\widetilde k$ both satisfy the hypotheses of
\Cref{thm_uniform_votes}(b), with equilibrium gaps $\Delta^*,\widetilde\Delta^*$.
If $\widetilde k/\widetilde K\ge k/K$ pointwise on $(0,1)$, then
$\widetilde\Delta^*\ge\Delta^*$. \guy{condense this with \ref{cor:rhr-canonical} and \ref{cor:crossing} for one key result in AMASS}
\end{proposition}

\proofin{app:omitted-proofs}

The hypothesis is far stronger than the conclusion needs. Since
$\widetilde\psi$ is strictly increasing, $\widetilde\Delta^*\ge\Delta^*$ holds
\emph{exactly} when $\widetilde r(\Delta^*)\ge r(\Delta^*)=1/(1-\Delta^*)$:
only the two electorates' responsiveness at the incumbent gap matters
(\Cref{fig:rhr}, dashed rule). It is therefore enough for the rhr order to
hold near $\Delta^*$, and the pointwise order on all of $(0,1)$ is merely the
cleanest sufficient condition---strictly stronger than the conclusion
needs.\footnote{A witness, against the uniform baseline $K(t)=t$, where
$\psi(t)=2t$ and $\Delta^*=\tfrac12$. Prescribe
$\widetilde\psi(t)=2t-3t(t-\tfrac3{10})(\tfrac7{10}-t)$ and recover
$\widetilde K$ from $\widetilde r=1/(\widetilde\psi(t)-t)$; since
$\psi=t+1/r$, the rhr order is exactly $\widetilde\psi\le\psi$, so the
cubic makes it hold on $(\tfrac3{10},\tfrac7{10})$ and fail outside.
\Cref{ass:reg} holds and $\widetilde\psi'\ge1.63$, so $\widetilde K$ is
regular; $\widetilde r(\tfrac12)=\tfrac{25}{11}>2$, and the equilibrium gap
moves out to $0.5311$. The conclusion of \Cref{prop:rh} holds although its
hypothesis fails on $60\%$ of $(0,1)$.}

% fig:rhr -- extracted from results_sections.tex.
% The float is self-contained: every macro it uses is defined inside it.
% Inputs from the active draft as \input{figures/fig_rhr}.
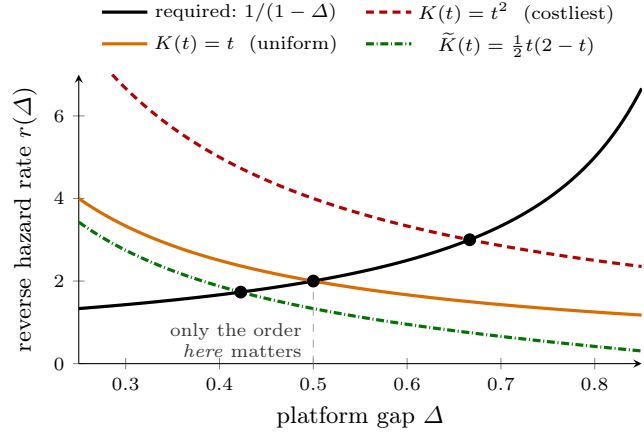
\begin{figure}[t]
\centering
%% -------------------------------------------------------------------------
%% The equilibrium as a crossing: r(D) = 1/(1-D), uniform voters, alpha=0.
%% This is exactly psi(D)=1 rewritten, since psi = t + 1/r.
%% Reverse hazard rates plotted:
%%   K = t^2            r = 2/t                    crossing at 2/3
%%   K = t   (uniform)  r = 1/t                    crossing at 1/2
%%   Ktilde = t(2-t)/2  r = 2(1-t)/(t(2-t))        crossing at 1-1/sqrt(3)
%% Edit the three \addplot expressions to swap in other cost distributions.
%% -------------------------------------------------------------------------
\begin{tikzpicture}
\begin{axis}[
  width=0.74\linewidth, height=5.4cm,
  xmin=0.25, xmax=0.85, ymin=0, ymax=7,
  xlabel={platform gap $\Delta$},
  ylabel={reverse hazard rate $r(\Delta)$},
  legend style={font=\scriptsize, at={(0.5,1.03)}, anchor=south, draw=none,
    fill=none, /tikz/every even column/.append style={column sep=8pt}},
  legend columns=2,
  tick label style={font=\scriptsize}, label style={font=\small},
  axis lines=left, clip=true, samples=200,
]
\addplot[black, very thick, domain=0.25:0.85] {1/(1-x)};
\addlegendentry{required: $1/(1-\Delta)$}
\addplot[red!70!black, densely dashed, very thick, domain=0.25:0.85] {2/x};
\addlegendentry{$K(t)=t^2$ \ (costliest)}
\addplot[orange!85!black, very thick, domain=0.25:0.85] {1/x};
\addlegendentry{$K(t)=t$ \ (uniform)}
\addplot[green!45!black, densely dashdotted, very thick, domain=0.25:0.85]
  {2*(1-x)/(x*(2-x))};
\addlegendentry{$\widetilde K(t)=\tfrac12 t(2-t)$}
% The comparison of \Cref{prop:rh} is settled at the incumbent crossing
% alone, not across the whole axis; mark that abscissa.
\addplot[gray!70, dashed, thin, forget plot]
  coordinates {(0.5,0) (0.5,2.0)};
% Label to the LEFT of the rule: to its right the K-tilde curve descends
% through the same band and the text runs straight through it. Two lines, so
% that it clears the y-axis -- the axis clips, and one line loses its first
% word. The band below y = 1.2 and left of the rule is otherwise empty.
\node[anchor=south east, align=right, font=\scriptsize, gray!45!black,
      inner sep=1.5pt, line width=0pt]
  at (axis cs:0.492,0.06) {only the order\\\emph{here} matters};
\addplot[only marks, mark=*, mark size=2.2pt, black, forget plot]
  coordinates {(0.66667,3.0) (0.5,2.0) (0.42265,1.73205)};
\end{axis}
\end{tikzpicture}
\caption{Equilibrium as a crossing (\Cref{cor:crossing}). The rising black
curve is the responsiveness a candidate \emph{needs} in order to justify a gap
$\Delta$; the falling curves are the responsiveness each cost distribution
\emph{supplies}. Equilibrium is where they meet (dots), at
$\Delta^*=\tfrac23$, $\tfrac12$ and $1-\tfrac1{\sqrt3}$. Raising the whole
supply curve---costs higher in the reverse-hazard-rate order---slides the
crossing to the right: \Cref{prop:rh}. The three curves here are ordered
along the whole axis, but that is more than the comparison needs: it is
settled at the incumbent crossing alone (dashed rule), which is why the
pointwise hypothesis of \Cref{prop:rh} can be weakened to a neighbourhood.
Note $\widetilde K$ lies \emph{below} $K(t)=t$ here even though it puts
costs stochastically higher (\Cref{rem:fosd}).}
\label{fig:rhr}
\end{figure}

\medskip
The logic behind \Cref{prop:rh} is not that costly elections have
fewer voters, but that they have \emph{more responsive} ones. When costs are
high in the reverse-hazard-rate sense, the voters who still turn out are
concentrated just above their indifference point, and a candidate who moves
outward converts many of them; the flank is worth fighting for, and she goes
to get it. When costs are low, her flank votes for her regardless, moving out
buys almost nothing, and she stays near the center to contest the middle. What
polarizes candidates is a flank that is \emph{nearly indifferent}, not a flank
that is small. \Cref{thm:rhr-alpha} carries the comparison to every
$\alpha>0$, so this is the $\alpha=0$ case of a channel that alienation does
not close.

It also settles a claim about turnout itself that would be natural to make,
and that is false.

\begin{remark}[Turnout \emph{per se} is not the driver]\label{rem:fosd}
By \Cref{prop:scale}, taking $z<1$ (which adds a mass $1-z$ of
\emph{never-voters}, whose cost exceeds the largest benefit any profile can
offer) drives turnout arbitrarily low while leaving the platforms exactly
where they were. So ``lower turnout $\Rightarrow$ more polarization'' cannot
hold in this model. Nor does the first-order stochastic order suffice, and regularity does not
rescue it: the two distributions below are both regular, with $\psi(t)=2t$
and $\widetilde\psi'\ge2$ throughout. Let $x\sim U(0,1)$ and compare
$K(t)=t$ with
\[
\widetilde K(t)=\tfrac12\,t(2-t),\qquad \widetilde k(t)=1-t
\]
(half of the electorate never votes and the rest have decreasing cost
density). Then $\widetilde K\le K$ pointwise, so costs are stochastically
higher (the mean cost rises from $\tfrac12$ to at least $\tfrac23$, the
never-voters' costs lying above every attainable benefit), and equilibrium
turnout falls from $0.375$ to $0.269$; yet the equilibrium gap \emph{shrinks},
from $\Delta^*=\tfrac12$ to $\Delta^*=1-\tfrac{1}{\sqrt3}\approx0.4226$.
Higher costs, lower
turnout, \emph{less} polarization. There is no contradiction with the comparative
static ($\widetilde K$ is \emph{lower} than $K$ in the reverse-hazard-rate
order), but the logic must be described as a statement about the shape of $K$,
not about the level of participation.
\end{remark}

The two distributions in \Cref{rem:fosd} sit on either side of $K(t)=t$ in
\Cref{fig:rhr} in the way the logic predicts: $\widetilde K$ has a
\emph{decreasing} density, so its remaining voters are the committed ones
rather than the marginal ones, its supply curve lies below, and its candidates
converge, even though fewer people vote. Turnout and polarization are both
outputs of the model; neither drives the other, and it is the shape of $K$
that drives both.

\begin{example}[Power costs: separation equals the mean cost]\label{ex:power}
Let $x\sim U(0,1)$, as throughout this subsection. For $K(t)=t^p$ with $p>0$
the reverse hazard rate $r(t)=p/t$ is increasing in
$p$, and \cref{eq:uv_foc} solves to
\[
\Delta^*=\frac{p}{1+p}=\mathbb E[\kappa]:
\]
the equilibrium separation equals the mean voting cost \emph{exactly}, for
every $p$. Raising $p$ raises $r$ pointwise, so \Cref{prop:rh} applies; on
this family its conclusion can be read straight off the formula.

Equilibrium turnout on this family is
$T=(1-\Delta^*)K(\Delta^*)+\int_0^{\Delta^*}K
 =(\Delta^*)^p\tfrac{1+2p}{(1+p)^2}$, which falls as $p$ rises:
$(\Delta^*,T)=(\tfrac13,0.513)$ at $p=\tfrac12$, $(\tfrac12,\tfrac38)$ at
$p=1$, $(\tfrac23,\tfrac{20}{81})$ at $p=2$ and $(\tfrac34,0.185)$ at $p=3$.
A single shift in the cost distribution therefore reproduces both trends of
\Cref{sec:intro} at once---rising polarization alongside falling
participation---without either one causing the other.
\end{example}

\Cref{cor_quarter} put an equilibrium at the quartiles. Within the power
family the same construction reaches \emph{any} symmetric pair of quantiles.
 \Guy{Reshef: the quantile half is the first caveat below. On existence, see \Cref{app:power-nonuniform}: no sufficiency result in the paper reaches $p\ne1$, but the profile is an equilibrium in every case I checked, and for a symmetric electorate what is left to prove is only own-side optimality.}

\begin{remark}[Prescribing the equilibrium profile]\label{rem:quantile-sweep}
Let $x\sim U(0,1)$ and fix a target profile $(q,1-q)$ with $q\in(0,\tfrac12)$.
Take the power costs
\[
K_q(t)\;=\;t^{\,p_q},
\qquad
p_q\;:=\;\frac1{2q}-1\;=\;\frac{1-2q}{2q}\;\in\;(0,\infty).
\]
Then $\psi(t)=t+K_q(t)/k_q(t)=\big(1+\tfrac1{p_q}\big)\,t$ is linear and
strictly increasing, so \Cref{thm_uniform_votes}(b) applies: the game has a
\emph{unique} equilibrium, and by \Cref{ex:power} its gap is
$\Delta^*=\tfrac{p_q}{1+p_q}=1-2q$. The candidates therefore sit at exactly
\[
c_1^*\;=\;q,\qquad c_2^*\;=\;1-q ,
\]
and the mean voting cost is $1-2q$, the separation itself.

So every interior symmetric profile is the unique equilibrium of exactly one
member of the family, and $p_q$ is strictly decreasing in $q$: the quartiles
of \Cref{cor_quarter} are $q=\tfrac14$, where $p_q=1$ and costs are uniform;
$q\uparrow\tfrac12$ sends $p_q\downarrow0$ and the platforms to the median;
$q\downarrow0$ sends $p_q\to\infty$ and the platforms to the endpoints. Those
two limits are approached but never attained, since neither $p=0$ nor
$p=\infty$ is an admissible cost distribution, which is the sense in which
costly voting breaks median convergence for \emph{every} member of the family.

Two caveats delimit the statement. The
distribution-freeness of \Cref{cor_quarter} (a quartile equilibrium for every
symmetric electorate with $\rho_v\le4$) is special to $p=1$, because uniform costs make $k$ constant and
the balance condition of \Cref{lem:balance} then compares voter mass alone;
for $p\ne1$ the marginal masses are $k$-weighted and the equilibrium quantile
depends on the shape of $v$. And the family one reaches for first, scaling
costs up and down ($K_z=zK$), cannot do this at all: rescaling leaves the
equilibrium exactly where it was (\Cref{prop:scale}), so any family that moves
the quantile must move the reverse hazard rate.
\end{remark}

\noindent Power costs against a \emph{non-uniform} electorate fall outside
every sufficiency result above, yet the profile is an equilibrium in each
case we checked, and the quantile stays close to the uniform-voter value;
\Cref{app:power-nonuniform} records what the computation shows and what is
left to prove.

\paragraph{Which way does causality run?}
The feedback is one-directional. Higher costs (in
the sense that the model can detect, \Cref{cor:rhr-canonical}) polarize the
candidates; conversely, a wider gap raises every voter's benefit from
participating, so turnout is \emph{increasing} in the imposed
separation: polarization raises turnout, never the reverse. Abstention is
therefore a cause of candidate divergence here, not a consequence of it. This
one-directionality is special to the baseline: under alienation the turnout
curve bends back (\Cref{prop:turnout}), and the reverse channel (extreme
candidates depressing participation) reappears.

\section{Alienated Voters}\label{sec:alien}

We now turn the first dial. Throughout this section $\alpha>0$ while
candidates still maximize vote share, so the convention is $\beta=0$. Alienation is a second route to divergence---not an independent one: without costly voting it would change nothing at all, since the factor $2+\alpha$ cancels out of the comparison between the candidates and the median equilibrium survives (\Cref{rem:alpha}). What it changes is \emph{who} turns out, and its
logic can be stated in one sentence: on the contested middle $[c_1,c_2]$ the two
distances sum to the constant $d_1(x)+d_2(x)\equiv\Delta$, so a penalty on the
distance to one's \emph{own} candidate bites hardest exactly at the midpoint%
---and pushes the candidates apart. The tent geometry (outer arms of slope
$\alpha$, inner arms of slope $2+\alpha$, a dead zone of abstention around
$m$) is laid out in \Cref{fig:anatomy}.

The section has three parts. \Cref{sec:alien:char} characterizes the
equilibrium at a fixed $\alpha$: the balance principle survives in a general
form (\Cref{prop:master}), and for uniform voters the equilibrium is unique
and symmetric (\Cref{thm:alien-unique}). \Cref{sec:alien:fixed} carries the
comparative statics of \Cref{sec:base:polar} to a fixed $\alpha$: the cost
channel and the transmission of mass polarization both survive
(\Cref{thm:rhr-alpha,prop:mass}). \Cref{sec:alien:polar} contains the
section's point: \emph{alienation polarizes}, with the gap strictly
increasing in $\alpha$ and tending to full separation (\Cref{thm:alien-mono}),
and it does so at a fixed level of participation (\Cref{cor:alien-turnout}),
so this is not the cost channel of \Cref{sec:base} renamed. Two observations
of \Cref{sec:base} do not survive alienation, and each is flagged where it
falls: the distribution-free quartile equilibrium (\Cref{rem:no-quantile})
and the one-directional link from polarization to turnout
(\Cref{prop:turnout}), the second of which is what lets the model speak to
the empirical literature.

\rmr{The section reads a bit like a collection of results, each of which a bit complicated.  \\
there are two main claims in this section, and all results should be organized to support them (or be removed): 1. all the observations from Section 3 still more or less hold for any $\alpha$. 2. Ceteris-Paribus, higher $\alpha$ means more polarization.   I think the first claim is mainly technical and the second one is more interesting.  Some results (also in other sections) seem to be about the interaction of parameters. E.g. that some cost distributions boost or dampen the mass-to-elite-polarization channel. This is a third-degree result that can stay in the Arxiv version (presented clearly) but not in the AAMAS version}\guy{Claude: adopted, with three claims rather than two, and every result in the section now named under one of them in the paragraph above. The reason for the third: claim~1 as you stated it is not true---two results here exist precisely to say which \Cref{sec:base} observations \emph{fail}. \Cref{prop:turnout} already said so in its own lead-in (``an artifact of $\alpha=0$''), and \Cref{rem:no-quantile} is the new answer to your uniform-costs note. Those two are also where the section earns its empirical reading, so filing them as exceptions to claim~1 would have buried the best material. Your reading of the other two was right and they are now labelled as such: \Cref{rem:deadzone} is machinery, and \Cref{prop:mass} is the mass-to-elite transmission rate, called a refinement of claim~1 in the opening and again after its statement.}
\rmr{split into subsections: characterizing eq existence for fixed $\alpha$; generalize results on cost and polarization for fixed $\alpha$; show higher $\alpha$ polarizes. Same partition for section 5.}\guy{Done for both sections. \Cref{thm:alien-mono} and \Cref{thm:beta-select} each straddled two parts and are split: \Cref{thm:alien-unique} and \Cref{prop:beta-foc} are the fixed-parameter halves. \Cref{n:prop:B3} is now stated in the body as \Cref{prop:rhr-beta}. The $\beta=1$ results are now in \Cref{sec:comp:margin}; see the reply there.}

% fig:anatomy -- extracted from results_sections.tex.
% The float is self-contained: every macro it uses is defined inside it.
% Inputs from the active draft as \input{figures/fig_anatomy}.
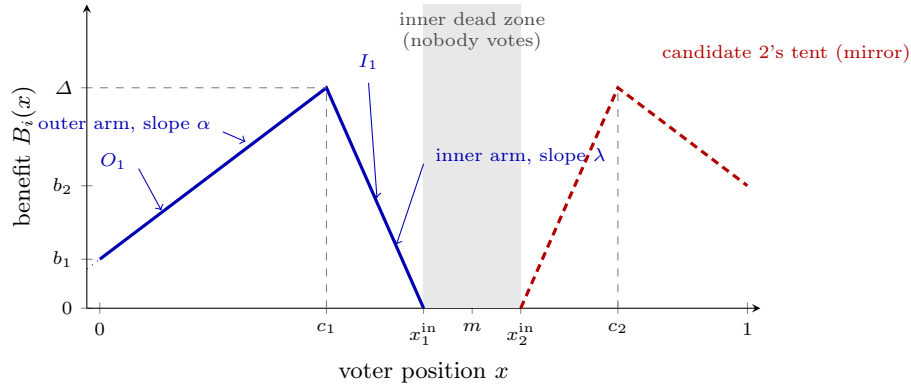
\begin{figure}[t]
\centering
%% -------------------------------------------------------------------------
%% Anatomy of the benefit tents: every geometric symbol in one picture.
%% Profile (c1,c2)=(0.35,0.80), alpha=1 (lambda=3), Delta=0.45.
%%   apex height Delta; outer slope alpha; inner slope lambda
%%   x1 = c1 + Delta/lambda = 0.50     inner zero of tent 1
%%   x0 = c1 - Delta/alpha  = -0.10    outer zero (off the space -> truncated)
%%   b1 = B1(0) = c2-(1+alpha)c1 = 0.10   boundary benefit at x=0
%%   y1 = c2 - Delta/lambda = 0.65;  b2 = B2(1) = 0.25
%%   dead zone = (x1,y1) = (0.50,0.65)
%% Fully parametric -- change c1,c2,alpha and recompute the six numbers.
%% -------------------------------------------------------------------------
\begin{tikzpicture}
\begin{axis}[
  width=0.86\linewidth, height=5.6cm,
  xmin=-0.02, xmax=1.02, ymin=0, ymax=0.62,
  xlabel={voter position $x$}, ylabel={benefit $B_i(x)$},
  xtick={0,0.35,0.50,0.575,0.65,0.80,1},
  xticklabels={$0$,$c_1$,$\zin{1}$,$m$,$\zin{2}$,$c_2$,$1$},
  ytick={0,0.10,0.25,0.45},
  yticklabels={$0$,$b_1$,$b_2$,$\Delta$},
  tick label style={font=\scriptsize}, label style={font=\small},
  axis lines=left, clip=false,
]
% dead zone shading
\addplot[draw=none, fill=gray!18] coordinates {(0.50,0) (0.65,0) (0.65,0.62) (0.50,0.62)} -- cycle;
\node[font=\scriptsize, gray!60!black, align=center] at (axis cs:0.575,0.565) {inner dead zone\\ (nobody votes)};
% tent 1: outer arm slope alpha=1 from (0,b1) to apex, inner arm slope 3 to x1
\addplot[blue!70!black, very thick] coordinates {(0,0.10) (0.35,0.45) (0.50,0)};
% tent 2: inner arm from y1 to apex, outer arm to (1,b2)
\addplot[red!70!black, very thick, densely dashed] coordinates {(0.65,0) (0.80,0.45) (1,0.25)};
% dashed continuation of tent-1 outer arm to its off-space zero x0
\addplot[blue!70!black, thin, dotted] coordinates {(-0.02,0.08) (0,0.10)};
% apex guides
\addplot[gray, thin, dashed, forget plot] coordinates {(0.35,0) (0.35,0.45)};
\addplot[gray, thin, dashed, forget plot] coordinates {(0.80,0) (0.80,0.45)};
\addplot[gray, thin, dashed, forget plot] coordinates {(0,0.45) (0.35,0.45)};
% arm labels: marginal masses and slopes, kept clear of the lines
\node[font=\scriptsize, blue!70!black, anchor=east] at (axis cs:0.055,0.30) {$O_1$};
\draw[->, blue!70!black, thin] (axis cs:0.06,0.285) -- (axis cs:0.095,0.20);
\node[font=\scriptsize, blue!70!black, anchor=east] at (axis cs:0.185,0.375) {outer arm, slope $\alpha$};
\draw[->, blue!70!black, thin] (axis cs:0.19,0.36) -- (axis cs:0.225,0.34);
\node[font=\scriptsize, blue!70!black, anchor=west] at (axis cs:0.505,0.31) {inner arm, slope $\lambda$};
\draw[->, blue!70!black, thin] (axis cs:0.50,0.295) -- (axis cs:0.458,0.13);
\node[font=\scriptsize, blue!70!black, anchor=south] at (axis cs:0.415,0.47) {$I_1$};
\draw[->, blue!70!black, thin] (axis cs:0.415,0.46) -- (axis cs:0.428,0.225);
\node[font=\scriptsize, red!70!black, anchor=west] at (axis cs:0.855,0.52) {candidate $2$'s tent (mirror)};
\end{axis}
\end{tikzpicture}
\caption{Anatomy of the benefit tents at the profile $(c_1,c_2)=(0.35,0.8)$
with $\alpha=1$. Candidate $1$'s tent (solid) peaks at height $\Delta$ over
$c_1$; its \emph{inner arm} falls at slope $\lambda=2+\alpha$ to the inner
zero $\zin{1}=c_1+\Delta/\lambda$, and its \emph{outer arm} at slope $\alpha$
toward the outer zero $\zout{1}=c_1-\Delta/\alpha$, here off the policy space, so
the arm is truncated at the boundary with \emph{boundary benefit}
$b_1=B_1(0)$. Candidate $2$'s tent (dashed) is the mirror image, with inner
zero $\zin{2}$ and boundary benefit $b_2=B_2(1)$. Between the inner zeros lies the
shaded \emph{inner dead zone} $(\zin{1},\zin{2})$, where both benefits are negative and
nobody votes; here $\zout{1}<0$, so there is no outer dead zone to draw. The
marginal masses of \Cref{prop:master} weight each point of an arm by
$k(B_1(x))\,v(x)$: $O_1$ on the outer arm, $I_1$ on the inner. The baseline of
\Cref{sec:base} is the case $\alpha=0$, in which the outer arms are flat at
height $\Delta$, the inner zeros merge at $m$, and the dead zone vanishes.}
\label{fig:anatomy}
\end{figure}

\subsection{Equilibrium at a Fixed $\alpha$}\label{sec:alien:char}

Alienation leaves the shape of the characterization intact: one first-order
condition for an arbitrary electorate (\Cref{prop:master}), and, for a uniform
one, a unique symmetric equilibrium (\Cref{thm:alien-unique}).

Here and in \Cref{sec:comp} write
\[
\lambda:=2+\alpha,
\qquad
b_1:=B_1(0)=c_2-(1+\alpha)c_1,
\qquad
b_2:=B_2(1)=(1+\alpha)c_2-c_1-\alpha
\]
for the \emph{boundary benefits}, the height of each tent's outer arm where
it meets the edge of the policy space (see \Cref{fig:anatomy}). At a symmetric profile
$c_1=\tfrac{1-\Delta}2$ they coincide, $b_1=b_2=b$, with
\begin{equation}\label{eq:b-delta}
b\;=\;\frac{\lambda\Delta-\alpha}{2},
\qquad\text{so}\qquad
\Delta-b\;=\;\frac{\alpha(1-\Delta)}{2},
\end{equation}
and $b\ge0$ exactly when $\Delta\ge\alpha/\lambda$.

Alienation opens two regions in which nobody votes, and only one of them is a
case the analysis has to track. The \emph{inner dead zone} between the tents'
inner zeros is present at every profile with $c_1<c_2$ and $\alpha>0$. An
\emph{outer dead zone} behind a candidate, where her most distant supporters
abstain, appears only when her tent's outer zero falls inside the policy space,
and it is that case distinction the first-order conditions must carry. From
here on ``dead zone'' unqualified means the outer one (\Cref{fig:anatomy}).
\Cref{lem:interior} survives alienation: every equilibrium is untied and
interior for every $\alpha\ge0$ and every $(v,k)$ (\Cref{n:lem:interior}).

\paragraph{The master condition, in general form}
One first-order condition covers alienation for every electorate and every
cost distribution, whether or not each candidate's tent reaches the edge of
the policy space.
One might hope that at equilibrium each tent reaches the boundary: for
candidate~$1$, that the outer zero $\zout{1}=c_1-\Delta/\alpha$ of her tent
(\Cref{fig:anatomy}) satisfies $\zout{1}\le0$, i.e.\ $(1+\alpha)c_1\le c_2$, so that
there is no outer dead zone,
and its mirror for candidate~$2$. For concentrated electorates this fails: the
equilibrium tents \emph{do} detach from the boundary (\Cref{rem:deadzone}).
The first-order condition valid in both regimes is the following.

\begin{proposition}[Master first-order condition]\label{prop:master}
Under alienation, every interior best response of candidate~1 satisfies \rmr{$x$ is already used for voter position. maybe change voter position to $\nu$}\guy{Claude: fixed, but by renaming the zeros rather than the voters. $\nu$ sits next to the density $v$ in every integrand in the paper---$\int_0^1p_i(\nu)v(\nu)\,d\nu$---which is a worse collision than the one it cures, and it would also give up the standard symbol for a voter's position. The zeros now carry the candidate subscript that every other object already uses ($c_i$, $d_i$, $B_i$, $I_i$, $O_i$, $u_i$), and \emph{in}/\emph{out} says which zero, so the old $x_1$---whose subscript meant \emph{inner} while the subscript on $c_1$ beside it meant \emph{candidate}---is gone. Candidate 2's zeros were $y_0,y_1$, the only per-candidate objects in the paper that changed letter instead of subscript; they gave integrals like $\beta\int_{y_1}^{y_0\wedge1}k(B_2)v\,dx$, with $y$ in the limits and $dx$ as the variable. $y$ is now free again. Four spellings of the clipped outer zero ($x_0^+$, $x_0^{\vee}$, $x_0\vee0$, and $x_{\mathrm{out}}(s)$ in the first supplement) collapsed into one, and the $z^{\vee}$ shorthand that supported them is gone. 210 occurrences across three files and two figures; the names are defined once in the preamble as \texttt{\char92 zin} and \texttt{\char92 zout}, so changing them again is a one-line edit.} \rmr{also, isn't this $O_1=I_1$?}\guy{Claude: yes, exactly---now said in the line after the proposition. The paper already knew it, but only in \Cref{lem:balance-gen}, which is where I found it and which is far too late for a reader meeting \cref{eq:master} here. Generalized while fixing: \Cref{prop:scale} and \Cref{lem:balance} each name their general form and are named back by it; \Cref{prop:master} was the only specialization naming neither. Swept the rest of the paper for the same omission afterwards and found four more unlinked pairs: \Cref{cor_uu}/\Cref{cor:alien-uu}, \Cref{cor_quarter}/\Cref{cor:quantile}, \Cref{thm:margin_median}/\Cref{claim:hd}, and \Cref{prop:rh}/\Cref{thm:rhr-alpha}; all four are now linked. The $\alpha=0$ collapse of \Cref{cor:alien-uu} was already stated in \texttt{verify/check\_alienation.py} but nowhere in the paper.}
\begin{equation}\label{eq:master}
\begin{gathered}
\int_{c_1}^{\zin{1}}\!k\big(B_1(x)\big)v(x)\,dx
\;=\;
\int_{\zout{1}\vee0}^{c_1}\!k\big(B_1(x)\big)v(x)\,dx,\\
\zin{1}=c_1+\tfrac{\Delta}{\lambda},\qquad \zout{1}=c_1-\tfrac{\Delta}{\alpha},
\end{gathered}
\end{equation}
with the mirror condition for candidate~2; here $a\vee b$ and $a\wedge b$ are the larger and the smaller of $a$ and $b$, so $\zout{1}\vee0$ is the outer zero clipped to the policy space. For uniform costs this reads
\begin{equation}\label{eq:master-unifcost}
V(\zin{1})+V(\zout{1}\vee0)=2V(c_1).
\end{equation}
When $\zout{1}\le0$ (no dead zone) $V(\zout{1}\vee0)=0$ and \cref{eq:master-unifcost}
reduces to $V(\zin{1})=2V(c_1)$; when $\zout{1}>0$ the extra term is genuinely present. \guy{add in AMASS in short}
\end{proposition}

\noindent This is \Cref{lem:balance} again. The left side of \cref{eq:master} is $I_1$ and the right side is $O_1$, with each arm cut at the tent's own zero instead of at the midpoint and the boundary; alienation moves the cuts, not the quantities being balanced. The limits collapse correctly at $\alpha=0$, where $\zin{1}=c_1+\Delta/2=m$ and $\zout{1}\to-\infty$, so $\zout{1}\vee0=0$. \Cref{lem:balance-gen} records the same identification once competitiveness is in play as well.

\begin{deferredproof}{\Cref{prop:master}}
$u_1=\int_{\zout{1}\vee0}^{\zin{1}}K(B_1)v$. The Leibniz terms at both endpoints vanish
because $B_1=0$ there and $K(0)=0$ (at a truncated lower limit $\zout{1}<0$ the
limit is the constant $0$ and contributes nothing either). On $[\zout{1}\vee0,c_1]$ we
have $\partial_{c_1}B_1=-(1+\alpha)$ and on $[c_1,\zin{1}]$,
$\partial_{c_1}B_1=+(1+\alpha)$, whence
$\partial_{c_1}u_1=(1+\alpha)\big[\int_{c_1}^{\zin{1}}k(B_1)v-\int_{\zout{1}\vee0}^{c_1}k(B_1)v\big]$.
Setting this to zero gives \cref{eq:master}.
\end{deferredproof}

\noindent Alienation also creates a case the baseline did not have---the
outer zero can fall inside the policy space---and the first claim of this
section needs that case to be harmless where the characterizations are
proved.

\begin{remark}[Dead-zone equilibria exist]\label{rem:deadzone}
Under \emph{uniform voters} dead zones never arise at a best response: if
$\zout{1}\ge0$ the tent lies entirely inside $[0,1]$ and
$u_1=\big(\tfrac1\alpha+\tfrac1\lambda\big)G(\Delta)$ (recall $G(t)=\int_0^tK$) depends on $c_1$ only
through $\Delta$, so $u_1$ is strictly decreasing in $c_1$ on the dead-zone
region $c_1\in[\tfrac{c_2}{1+\alpha},c_2)$ and the candidate strictly prefers
to move out of it. For non-uniform $v$ this argument fails and dead-zone
equilibria do occur. The supplementary analysis sharpens both halves: \Cref{n:prop:C2} shows a
stationary dead zone forces $\rho_v\ge\tfrac{2+\alpha}\alpha$ for \emph{any}
$K$, that uniform voters never have dead-zone equilibria for any $\alpha$ and
any $K$, and, under uniform costs, gives the criterion
$V\big(\tfrac{2(1+\alpha)}{\lambda^2}\big)>2V\big(\tfrac1\lambda\big)$ under
which a symmetric $v$ admits a stationary profile with a dead zone. For
$v=\mathrm{Beta}(8,8)$ and uniform costs the symmetric stationary profile is
unique, $(0.3812,0.6188)$ at $\alpha=2$ and $(0.3748,0.6252)$ at $\alpha=4$,
and is an equilibrium (mutual global best responses, verified numerically;
\Cref{n:rem:C2}); in both $\zout{1}>0$ (respectively $0.26$ and $0.31$), the condition
$V(\zin{1})=2V(c_1)$ fails by $2.3\cdot10^{-2}$ and $6.2\cdot10^{-2}$, and
\cref{eq:master-unifcost} holds to $10^{-9}$. The same happens for truncated
normals with small variance and for $\mathrm{Beta}(20,20)$.
\end{remark}

The dead-zone case being real, sufficiency also changes shape. Call the
\emph{umbrella}\rmr{you actually used this term before}\guy{Claude: it was, once, in \Cref{rem:exist}---now glossed there rather than named, so this is the first occurrence. Swept the other 134 terms the paper sets in \texttt{\char92 emph} for the same problem: the only earlier uses left are the two notation tables, which are glossaries and point forward on purpose, and ``dead zone'', which \Cref{sec:model:ext} introduces in scare quotes before this section splits it into inner and outer.} of a setting the range of density ratios $\rho_v$ within
which the first-order conditions are known to characterize equilibria. The
baseline umbrella $\rho_v\le4$ (\Cref{prop:fourflat}) deforms to
$\rho_v\le\min\big\{\tfrac{2\lambda}{1+\alpha},\,\tfrac{\lambda}{\alpha}\big\}$%
---the first branch controlling concavity, the second the dead zone---with
existence for symmetric electorates below the bound (\Cref{n:thm:A3,n:cor:A3exist}).
Both branches \emph{decrease} in $\alpha$: alienation shrinks the
umbrella rather than widening it, foreshadowing a contrast with
competitiveness (\Cref{tab:summary}).

\paragraph{Uniform voters: a one-line characterization}
For a uniform electorate the characterization collapses to a single equation
in the reverse hazard rate, \cref{eq:sym-foc} below, whose unique root is the
equilibrium gap of \Cref{thm:alien-unique}.
Under uniform voters every equilibrium is untied and interior
(\Cref{n:lem:interior}) and, by \Cref{rem:deadzone}, has no dead zone, so
specializing \cref{eq:master} to $v\equiv1$ the two first-order conditions
read $K(b_i)=\tfrac2\lambda K(\Delta)$, $i=1,2$. Hence $K(b_1)=K(b_2)$, so
$b_1=b_2$ ($K$ is strictly increasing on $[0,1]$), and since
$b_2-b_1=\alpha(c_1+c_2-1)$, for $\alpha>0$ every equilibrium is symmetric
(this is the necessity step of \Cref{n:thm:A2} at $\beta=0$). Its gap
therefore solves
\begin{equation}\label{eq:sym-foc}
\lambda\,K(b)=2\,K(\Delta),
\qquad\text{equivalently}\qquad
\int_{b}^{\Delta}r(t)\,dt\;=\;\log\frac\lambda2\;=\;\log\Big(1+\frac\alpha2\Big),
\end{equation}
with $b$ as in \cref{eq:b-delta}. The second form makes
\Cref{cor:rhr-canonical} visible: only $r$ enters. At $\alpha=0$ both sides
vanish; dividing by $\alpha$ and letting $\alpha\downarrow0$ recovers
$K(\Delta)=(1-\Delta)k(\Delta)$, the baseline condition \eqref{eq:uv_foc}, as
it must. The theorem below is therefore stated for $\alpha>0$; at $\alpha=0$
the model is the baseline, and under super-regularity its unique
equilibrium is the one of \Cref{thm_uniform_votes}(b) (\Cref{cor:crossing}).\footnote{%
A reader who simply substitutes $\alpha=0$ into \cref{eq:sym-foc} gets the
identity $2K(\Delta)=2K(\Delta)$ and no information; the content is at first
order in $\alpha$.}

\begin{theorem}[Uniform voters: the equilibrium at each $\alpha$]\label{thm:alien-unique}
Let $x\sim U(0,1)$ and let $k$ satisfy \Cref{ass:reg} with $K$ super-regular.
Then for every $\alpha>0$, \cref{eq:sym-foc} has a unique root
$\Delta^*(\alpha)\in\big(\tfrac\alpha\lambda,1\big)$, and the profile
$\big(\tfrac{1-\Delta^*}{2},\tfrac{1+\Delta^*}{2}\big)$ is the unique Nash
equilibrium up to relabeling. \guy{add in AMASS in short}
\end{theorem}

\begin{deferredproof}{\Cref{thm:alien-unique}}
Write $\Psi(\Delta):=\int_{b(\Delta)}^{\Delta}r-\log(\lambda/2)$ on
$\big[\tfrac\alpha\lambda,1\big)$, so that \cref{eq:sym-foc} reads
$\Psi(\Delta)=0$.

\emph{Necessity.} Every equilibrium is symmetric with gap a root of
\cref{eq:sym-foc}, as shown before the theorem.

\emph{Uniqueness of the root.} By \cref{eq:b-delta} the interval
$[b(\Delta),\Delta]$ has length $\tfrac\alpha2(1-\Delta)$, strictly decreasing
in $\Delta$, and its left endpoint $b(\Delta)$ is strictly increasing in
$\Delta$ (slope $\lambda/2>0$). Hence as $\Delta$ grows the interval both
shifts right and shrinks; since $r$ is non-increasing and positive,
$\int_b^\Delta r$ is strictly decreasing, so $\Psi$ is strictly decreasing. At
$\Delta=\alpha/\lambda$ we have $b=0$ and $\int_0^\Delta r=+\infty$ (as
$\log K(0^+)=-\infty$), so $\Psi>0$ there; at $\Delta\uparrow1$ we have
$b\uparrow1$, the interval degenerates, and $\Psi\to-\log(\lambda/2)<0$. A
unique root $\Delta^*$ follows.

\emph{The root is an equilibrium.} Fix the opponent at $c_2$ and let candidate
1 choose $c_1\in[0,c_2)$. In the no-dead-zone region
$u_1=A\,G(\Delta)-\tfrac1\alpha G(b_1)$ with
$A=\tfrac{2(1+\alpha)}{\alpha\lambda}$, so
\[
\frac{\partial u_1}{\partial c_1}
=\frac{1+\alpha}{\alpha}\Big[K(b_1)-\tfrac2\lambda K(\Delta)\Big],
\qquad\text{whose sign is that of}\quad
\log\tfrac\lambda2-\int_{b_1}^{\Delta}r .
\]
Moving $c_1$ toward the opponent decreases $\Delta$ at unit rate and $b_1$ at
rate $1+\alpha$, so the interval $[b_1,\Delta]$ lengthens and moves left; with
$r$ non-increasing, $\int_{b_1}^\Delta r$ is strictly increasing in $c_1$.
Hence $\partial_{c_1}u_1$ changes sign exactly once, from $+$ to $-$:
$u_1(\cdot,c_2)$ is strictly single-peaked on the own side, and by
\Cref{rem:deadzone} the dead-zone region contains no maximizer. For a
cross-over deviation note that, at a fixed distance $t$ from the
opponent, the payoff is
$A\,G(t)-\tfrac1\alpha G\big((t-\alpha \varrho)^+\big)$, where $\varrho$ is
the room behind the deviator ($\varrho=c_2-t$ on the left, $\varrho=1-c_2-t$
on the right); this is increasing in $\varrho$, and at the
symmetric profile $c_2>\tfrac12$ gives $c_2-t>1-c_2-t$. Hence every far-side
deviation is dominated by its mirror image on the own side, and the tie yields
$0$. The profile is therefore a global mutual best response.
\end{deferredproof}

\noindent Super-regularity is used only for
uniqueness and for the equilibrium property: the boundary values of $\Psi$
in the proof use no monotonicity, so \cref{eq:sym-foc} has a root for every
$K$.

\begin{corollary}\label{cor:alien-uu}
If $x,\kappa\sim U(0,1)$ then
$\Delta^*(\alpha)=\dfrac{2+\alpha}{4+\alpha}$, attained at the symmetric
profile $\big(\tfrac1{4+\alpha},\tfrac{3+\alpha}{4+\alpha}\big)$, with
equilibrium turnout $\tfrac{2(3+\alpha)}{(4+\alpha)^2}$.
\end{corollary}

\noindent At $\alpha=0$ this is \Cref{cor_uu}: the gap is $\tfrac12$, the
profile is $\big(\tfrac14,\tfrac34\big)$, and turnout is $\tfrac38$. The
benchmark is the left endpoint of a one-parameter family, and every quantity
in it moves monotonically from there.

\rmr{what about uniform costs? is there a corresponding result to the quartile equilibrium? Prop~\ref{prop:mass} says something about it but very unclear}\guy{Claude: no, and \Cref{rem:no-quantile} now says so with the reason. The short version: \Cref{cor_quarter} is distribution-free only because at $\alpha=0$ the inner zero sits at the midpoint, whose quantile a symmetric electorate fixes at $\tfrac12$ without reference to $V$. Alienation moves that zero inward to a point that depends on $c_1$, and $V$ enters the condition properly. What survives is the \emph{equilibrium}, not the fixed quantile: every profile in the table below is a mutual global best response. On \Cref{prop:mass}---that was your separate note, and I think its lead-in was the cause; see there.}

\begin{remark}[Uniform costs: no fixed-quantile analogue]\label{rem:no-quantile}\stat{numerical observation}
The quartile profile of \Cref{cor_quarter} is the same for every symmetric
electorate, and alienation removes the reason. Under uniform costs
\cref{eq:master-unifcost} reads $V(\zin{1})+V(\zout{1}\vee0)=2V(c_1)$. At
$\alpha=0$ the inner zero is the midpoint and the outer zero lies off the
policy space, so for a symmetric electorate the condition is
$\tfrac12=2V(c_1)$: the distribution enters only through $V(m)=\tfrac12$, and
every symmetric electorate lands on the quartiles. For $\alpha>0$ the inner
zero moves inward to $\zin{1}=c_1+\Delta/\lambda$, which moves with $c_1$, so
$V$ enters in its own right and the equilibrium quantile becomes a property of
the electorate. Across uniform, $\mathrm{Beta}(2,2)$, $\mathrm{Beta}(5,5)$,
arcsine and cosine-peaked electorates, $V(c_1^*)$ agrees to $10^{-12}$ at
$\alpha=0$ and then separates: it spans $[0.194,0.207]$ at $\alpha=1$ and
$[0.125,0.153]$ at $\alpha=4$, against $1/(4+\alpha)$ for uniform voters
(\Cref{cor:alien-uu}). Each of those profiles is still an equilibrium---we
verified mutual global best responses---so what alienation costs is the
distribution-freeness of the quantile, not the equilibrium.
\end{remark}
\subsection{Costs and Polarization at a Fixed $\alpha$}\label{sec:alien:fixed}

At a fixed $\alpha$ the comparative statics of \Cref{sec:base:polar} carry
over---the cost channel (\Cref{thm:rhr-alpha}) and the transmission of mass
polarization (\Cref{prop:mass})---while the one-directional link from
polarization to turnout does not (\Cref{prop:turnout}).

The cost channel of \Cref{sec:base:polar} survives alienation intact: at
every $\alpha$ the two forces push the same way.

\begin{theorem}[The cost channel survives alienation]\label{thm:rhr-alpha}
Let $x\sim U(0,1)$ and let $k,\widetilde k$ both satisfy the hypotheses of
\Cref{thm:alien-unique}. If $\widetilde r\ge r$ pointwise on $(0,1)$---i.e.\
$\widetilde K$ is higher than $K$ in the reverse-hazard-rate order---then
$\widetilde\Delta^*(\alpha)\ge\Delta^*(\alpha)$ for every $\alpha>0$, and for
$\alpha=0$ by \Cref{prop:rh}. As there, the comparison is also one of
stochastically higher costs whenever $\widetilde K(1)\le K(1)$. \guy{key result in AMASS}
\end{theorem}

\begin{proof}
$\widetilde\Psi(\Delta)=\int_{b(\Delta)}^{\Delta}\widetilde r-\log(\lambda/2)
\ge\Psi(\Delta)$ for every $\Delta$, and both are strictly decreasing with a
unique root (proof of \Cref{thm:alien-unique}). Hence
$\widetilde\Psi(\Delta^*)\ge\Psi(\Delta^*)=0=\widetilde\Psi(\widetilde\Delta^*)$
and therefore $\widetilde\Delta^*\ge\Delta^*$.
\end{proof}

At $\alpha=0$ this is exactly the baseline comparative static; the point is
that neither force switches the other off: at any fixed $\alpha$, higher
costs still polarize, and at any fixed $K$, stronger alienation still
polarizes. (The theorem orders levels, not increments: along the power family
$K(t)=t^p$ the gain in $\Delta^*$ from raising $\alpha$ is smaller at larger
$p$.)

\paragraph{Mass polarization under alienation}
\Cref{rem:no-quantile} leaves this section unable to say \emph{where} the
equilibrium sits for a general electorate. It can still say exactly how that
equilibrium \emph{moves} when the electorate itself polarizes---how much elite
polarization a given amount of mass polarization buys, and what fixes the
exchange rate. This is a question about $V$, not about $\alpha$: the answer
below holds at every $\alpha$, the baseline included.

\rmr{proposition very unclear. recall: every result either says that $\alpha$ polarizes, or shows that other parameters have same effect as before. What does this proposition show?}\guy{Claude: neither of the two, which is exactly why it would not sit still. It is about a third parameter---$\sigma$, the spread of the electorate---and it holds at every $\alpha$ including $0$. The lead-in was the real problem: it called the proposition ``an exact substitute'' for the quartile characterization, and a comparative static is not a substitute for a location. Rewritten to say what question it answers, with the answer summarized after the statement. Its own content is the mass-to-elite transmission rate: exactly proportional for power costs, and amplified or damped according to whether the elasticity of $K$ rises or falls. That is your ``third-degree'' result from the note on \Cref{sec:alien}, so it is labelled as a refinement rather than left at the same level as the headline.}
\begin{proposition}[Mass polarization scales elite polarization]\label{prop:mass}
Let $V_\sigma$ be the law of $\tfrac12+\sigma(a-\tfrac12)$ where $a\sim V_0$
and $\sigma\in(0,1]$, so that $\sigma$ measures the spread of the electorate
(for $\sigma<1$ the support of $V_\sigma$ is
$[\tfrac{1-\sigma}2,\tfrac{1+\sigma}2]$, outside which \Cref{ass:reg} fails;
the proof first shows that no best response lies there). Then
the game $(V_\sigma,K,\alpha)$ is isomorphic to $(V_0,K_\sigma,\alpha)$ under
$c=\tfrac12+\sigma(a-\tfrac12)$, where $K_\sigma(t):=K(\sigma t)$.
Consequently:
\begin{enumerate}
\item[(a)] \emph{(power costs)} for $K(t)=t^p$ (\Cref{ex:power}),
  $K_\sigma=\sigma^pK$, so by \Cref{prop:scale-gen} the equilibrium positions
  in $a$-space do not move and $\Delta^*(\sigma)=\sigma\cdot\Delta^*(1)$ for
  every $\alpha\ge0$ and every $V_0$: elite polarization is \emph{exactly
  proportional} to mass polarization (uniform costs are $p=1$);
\item[(b)] \emph{(general costs; $a\sim U(0,1)$ and $K$ super-regular,
  as in \Cref{thm:alien-unique})} the reverse hazard rate of $K_\sigma$ is
  $r_\sigma(t)=\eta(\sigma t)/t$, where $\eta(u)=u\,k(u)/K(u)$ is the
  elasticity of $K$. If $\eta$ is non-decreasing then $\sigma\mapsto r_\sigma$
  increases pointwise, so by \Cref{thm:rhr-alpha} the equilibrium gap
  $\Delta_a(\sigma)$ of $(V_0,K_\sigma,\alpha)$ increases and
  $\Delta^*(\sigma)=\sigma\Delta_a(\sigma)$ grows
  \emph{super}-proportionally: a more dispersed electorate is amplified.
  If $\eta$ is non-increasing, it is damped. Constant elasticity---the power
  family of part~(a)---is the exact knife-edge.
\end{enumerate}
\end{proposition}

\noindent So the transmission from mass to elite polarization has a rate, and
the cost distribution sets it: proportional under power costs, amplified when
the elasticity of $K$ rises and damped when it falls. It says nothing about
whether $\alpha$ polarizes, and holds equally at $\alpha=0$.

\begin{deferredproof}{\Cref{prop:mass}}
No best response lies outside $\operatorname{supp}V_\sigma$, so the strategy
space may be taken to be that interval, on which the affine map is a
bijection. Indeed, let $L:=\min\operatorname{supp}V_\sigma$, write
$B_i(x;c_1,c_2)$ for the benefit at the profile $(c_1,c_2)$, and suppose
$c_1<L$, wherever $c_2$ lies. Every voter $x$ in the support has $x>c_1$, so
$\partial B_1(x)/\partial c_1=1+\alpha>0$ and
$\partial B_2(x)/\partial c_1=-1<0$ pointwise; as $K$ is non-decreasing,
$u_1$ is non-decreasing and $u_2$ non-increasing in $c_1$ on $[c_1,L]$, whence
$U_1(c_1,c_2)\le U_1(L,c_2)$ for every $\beta\ge0$. Moreover $c_1$ is not a
best response. If $c_2\ne L$, the voter at $L$ has
$B_1(L;L,c_2)=|L-c_2|\in(0,1)$, so by continuity $B_1(\cdot;L,c_2)\in(0,1)$ on
a set of positive $V_\sigma$-measure; there $K$ is strictly increasing
(\Cref{ass:reg}) and $B_1(\cdot;c_1,c_2)=B_1(\cdot;L,c_2)-(1+\alpha)(L-c_1)$
is pointwise smaller, hence $u_1(c_1,c_2)<u_1(L,c_2)$ and
$U_1(c_1,c_2)<U_1(L,c_2)$. If $c_2=L$, then
$B_1(x;c_1,L)=(c_1-L)-\alpha(x-c_1)<0$ on the support, so $u_1(c_1,L)=0$,
while $B_2(L;c_1,L)=L-c_1>0$ gives $u_2(c_1,L)>0$: for $\beta>0$ the tied
profile $(L,L)$, at which both shares vanish, is strictly better, and for
$\beta=0$ any $c_1'$ in the interior of the support, where
$B_1(c_1';c_1',L)=c_1'-L>0$, earns a strictly positive share. The right-hand
case is the mirror image.

Under the affine map $d_i(x)=\sigma|a-a_i|$, so $B_i(x)=\sigma B_i^0(a)$ and
$u_i=\int K(\sigma B_i^0(a))\,dV_0(a)$, which is the payoff of the game
$(V_0,K_\sigma,\alpha)$. Part (a) is \Cref{prop:scale-gen} applied to
$K_\sigma=\sigma^pK$ (legitimate since all benefits lie in $[0,1]$ and
$\sigma\le1$). Part (b) computes $r_\sigma(t)=\sigma k(\sigma t)/K(\sigma t)
=\eta(\sigma t)/t$ and applies \Cref{thm:rhr-alpha} to $(V_0,K_\sigma,\alpha)$.
Since $r_\sigma(t)=\sigma r(\sigma t)$ is non-increasing whenever $r$ is, and
$K_\sigma(1)=K(\sigma)\in(0,1]$, every $K_\sigma$ satisfies the hypotheses of
that theorem. That $\Delta^*(\sigma)$ itself rises in $\sigma$ at every
$\alpha>0$, whether the transmission is amplified or damped, follows from
writing \cref{eq:sym-foc} for $(V_0,K_\sigma,\alpha)$ in $x$-space:
substituting $u=\sigma t$ and using \cref{eq:b-delta}, $\Delta^*$ solves
$\int_{\Delta^*-\frac\alpha2(\sigma-\Delta^*)}^{\Delta^*}r\,du=\log\frac\lambda2$,
whose left side is strictly increasing in $\sigma$ and, for $r$
non-increasing, strictly decreasing in $\Delta^*$.
\end{deferredproof}

In words: \emph{stretching a uniform electorate always produces more polarized
candidates}, at every level of alienation, and the cost distribution decides
whether the transmission is one-for-one (constant elasticity, e.g.\ power-law
costs), amplified or damped. An increasing
cost elasticity (participation that grows relatively more responsive as the
stakes rise) amplifies mass polarization; a decreasing one damps it; power costs
transmit it exactly. Relative to the interquartile comparison of
\Cref{sec:base:polar}, which needed uniform costs, symmetry and $\rho_v\le4$,
part (a) keeps only the power shape and holds at every $\alpha\ge0$, while
part (b) drops that too, at the price of a uniform electorate; both compare affine
stretches of a fixed electorate rather than arbitrary pairs of symmetric ones. At $\alpha=0$ the transmission has a
one-line closed form: applying \Cref{cor:crossing} to $(V_0,K_\sigma,0)$, the
equilibrium gap solves
\[
r(\Delta^*)\;=\;\frac1{\sigma-\Delta^*}\,,
\]
so widening the electorate relaxes the demand side of the crossing, and
$\Delta^*$ is strictly increasing in $\sigma$ whenever $r$ is
non-increasing. The elasticity of $K$ governs only the \emph{rate} of
transmission, never its direction. (Uniform costs: $\Delta^*=\sigma/2$
exactly.)

\paragraph{Turnout: the reverse channel reappears}
The one-directional causality of \Cref{sec:base:polar}, in which polarization
raises turnout and never the reverse, is an artifact of $\alpha=0$.

\begin{proposition}[Non-monotone turnout under alienation]\label{prop:turnout}
Let $x\sim U(0,1)$, let $\alpha>0$, and let $T(\Delta)=u_1+u_2$ be turnout at
the symmetric profile with gap $\Delta\ge\alpha/\lambda$. Then
\[
T'(\Delta)=\frac2\alpha\Big[\tfrac{2(1+\alpha)}{\lambda}K(\Delta)-\tfrac\lambda2K(b)\Big],
\quad\text{so}\quad
T'(\Delta)<0\iff\tfrac{K(b)}{K(\Delta)}>\tfrac{4(1+\alpha)}{\lambda^2}.
\]
Consequently:
\begin{enumerate}
\item[(a)] at any equilibrium, \cref{eq:sym-foc} gives
  $K(b)/K(\Delta)=2/\lambda$ and hence
  $T'(\Delta^*)=\tfrac2\lambda K(\Delta^*)>0$: locally, more polarized
  candidates still raise turnout;
\item[(b)] but for $\alpha>0$ turnout is \emph{not} globally increasing: for
  every $K$, $T'(1)=-\tfrac\alpha\lambda K(1)<0$, so the turnout curve turns
  over somewhere between the equilibrium and full separation. In the
  doubly-uniform benchmark $T$ is single-peaked with peak at
  $\bar\Delta(\alpha)=\tfrac{(2+\alpha)^2}{\alpha^2+6\alpha+4}<1$, and
  strictly decreasing on $(\bar\Delta,1]$. E.g.\ $\bar\Delta(1)=\tfrac9{11}
  \approx0.818$ while $\Delta^*(1)=0.6$.
\end{enumerate}
\end{proposition}

\begin{deferredproof}{\Cref{prop:turnout}}
Differentiate $T=2\big[A\,G(\Delta)-\tfrac1\alpha G(b)\big]$ using
$\mathrm db/\mathrm d\Delta=\lambda/2$ and $G'=K$; part~(a) substitutes
\cref{eq:sym-foc}. For (b), at $\Delta=1$ we have $b=1$ by
\cref{eq:b-delta}, so the displayed derivative gives
$T'(1)=\tfrac2\alpha K(1)\big[\tfrac{2(1+\alpha)}\lambda-\tfrac\lambda2\big]
=-\tfrac\alpha\lambda K(1)<0$, while $T'(\Delta^*)>0$ by (a); since $T'$ is
continuous, $T$ turns over in $(\Delta^*,1)$. For the benchmark put
$K=\mathrm{id}$: $T'<0$ iff
$b/\Delta>4(1+\alpha)/\lambda^2$, and substituting $b=(\lambda\Delta-\alpha)/2$
and solving gives the stated $\bar\Delta$ (the ratio $b/\Delta$ is increasing
in $\Delta$, so $T$ is single-peaked). That
$\Delta^*(\alpha)=\tfrac{2+\alpha}{4+\alpha}<\bar\Delta(\alpha)$ is consistent
with (a).
\end{deferredproof}

% fig:turnout -- extracted from results_sections.tex.
% The float is self-contained: every macro it uses is defined inside it.
% Inputs from the active draft as \input{figures/fig_turnout}.
\begin{figure}[t]
\centering
%% -------------------------------------------------------------------------
%% Turnout as a function of an IMPOSED separation, doubly uniform.
%% Closed forms (K = id, v = 1), fully editable:
%%   alpha = 0 : T = D - D^2/2
%%   alpha > 0 : T = 2(1+a)/(a*lam)*D^2 - max(0,(lam*D-a)/2)^2 / a,  lam = 2+a
%% Equilibrium  D* = (2+a)/(4+a),  T* = 2(3+a)/(4+a)^2
%% Turnout peak Dbar = (2+a)^2/(a^2+6a+4)
%% -------------------------------------------------------------------------
\begin{tikzpicture}
\begin{axis}[
  width=0.72\linewidth, height=5.4cm,
  xmin=0, xmax=1, ymin=0, ymax=0.70,
  ytick={0,0.1,0.2,0.3,0.4,0.5},
  xlabel={imposed separation $\Delta$}, ylabel={turnout $T$},
  legend style={font=\scriptsize, at={(0.5,0.995)}, anchor=north, draw=none,
    fill=none, legend columns=3, /tikz/every even column/.append style={column sep=7pt}},
  tick label style={font=\scriptsize}, label style={font=\small},
  axis lines=left,
]
\addplot[blue!70!black, very thick, domain=0:1, samples=100] {x-x^2/2};
\addlegendentry{$\alpha=0$}
\addplot[orange!85!black, densely dashed, very thick, domain=0:1, samples=200]
  {2*(1+0.5)/(0.5*2.5)*x^2 - (max(0,(2.5*x-0.5)/2))^2/0.5};
\addlegendentry{$\alpha=0.5$}
\addplot[green!45!black, densely dashdotted, very thick, domain=0:1, samples=200]
  {2*(1+1)/(1*3)*x^2 - (max(0,(3*x-1)/2))^2/1};
\addlegendentry{$\alpha=1$}
\addplot[red!70!black, dotted, very thick, domain=0:1, samples=200]
  {2*(1+2)/(2*4)*x^2 - (max(0,(4*x-2)/2))^2/2};
\addlegendentry{$\alpha=2$}
% equilibria (stars) and turnout peaks (triangles), both named in the legend
\addplot[only marks, mark=star, mark size=3.4pt, black] coordinates
  {(0.5,0.375) (0.5556,0.3457) (0.6,0.32) (0.6667,0.2778)};
\addlegendentry{equilibrium $\Delta^*(\alpha)$}
\addplot[only marks, mark=triangle*, mark size=3.4pt, black!55] coordinates
  {(0.8621,0.4138) (0.8182,0.3636) (0.8,0.3)};
\addlegendentry{turnout peak $\bar\Delta(\alpha)$}
\end{axis}
\end{tikzpicture}
\caption{Turnout against an \emph{imposed} separation in the doubly-uniform
benchmark (\Cref{prop:turnout}). Stars mark the equilibrium $\Delta^*(\alpha)$,
triangles the turnout peak $\bar\Delta(\alpha)$. The equilibrium always lies on
the rising branch, so locally polarization raises turnout; but once $\alpha>0$
the curve turns over, and candidates pushed beyond $\bar\Delta$ by forces
outside the model depress participation. At $\alpha=0$ turnout is increasing
throughout.}
\label{fig:turnout}
\end{figure}
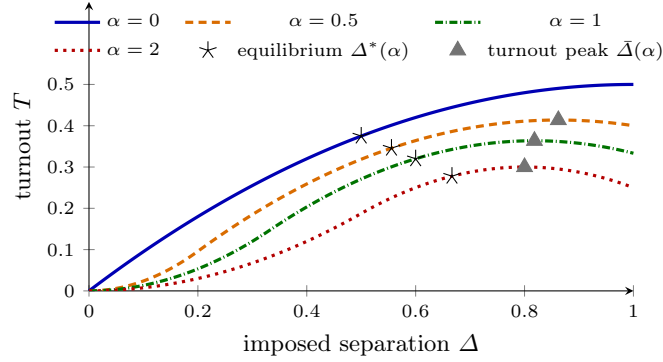

\Cref{prop:turnout} is the one result here whose implication for polarization
is indirect, so it is worth stating plainly. It does \emph{not} say that
alienation lowers equilibrium polarization---\Cref{thm:alien-mono} says the
opposite. It says that the \emph{feedback} from polarization to turnout, which
is unambiguously positive in the baseline, changes sign once the gap grows
beyond $\bar\Delta(\alpha)$ (\Cref{fig:turnout}). Equilibrium candidates never get there on their
own. But a model in which some outside force (primaries, activists, party
discipline) pushes candidates past $\bar\Delta$ predicts falling turnout, and
in that regime polarization and participation move in opposite
directions: the pattern Rogowski~\cite{rogowski2014electoral} estimates from
joint measures of voter preferences and candidate platforms in U.S. House and
Senate races, where greater ideological conflict between the candidates lowers
turnout. Alienation together with voting costs can thus explain how candidate
polarization can both increase and decrease participation. Both directions
come out of the same two parameters, so neither is an artifact of a separate
modeling choice.
\rmr{`Alienation together with voting costs can thus explain how candidate polarization can both increase and decrease participation'}\guy{Claude: adopted, nearly verbatim---it states the claim where the old sentence only gestured at a debate.}

\subsection{Higher $\alpha$ Polarizes}\label{sec:alien:polar}

The main result of this section is that higher alienation polarizes more:
\begin{theorem}[Alienation drives divergence, for every super-regular cost distribution]\label{thm:alien-mono}
In the setting of \Cref{thm:alien-unique}, the equilibrium gap
$\Delta^*(\alpha)$ is strictly increasing in $\alpha$, with
$\Delta^*(\alpha)\to1$ as $\alpha\to\infty$. \guy{key result in AMASS}
\end{theorem}

\begin{deferredproof}{\Cref{thm:alien-mono}}
Let $\Psi$ be as in the proof of \Cref{thm:alien-unique}, strictly decreasing
in $\Delta$ with root $\Delta^*$. With $\lambda=2+\alpha$ and
$\partial_\alpha b=-\tfrac{1-\Delta}2$,
$\partial_\alpha\Psi=r(b)\tfrac{1-\Delta}2-\tfrac1\lambda$. At a root,
super-regularity gives
$\log\tfrac\lambda2=\int_b^\Delta r\le r(b)(\Delta-b)=r(b)\tfrac{\alpha(1-\Delta)}2$,
so $r(b)\tfrac{1-\Delta}2\ge\tfrac1\alpha\log\tfrac\lambda2$ and
\[
\frac{\partial\Psi}{\partial\alpha}\Big|_{\Delta=\Delta^*}
\;\ge\;\frac1\alpha\log\Big(1+\frac\alpha2\Big)-\frac1{2+\alpha}
\;=\;\frac{\lambda\log(\lambda/2)-\alpha}{\alpha\lambda}\;>\;0,
\]
the last inequality being $(1+s)\log(1+s)>s$ for $s=\alpha/2>0$. Since $\Psi$
is strictly decreasing in $\Delta$, the implicit function theorem gives
$\mathrm d\Delta^*/\mathrm d\alpha>0$. Finally $\Delta^*>\alpha/\lambda\to1$.
\end{deferredproof}

\Cref{thm:alien-mono} is this section's headline, and the general form of
the doubly-uniform benchmark $\Delta^*(\alpha)=\tfrac{2+\alpha}{4+\alpha}$ of
\Cref{cor:alien-uu}: \emph{stronger alienation strictly polarizes}, from the baseline gap at $\alpha=0$ all the
way to full separation as $\alpha\to\infty$, for uniform voters and every
super-regular cost distribution (\Cref{fig:alien_delta}). Alienation simultaneously depresses turnout%
---in the doubly-uniform benchmark, from $\tfrac38$ to $0$---so the model
produces the empirically familiar pairing of rising elite polarization with
falling participation from a single voter-side parameter. For a uniform
electorate the direction is unambiguous and needs no condition on the cost
distribution beyond super-regularity: \emph{alienation
increases polarization, always, and without bound}. It is the strongest
single force in the paper. The cost channel moves $\Delta^*$ within $(0,1)$
according to the shape of $K$, whereas alienation drives it to $1$ regardless
of $K$. The uniform-electorate hypothesis is not decorative: for the
concentrated electorate of \Cref{rem:deadzone} the equilibrium gap also grows
with $\alpha$ ($\Delta^*=0.238$ at $\alpha=2$ and $0.250$ at $\alpha=4$ for
$\mathrm{Beta}(8,8)$ and uniform costs), but far more slowly, and we have no
general-electorate theorem.

% fig:alien_delta -- extracted from results_sections.tex.
% The float is self-contained: every macro it uses is defined inside it.
% Inputs from the active draft as \input{figures/fig_alien_delta}.
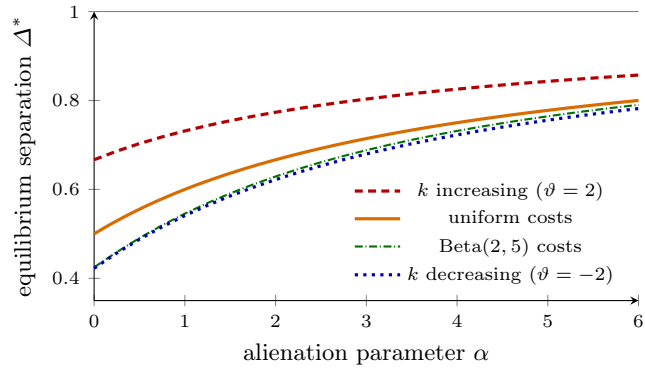
\begin{figure}[t]
\centering
%% -------------------------------------------------------------------------
%% Equilibrium separation as a function of alienation, uniform voters.
%% The four curves solve lambda*K(b) = 2*K(Delta), b=(lambda*Delta-alpha)/2,
%% for four cost distributions.  Coordinates generated by verify/figdata.py
%% (check_alienation.alien_gap); rerun it and paste new coordinates to change the figure.
%% The "uniform costs" curve is the closed form (2+a)/(4+a) and is plotted
%% as an expression, so it can be edited directly.
%% -------------------------------------------------------------------------
\begin{tikzpicture}
\begin{axis}[
  width=0.72\linewidth, height=5.4cm,
  xmin=0, xmax=6, ymin=0.35, ymax=1.0,
  xlabel={alienation parameter $\alpha$},
  ylabel={equilibrium separation $\Delta^*$},
  legend style={font=\scriptsize, at={(0.98,0.02)}, anchor=south east, draw=none, fill=none},
  tick label style={font=\scriptsize}, label style={font=\small},
  axis lines=left, clip=false,
]
\addplot[gray, thin, forget plot, domain=0:6] {1};
\addplot[red!70!black, densely dashed, very thick] coordinates {(0.000,0.6667)(0.250,0.6861)(0.500,0.7031)(0.750,0.7181)(1.000,0.7315)(1.250,0.7436)(1.500,0.7545)(1.750,0.7644)(2.000,0.7735)(2.250,0.7818)(2.500,0.7895)(2.750,0.7966)(3.000,0.8032)(3.250,0.8093)(3.500,0.8151)(3.750,0.8205)(4.000,0.8255)(4.250,0.8303)(4.500,0.8348)(4.750,0.8390)(5.000,0.8430)(5.250,0.8468)(5.500,0.8504)(5.750,0.8539)(6.000,0.8571)};
\addlegendentry{$k$ increasing ($\vartheta=2$)}
\addplot[orange!85!black, very thick, domain=0:6, samples=60] {(2+x)/(4+x)};
\addlegendentry{uniform costs}
\addplot[green!45!black, densely dashdotted, thick] coordinates {(0.000,0.4247)(0.250,0.4591)(0.500,0.4906)(0.750,0.5194)(1.000,0.5455)(1.250,0.5694)(1.500,0.5911)(1.750,0.6110)(2.000,0.6291)(2.250,0.6458)(2.500,0.6611)(2.750,0.6752)(3.000,0.6882)(3.250,0.7002)(3.500,0.7114)(3.750,0.7218)(4.000,0.7315)(4.250,0.7406)(4.500,0.7490)(4.750,0.7570)(5.000,0.7644)(5.250,0.7715)(5.500,0.7781)(5.750,0.7843)(6.000,0.7902)};
\addlegendentry{$\mathrm{Beta}(2,5)$ costs}
\addplot[blue!70!black, dotted, very thick] coordinates {(0.000,0.4226)(0.250,0.4569)(0.500,0.4879)(0.750,0.5158)(1.000,0.5412)(1.250,0.5642)(1.500,0.5852)(1.750,0.6044)(2.000,0.6220)(2.250,0.6382)(2.500,0.6532)(2.750,0.6670)(3.000,0.6797)(3.250,0.6916)(3.500,0.7027)(3.750,0.7130)(4.000,0.7226)(4.250,0.7317)(4.500,0.7402)(4.750,0.7481)(5.000,0.7557)(5.250,0.7627)(5.500,0.7694)(5.750,0.7758)(6.000,0.7818)};
\addlegendentry{$k$ decreasing ($\vartheta=-2$)}
\end{axis}
\end{tikzpicture}
\caption{Equilibrium separation under alienation for uniform voters and four
cost distributions (\Cref{thm:alien-mono}). Here
$k_\vartheta(t)=\vartheta t+\big(1-\tfrac\vartheta2\big)$, $\vartheta\in[-2,2]$,
is the linear cost family of \Cref{app:turnout}. Every curve is increasing in
$\alpha$ and tends to $1$. The three linear-family curves are ordered as
\Cref{thm:rhr-alpha} predicts, $r_\vartheta$ being pointwise increasing in
$\vartheta$; $\mathrm{Beta}(2,5)$ is not comparable with them in the
reverse-hazard-rate order (its rate crosses the uniform one), so its position
is not an instance of the theorem. The $\alpha=0$ intercepts
are $\tfrac23$, $\tfrac12$, $0.425$ and $1-\tfrac1{\sqrt3}$.}
\label{fig:alien_delta}
\end{figure}

\paragraph{Why alienation pushes candidates apart}
The logic is worth isolating, because it is different from the cost
channel. At $\alpha=0$ a candidate's flank is \emph{captive and uniform}:
every voter behind $c_1$ faces exactly the same benefit $\Delta$, whatever
their position. Alienation destroys that. For $x<c_1$,
\[
B_1(x)=\Delta-\alpha\,(c_1-x),
\]
so a candidate's own extremists---the voters furthest behind her---are the
ones with the least at stake and the first to abstain. She now has a second,
quite separate reason to move outward: not only to widen $\Delta$ for the
whole flank, but to \emph{go out and get her own flank}, whose benefit rises at
rate $\alpha$ as she approaches it. The tension in \Cref{cor:crossing} is
thereby shifted in favor of separation at every cost distribution, and
\Cref{eq:sym-foc} records the new balance: the boundary voter's turnout rate
must be held at the fixed fraction $K(b)/K(\Delta)=2/(2+\alpha)$ of the
flank's, a fraction that falls as $\alpha$ grows and can only be restored by
moving out. This is Hinich and Ordeshook's alienation channel
\cite{hinich1969abstentions} recovered inside a costly-voting model. It is a
distinct channel---one changes \emph{who} is at the margin, the other
\emph{how many}---but the two are not independent: both branches of the
sufficiency umbrella shrink as $\alpha$ grows, so alienation makes the cost
channel harder to certify, not merely different from it.
\rmr{at the least, channels interact (negatively), not orthogonal.}\guy{Claude: agreed, \emph{orthogonal} overclaimed it---and the paper contradicts it 150 lines earlier, where both umbrella branches are shown to decrease in $\alpha$. Now says distinct but interacting, with that as the evidence. I checked the other independence claims too: the one in the footnote on i.i.d.\ costs is hedged and the discussion now qualifies it at length, so it stands.}

\noindent One reading of \Cref{thm:alien-mono} has to be ruled out.
Raising $\alpha$ lowers every voter's benefit and therefore lowers turnout,
so the result might be the cost channel of \Cref{sec:base} under another
name. It is not, and \Cref{rem:fosd} already says why: by \Cref{prop:scale}
the \emph{level} of participation can be set anywhere without moving the
equilibrium at all, so ``turnout fell, hence the candidates diverged'' is not
a mechanism this model has. The constructive version is immediate.

\begin{corollary}[Polarization at unchanged turnout]\label{cor:alien-turnout}
In the setting of \Cref{thm:alien-mono}, fix $\alpha>0$ and write
$T^*(\alpha)$ for equilibrium turnout. Then there are positive multiples
$K_0$ and $K_\alpha$ of $K$, both admissible, such that the games
$(V,K_0,0)$ and $(V,K_\alpha,\alpha)$ have the \emph{same} equilibrium
turnout while their equilibrium gaps are $\Delta^*(0)$ and
$\Delta^*(\alpha)>\Delta^*(0)$. \guy{have in short in AMASS}
\end{corollary}

\begin{proof}
By \Cref{prop:scale}, replacing $K$ by $zK$ for admissible $z>0$ leaves the
equilibrium set unchanged and multiplies turnout by $z$. If
$T^*(\alpha)\le T^*(0)$, take $K_\alpha=K$ and $K_0=zK$ with
$z=T^*(\alpha)/T^*(0)$; otherwise take $K_0=K$ and $K_\alpha=zK$ with
$z=T^*(0)/T^*(\alpha)$. Either way $z\le1$, so $zK(1)\le K(1)\le1$ and $zK$
is a cost distribution, and the two games have equal equilibrium turnout by
construction. Neither scaling moves its own equilibrium, so the gaps are
still $\Delta^*(0)$ and $\Delta^*(\alpha)$, and the inequality is
\Cref{thm:alien-mono}.
\end{proof}

\noindent Scaling the \emph{larger} of the two turnouts down, rather than
the smaller up, is what keeps $z\le1$; it costs nothing and avoids assuming
that $T^*$ falls in $\alpha$, which we have not proved
(\Cref{prop:turnout} concerns $T$ as a function of $\Delta$ at fixed
$\alpha$). In the doubly-uniform benchmark the rescaling is explicit: by
\Cref{cor:alien-uu}, $T^*(\alpha)=\tfrac{2(3+\alpha)}{(4+\alpha)^2}$
falls from $\tfrac38$, so $z(\alpha)=\tfrac{16(3+\alpha)}{3(4+\alpha)^2}$,
and at $\alpha=1$ both games turn out $\tfrac8{25}$ of the electorate while
the gap is $\tfrac35$ under alienation against $\tfrac12$ without it. The
gap between the matched games can be much wider: for $K(t)=t^{1/2}$ at
$\alpha=4$ both turn out $\approx0.313$, with gaps $\approx0.692$ and
$\tfrac13$.

\paragraph{Telling the two channels apart}
Higher costs and stronger alienation both raise polarization and both depress
turnout, so the two are observationally similar in aggregate data. They
differ cross-sectionally: alienation reshapes \emph{who} abstains---%
participation falls fastest among voters far from both candidates, opening
the inner dead zone around $m$ visible in \Cref{fig:benefit-alien}, whereas a cost shift moves
participation for ideologues and centrists alike. The model thus offers a
concrete empirical signature separating ``voting became harder'' from
``voters became disaffected''. The signature is a shape: since
$p_i=K(B_i^+)$ and $B_i$ is a tent peaked at $c_i$, alienation concentrates
the active electorate into two humps centred on the candidates, falling away
on both sides and vanishing outside $[\zout{1}\vee0,\zin{1}]$. At $\alpha=0$ there is no
outer edge to fall away from: the flank is flat at the full gain $\Delta$,
so the active electorate simply tracks $v$ there. A cost shift scales that
picture without changing its shape; alienation makes it bimodal.
\rmr{so in the latter case we should see a more `bimodal' distribution of active voters?}\guy{Claude: yes, and it is worth stating---it sharpens the signature from ``who abstains'' to a shape an empiricist can look for. It follows from the tent geometry already in \Cref{fig:anatomy}: two humps at the candidates under alienation, against a flat flank at $\alpha=0$.}

\section{Competitive Candidates}\label{sec:comp}

The second extension changes the candidates rather than the voters:
$U_i=u_i-\beta u_{-i}$, spanning vote-share maximization ($\beta=0$) to pure
margin maximization ($\beta=1$). \Cref{rem:beta} showed that under full
participation $\beta$ would be irrelevant; everything in this section is
therefore attributable to abstention.

The message is that the separation between the two objectives is
\emph{de}-polarizing, that it runs against both channels of
\Cref{sec:base,sec:alien}, and that it is continuous in $\beta$ up to a
single knife-edge at $\beta=1$. Accordingly we take $\beta\in[0,1)$
throughout the first three parts, which follow those of \Cref{sec:alien}
result for result,
and treat the margin game $\beta=1$ as a special case at the end.
\Cref{sec:comp:char} characterizes the equilibrium at a fixed $\beta$: it is
unique and symmetric (\Cref{thm:beta-suff}), and under uniform costs the
mass-gap law locates it for any electorate (\Cref{prop:massgap,cor:quantile}).
\Cref{sec:comp:fixed} carries the comparative statics of
\Cref{sec:base:polar} to a fixed $\beta$: the cost channel and the
transmission of mass polarization survive (\Cref{prop:rhr-beta,prop:mass-gen}),
and the turnout curve is the same at every $\beta$ (\Cref{prop:turnout-beta}).
\Cref{sec:comp:polar} contains the section's point: the gap falls in $\beta$ at
every level of alienation (\Cref{prop:beta-mono}), and does so at a fixed level
of participation (\Cref{cor:beta-turnout}).
At $\beta=1$ (\Cref{sec:comp:margin}) this picture breaks: the equilibrium
set becomes a band rather than a point (\Cref{prop:margin_uniform}), the
$\beta\uparrow1$ limit settles on its \emph{most polarized} profile rather
than the median (\Cref{thm:beta-select}), the median fails as soon as the
electorate is hollow at the centre (\Cref{thm:quartic}), and for some
electorates no pure equilibrium survives at all (\Cref{prop:nonexistence}).
\guy{New, for you to check: \Cref{prop:mass-gen,prop:turnout-beta,cor:beta-turnout} are the $\beta$-counterparts of \Cref{prop:mass,prop:turnout,cor:alien-turnout}, which \Cref{sec:comp} lacked; all three are verified numerically. \Cref{cor:beta-uu} is the closed form split out of \Cref{prop:beta-foc} to mirror \Cref{cor:alien-uu}.}

\rmr{what is the main message of this section, that all results should support?}\guy{Claude: the section did state one, but it reached four of its eleven results---\Cref{thm:beta-suff}, \Cref{prop:massgap}, \Cref{cor:quantile}, \Cref{thm:quartic} and \Cref{prop:nonexistence} all sat outside it. Restated above as continuity plus a knife-edge, which reaches all but the two refinements, and those are named as refinements. Separately: the section was carrying the paper's conclusion. ``What drives polarization'' with \Cref{tab:drivers}, and ``The baseline and the two dials together'' with \Cref{tab:summary}, are paper-level summaries, and a reader looking for this section's message reached a table covering all five settings instead. Both have moved to the head of \Cref{sec:discussion}, which already opened by referring back to \Cref{tab:drivers}; its old first sentence was then redundant and is gone.}

With $\alpha$ and $\beta$ both in play, the scale invariance of
\Cref{sec:base} can be stated in full generality.

\begin{proposition}[Scale invariance, general form]\label{prop:scale-gen}
\Cref{prop:scale} holds verbatim for every $\alpha\ge0$ and every
$\beta\in[0,1]$.
\end{proposition}

\begin{proof}
Neither parameter enters the proof of \Cref{prop:scale}. Alienation acts only
through $B_i$ in \cref{eq:benefit-gen}, which does not involve $K$ and is
therefore unchanged; and for any $\beta$, multiplying $u_1$ and $u_2$ by the
same constant $z>0$ multiplies $U_i=u_i-\beta u_{-i}$ by $z$ as well. Positive
affine transformations of payoffs leave $\arg\max$ unchanged.
\end{proof}

\subsection{Equilibrium at a Fixed $\beta$}\label{sec:comp:char} \rmr{I suggest define $\beta \in [0,1)$ throughout the section, and separately discuss the case of $\beta=1, \alpha>0$ }\guy{Done: $\beta\in[0,1)$ in \Cref{sec:comp:char,sec:comp:fixed,sec:comp:polar}, and \Cref{sec:comp:margin} treats $\beta=1$, announced in the section opening. I included $\alpha=0$ there as well (\Cref{thm:margin_median}), so that nothing about $\beta=1$ is left in the body; and \Cref{thm:beta-select} moved with it, since it says where the $\beta<1$ equilibria end up inside the $\beta=1$ band.}

The logic is visible already in \Cref{lem:balance}. A candidate has two
ways to add to her tally: convert an abstainer on her own flank, or take a
voter away from her opponent in the contested middle. The distributive-politics
literature has long framed this choice as a targeting problem
\cite{cox1986electoral,lindbeck1987balanced,cox2010swing} and which
Glaeser, Ponzetto and Shapiro~\cite{glaeser2005strategic} make the engine of candidate divergence. The two differ in their effect on the
opponent. Converting an abstainer on one's own flank means moving \emph{out},
which widens the separation and so raises the benefit of voting---and with it
the turnout---on the opponent's flank as well: the first strategy actually
\emph{helps} the opponent. Taking a voter from the contested middle costs the
opponent a vote outright. To a candidate who weighs the opponent's votes
negatively the two are therefore not merely unequal but opposite in sign.
That is exactly the asymmetry the balance
condition prices, and it can now be stated for every $\alpha$ and $\beta$.

\begin{lemma}[Balance conditions, general form]\label{lem:balance-gen}
Let $v,k$ satisfy \Cref{ass:reg}, and let $\alpha\ge0$, $\beta\in[0,1)$. Every
interior stationary profile $(c_1,c_2)$ of the game (\Cref{def:stationary})
satisfies, for each candidate $i$,
\[
(1+\alpha)\,\big(I_i-O_i\big)\;+\;\beta\,\big(I_{-i}+O_{-i}\big)\;=\;0 .
\]
Summing and differencing the two conditions gives a \emph{ratio law} and an
\emph{asymmetry law},
\[
(1+\alpha+\beta)\,I=(1+\alpha-\beta)\,O,
\qquad
(1+\alpha-\beta)(I_1-I_2)=(1+\alpha+\beta)(O_1-O_2):
\]
the aggregate inner share of marginal voters is the universal constant
$\tfrac{1+\alpha-\beta}{2(1+\alpha)}$, for every $v$ and $k$. At
$(\alpha,\beta)=(0,0)$ this is \Cref{lem:balance}; at $\beta=0$ it is
\Cref{prop:master}, summed over the two candidates.
\end{lemma}

\noindent{\small\itshape Proof: \Cref{m:prop:ratio}.}

\medskip
The restriction to interior profiles matters only in the joint case: ties and
boundary positions are excluded at $\beta=0$ for every $\alpha$, at $\alpha=0$
for every $\beta<1$, and for uniform voters at every $(\alpha,\beta)$ with
$\beta<1$ (\Cref{n:lem:interior}); with $\alpha,\beta>0$ and a non-uniform
electorate they are not excluded in general. So the inner marginal mass carries weight $1+\alpha+\beta$ and the outer mass
weight $1+\alpha-\beta$, and the equilibrium inner share falls from
$\tfrac12$ at $\beta=0$ (\Cref{lem:balance}) toward $\tfrac{\alpha}{2(1+\alpha)}$
as $\beta\uparrow1$. As competitiveness rises the contested middle becomes worth more
relative to the flank, and the candidates move in to contest it: \emph{this is
the paper's only de-polarizing force}. It also explains why $\beta$ is inert
under full participation (\Cref{rem:beta}): with nobody abstaining there is no
flank to convert, every vote is taken from the opponent, and the distinction
the parameter prices does not exist.

What the model adds to that choice is a price: persuasion is worth
$\tfrac{1+\alpha+\beta}{1+\alpha-\beta}$ units of mobilization, for every $v$,
every $K$ and every $\beta<1$ (\Cref{lem:balance-gen}). It also corrects the
folk version of the claim: a flank is worth mobilizing because its turnout is
\emph{responsive}, not because it is low (\Cref{prop:scale-gen,rem:fosd}).

\paragraph{Both extensions, non-uniform electorates}
With both extensions on and a non-uniform electorate the characterization is
partial, and we record what is known. Under uniform costs the balance
conditions of \Cref{lem:balance-gen} still characterize equilibria, up to the
two cross-over checks, whenever $\rho_v\le\min\{f_\beta,g_\beta\}$ with the
thresholds of \Cref{tab:summary} (\Cref{m:thm:D1}); competition moves the two
branches in opposite directions, widening the concavity branch $f_\beta$ and
narrowing the dead-zone branch $g_\beta$. A symmetric stationary profile
exists for every $\beta<1$ (\Cref{m:prop:roots}), and for a symmetric
electorate strictly inside the vote-share umbrella, with a unique symmetric
stationary profile there, the symmetric equilibrium of \Cref{n:cor:A3exist}
continues into small $\beta>0$ (\Cref{m:thm:D1exist}); numerically it
persists to moderate $\beta$, inside and outside the umbrella
(\Cref{m:rem:jointnum}). Competition also makes alienation dead zones easier
to sustain: the density ratio above which a stationary profile can have one
falls from $\lambda/\alpha$ at $\beta=0$ to $g_\beta$ (\Cref{m:prop:D3}).
Beyond uniform costs the concavity argument meets an obstruction that we
isolate rather than resolve (\Cref{m:rem:D1k}).

\paragraph{Uniform voters}
For a uniform electorate the characterization again collapses to a single
equation, \cref{eq:beta-foc} below, whose unique root is the equilibrium gap
of \Cref{thm:beta-suff}.

\begin{proposition}[The stationarity condition for $\beta<1$]\label{prop:beta-foc}
Let $x\sim U(0,1)$, $\alpha>0$, $\beta\in[0,1)$, and let $K$ satisfy
\Cref{ass:reg}. Every symmetric interior stationary profile of the
$\beta$-game satisfies
\begin{equation}\label{eq:beta-foc}
(1+\alpha-\beta)\,K(b)\;=\;(1-\beta)\,\frac{2(1+\alpha)}{2+\alpha}\,K(\Delta).
\end{equation}
\end{proposition}

\begin{deferredproof}{\Cref{prop:beta-foc}}
Under uniform voters no untied symmetric stationary profile has a dead zone,
for any $\alpha>0$, any $\beta<1$ and any $K$: if $b_1=b_2=b\le0$ then
$u_1=u_2=A\,G(\Delta)$ with $A=\tfrac{2(1+\alpha)}{\alpha\lambda}$
(\Cref{rem:deadzone}), so $\partial U_1/\partial c_1=-(1-\beta)A\,K(\Delta)<0$
for $\Delta\in(0,1)$---the argument of \Cref{n:prop:C2}(b), which uses only
the first-order conditions. Hence $b_1,b_2>0$ and
$u_1=A\,G(\Delta)-\tfrac1\alpha G(b_1)$,
$u_2=A\,G(\Delta)-\tfrac1\alpha G(b_2)$. Since $\partial_{c_1}\Delta=-1$,
$\partial_{c_1}b_1=-(1+\alpha)$ and $\partial_{c_1}b_2=-1$,
\[
\frac{\partial U_1}{\partial c_1}
=-(1-\beta)A\,K(\Delta)+\frac{1+\alpha}{\alpha}K(b_1)-\frac{\beta}{\alpha}K(b_2).
\]
At a symmetric profile $b_1=b_2=b$, and setting the derivative to zero yields
\cref{eq:beta-foc}.
\end{deferredproof}

\Cref{prop:beta-foc} characterizes stationary profiles; under
super-regularity its root is in fact the game's unique equilibrium, so the characterization is
not an artifact of restricting to symmetric or stationary profiles. Define
$\psi_\beta(t):=t+\tfrac{1+\beta}{1-\beta}\,K(t)/k(t)$---the baseline $\psi$
with the term $K/k=1/r$ weighted by how much a converted opponent voter is worth
relative to a converted abstainer. The $\alpha\downarrow0$ limit of
\cref{eq:beta-foc} then reads $\psi_\beta(\Delta)=1$, which is
\cref{eq:uv_foc} with competition priced in.
(As in \cref{eq:sym-foc}, substituting $\alpha=0$ outright yields an identity;
the content is at first order in $\alpha$.)

\noindent The next theorem is the continuity half of the section's message:
below the knife-edge the equilibrium is unique, for every $\alpha\ge0$.

\begin{theorem}[Existence and uniqueness for every $\beta<1$]\label{thm:beta-suff}
Let $x\sim U(0,1)$, $\beta\in[0,1)$, and let $K$ be super-regular.
\begin{enumerate}
\item[(a)] If $\alpha=0$, the game has a unique Nash equilibrium: symmetric,
with gap the unique root of $\psi_\beta(\Delta)=1$.
\item[(b)] If $\alpha>0$, the game has a unique Nash equilibrium: symmetric,
with no dead zone, with gap the unique root of \cref{eq:beta-foc}, and
strictly more polarized than every symmetric margin-game profile:
$\Delta^*(\alpha,\beta)>\tfrac{\alpha}{2+\alpha}$ (\Cref{prop:margin_uniform}). \guy{combine with \ref{lem:balance-gen} and \ref{prop:beta-foc} in AMASS in short to present equilibrium conditions}
\end{enumerate}
\end{theorem}

\noindent{\small\itshape Proof: \Cref{n:thm:A1,n:thm:A2}.}

\begin{corollary}\label{cor:beta-uu}
If $x,\kappa\sim U(0,1)$ then
\begin{equation}\label{eq:beta-closed}
\Delta^*(\alpha,\beta)=
\frac{\alpha(2+\alpha)(1+\alpha-\beta)}
     {(1+\alpha-\beta)(2+\alpha)^2-4(1-\beta)(1+\alpha)},
\qquad
\Delta^*(0,\beta)=\frac{1-\beta}{2},
\end{equation}
the second entry being the $\alpha\downarrow0$ limit of the first, equivalently
the root of $\psi_\beta(\Delta)=1$ for uniform costs.
\end{corollary}

\begin{proof}
Put $K=\mathrm{id}$ and $b=(\lambda\Delta-\alpha)/2$ in \cref{eq:beta-foc} and
solve the resulting linear equation in $\Delta$. At $\beta=0$ this returns
$\tfrac{2+\alpha}{4+\alpha}$ (\Cref{cor:alien-uu}), and as $\beta\to1$ it tends
to $\tfrac\alpha{2+\alpha}$, the top of the band of \Cref{sec:comp:margin}.
\end{proof}

\paragraph{Uniform costs, arbitrary electorates}
The quartile structure of \Cref{cor_quarter} also generalizes, giving the
distribution-free face of the same de-polarization.

\begin{proposition}[Mass-gap law]\label{prop:massgap}
Let $\alpha=0$, $\kappa\sim U(0,1)$, $\beta\in[0,1)$, and let $v$ satisfy
\Cref{ass:reg}. Every Nash equilibrium satisfies
\[
V(c_2)-V(c_1)=\frac{1-\beta}2,
\qquad
2V(c_1)=V(m)+\beta\big(1-V(m)\big):
\]
the candidates trap voter mass exactly $\tfrac{1-\beta}2$, independently of
$v$. If moreover $\rho_v\le\tfrac4{1-\beta}$, each candidate's payoff is
concave on her side, and equilibria are exactly the solutions passing the two
cross-over checks.
\end{proposition}

\noindent{\small\itshape Proof: \Cref{n:prop:C4}.}

\begin{corollary}[Quantile equilibria]\label{cor:quantile}
In the setting of \Cref{prop:massgap} with $v$ symmetric about $\tfrac12$,
the \emph{quantile profile} $(c_1^*,c_2^*)$ with $V(c_1^*)=\tfrac{1+\beta}4$,
$V(c_2^*)=\tfrac{3-\beta}4$ is a Nash equilibrium under explicit bounds on
$v$ (\Cref{m:thm:E1});
for convex symmetric electorates the single condition
$v(c_2^*)/v(\tfrac12)\le\tfrac4{1-\beta}$ suffices. In particular, for the
arcsine electorate $v=\mathrm{Beta}(\tfrac12,\tfrac12)$,
\[
c_2^*=\sin^2\!\Big(\frac{\pi(3-\beta)}{8}\Big),
\qquad
\frac{v(c_2^*)}{v(\tfrac12)}=\frac1{\sin\big(\tfrac{\pi(3-\beta)}4\big)}
\in(1,\sqrt2\,],
\]
so a pure equilibrium exists for \emph{every} $\beta\in[0,1)$.
\end{corollary}

Note how the two conditions of \Cref{prop:massgap} interpolate the whole
paper's geometry: at $\beta=0$ they are \Cref{thm_uniform_costs}; as
$\beta\uparrow1$ the trapped mass vanishes and the quantiles collapse onto
the median of \Cref{thm:margin_median}, moving from the quartiles to the
median along an explicit path of quantile pairs.

\subsection{Costs and Polarization at a Fixed $\beta$}\label{sec:comp:fixed}

At a fixed $\beta$ the comparative statics of \Cref{sec:base:polar} carry
over as they do at a fixed $\alpha$---the cost channel (\Cref{prop:rhr-beta})
and the transmission of mass polarization (\Cref{prop:mass-gen})---and the
turnout curve of \Cref{prop:turnout} does not depend on $\beta$ at all
(\Cref{prop:turnout-beta}).

\begin{proposition}[The cost channel survives competition]\label{prop:rhr-beta}
In the setting of \Cref{thm:beta-suff}, let $\widetilde K$ be a second
super-regular cost distribution with $\widetilde r\ge r$ pointwise on
$(0,1)$. Then $\widetilde\Delta^*(\alpha,\beta)\ge\Delta^*(\alpha,\beta)$ for
every $\alpha\ge0$ and every $\beta\in[0,1)$. \guy{key result in AMASS}
\end{proposition}

\noindent{\small\itshape Proof: \Cref{n:prop:B3}.}

\medskip
At $\beta=0$ this is \Cref{thm:rhr-alpha}, and at $(\alpha,\beta)=(0,0)$ it
is \Cref{prop:rh}: more responsive electorates sustain weakly more
polarization, uniformly in competitiveness.

\paragraph{Mass polarization under competition}
The transmission of mass polarization is unchanged as well.

\begin{proposition}[Mass polarization scales elite polarization, general form]\label{prop:mass-gen}
\Cref{prop:mass} holds verbatim for every $\beta\in[0,1)$: the game
$(V_\sigma,K,\alpha,\beta)$ is isomorphic to $(V_0,K_\sigma,\alpha,\beta)$, and
in part~(b) the setting is that of \Cref{thm:beta-suff}.
\end{proposition}

\begin{proof}
Under the affine map both benefits scale, $B_i=\sigma B_i^0$, so both vote
shares are those of $(V_0,K_\sigma,\alpha)$, and hence so is
$U_i=u_i-\beta u_{-i}$; the argument that no best response lies outside
$\operatorname{supp}V_\sigma$ is already made for every $\beta\ge0$ in the
proof of \Cref{prop:mass}. Part~(a) then follows from \Cref{prop:scale-gen},
and part~(b) from \Cref{prop:rhr-beta} in place of \Cref{thm:rhr-alpha}, since
$K_\sigma$ is super-regular whenever $K$ is.
\end{proof}

\noindent So under power costs elite polarization is exactly proportional to
mass polarization at every level of competition, and otherwise the elasticity
of $K$ decides whether it is amplified or damped. At the margin endpoint mass
polarization acts differently, only through the shape of the electorate at
the centre (\Cref{thm:quartic}).

\paragraph{Turnout: competition lowers it}
Turnout depends on the platforms but not on $\beta$, which enters only the
candidates' objectives; so the turnout curve of \Cref{prop:turnout} is the
same at every level of competition, and what competition changes is where on
it the equilibrium sits.

\begin{proposition}[Turnout under competition]\label{prop:turnout-beta}
In the setting of \Cref{thm:beta-suff}, let $\alpha>0$. The formula for
$T'(\Delta)$ and part~(b) of \Cref{prop:turnout} hold unchanged for every
$\beta\in[0,1)$, and at the equilibrium
\[
T'(\Delta^*)=\frac{2(1+\alpha)(1+\beta)}{\lambda(1+\alpha-\beta)}\,K(\Delta^*)>0 .
\]
Equilibrium turnout $T\big(\Delta^*(\alpha,\beta)\big)$ is strictly decreasing
in $\beta$.
\end{proposition}

\begin{proof}
$T$ is the sum of the two vote shares at the symmetric profile, and neither
involves $\beta$, so the formula for $T'(\Delta)$ and part~(b) carry over. At
the equilibrium there is no dead zone (\Cref{thm:beta-suff}), and
\cref{eq:beta-foc} gives
$K(b)/K(\Delta^*)=\tfrac{2(1+\alpha)(1-\beta)}{\lambda(1+\alpha-\beta)}$;
substituting into $T'(\Delta)$ gives the display, whose sign is that of
$2(1+\alpha-\beta)-\lambda(1-\beta)=\alpha(1+\beta)>0$. Finally, the root
$\Delta^*(\alpha,\beta)$ of \cref{eq:beta-foc} is unique, hence continuous in
$\beta$, and strictly decreasing in $\beta$ (\Cref{prop:beta-mono}); as
$\beta$ rises it therefore sweeps an interval at every point of which
$T'>0$, and $T(\Delta^*)$ falls.
\end{proof}

\noindent At $\beta=0$ the display is part~(a) of \Cref{prop:turnout}. Locally,
more polarized candidates raise turnout at every level of competition; and
since competition narrows the gap, it lowers turnout too.

\subsection{Higher $\beta$ De-polarizes}\label{sec:comp:polar}

Under the hypotheses that delivered uniqueness, the equilibrium gap is
monotone in $\alpha$ and $\beta$ at once: competition for the margin pulls the
candidates together at every level of alienation, and alienation pushes them
apart at every level of competition.

\begin{proposition}\label{prop:beta-mono}
In the setting of \Cref{thm:beta-suff}, the equilibrium gap
$\Delta^*(\alpha,\beta)$ is strictly decreasing in $\beta$ on $[0,1)$ and
(for $\alpha>0$) strictly increasing in $\alpha$. \guy{key result in AMASS}
\end{proposition}

\noindent{\small\itshape Proof: \Cref{n:prop:B2}.}

\medskip
Competition for the margin thus pulls candidates together exactly as costly
voting and alienation push them apart; in the doubly-uniform benchmark the
closed form \eqref{eq:beta-closed} prices the trade-off between $\alpha$ and
$\beta$ exactly.

\noindent Under alienation, competition lowers turnout as well as
polarization (\Cref{prop:turnout-beta}), so the reading ruled out for
alienation must be ruled out here too: that the candidates converge because
fewer people vote. The argument of \Cref{cor:alien-turnout} applies
unchanged.

\begin{corollary}[De-polarization at unchanged turnout]\label{cor:beta-turnout}
In the setting of \Cref{thm:beta-suff}, fix $\beta\in(0,1)$. Then there are
positive multiples $K_0$ and $K_\beta$ of $K$, both admissible, such that the
games $(V,K_0,\alpha,0)$ and $(V,K_\beta,\alpha,\beta)$ have the \emph{same}
equilibrium turnout while their equilibrium gaps are $\Delta^*(\alpha,0)$ and
$\Delta^*(\alpha,\beta)<\Delta^*(\alpha,0)$.
\end{corollary}

\begin{proof}
As for \Cref{cor:alien-turnout}, with \Cref{prop:scale-gen} and
\Cref{prop:beta-mono} in place of \Cref{prop:scale} and \Cref{thm:alien-mono}.
\end{proof}

\noindent Super-regularity is a genuine
hypothesis, not a convenience:

\begin{example}[A fold: competition can polarize]\label{ex:fold}
There is an explicit cost distribution with piecewise-linear $K/k$
(\Cref{n:ex:phi}) whose vote-share game has \emph{three}
equilibria, with gaps $\tfrac13$, $\tfrac{19}{50}$ and $\tfrac{11}{20}$. Along
the middle branch $\mathrm d\Delta^*/\mathrm d\beta>0$: more competitive
candidates, \emph{more} polarization. At $\beta^*=\tfrac3{47}\approx0.064$
the middle and upper branches collide and annihilate in a fold,
and the equilibrium set jumps discontinuously from three points to one. (We
verified the middle branch is a genuine global equilibrium, to deviation gain
$2\cdot10^{-12}$.)
\end{example}

The example is not a curiosity about multiplicity; it locates precisely where
the de-polarizing logic can break. Monotonicity of $K/k$ is what guarantees
that the ``supply'' curve of \Cref{cor:crossing} meets the ``demand'' curve
once and from above. When it does not, an increase in $\beta$ can slide the
crossing the wrong way along a branch where supply cuts demand from below---%
and such a branch is fragile, which is why it is the one that dies in the
fold. Two lessons for the polarization story: comparative statics in $\beta$
are safe exactly under the same hypothesis that gives uniqueness, and a small
increase in competitiveness can produce a \emph{discontinuous} collapse in
polarization when the fold is reached: at $\beta^*=\tfrac3{47}$ the upper
branch, down from $\tfrac{11}{20}$ at $\beta=0$ to $\tfrac12$, disappears and
the equilibrium gap drops to $\tfrac{11}{36}\approx0.31$.

\noindent Where the decline ends as $\beta\uparrow1$---not at the median---is
taken up in \Cref{sec:comp:margin}.

\subsection{The Margin Game ($\beta=1$) \guy{reference in short in AMASS}}\label{sec:comp:margin}

At the endpoint $\beta=1$ the game is zero-sum, and without alienation the
classical Hotelling--Downs prediction of \Cref{claim:hd} returns in full
force---costs and all.

\begin{theorem}[Without alienation, the margin game selects the median]\label{thm:margin_median}
Let $\alpha=0$, let $v,k$ satisfy \Cref{ass:reg}, and let
$c_{\mathrm{med}}=V^{-1}(\tfrac12)$. Then
$(c_{\mathrm{med}},c_{\mathrm{med}})$ is the unique Nash equilibrium of the
margin game, and every unilateral deviation strictly lowers the deviator's
margin.
\end{theorem}
\proofin{app:omitted-proofs}

Together with \Cref{sec:base:char}, this cleanly separates the two
objectives: costly participation moves \emph{vote-share} maximizers away from
the median for every distribution satisfying \Cref{ass:reg}, yet leaves
\emph{margin}
maximizers exactly at the median, for all distributions. With full alienation,
however, the margin game degenerates:

\begin{proposition}[The margin band under alienation]\label{prop:margin_uniform}
Suppose $\beta=1$. Let $x\sim U(0,1)$ and let $K$ satisfy \Cref{ass:reg}. For every
$\alpha\ge0$, the pure Nash equilibria of the margin game are exactly
\[
E(\alpha)=\Big[\tfrac1{2+\alpha},\ \tfrac{1+\alpha}{2+\alpha}\Big]^{2},
\]
each with margin $0$: both platforms may sit anywhere in a band of width
$\Delta_0:=\tfrac{\alpha}{2+\alpha}$ around the median, independent of the
cost distribution.
\end{proposition}
\proofin{app:omitted-proofs}

The band contains a continuum of profiles, from the tied median profile to
the maximally separated $\big(\tfrac1\lambda,\tfrac{1+\alpha}\lambda\big)$%
---and even tied profiles $c_1=c_2$ in which (almost) nobody votes and both
candidates, caring only about the margin, are content. Which of these
predictions should one believe? \Cref{thm:beta-select} below answers the question,
and the answer is not the median.

Both results are knife-edges---\Cref{thm:margin_median} in the objective,
\Cref{prop:margin_uniform} in the voter density. \Cref{thm:beta-select} shows
which point of the band survives a small vote-share motive, and
\Cref{thm:quartic} what happens to the median when the electorate is not
uniform.

\noindent The decline in $\beta$ of \Cref{sec:comp:polar} ends at a
specific point of the band.

\begin{theorem}[Selection in the band]\label{thm:beta-select}
In the setting of \Cref{prop:beta-foc}, as $\beta\uparrow1$ the right-hand
side of \cref{eq:beta-foc} vanishes, forcing $K(b)\to0$ and hence $b\to0$: for every cost distribution $K$ and every $\alpha>0$ the selected
profile converges to
\[
(c_1,c_2)\longrightarrow\Big(\tfrac1{2+\alpha},\ \tfrac{1+\alpha}{2+\alpha}\Big),
\qquad\Delta\longrightarrow\frac{\alpha}{2+\alpha},
\]
which is the \emph{most polarized} profile of the band $E(\alpha)$---and also
its turnout-maximizing profile, since inside the band both tents are interior
and $T=2\big(\tfrac1\alpha+\tfrac1\lambda\big)G(\Delta)$ increases in $\Delta$.
\end{theorem}

\begin{proof}
In \cref{eq:beta-foc} the factor $1+\alpha-\beta$ tends to $\alpha>0$ while the
right-hand side tends to $0$, so $K(b)\to0$; as $K$ is strictly increasing
with $K(0)=0$, $b\to0$, and $b=(\lambda\Delta-\alpha)/2$ from \cref{eq:b-delta}
gives $\Delta\to\alpha/\lambda$. The profile is symmetric, so
$c_1=\tfrac{1-\Delta}2\to\tfrac1\lambda$, the lower edge of $E(\alpha)$ in
\Cref{prop:margin_uniform}.
\end{proof}

Thus the indeterminacy of the band is resolved in a specific direction: an
arbitrarily small vote-share motive added to a margin-maximizing candidate
selects maximal polarization, not the median (\Cref{fig:beta_select}). The
median is the \emph{least} robust point of the band. At $\alpha=0$ the band is
a single point and the limit is the median; \cref{eq:beta-closed} shows the
collapse is linear, $\Delta^*=(1-\beta)/2$.

% fig:beta_select -- extracted from results_sections.tex.
% The float is self-contained: every macro it uses is defined inside it.
% Inputs from the active draft as \input{figures/fig_beta_select}.
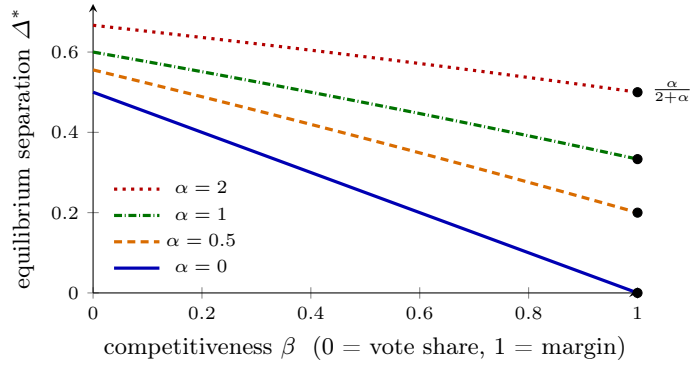
\begin{figure}[t]
\centering
%% -------------------------------------------------------------------------
%% Selection in the band: Delta*(alpha,beta), doubly uniform.  Closed form,
%% fully editable -- edit the alpha values in the four \addplot lines.
%%   alpha > 0 : a(2+a)(1+a-b) / [ (1+a-b)(2+a)^2 - 4(1-b)(1+a) ]
%%   alpha = 0 : (1-b)/2
%% Band tops at beta=1:  alpha/(2+alpha).
%% -------------------------------------------------------------------------
\begin{tikzpicture}
\begin{axis}[
  width=0.72\linewidth, height=5.4cm,
  xmin=0, xmax=1, ymin=0, ymax=0.72,
  xlabel={competitiveness $\beta$ \ (0 = vote share, 1 = margin)},
  ylabel={equilibrium separation $\Delta^*$},
  legend style={font=\scriptsize, at={(0.02,0.02)}, anchor=south west, draw=none, fill=none},
  tick label style={font=\scriptsize}, label style={font=\small},
  axis lines=left, clip=false,
]
\addplot[red!70!black, dotted, very thick, domain=0:1, samples=100]
  {2*(2+2)*(1+2-x)/((1+2-x)*(2+2)^2-4*(1-x)*(1+2))};
\addlegendentry{$\alpha=2$}
\addplot[green!45!black, densely dashdotted, very thick, domain=0:1, samples=100]
  {1*(2+1)*(1+1-x)/((1+1-x)*(2+1)^2-4*(1-x)*(1+1))};
\addlegendentry{$\alpha=1$}
\addplot[orange!85!black, densely dashed, very thick, domain=0:1, samples=100]
  {0.5*(2+0.5)*(1+0.5-x)/((1+0.5-x)*(2+0.5)^2-4*(1-x)*(1+0.5))};
\addlegendentry{$\alpha=0.5$}
\addplot[blue!70!black, very thick, domain=0:1, samples=2] {(1-x)/2};
\addlegendentry{$\alpha=0$}
\addplot[only marks, mark=*, mark size=1.7pt, black, forget plot] coordinates
  {(1,0.5) (1,0.3333) (1,0.2) (1,0)};
\node[anchor=west, font=\scriptsize] at (axis cs:1.01,0.5) {$\frac{\alpha}{2+\alpha}$};
\end{axis}
\end{tikzpicture}
\caption{Equilibrium separation as competitiveness rises, doubly uniform
(\Cref{thm:beta-select}). Increasing $\beta$ pulls the candidates together, but
for $\alpha>0$ the limit at $\beta=1$ is $\tfrac\alpha{2+\alpha}$---the
\emph{top} of the margin band $E(\alpha)$, not the median. Only at $\alpha=0$
does the limit reach the median, and there the collapse is linear.}
\label{fig:beta_select}
\end{figure}

Under super-regularity the collapse of \Cref{thm:beta-select} has an exact
rate: the approach to the band top is linear in $1-\beta$ whenever
$k(0^+)\in(0,\infty)$, and of order $(1-\beta)^{1/p}$ when $K(t)\sim Ct^p$
near $0$; at $\alpha=0$ the rate $\tfrac{1-\beta}2$ of \cref{eq:beta-closed}
is universal, $\Delta^*=\tfrac{1-\beta}2\,(1+o(1))$ for every $K$ with
$k(0^+)\in(0,\infty)$ (\Cref{n:prop:B4}).

\paragraph{Mass polarization breaks the median for margin maximizers}
With $\alpha>0$ the median solves the margin first-order condition but need not
be an equilibrium. The second-order condition at the median is degenerate
(identically zero), so the question is settled at fourth order; the answer
is a clean dichotomy in the \emph{curvature of the electorate at the center}.

\begin{theorem}[Median dichotomy]\label{thm:quartic}
Let $\kappa\sim U(0,1)$, $\alpha>0$, and let $v$ be symmetric about $\tfrac12$ and
$C^2$ near $\tfrac12$. Consider the margin obtained by a candidate deviating to
$\tfrac12-\delta$ against an opponent at the median. Then
\[
w_1(\delta)\;=\;C(\alpha)\,v''\!\big(\tfrac12\big)\,\delta^4+o(\delta^4),
\qquad
C(\alpha)=\frac{1+\alpha}{2\alpha(2+\alpha)}
+\frac16\Big(\frac1{\alpha^2}-\frac1{(2+\alpha)^2}\Big)>0 .
\]
Hence the median profile is a strict local equilibrium of the margin game if
$v''(\tfrac12)<0$ (an electorate peaked at the center), and is \emph{not} an
equilibrium if $v''(\tfrac12)>0$ (a hollow, polarized electorate). Uniform
voters, $v''\equiv0$, are exactly the knife-edge that produces the band
$E(\alpha)$.
\end{theorem}

\begin{deferredproof}{\Cref{thm:quartic}}
Write $q=\tfrac12$, $c_1=q-\delta$, $c_2=q$, so $\Delta=\delta$ and
$m=q-\delta/2$. The two benefit tents are reflections of one another about
$m$: $B_2(x)=B_1(2m-x)$. With uniform costs
$w_1=\int B_1^+(x)[v(x)-v(2m-x)]dx$, and since $v$ is symmetric about $q$,
$v(2m-x)=v(1-\delta-x)=v(x+\delta)$, so
\[
w_1=\int B_1^+(x)[v(x)-v(x+\delta)]dx
=-\delta\!\int\!B_1^+v'\;-\;\frac{\delta^2}2\!\int\!B_1^+v''\;+\;O(\delta^5).
\]
$B_1$ is a tent with apex $c_1$, height $\delta$, left half-width
$\delta/\alpha$ and right half-width $\delta/\lambda$, so
$\int B_1^+=\tfrac{\delta^2}2A$ with $A=\tfrac1\alpha+\tfrac1\lambda$, and
$\int B_1^+(x-c_1)=\tfrac{\delta^3}6\big(\tfrac1{\lambda^2}-\tfrac1{\alpha^2}\big)$.
Expanding $v'(x)=v'(c_1)+v''(c_1)(x-c_1)+O((x-c_1)^2)$ and using $v'(q)=0$,
$v'(c_1)=-\delta v''(q)+O(\delta^2)$,
\begin{align*}
-\delta\!\int\!B_1^+v'
&=\delta^4v''(q)\Big[\tfrac A2-\tfrac16\big(\tfrac1{\lambda^2}-\tfrac1{\alpha^2}\big)\Big]+O(\delta^5),\\
-\frac{\delta^2}2\!\int\!B_1^+v''
&=-\frac{\delta^4}4v''(q)\,A+O(\delta^5).
\end{align*}
Adding gives
$w_1=\delta^4v''(q)\big[\tfrac A4+\tfrac16(\tfrac1{\alpha^2}-\tfrac1{\lambda^2})\big]+O(\delta^5)$,
which is the claim since $A/4=\tfrac{1+\alpha}{2\alpha(2+\alpha)}$.
\end{deferredproof}

The formula is sharp: for $\alpha=1$, $\delta=0.02$ it predicts
$w_1=-9.244\cdot10^{-7}$ for $v=\mathrm{Beta}(2,2)$ and $+1.962\cdot10^{-7}$
for $v=\mathrm{Beta}(\tfrac12,\tfrac12)$, against simulated values
$-9.244\cdot10^{-7}$ and $+1.967\cdot10^{-7}$; for uniform voters the simulated
margin is zero to machine precision at every $\delta$ (\Cref{fig:dichotomy}).
Uniform costs enter only through the constant. For a general $K$ the same
reflection argument gives
\[
w_1(\delta)=v''\big(\tfrac12\big)\Big[\tfrac{\delta^2}2\Big(\tfrac1\alpha+\tfrac1\lambda\Big)\int_0^{\delta}\!K
+\delta\Big(\tfrac1{\alpha^2}-\tfrac1{\lambda^2}\Big)\int_0^{\delta}\!K(t)\,(\delta-t)\,dt\Big]
+o\Big(\delta^2\!\int_0^{\delta}\!K\Big),
\]
whose bracket is positive because $\lambda>\alpha$: the dichotomy in the sign
of $v''(\tfrac12)$ holds for every cost distribution satisfying
\Cref{ass:reg}, and at $K=\mathrm{id}$ the bracket is $C(\alpha)\delta^4$.
\Cref{thm:quartic} is the margin-game counterpart of the mass-polarization
result of \Cref{sec:base:polar}, and it sharpens what ``mass polarization''
has to mean. In \Cref{sec:base:polar} the relevant statistic was a global
one, the interquartile range. Here it is purely \emph{local}: only the sign of
$v''$ at the median matters, and the size of $C(\alpha)$ is irrelevant to the
conclusion. An electorate can be enormously dispersed and still hold the
candidates at the median against small deviations, provided it is peaked there; and it can be barely
dispersed at all and lose the median, provided it is locally hollow. For
margin maximizers, then, what drives elite polarization is not how spread out
the voters are but whether the center is \emph{empty}; the classical
Downsian prediction fails as soon as it is, the moment $\alpha>0$.

% fig:dichotomy -- extracted from results_sections.tex.
% The float is self-contained: every macro it uses is defined inside it.
% Inputs from the active draft as \input{figures/fig_dichotomy}.
\begin{figure}[t]
\centering
%% -------------------------------------------------------------------------
%% Median (in)stability in the margin game: deviator's margin w1(delta).
%% alpha = 1, uniform costs, opponent at the median, deviator at 1/2 - delta.
%% Coordinates from verify/figdata.py (exact incomplete-beta evaluation,
%% certificate.vote_share_exact, of
%%   w1 = int B1^+(x) [ v(x) - v(1-delta-x) ] dx ).
%% The dashed curves are the quartic law of Theorem: C(alpha)*v''(1/2)*delta^4
%% with C(1) = 0.481481, v''(1/2) = -12 for Beta(2,2) and +8/pi = 2.5465 for
%% Beta(1/2,1/2); edit those constants to change the prediction curves.
%% -------------------------------------------------------------------------
\begin{tikzpicture}
\begin{axis}[
  width=0.72\linewidth, height=5.4cm,
  xmin=0, xmax=0.2, ymin=-1.4e-3, ymax=0.45e-3,
  xlabel={deviation $\delta$ from the median},
  ylabel={deviator's margin $w_1$},
  scaled y ticks=false, yticklabel style={/pgf/number format/fixed,
    /pgf/number format/precision=4, font=\scriptsize},
  tick label style={font=\scriptsize}, label style={font=\small},
  legend style={font=\scriptsize, at={(0.02,0.03)}, anchor=south west, draw=none,
    fill=white, fill opacity=0.85, text opacity=1},
  xtick={0,0.05,0.1,0.15,0.2},
  xticklabel style={/pgf/number format/fixed},
  axis lines=left,
]
\addplot[black, thin, forget plot, domain=0:0.2] {0};
\addplot[green!45!black, densely dashdotted, very thick] coordinates {(0.0050,7.654325e-10)(0.0148,5.882924e-08)(0.0246,4.506941e-07)(0.0344,1.729550e-06)(0.0442,4.736707e-06)(0.0540,1.061677e-05)(0.0638,2.083944e-05)(0.0736,3.722846e-05)(0.0834,6.199911e-05)(0.0932,9.780615e-05)(0.1030,1.478050e-04)(0.1128,2.157297e-04)(0.1226,3.059927e-04)(0.1324,4.238143e-04)(0.1422,5.753924e-04)(0.1520,7.681283e-04)(0.1618,1.010934e-03)(0.1716,1.314664e-03)(0.1814,1.692731e-03)(0.1912,2.162040e-03)};
\addlegendentry{$\mathrm{Beta}(\frac12,\frac12)$: $v''(\frac12)>0$, median fails}
\addplot[orange!85!black, very thick] coordinates {(0.0050,0.000000e+00)(0.0148,0.000000e+00)(0.0246,0.000000e+00)(0.0344,0.000000e+00)(0.0442,0.000000e+00)(0.0540,0.000000e+00)(0.0638,0.000000e+00)(0.0736,0.000000e+00)(0.0834,0.000000e+00)(0.0932,0.000000e+00)(0.1030,0.000000e+00)(0.1128,0.000000e+00)(0.1226,0.000000e+00)(0.1324,0.000000e+00)(0.1422,0.000000e+00)(0.1520,0.000000e+00)(0.1618,0.000000e+00)(0.1716,0.000000e+00)(0.1814,0.000000e+00)(0.1912,0.000000e+00)};
\addlegendentry{uniform: $v''(\frac12)=0$, knife-edge}
\addplot[blue!70!black, very thick] coordinates {(0.0050,-3.606460e-09)(0.0148,-2.772092e-07)(0.0246,-2.115930e-06)(0.0344,-8.090858e-06)(0.0442,-2.205210e-05)(0.0540,-4.912877e-05)(0.0638,-9.572901e-05)(0.0736,-1.695400e-04)(0.0834,-2.795278e-04)(0.0932,-4.359378e-04)(0.1030,-6.502940e-04)(0.1128,-9.353997e-04)(0.1226,-1.305337e-03)(0.1324,-1.775468e-03)(0.1422,-2.362431e-03)(0.1520,-3.084148e-03)(0.1618,-3.959815e-03)(0.1716,-5.009911e-03)(0.1814,-6.256191e-03)(0.1912,-7.721690e-03)};
\addlegendentry{$\mathrm{Beta}(2,2)$: $v''(\frac12)<0$, median holds}
% quartic prediction C(1)*v''(1/2)*delta^4
\addplot[black, dashed, thin, forget plot, domain=0:0.2, samples=60] {0.481481*(-12)*x^4};
\addplot[black, dashed, thin, forget plot, domain=0:0.2, samples=60] {0.481481*(2.5465)*x^4};
\end{axis}
\end{tikzpicture}
\caption{The quartic law of \Cref{thm:quartic} at $\alpha=1$ with uniform
costs: a candidate deviating to $\tfrac12-\delta$ against an opponent at the
median gains iff the electorate is hollow at the center. Thin dashed lines are
the prediction $C(\alpha)v''(\tfrac12)\delta^4$; they are indistinguishable
from the computed curves for small $\delta$. For uniform voters the margin is
zero to machine precision at every $\delta$---the knife-edge that manufactures
the band $E(\alpha)$.}
\label{fig:dichotomy}
\end{figure}
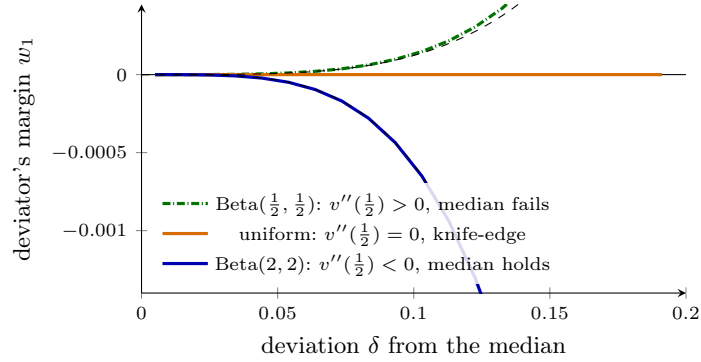

\paragraph{Existence can fail}
Where alienation meets a hollow electorate, the margin game can have no pure
equilibrium at all, and a certificate rules one out for the arcsine electorate
at $\alpha=1$ and for every $\beta\ge0.9168$.

\begin{proposition}[Non-existence]\label{prop:nonexistence}\stat{computer-assisted}
Let $\alpha=1$ and $\kappa\sim U(0,1)$. For $v=\mathrm{Beta}(\tfrac12,\tfrac12)$ and
for $v=\mathrm{Beta}(0.8,0.8)$ the margin game has no pure Nash equilibrium
(\Cref{n:thm:C1}); for the former, the same holds for every $\beta\ge0.9168$
(\Cref{m:thm:E2}).
\end{proposition}

\noindent Each statement is a finite computation at a fixed $(\alpha,v)$. We
expect, but do not prove, that non-existence persists for every sufficiently
hollow $v$ and every $\alpha$ above a $v$-dependent threshold.

\begin{proof}[Sketch; the full certificate is \Cref{n:thm:C1,m:thm:E2}]
The margin game is symmetric and zero-sum, so a pure equilibrium exists if and
only if some position has security value zero,
$\max_{c\in[0,1]}\min_{c'\in[0,1]}w(c,c')=0$, writing $w(c,c')$ for the margin
of a candidate at $c$ against one at $c'$; since $w(c,c)=0$ and $w$ is
antisymmetric, the security value is always $\le0$. For $\alpha=1$ and uniform
costs it equals $-8.2244\cdot10^{-3}$ (at $c=0.829$) for the arcsine density
$v=\mathrm{Beta}(\tfrac12,\tfrac12)$ and $-3.0238\cdot10^{-3}$ for
$\mathrm{Beta}(0.8,0.8)$, so \emph{every} position can be strictly beaten
(\Cref{fig:security}); for $\mathrm{Beta}(2,2)$ the grid security value
attains $0$ at the median and only there, as \Cref{thm:quartic} leads one to
expect. A Lipschitz bound on $w$ turns the finite grid computation into
an exact statement about the continuum, giving a computer-assisted proof. A
mixed equilibrium always exists by Glicksberg's fixed-point theorem, with value $0$.
\end{proof}

The certificate extends off the boundary $\beta=1$: no pure equilibrium exists
for any $\beta\ge0.9168$ (\Cref{m:thm:E2}), while a pure equilibrium is still
found numerically at $\beta=0.90$, placing the true threshold in
$(0.90,0.9168]$ with the numerics placing that threshold $\beta_c$ between $0.908$ and
$0.909$. The
contrast with $\alpha=0$ is stark: there, under uniform costs, the very
arcsine electorate that admits no pure equilibrium at $\alpha=1$ keeps one for
\emph{every} $\beta<1$, along with every other convex symmetric electorate
meeting the density bound of \Cref{cor:quantile}. Non-existence is therefore a phenomenon of
alienation and competitiveness \emph{together}, not of competitiveness alone.

% fig:security -- extracted from results_sections.tex.
% The float is self-contained: every macro it uses is defined inside it.
% Inputs from the active draft as \input{figures/fig_security}.
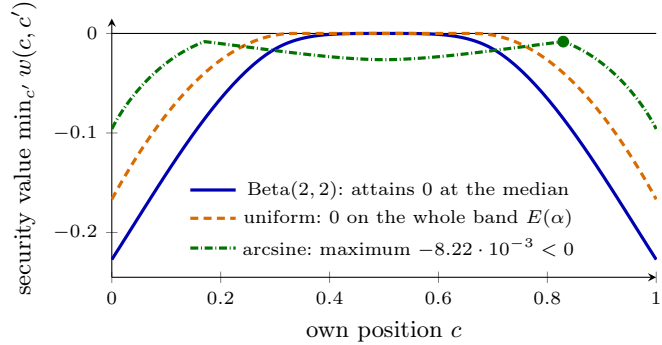
\begin{figure}[t]
\centering
%% -------------------------------------------------------------------------
%% Zero-sum security value of the margin game, alpha = 1, uniform costs.
%% sec(c) = min_{c'} w(c,c').  A pure equilibrium exists iff max_c sec(c)=0.
%% Coordinates regenerated with the batch-1/2 engine (exact incomplete-beta
%% quadrature); a naive grid rule is unusable here -- see Remark on the
%% -0.387 artifact.  Regenerate with verify/figdata.py (certificate.vote_share_exact).
%% -------------------------------------------------------------------------
\begin{tikzpicture}
\begin{axis}[
  width=0.72\linewidth, height=5.0cm,
  xmin=0, xmax=1, ymin=-0.245, ymax=0.015,
  xlabel={own position $c$},
  ylabel={security value $\min_{c'}w(c,c')$},
  legend style={font=\scriptsize, at={(0.5,0.03)}, anchor=south, draw=none,
    fill=white, fill opacity=0.85, text opacity=1},
  tick label style={font=\scriptsize}, label style={font=\small},
  axis lines=left,
]
\addplot[black, thin, forget plot, domain=0:1] {0};
\addplot[blue!70!black, very thick] coordinates {(0.0000,-0.227208)(0.0042,-0.223637)(0.0083,-0.220054)(0.0125,-0.216461)(0.0167,-0.212859)(0.0208,-0.209249)(0.0250,-0.205633)(0.0292,-0.202012)(0.0333,-0.198388)(0.0375,-0.194761)(0.0417,-0.191135)(0.0458,-0.187509)(0.0500,-0.183885)(0.0542,-0.180265)(0.0583,-0.176650)(0.0625,-0.173041)(0.0667,-0.169439)(0.0708,-0.165847)(0.0750,-0.162264)(0.0792,-0.158693)(0.0833,-0.155135)(0.0875,-0.151590)(0.0917,-0.148061)(0.0958,-0.144548)(0.1000,-0.141053)(0.1042,-0.137577)(0.1083,-0.134121)(0.1125,-0.130686)(0.1167,-0.127273)(0.1208,-0.123885)(0.1250,-0.120521)(0.1292,-0.117182)(0.1333,-0.113871)(0.1375,-0.110588)(0.1417,-0.107335)(0.1458,-0.104111)(0.1500,-0.100919)(0.1542,-0.097759)(0.1583,-0.094633)(0.1625,-0.091542)(0.1667,-0.088486)(0.1708,-0.085467)(0.1750,-0.082485)(0.1792,-0.079542)(0.1833,-0.076639)(0.1875,-0.073776)(0.1917,-0.070955)(0.1958,-0.068176)(0.2000,-0.065441)(0.2042,-0.062749)(0.2083,-0.060103)(0.2125,-0.057504)(0.2167,-0.054951)(0.2208,-0.052446)(0.2250,-0.049990)(0.2292,-0.047584)(0.2333,-0.045228)(0.2375,-0.042925)(0.2417,-0.040675)(0.2458,-0.038480)(0.2500,-0.036340)(0.2542,-0.034258)(0.2583,-0.032235)(0.2625,-0.030273)(0.2667,-0.028372)(0.2708,-0.026534)(0.2750,-0.024761)(0.2792,-0.023056)(0.2833,-0.021418)(0.2875,-0.019852)(0.2917,-0.018358)(0.2958,-0.016940)(0.3000,-0.015600)(0.3042,-0.014340)(0.3083,-0.013158)(0.3125,-0.012050)(0.3167,-0.011015)(0.3208,-0.010047)(0.3250,-0.009144)(0.3292,-0.008304)(0.3333,-0.007523)(0.3375,-0.006798)(0.3417,-0.006128)(0.3458,-0.005508)(0.3500,-0.004936)(0.3542,-0.004410)(0.3583,-0.003927)(0.3625,-0.003485)(0.3667,-0.003081)(0.3708,-0.002714)(0.3750,-0.002380)(0.3792,-0.002078)(0.3833,-0.001806)(0.3875,-0.001562)(0.3917,-0.001343)(0.3958,-0.001148)(0.4000,-0.000975)(0.4042,-0.000822)(0.4083,-0.000688)(0.4125,-0.000571)(0.4167,-0.000470)(0.4208,-0.000383)(0.4250,-0.000308)(0.4292,-0.000245)(0.4333,-0.000193)(0.4375,-0.000149)(0.4417,-0.000113)(0.4458,-0.000084)(0.4500,-0.000061)(0.4542,-0.000043)(0.4583,-0.000029)(0.4625,-0.000019)(0.4667,-0.000012)(0.4708,-0.000007)(0.4750,-0.000004)(0.4792,-0.000002)(0.4833,-0.000001)(0.4875,-0.000000)(0.4917,-0.000000)(0.4958,-0.000000)(0.5000,0.000000)(0.5042,-0.000000)(0.5083,-0.000000)(0.5125,-0.000000)(0.5167,-0.000001)(0.5208,-0.000002)(0.5250,-0.000004)(0.5292,-0.000007)(0.5333,-0.000012)(0.5375,-0.000019)(0.5417,-0.000029)(0.5458,-0.000043)(0.5500,-0.000061)(0.5542,-0.000084)(0.5583,-0.000113)(0.5625,-0.000149)(0.5667,-0.000193)(0.5708,-0.000245)(0.5750,-0.000308)(0.5792,-0.000383)(0.5833,-0.000470)(0.5875,-0.000571)(0.5917,-0.000688)(0.5958,-0.000822)(0.6000,-0.000975)(0.6042,-0.001148)(0.6083,-0.001343)(0.6125,-0.001562)(0.6167,-0.001806)(0.6208,-0.002078)(0.6250,-0.002380)(0.6292,-0.002714)(0.6333,-0.003081)(0.6375,-0.003485)(0.6417,-0.003927)(0.6458,-0.004410)(0.6500,-0.004936)(0.6542,-0.005508)(0.6583,-0.006128)(0.6625,-0.006798)(0.6667,-0.007523)(0.6708,-0.008304)(0.6750,-0.009144)(0.6792,-0.010047)(0.6833,-0.011015)(0.6875,-0.012050)(0.6917,-0.013158)(0.6958,-0.014340)(0.7000,-0.015600)(0.7042,-0.016940)(0.7083,-0.018358)(0.7125,-0.019852)(0.7167,-0.021418)(0.7208,-0.023056)(0.7250,-0.024761)(0.7292,-0.026534)(0.7333,-0.028372)(0.7375,-0.030273)(0.7417,-0.032235)(0.7458,-0.034258)(0.7500,-0.036340)(0.7542,-0.038480)(0.7583,-0.040675)(0.7625,-0.042925)(0.7667,-0.045228)(0.7708,-0.047584)(0.7750,-0.049990)(0.7792,-0.052446)(0.7833,-0.054951)(0.7875,-0.057504)(0.7917,-0.060103)(0.7958,-0.062749)(0.8000,-0.065441)(0.8042,-0.068176)(0.8083,-0.070955)(0.8125,-0.073776)(0.8167,-0.076639)(0.8208,-0.079542)(0.8250,-0.082485)(0.8292,-0.085467)(0.8333,-0.088486)(0.8375,-0.091542)(0.8417,-0.094633)(0.8458,-0.097759)(0.8500,-0.100919)(0.8542,-0.104111)(0.8583,-0.107335)(0.8625,-0.110588)(0.8667,-0.113871)(0.8708,-0.117182)(0.8750,-0.120521)(0.8792,-0.123885)(0.8833,-0.127273)(0.8875,-0.130686)(0.8917,-0.134121)(0.8958,-0.137577)(0.9000,-0.141053)(0.9042,-0.144548)(0.9083,-0.148061)(0.9125,-0.151590)(0.9167,-0.155135)(0.9208,-0.158693)(0.9250,-0.162264)(0.9292,-0.165847)(0.9333,-0.169439)(0.9375,-0.173041)(0.9417,-0.176650)(0.9458,-0.180265)(0.9500,-0.183885)(0.9542,-0.187509)(0.9583,-0.191135)(0.9625,-0.194761)(0.9667,-0.198388)(0.9708,-0.202012)(0.9750,-0.205633)(0.9792,-0.209249)(0.9833,-0.212859)(0.9875,-0.216461)(0.9917,-0.220054)(0.9958,-0.223637)(1.0000,-0.227208)};
\addlegendentry{$\mathrm{Beta}(2,2)$: attains $0$ at the median}
\addplot[orange!85!black, densely dashed, very thick] coordinates {(0.0000,-0.166667)(0.0042,-0.162526)(0.0083,-0.158438)(0.0125,-0.154401)(0.0167,-0.150417)(0.0208,-0.146484)(0.0250,-0.142604)(0.0292,-0.138776)(0.0333,-0.135000)(0.0375,-0.131276)(0.0417,-0.127604)(0.0458,-0.123984)(0.0500,-0.120417)(0.0542,-0.116901)(0.0583,-0.113438)(0.0625,-0.110026)(0.0667,-0.106667)(0.0708,-0.103359)(0.0750,-0.100104)(0.0792,-0.096901)(0.0833,-0.093750)(0.0875,-0.090651)(0.0917,-0.087604)(0.0958,-0.084609)(0.1000,-0.081667)(0.1042,-0.078776)(0.1083,-0.075938)(0.1125,-0.073151)(0.1167,-0.070417)(0.1208,-0.067734)(0.1250,-0.065104)(0.1292,-0.062526)(0.1333,-0.060000)(0.1375,-0.057526)(0.1417,-0.055104)(0.1458,-0.052734)(0.1500,-0.050417)(0.1542,-0.048151)(0.1583,-0.045938)(0.1625,-0.043776)(0.1667,-0.041667)(0.1708,-0.039609)(0.1750,-0.037604)(0.1792,-0.035651)(0.1833,-0.033750)(0.1875,-0.031901)(0.1917,-0.030104)(0.1958,-0.028359)(0.2000,-0.026667)(0.2042,-0.025026)(0.2083,-0.023437)(0.2125,-0.021901)(0.2167,-0.020417)(0.2208,-0.018984)(0.2250,-0.017604)(0.2292,-0.016276)(0.2333,-0.015000)(0.2375,-0.013776)(0.2417,-0.012604)(0.2458,-0.011484)(0.2500,-0.010417)(0.2542,-0.009401)(0.2583,-0.008438)(0.2625,-0.007526)(0.2667,-0.006667)(0.2708,-0.005859)(0.2750,-0.005104)(0.2792,-0.004401)(0.2833,-0.003750)(0.2875,-0.003151)(0.2917,-0.002604)(0.2958,-0.002109)(0.3000,-0.001667)(0.3042,-0.001276)(0.3083,-0.000937)(0.3125,-0.000651)(0.3167,-0.000417)(0.3208,-0.000234)(0.3250,-0.000104)(0.3292,-0.000026)(0.3333,-0.000000)(0.3375,-0.000000)(0.3417,-0.000000)(0.3458,-0.000000)(0.3500,-0.000000)(0.3542,-0.000000)(0.3583,-0.000000)(0.3625,-0.000000)(0.3667,-0.000000)(0.3708,-0.000000)(0.3750,-0.000000)(0.3792,-0.000000)(0.3833,-0.000000)(0.3875,-0.000000)(0.3917,-0.000000)(0.3958,-0.000000)(0.4000,-0.000000)(0.4042,-0.000000)(0.4083,-0.000000)(0.4125,-0.000000)(0.4167,-0.000000)(0.4208,-0.000000)(0.4250,-0.000000)(0.4292,-0.000000)(0.4333,-0.000000)(0.4375,-0.000000)(0.4417,-0.000000)(0.4458,-0.000000)(0.4500,-0.000000)(0.4542,-0.000000)(0.4583,-0.000000)(0.4625,-0.000000)(0.4667,-0.000000)(0.4708,-0.000000)(0.4750,-0.000000)(0.4792,-0.000000)(0.4833,-0.000000)(0.4875,-0.000000)(0.4917,-0.000000)(0.4958,-0.000000)(0.5000,-0.000000)(0.5042,-0.000000)(0.5083,-0.000000)(0.5125,-0.000000)(0.5167,-0.000000)(0.5208,-0.000000)(0.5250,-0.000000)(0.5292,-0.000000)(0.5333,-0.000000)(0.5375,-0.000000)(0.5417,-0.000000)(0.5458,-0.000000)(0.5500,-0.000000)(0.5542,-0.000000)(0.5583,-0.000000)(0.5625,-0.000000)(0.5667,-0.000000)(0.5708,-0.000000)(0.5750,-0.000000)(0.5792,-0.000000)(0.5833,-0.000000)(0.5875,-0.000000)(0.5917,-0.000000)(0.5958,-0.000000)(0.6000,-0.000000)(0.6042,-0.000000)(0.6083,-0.000000)(0.6125,-0.000000)(0.6167,-0.000000)(0.6208,-0.000000)(0.6250,-0.000000)(0.6292,-0.000000)(0.6333,-0.000000)(0.6375,-0.000000)(0.6417,-0.000000)(0.6458,-0.000000)(0.6500,-0.000000)(0.6542,-0.000000)(0.6583,-0.000000)(0.6625,-0.000000)(0.6667,-0.000000)(0.6708,-0.000026)(0.6750,-0.000104)(0.6792,-0.000234)(0.6833,-0.000417)(0.6875,-0.000651)(0.6917,-0.000937)(0.6958,-0.001276)(0.7000,-0.001667)(0.7042,-0.002109)(0.7083,-0.002604)(0.7125,-0.003151)(0.7167,-0.003750)(0.7208,-0.004401)(0.7250,-0.005104)(0.7292,-0.005859)(0.7333,-0.006667)(0.7375,-0.007526)(0.7417,-0.008437)(0.7458,-0.009401)(0.7500,-0.010417)(0.7542,-0.011484)(0.7583,-0.012604)(0.7625,-0.013776)(0.7667,-0.015000)(0.7708,-0.016276)(0.7750,-0.017604)(0.7792,-0.018984)(0.7833,-0.020417)(0.7875,-0.021901)(0.7917,-0.023438)(0.7958,-0.025026)(0.8000,-0.026667)(0.8042,-0.028359)(0.8083,-0.030104)(0.8125,-0.031901)(0.8167,-0.033750)(0.8208,-0.035651)(0.8250,-0.037604)(0.8292,-0.039609)(0.8333,-0.041667)(0.8375,-0.043776)(0.8417,-0.045937)(0.8458,-0.048151)(0.8500,-0.050417)(0.8542,-0.052734)(0.8583,-0.055104)(0.8625,-0.057526)(0.8667,-0.060000)(0.8708,-0.062526)(0.8750,-0.065104)(0.8792,-0.067734)(0.8833,-0.070417)(0.8875,-0.073151)(0.8917,-0.075937)(0.8958,-0.078776)(0.9000,-0.081667)(0.9042,-0.084609)(0.9083,-0.087604)(0.9125,-0.090651)(0.9167,-0.093750)(0.9208,-0.096901)(0.9250,-0.100104)(0.9292,-0.103359)(0.9333,-0.106667)(0.9375,-0.110026)(0.9417,-0.113438)(0.9458,-0.116901)(0.9500,-0.120417)(0.9542,-0.123984)(0.9583,-0.127604)(0.9625,-0.131276)(0.9667,-0.135000)(0.9708,-0.138776)(0.9750,-0.142604)(0.9792,-0.146484)(0.9833,-0.150417)(0.9875,-0.154401)(0.9917,-0.158437)(0.9958,-0.162526)(1.0000,-0.166667)};
\addlegendentry{uniform: $0$ on the whole band $E(\alpha)$}
\addplot[green!45!black, densely dashdotted, very thick] coordinates {(0.0000,-0.095973)(0.0042,-0.091562)(0.0083,-0.087516)(0.0125,-0.083705)(0.0167,-0.080081)(0.0208,-0.076618)(0.0250,-0.073296)(0.0292,-0.070103)(0.0333,-0.067029)(0.0375,-0.064066)(0.0417,-0.061206)(0.0458,-0.058444)(0.0500,-0.055775)(0.0542,-0.053194)(0.0583,-0.050699)(0.0625,-0.048286)(0.0667,-0.045952)(0.0708,-0.043693)(0.0750,-0.041509)(0.0792,-0.039396)(0.0833,-0.037352)(0.0875,-0.035377)(0.0917,-0.033467)(0.0958,-0.031621)(0.1000,-0.029839)(0.1042,-0.028118)(0.1083,-0.026457)(0.1125,-0.024855)(0.1167,-0.023312)(0.1208,-0.021825)(0.1250,-0.020394)(0.1292,-0.019018)(0.1333,-0.017695)(0.1375,-0.016426)(0.1417,-0.015210)(0.1458,-0.014045)(0.1500,-0.012931)(0.1542,-0.011867)(0.1583,-0.010852)(0.1625,-0.009886)(0.1667,-0.008969)(0.1708,-0.008224)(0.1750,-0.008504)(0.1792,-0.008784)(0.1833,-0.009066)(0.1875,-0.009351)(0.1917,-0.009636)(0.1958,-0.009922)(0.2000,-0.010210)(0.2042,-0.010499)(0.2083,-0.010789)(0.2125,-0.011079)(0.2167,-0.011371)(0.2208,-0.011663)(0.2250,-0.011956)(0.2292,-0.012249)(0.2333,-0.012543)(0.2375,-0.012837)(0.2417,-0.013131)(0.2458,-0.013426)(0.2500,-0.013720)(0.2542,-0.014014)(0.2583,-0.014308)(0.2625,-0.014602)(0.2667,-0.014895)(0.2708,-0.015188)(0.2750,-0.015480)(0.2792,-0.015772)(0.2833,-0.016062)(0.2875,-0.016352)(0.2917,-0.016641)(0.2958,-0.016928)(0.3000,-0.017214)(0.3042,-0.017499)(0.3083,-0.017782)(0.3125,-0.018064)(0.3167,-0.018343)(0.3208,-0.018621)(0.3250,-0.018897)(0.3292,-0.019171)(0.3333,-0.019442)(0.3375,-0.019711)(0.3417,-0.019977)(0.3458,-0.020240)(0.3500,-0.020501)(0.3542,-0.020759)(0.3583,-0.021013)(0.3625,-0.021264)(0.3667,-0.021512)(0.3708,-0.021756)(0.3750,-0.021996)(0.3792,-0.022233)(0.3833,-0.022465)(0.3875,-0.022692)(0.3917,-0.022915)(0.3958,-0.023134)(0.4000,-0.023347)(0.4042,-0.023556)(0.4083,-0.023758)(0.4125,-0.023956)(0.4167,-0.024148)(0.4208,-0.024334)(0.4250,-0.024513)(0.4292,-0.024686)(0.4333,-0.024852)(0.4375,-0.025011)(0.4417,-0.025163)(0.4458,-0.025307)(0.4500,-0.025443)(0.4542,-0.025571)(0.4583,-0.025691)(0.4625,-0.025802)(0.4667,-0.025926)(0.4708,-0.026034)(0.4750,-0.026125)(0.4792,-0.026199)(0.4833,-0.026257)(0.4875,-0.026299)(0.4917,-0.026326)(0.4958,-0.026339)(0.5000,-0.026339)(0.5042,-0.026339)(0.5083,-0.026326)(0.5125,-0.026299)(0.5167,-0.026257)(0.5208,-0.026199)(0.5250,-0.026125)(0.5292,-0.026034)(0.5333,-0.025926)(0.5375,-0.025802)(0.5417,-0.025691)(0.5458,-0.025571)(0.5500,-0.025443)(0.5542,-0.025307)(0.5583,-0.025163)(0.5625,-0.025011)(0.5667,-0.024852)(0.5708,-0.024686)(0.5750,-0.024513)(0.5792,-0.024334)(0.5833,-0.024148)(0.5875,-0.023956)(0.5917,-0.023758)(0.5958,-0.023556)(0.6000,-0.023347)(0.6042,-0.023134)(0.6083,-0.022915)(0.6125,-0.022692)(0.6167,-0.022465)(0.6208,-0.022233)(0.6250,-0.021996)(0.6292,-0.021756)(0.6333,-0.021512)(0.6375,-0.021264)(0.6417,-0.021013)(0.6458,-0.020759)(0.6500,-0.020501)(0.6542,-0.020240)(0.6583,-0.019977)(0.6625,-0.019711)(0.6667,-0.019442)(0.6708,-0.019171)(0.6750,-0.018897)(0.6792,-0.018621)(0.6833,-0.018343)(0.6875,-0.018064)(0.6917,-0.017782)(0.6958,-0.017499)(0.7000,-0.017214)(0.7042,-0.016928)(0.7083,-0.016641)(0.7125,-0.016352)(0.7167,-0.016062)(0.7208,-0.015772)(0.7250,-0.015480)(0.7292,-0.015188)(0.7333,-0.014895)(0.7375,-0.014602)(0.7417,-0.014308)(0.7458,-0.014014)(0.7500,-0.013720)(0.7542,-0.013426)(0.7583,-0.013131)(0.7625,-0.012837)(0.7667,-0.012543)(0.7708,-0.012249)(0.7750,-0.011956)(0.7792,-0.011663)(0.7833,-0.011371)(0.7875,-0.011079)(0.7917,-0.010789)(0.7958,-0.010499)(0.8000,-0.010210)(0.8042,-0.009922)(0.8083,-0.009636)(0.8125,-0.009351)(0.8167,-0.009066)(0.8208,-0.008784)(0.8250,-0.008504)(0.8292,-0.008224)(0.8333,-0.008969)(0.8375,-0.009886)(0.8417,-0.010852)(0.8458,-0.011867)(0.8500,-0.012931)(0.8542,-0.014045)(0.8583,-0.015210)(0.8625,-0.016426)(0.8667,-0.017695)(0.8708,-0.019018)(0.8750,-0.020394)(0.8792,-0.021825)(0.8833,-0.023312)(0.8875,-0.024855)(0.8917,-0.026457)(0.8958,-0.028118)(0.9000,-0.029839)(0.9042,-0.031621)(0.9083,-0.033467)(0.9125,-0.035377)(0.9167,-0.037352)(0.9208,-0.039396)(0.9250,-0.041509)(0.9292,-0.043693)(0.9333,-0.045952)(0.9375,-0.048286)(0.9417,-0.050699)(0.9458,-0.053194)(0.9500,-0.055775)(0.9542,-0.058444)(0.9583,-0.061206)(0.9625,-0.064066)(0.9667,-0.067029)(0.9708,-0.070103)(0.9750,-0.073296)(0.9792,-0.076618)(0.9833,-0.080081)(0.9875,-0.083705)(0.9917,-0.087516)(0.9958,-0.091562)(1.0000,-0.095973)};
\addlegendentry{arcsine: maximum $-8.22\cdot10^{-3}<0$}
\addplot[only marks, mark=*, mark size=2pt, green!45!black, forget plot]
  coordinates {(0.8292,-0.008224)};
\end{axis}
\end{tikzpicture}
\caption{The certificate behind \Cref{prop:nonexistence}. The margin game is
symmetric zero-sum, so a pure equilibrium exists iff some position has security
value $0$. For a centrally peaked electorate the median achieves it; for
uniform voters every point of the band does; for the arcsine (hollow)
electorate the maximum is strictly negative at $c=0.829$, so \emph{every}
position can be strictly beaten and no pure equilibrium exists. Computing this
curve requires quadrature that respects the endpoint singularity of $v$: a
naive grid rule returns $-0.387$ for the arcsine maximum, a factor of $47$ too
negative (\Cref{m:rem:artifact}).}
\label{fig:security}
\end{figure}

Two remarks complete the picture. First, equilibria of nearly-zero-sum games
converge onto margin-game equilibria (upper hemicontinuity,
\Cref{n:thm:B5}). \Cref{cor:quantile} shows the statement is
not vacuous at $\alpha=0$, where the approximating equilibria exist all the
way to the limit and select the median. Second, at $\alpha=0$ we conjecture
pure existence for \emph{every} $(v,k)$ and $\beta<1$
(\Cref{m:conj:alpha0}); the certified failure at $\alpha=1$ shows that no
proof can survive a relaxation of $\alpha=0$. Together with
\Cref{thm:quartic}, this gives the picture for the margin game:
the median survives alienation, at least locally, when the electorate is
centrally peaked; when it is hollow, the median fails and, at least for the densities examined
here, no pure equilibrium takes its place.

%%%%%%%%%%%%%%%%%%%%%%%%%%%%%%%%%%%%%%%%%%%%%%%%%%%%%%%%%%%%%%%%%%%%%%%%%%%%%%
%%%%%%%%%%%%%%%%%%%%%%%%%%%%%%%%%%%%%%%%%%%%%%%%%%%%%%%%%%%%%%%%%%%%%%%%%%%%%%
%% Discussion and Further Research
%%%%%%%%%%%%%%%%%%%%%%%%%%%%%%%%%%%%%%%%%%%%%%%%%%%%%%%%%%%%%%%%%%%%%%%%%%%%%%
\section{Discussion and Further Research}\label{sec:discussion}

\paragraph{What drives polarization}
\Cref{tab:drivers} collects the answer the paper was built to give. Three
forces raise equilibrium polarization---a more dispersed electorate, costlier
voting in the reverse-hazard-rate sense, and stronger alienation; a hollow
center removes the median as a margin-game equilibrium; one force lowers
polarization---competition for the margin; and one candidate that every
intuition nominates, the level of turnout, provably does nothing at all. The
two voter-side channels do not interfere with each other
(\Cref{thm:rhr-alpha}): each still polarizes at every level of the other, and
both survive the presence of competitiveness
(\Cref{prop:beta-mono,prop:rhr-beta}). The one de-polarizing
force is also the only one whose comparative static is shown to reverse sign
when super-regularity fails (\Cref{ex:fold}), an asymmetry worth
keeping in mind when reading the model empirically.

\begin{table}[!t]
\centering
\caption{Every force in the model, and which way it moves equilibrium
polarization $\Delta^*$. The last row is not a gap in the analysis but a
theorem: the level of participation is an output of the model, never an input
to it.}
\label{tab:drivers}
\setlength{\tabcolsep}{3.5pt}\footnotesize
\renewcommand{\arraystretch}{1.15}
\begin{tabular}{@{}lc R{3.2cm} R{2.5cm}@{}}
\toprule
Increase in\ldots & $\Delta^*$ & Holds when & Where \\
\midrule
electorate spread (IQR) & $\uparrow$ & unif.\ costs, $v$ symmetric, $\rho_v\le4$ & \Cref{cor_quarter} \\
electorate scale $\sigma$ & $\uparrow$ & unif.\ voters, $K$ super-reg.; exact for power costs (any $v$) & \Cref{prop:mass} \\
voting costs (rhr order) & $\uparrow$ & unif.\ voters, $K$ super-reg., any $(\alpha,\beta)$ & \Cref{prop:rh,thm:rhr-alpha,prop:rhr-beta} \\
alienation $\alpha$ & $\uparrow$ & unif.\ voters, $K$ super-reg., $\alpha>0$; $\Delta^*\to1$ & \Cref{thm:alien-mono} \\
competitiveness $\beta$ & $\downarrow$ & unif.\ voters, monotone rhr \emph{only} & \Cref{prop:beta-mono,ex:fold} \\
hollow center, $v''(\tfrac12)>0$ & breaks median & margin game, unif.\ costs, $\alpha>0$, $v$ symmetric & \Cref{thm:quartic} \\
cost \emph{level}, $K\mapsto zK$ & --- & always: no effect & \Cref{prop:scale-gen,rem:fosd} \\
\bottomrule
\end{tabular}
\end{table}

\paragraph{The baseline and the two dials together}
\Cref{tab:summary} collects the doubly-uniform benchmarks. Reading it as a
story: costly voting moves vote-share maximizers from the median to the
quartiles; alienation pushes them further out, without bound; competitiveness
pulls them back in, linearly to the median when voters are not alienated,
but only to the top of the margin band when they are. The classical Downsian
prediction survives costly voting only under the margin objective, and even
then only when voters are unalienated or, locally, when the electorate is
centrally peaked;
for alienated voters facing a hollow electorate, nothing need survive---not
even equilibrium existence. Polarization in this model is thus neither an
inevitability nor an accident: it is the resultant of two voter-side
forces---costly voting and alienation---and one candidate-side force, each
with an explicit comparative static.

\begin{table}[!t]
\centering
\caption{The baseline and the two dials at a glance. Second column: the
equilibrium in the doubly-uniform benchmark ($x,\kappa\sim U(0,1)$), with
$\lambda=2+\alpha$. Third column: sufficient density-ratio bound on $v$ for
the first-order conditions to characterize equilibria under uniform costs,
where
$f_\beta=\tfrac{2\lambda(1+\alpha)+\beta}{(1+\alpha)^2}$ (concavity) and
$g_\beta=\tfrac{(1+\alpha)\lambda}{\alpha(1+\alpha+\beta)+\beta\lambda}$
(dead zone), from \Cref{m:thm:D1}; the bounds come from
different arguments and are not claimed tight. The row `both' is stated for
$\alpha>0$ and does not reduce to the row above it as $\alpha\downarrow0$:
$f_\beta\to4+\beta\leq \tfrac4{1-\beta}$, because the $\alpha=0$ bound uses an
exact cancellation at the midpoint that the worst-case argument of
\Cref{m:thm:D1} forgoes (\Cref{m:rem:D1shape}), and $g_\beta\to1/\beta$
bounds a dead-zone piece of the strategy space that vanishes at $\alpha=0$.}
\label{tab:summary}
\setlength{\tabcolsep}{2.6pt}\footnotesize
\begin{tabular}{@{}lll R{2.9cm}@{}}
\toprule
Setting & $\Delta^*$ (doubly unif.) & Umbrella $\rho_v\le$ & Where \\
\midrule
baseline $(0,0)$ & $\tfrac12$, at $(\tfrac14,\tfrac34)$ & $4$ & \Cref{cor_uu,prop:fourflat} \\[2pt]
alienation $(\alpha,0)$ & $\tfrac{2+\alpha}{4+\alpha}$ & $\min\{\tfrac{2\lambda}{1+\alpha},\tfrac\lambda\alpha\}$ & \Cref{cor:alien-uu,n:thm:A3} \\[2pt]
competition $(0,\beta)$ & $\tfrac{1-\beta}2$ (quantile profile) & $\tfrac4{1-\beta}$ & \Cref{cor:beta-uu,prop:massgap} \\[2pt]
both $(\alpha,\beta)$ & \cref{eq:beta-closed} & $\min\{f_\beta,g_\beta\}$ & \Cref{cor:beta-uu,m:thm:D1} \\[2pt]
margin limit $(\alpha,1)$ & band: $\Delta\in\big[0,\tfrac\alpha\lambda\big]$ & knife-edge & \Cref{prop:margin_uniform,thm:quartic} \\
\bottomrule
\end{tabular}
\end{table}

\paragraph{Open questions}
Several questions remain open. On
sufficiency: sharp thresholds for the density-ratio umbrellas of
\Cref{tab:summary}, and a sufficiency theory beyond uniform costs
(\Cref{n:prop:A4} covers non-increasing $k$ at $\beta=0$; the obstruction at
$\beta>0$ is isolated in \Cref{m:rem:D1k}). On alienation: a
general-electorate version of \Cref{thm:alien-mono}, whose conclusion holds
numerically for the concentrated electorate of \Cref{rem:deadzone} but is
proved only for uniform voters. On existence: under
\Cref{ass:reg} we conjecture pure existence for every $(v,k)$ in the baseline
and, at $\alpha=0$, for every $\beta<1$ (\Cref{rem:exist,m:conj:alpha0}),
the first of these open already for uniform costs once $\rho_v>4$; closing the certified gap $(0.90,0.9168]$
of \Cref{m:thm:E2}; and characterizing the mixed equilibria of the margin
game where pure ones fail. On the margin game with alienation: the exact
non-existence region in $(\alpha,\beta,v)$, and whether the quantile path of
\Cref{cor:quantile} has an $\alpha>0$ analogue. The certificate technique
behind \Cref{m:thm:E2} may be of independent interest for other spatial games
with continuous strategy spaces.

Empirically, where political polarization has
risen~\cite{pew2017partisan,boxell2024cross} alongside declining
turnout~\cite{hooghe2017tipping,franklin2004voter}, the model offers a channel
complementary to media- and network-based accounts
\cite{prior2013media,kubin2021role}, with a sharpened prediction: the
association should run through the \emph{composition} of the abstaining
margin rather than the participation rate, and the model separates ``voting
became harder'' (uniform shifts in participation) from ``voters became
disaffected'' (abstention concentrated among the ideologically
distant)---distinctions the turnout literature is equipped to test
\cite{leighley2014who,fowler2013electoral}. A first piece of evidence on the
composition side exists: across Dutch municipalities, lower turnout goes with
a more polarized distribution of the votes cast, the parties furthest from the
center gaining most, consistent with centrist voters being the ones who
abstain~\cite{alyussef2026turnout}---the abstaining margin of the baseline
model, where abstention concentrates between the candidates.

That test has a confound the present model assumes away. Costs here are drawn
i.i.d.\ and independently of position (\Cref{sec:model}), and the assumption
does real work: written out for the baseline, the balance condition of
\Cref{lem:balance} has an outer term that factorizes as $k(\Delta)V(c_1)$,
one cost density times one voter mass, because
every voter on the flank both faces the same benefit $\Delta$ and draws from the
same $K$. Replace $K$ by a family $K_x$ and that product becomes
$\int_0^{c_1}k_x(\Delta)v(x)\,dx$, which does not split; \Cref{prop:scale} and
the balance principle itself survive, but \Cref{cor:rhr-canonical} weakens from
``one scalar function summarizes the cost side'' to a statement about a field
$\{r_x\}$. Note that this is not the position-dependence the model already has.
Alienation acts through the benefit, via $d_i$, so its bite moves when the
candidates move; a cost--position correlation is pinned to the voter and does
not move at all. The paper treats endogenous position-dependence thoroughly and
exogenous position-dependence not at all, and no choice of $\alpha$ substitutes
for the latter. The direction is signed, moreover, and the two channels
disagree: alienation predicts that a candidate's own extremists are the first to
abstain, whereas if the ideologically committed are simply cheaper to mobilize,
then $k_x$ is shifted toward low costs at extreme $x$, the flank is more
responsive than \Cref{prop:rh} assumes, and the model \emph{understates}
polarization. Which of the two dominates is an empirical question, and it is the
same question as whether a composition shift in the abstaining margin identifies
the strategic channel or merely reflects who finds voting cheap.

Finally, the costly-turnout framework introduced here can be naturally
combined with additional dynamics. One extension is to allow for non-linear or
probabilistic voter utilities, as in \cite{coughlin2015probabilistic}, which
may better capture how voters evaluate platforms. Another is to embed the
model in a dynamic environment with repeated elections, building on iterative
voting models \cite{meir2017iterative}; such a setting could generate
endogenous status quo effects that shape participation incentives, in the
spirit of \cite{shapiro2017reality}. Further natural directions include more
than two candidates and endogenous entry, where the alienation dead zones
suggest room for centrist or extremist entrants; asymmetric candidates
(valence or cost advantages), for which our asymmetry law
(\Cref{lem:balance-gen}) already prices the trade-off; and finite electorates with
pivotality considerations in the spirit of the Calculus of Voting
\cite{riker1968theory}, connecting the present continuum model back to the
classical turnout literature.

%%%%%%%%%%%%%%%%%%%%%%%%%%%%%%%%%%%%%%%%%%%%%%%%%%%%%%%%%%%%%%%%%%%%%%%%%%%%%%
\appendix
% From here on every \section is an appendix, so cleveref should say
% "Appendix B" and not "Section B".
\crefalias{section}{appendix}
\crefalias{subsection}{appendix}

\section{Deferred Proofs}\label{app:proofs}
\noindent{\small\itshape This appendix collects the proofs deferred from
\Cref{sec:base,sec:alien,sec:comp}. Proofs of
\Cref{lem:interior,thm_uniform_costs,prop:fourflat,cor_quarter,thm_uniform_votes,prop:rh,thm:margin_median,prop:margin_uniform}
are in \Cref{app:omitted-proofs}. The remaining proofs, together with the
supplementary results, worked examples and certificates they draw on, occupy
\Crefrange{n:sec:prelim}{m:sec:G}.}
\medskip
\printdeferredproofs

\clearpage
\section{Omitted Proofs: Baseline and Margin Game}\label{app:omitted-proofs}

This appendix proves the results of \Cref{sec:base,sec:comp} that are stated
without proof in the main text; proofs of the remaining results are in
\Cref{app:proofs} and in \Crefrange{n:sec:prelim}{m:sec:G}. For the reader's convenience we
restate each result before its proof.

\subsection{Preliminaries: Payoffs and Their Derivatives}\label{app:prelim}

Throughout, fix a profile $(c_1,c_2)$ with $c_1<c_2$ and write $\Delta:=c_2-c_1$ and $m:=\frac{c_1+c_2}{2}$. Voters on $[0,m)$ either vote for candidate $1$ or abstain, and voters on $(m,1]$ either vote for candidate $2$ or abstain. Evaluating the gain $d_{-i}(x)-d_i(x)$ piecewise---it equals $\Delta$ outside $[c_1,c_2]$ and $|c_1+c_2-2x|$ inside---we obtain
\begin{align}
u_1(c_1,c_2)&=K(\Delta)\,V(c_1)+\int_{c_1}^{m}K(c_1+c_2-2x)\,v(x)\,dx,\label{eq:u1_general}\\
u_2(c_1,c_2)&=K(\Delta)\,\big(1-V(c_2)\big)+\int_{m}^{c_2}K(2x-c_1-c_2)\,v(x)\,dx.\label{eq:u2_general}
\end{align}
Since $v$ and $k$ are continuous, $u_1$ is differentiable in $c_1$ on $(0,c_2)$. Differentiating \cref{eq:u1_general} by the Leibniz rule, the term $K(\Delta)v(c_1)$ produced by the first summand cancels against the lower-limit term of the integral, and the upper-limit term vanishes because $K(0)=0$; what remains is
\begin{equation}\label{eq:du1_general}
\frac{\partial u_1}{\partial c_1}(c_1,c_2)
=\int_{c_1}^{m}k(c_1+c_2-2x)\,v(x)\,dx\;-\;k(\Delta)\,V(c_1),
\end{equation}
and, symmetrically,
\begin{equation}\label{eq:du2_general}
\frac{\partial u_2}{\partial c_2}(c_1,c_2)
=k(\Delta)\,\big(1-V(c_2)\big)\;-\;\int_{m}^{c_2}k(2x-c_1-c_2)\,v(x)\,dx.
\end{equation}
Both formulas remain valid as one-sided derivatives at $c_1=0$ and $c_2=1$.

Two specializations will be used repeatedly. First, for \emph{uniform costs} ($K(t)=t$, $k\equiv1$), \crefrange{eq:du1_general}{eq:du2_general} become
\begin{equation}\label{eq:du_ucost}
\frac{\partial u_1}{\partial c_1}=V(m)-2V(c_1),
\qquad
\frac{\partial u_2}{\partial c_2}=1+V(m)-2V(c_2),
\end{equation}
and integrating \cref{eq:u1_general} by parts gives the compact representation
\begin{equation}\label{eq:u_2intV}
u_1(c_1,c_2)=2\int_{c_1}^{m}V(x)\,dx,
\qquad
u_2(c_1,c_2)=2\int_{m}^{c_2}\big(1-V(x)\big)\,dx .
\end{equation}
(The analogous formulas with the roles of the candidates reversed hold when $c_1>c_2$; they follow from the reflection $x\mapsto1-x$, which maps $V$ to the CDF $\widetilde V(x)=1-V(1-x)$, whose density $\widetilde v(x)=v(1-x)$ has the same bounds as $v$.)

Second, for \emph{uniform voters} ($v\equiv1$), substituting $y=c_1+c_2-2x$ in \cref{eq:u1_general} gives, with $G(t):=\int_0^tK(y)\,dy$,
\begin{equation}\label{eq:u_uvoters}
u_1(c_1,c_2)=c_1K(\Delta)+\tfrac12G(\Delta),
\qquad
u_2(c_1,c_2)=(1-c_2)K(\Delta)+\tfrac12G(\Delta),
\end{equation}
so that
\begin{equation}\label{eq:du_uvoters}
\frac{\partial u_1}{\partial c_1}=\tfrac12K(\Delta)-c_1k(\Delta),
\qquad
\frac{\partial u_2}{\partial c_2}=(1-c_2)k(\Delta)-\tfrac12K(\Delta).
\end{equation}

\medskip
\noindent\textbf{\Cref{lem:interior} (restated).}\enspace
\emph{In the baseline model, under \Cref{ass:reg}, every pure Nash equilibrium $(c_1,c_2)$ satisfies $c_1\neq c_2$ and, labeling the candidates so that $c_1<c_2$, $0<c_1<c_2<1$.}

\begin{proof}
\emph{No ties.} At a tied profile $c_1=c_2=c$ both candidates receive $0$. If $c>0$, let candidate $1$ deviate to any $c_1'\in(0,c)$; by \cref{eq:u1_general}, her payoff is at least $\int_{c_1'}^{m'}K(c_1'+c-2x)v(x)\,dx>0$, since on the interior of that interval the argument of $K$ lies in $(0,\Delta')\subseteq(0,1)$, where $K>0$, and $v>0$. If $c=0$, the mirror deviation to any $c_1'\in(0,1)$ works. Hence tied profiles are never equilibria.

\emph{No boundary positions.} Suppose $c_1=0<c_2$. By \cref{eq:du1_general}, the one-sided derivative at $c_1=0$ equals $\int_{0}^{c_2/2}k(c_2-2x)v(x)\,dx>0$ (the argument of $k$ ranges over $(0,c_2)\subseteq(0,1)$, where $k>0$, and $v>0$), so a small move to the right strictly increases $u_1$, and $c_1=0$ is not a best response. The mirror argument using \cref{eq:du2_general} rules out $c_2=1$.
\end{proof}

\medskip
\noindent\textbf{\Cref{thm_uniform_costs} (restated).}\enspace
\emph{Let $\kappa\sim U(0,1)$ and let $v$ satisfy \Cref{ass:reg}. Then every Nash equilibrium $(c_1,c_2)$ satisfies}
\[
V(c_1)=V(c_2)-\tfrac12=\tfrac12\,V\!\left(\tfrac{c_1+c_2}{2}\right).
\]

\begin{proof}
By \Cref{lem:interior} and \cref{eq:du_ucost}, an equilibrium satisfies
\begin{equation}\label{eq:FOC_ucost}
V(m)=2V(c_1)
\qquad\text{and}\qquad
V(m)=2V(c_2)-1 .
\end{equation}
Combining the two equalities yields $V(c_1)=V(c_2)-\tfrac12=\tfrac12V(m)$, as required.
\end{proof}

\medskip
\noindent\textbf{\Cref{prop:fourflat} (restated).}\enspace
\emph{Let $\kappa\sim U(0,1)$ and suppose $\rho_v\le4$. Then each candidate's payoff is concave in her own position on each side of the opponent. Consequently, $(c_1,c_2)$ with $0<c_1<c_2<1$ is a Nash equilibrium if and only if it satisfies the conditions of \Cref{thm_uniform_costs} and neither candidate gains by a cross-over deviation to the far side of her opponent.}

\begin{proof}
Fix $c_2$ and consider $c_1'\in[0,c_2)$. By \cref{eq:du_ucost}, $\partial u_1/\partial c_1'=V(m')-2V(c_1')$ with $m'=\frac{c_1'+c_2}{2}$, whose derivative in $c_1'$ is
\[
\tfrac12\,v(m')-2\,v(c_1')\;\le\;\tfrac12\sup v-2\inf v\;\le\;0
\]
whenever $\rho_v\le4$. Hence $u_1(\cdot,c_2)$ has a non-increasing derivative, i.e.,\ is concave, on $[0,c_2]$. For $c_1'\in(c_2,1]$ the same argument applies after the reflection $x\mapsto1-x$ (which preserves the density ratio), so $u_1(\cdot,c_2)$ is also concave on $[c_2,1]$. The argument for candidate $2$ is symmetric.

($\Rightarrow$) If $(c_1,c_2)$ is an equilibrium then the balance conditions hold by \Cref{thm_uniform_costs}, and the cross-over inequalities hold trivially, since they compare the equilibrium payoff with payoffs of feasible deviations.

($\Leftarrow$) Suppose the conditions of \Cref{thm_uniform_costs} hold at an interior profile. Then $\partial u_1/\partial c_1=V(m)-2V(c_1)=0$, so by concavity $c_1$ maximizes $u_1(\cdot,c_2)$ on $[0,c_2]$. Note also that $\rho_v\le4$ forces $\inf v\ge\tfrac14\sup v\ge\tfrac14$ (as $\sup v\ge\int_0^1v=1$), so $V(c_1)\ge c_1/4>0$ and, by \cref{eq:u_2intV}, $u_1(c_1,c_2)=2\int_{c_1}^mV>0$; in particular the tie $c_1'=c_2$ (payoff $0$) is not profitable. The cross-over inequality covers all $c_1'\in[c_2,1]$. Hence $c_1$ is a global best response, and symmetrically for $c_2$, so the profile is an equilibrium.
\end{proof}

\medskip
\noindent\textbf{\Cref{cor_quarter} (restated).}\enspace
\emph{Let $\kappa\sim U(0,1)$ and suppose $\rho_v\le4$. Then any profile with $V(c_1)=\tfrac14$, $V(c_2)=\tfrac34$ and $V\big(\tfrac{c_1+c_2}2\big)=\tfrac12$ is a Nash equilibrium.}

\begin{proof}
Write $s:=\inf v$; as noted above, $\rho_v\le4$ implies $s\ge\tfrac14$ and $\sup v\le 4s$. The profile is interior ($0<V(c_1)<V(c_2)<1$ and $V$ is strictly increasing) and satisfies \cref{eq:FOC_ucost}: $V(m)=\tfrac12=2V(c_1)$ and $2V(c_2)-1=\tfrac12=V(m)$. By \Cref{prop:fourflat} it therefore suffices to verify the two cross-over inequalities.

\emph{A lower bound on the candidates' payoffs.} Since $V(x)\ge V(c_1)+s(x-c_1)=\tfrac14+s(x-c_1)$, \cref{eq:u_2intV} gives
\[
u_1(c_1,c_2)=2\int_{c_1}^{m}V(x)\,dx
\;\ge\;2\left[\frac14\cdot\frac\Delta2+\frac s2\Big(\frac\Delta2\Big)^2\right]
=\frac\Delta4+\frac{s\Delta^2}{4}.
\]
Moreover $\tfrac12=V(c_2)-V(c_1)=\int_{c_1}^{c_2}v\le 4s\Delta$, i.e.,\ $\Delta\ge\tfrac1{8s}$, whence
\begin{equation}\label{eq:quartile_lb}
u_1(c_1,c_2)\;\ge\;\frac{1}{32s}+\frac{1}{256s}\;>\;\frac{1}{32s}.
\end{equation}
The same bound holds for $u_2(c_1,c_2)$, using $1-V(x)\ge\tfrac14+s(c_2-x)$ in \cref{eq:u_2intV}.

\emph{An upper bound on cross-over payoffs.} Consider a deviation of candidate $1$ to $c_1'=c_2+t$ with $t\in(0,1-c_2]$. Every voter who votes for the deviator lies to the right of the new midpoint $m'=c_2+\tfrac t2$, and votes with probability at most $K(t)=t$. Hence
\[
u_1(c_1',c_2)\;\le\;t\,\big(1-V(c_2+\tfrac t2)\big)
\;\le\;t\Big(\frac14-\frac{st}{2}\Big)
\;\le\;\max_{\tau\ge0}\ \tau\Big(\frac14-\frac{s\tau}{2}\Big)
=\frac{1}{32s},
\]
where the second inequality uses $1-V(c_2)=\tfrac14$ and $V(c_2+\tfrac t2)-V(c_2)\ge\tfrac{st}2$. Combining with \cref{eq:quartile_lb}, no cross-over deviation of candidate $1$ is profitable. The bound for candidate $2$ is the mirror image (using $V(c_1)=\tfrac14$ and $V(c_1)-V(c_1-\tfrac t2)\ge \tfrac{st}2$). By \Cref{prop:fourflat}, the profile is a Nash equilibrium.
\end{proof}

\begin{remark}\label{rem:tight}
The bound $\rho_v\le4$ is tight for this argument in the following sense: at $\rho_v=4$ the deviator's payoff can be exactly flat over an interval of alternative best responses. For instance, for the symmetric step density $v=\tfrac58$ on $[0,\tfrac25)\cup(\tfrac35,1]$ and $v=\tfrac52$ on $[\tfrac25,\tfrac35]$ (density ratio exactly $4$), the quartile profile is $\big(\tfrac25,\tfrac35\big)$, and one computes from \cref{eq:u_2intV} that $u_1(c_1',\tfrac35)=\tfrac3{40}$ for \emph{every} $c_1'\in[\tfrac15,\tfrac25]$: candidate $1$ is indifferent across a whole interval of positions, and the quartile profile is a (non-strict) equilibrium.
\end{remark}

\medskip
\noindent\textbf{\Cref{cor_quarter}, symmetric case (restated).}\enspace
\emph{Let $\kappa\sim U(0,1)$ and let $v$ be symmetric around $\tfrac12$ with $\rho_v\le4$. Then the quartile profile $\big(V^{-1}(\tfrac14),V^{-1}(\tfrac34)\big)$ is a Nash equilibrium; in particular an equilibrium exists. Moreover, any profile satisfying the conditions of \Cref{cor_quarter} must be symmetric, $c_2=1-c_1$.}

\begin{proof}
Symmetry of $v$ gives $V(x)+V(1-x)=1$ for all $x$, hence $V^{-1}(\tfrac34)=1-V^{-1}(\tfrac14)$; the midpoint of the quartile profile is then $\tfrac12$ and $V(\tfrac12)=\tfrac12$, so all three conditions of \Cref{cor_quarter} hold and the profile is an equilibrium. Conversely, if $V(c_1)=\tfrac14$ and $V(c_2)=\tfrac34$, then $V(1-c_1)=1-V(c_1)=\tfrac34=V(c_2)$; since $\rho_v\le4$ forces $\inf v>0$, $V$ is strictly increasing, so $c_2=1-c_1$.
\end{proof}

\medskip
\noindent\textbf{\Cref{thm_uniform_votes} (restated).}\enspace
\emph{Let $x\sim U(0,1)$ and let $k$ satisfy \Cref{ass:reg}. (a) Every Nash equilibrium is symmetric, $c_1=1-c_2$, with gap $\Delta\in(0,1)$ solving $K(\Delta)=(1-\Delta)k(\Delta)$. (b) If moreover $\psi(t)=t+K(t)/k(t)$ is strictly increasing on $(0,1)$, then this equation has a unique solution $\Delta^*\in(0,1)$, and $\big(\tfrac{1-\Delta^*}2,\tfrac{1+\Delta^*}2\big)$ is the unique Nash equilibrium up to relabeling.}

\begin{proof}
\emph{(a)} By \Cref{lem:interior}, an equilibrium is interior with $0<\Delta<1$, so both derivatives in \cref{eq:du_uvoters} vanish:
\begin{equation}\label{eq:uv_focs}
\tfrac12K(\Delta)=c_1\,k(\Delta)
\qquad\text{and}\qquad
\tfrac12K(\Delta)=(1-c_2)\,k(\Delta).
\end{equation}
Since $\Delta\in(0,1)$ and $k>0$ on $(0,1)$, we have $K(\Delta)>0$; the first equation then forces $k(\Delta)>0$, so comparing the two equations gives $c_1=1-c_2$, whence $c_1=\tfrac{1-\Delta}2$, and \cref{eq:uv_focs} becomes $K(\Delta)=(1-\Delta)k(\Delta)$.

\emph{(b)} On $(0,1)$, where $k>0$, the equation $K(\Delta)=(1-\Delta)k(\Delta)$ is equivalent to $\psi(\Delta)=1$.

\emph{Existence and uniqueness of the root.} $\psi$ is continuous on $(0,1)$ and strictly increasing, so $L:=\lim_{t\downarrow0}\psi(t)$ exists. Suppose $L>0$; then $K(t)/k(t)\ge L/2$ for all sufficiently small $t$, i.e.,\ $(\log K)'(t)=k(t)/K(t)\le 2/L$, which upon integrating from $t$ to a fixed small $t_0$ gives $K(t)\ge K(t_0)\,e^{-2t_0/L}>0$ as $t\downarrow0$---contradicting $K(0)=0$ and the continuity of $K$. Hence $L=0<1$. At the other end, $\psi$ exceeds $1$ somewhere on $(0,1)$: otherwise $K(t)/k(t)\le 1-t$, i.e.\ $k(t)/K(t)\ge 1/(1-t)$, for all $t\in(0,1)$, and integrating from $t$ to $s<1$ gives $\log K(s)-\log K(t)=\int_t^{s}k/K\ge\log\tfrac{1-t}{1-s}\to\infty$ as $s\uparrow1$, contradicting $K(s)\le K(1)\le1$. (No bound on $k$ is needed: \Cref{ass:reg} permits $k$ unbounded at the endpoints, as in \Cref{ex:power} with $p<1$.) By the intermediate value theorem and strict monotonicity, $\psi(\Delta)=1$ has a unique root $\Delta^*\in(0,1)$.

\emph{The symmetric profile is an equilibrium.} Let $c_1=\tfrac{1-\Delta^*}2$, $c_2=\tfrac{1+\Delta^*}2$, and consider deviations of candidate $1$ (the case of candidate $2$ is the mirror image). Parameterize a deviation on the \emph{own} side of the opponent by its distance $t\in(0,c_2]$ from $c_2$: by \cref{eq:u_uvoters}, the deviation payoff is $u_L(t)=(c_2-t)K(t)+\tfrac12G(t)$, with derivative
\[
u_L'(t)=(c_2-t)\,k(t)-\tfrac12K(t)
=k(t)\left[c_2-\frac{t+\psi(t)}{2}\right].
\]
Since $t\mapsto t+\psi(t)$ is strictly increasing and equals $2c_2=1+\Delta^*$ exactly at $t=\Delta^*$ (because $\psi(\Delta^*)=1$), the bracket is positive for $t<\Delta^*$ and negative for $t>\Delta^*$: $u_L$ is strictly single-peaked with maximum at $t=\Delta^*$, i.e.,\ at $c_1$ itself. A deviation \emph{across} the opponent to $c_1'=c_2+t$, $t\in(0,1-c_2]$, yields by \cref{eq:u_uvoters} the payoff $u_R(t)=(1-c_2-t)K(t)+\tfrac12G(t)\le u_L(t)$, since $1-c_2<c_2$; hence its supremum is at most $\max_tu_L(t)=u_1(c_1,c_2)$. Finally, the tie $c_1'=c_2$ yields $0<u_1(c_1,c_2)$. So no deviation is profitable, and the profile is an equilibrium.

\emph{Uniqueness.} By part (a), any equilibrium is symmetric with gap solving $\psi(\Delta)=1$, i.e.,\ $\Delta=\Delta^*$.
\end{proof}

\medskip
\noindent\textbf{\Cref{prop:rh} (restated).}\enspace
\emph{Let $x\sim U(0,1)$ and let $k,\widetilde k$ both satisfy the hypotheses of \Cref{thm_uniform_votes}(b), with equilibrium gaps $\Delta^*,\widetilde\Delta^*$. If $\widetilde k/\widetilde K\ge k/K$ pointwise on $(0,1)$, then $\widetilde\Delta^*\ge\Delta^*$. If in addition $\widetilde K(1)\le K(1)$, then $\widetilde K\le K$ pointwise.}

\begin{proof}
\emph{First claim.} Since $k,\widetilde k>0$ on $(0,1)$, the hypothesis gives $\widetilde K(t)/\widetilde k(t)\le K(t)/k(t)$, hence $\widetilde\psi(t)=t+\widetilde K(t)/\widetilde k(t)\le t+K(t)/k(t)=\psi(t)$ for all $t\in(0,1)$, so
\[
\widetilde\psi(\Delta^*)\;\le\;\psi(\Delta^*)=1=\widetilde\psi(\widetilde\Delta^*),
\]
and since $\widetilde\psi$ is strictly increasing, $\Delta^*\le\widetilde\Delta^*$. No condition on $\widetilde K(1)$ is used.

\emph{Second claim.} The reverse hazard rate is the derivative of $\log K$, and $\int_t^1 k/K<\infty$ for $t\in(0,1)$ because $K\ge K(t)>0$ on $[t,1]$ and $\int_t^1 k=K(1)-K(t)$. Hence, for $t\in(0,1)$, using $\widetilde K(1)\le K(1)$,
\[
\log\widetilde K(t)=\log\widetilde K(1)-\int_t^1\frac{\widetilde k(s)}{\widetilde K(s)}\,ds
\;\le\;\log K(1)-\int_t^1\frac{k(s)}{K(s)}\,ds=\log K(t),
\]
so $\widetilde K\le K$ pointwise, i.e.,\ costs under $\widetilde K$ are stochastically higher.
\end{proof}

\medskip
\noindent\textbf{\Cref{prop:margin_uniform} (restated).}\enspace
\emph{Under uniform voters, $k>0$ on $(0,1)$, and alienation with parameter $\alpha\ge0$, the pure Nash equilibria of the margin game are exactly $E=\big[\tfrac{1}{2+\alpha},\tfrac{1+\alpha}{2+\alpha}\big]^{2}$, each with margin $0$, independently of $K$.}

\begin{proof}
We first treat $\alpha>0$; the case $\alpha=0$ is handled at the end. Take any profile with $c_1\le c_2$ and let $G(t)=\int_0^tK$. Candidate $1$'s benefit tent has arms of slope $\alpha$ (outer) and $2+\alpha$ (inner); changing variables along each arm, and subtracting the part of the outer arm that is truncated by the boundary $0$, gives
\[
u_1=\Big(\frac1\alpha+\frac1{2+\alpha}\Big)G(\Delta)-\frac1\alpha G(b_1),
\qquad
b_1:=\big(B_1(0)\big)^+=(\Delta-\alpha c_1)^+ ,
\]
and symmetrically for $u_2$ with $b_2=\big(B_2(1)\big)^+=\big(\Delta-\alpha(1-c_2)\big)^+$
(throughout this proof, $b_1,b_2$ denote the boundary benefits of
\Cref{sec:alien:char} truncated at zero). Hence the margin is
\begin{equation}\label{eq:margin_G}
w_1\;=\;u_1-u_2\;=\;\frac1\alpha\Big(G(b_2)-G(b_1)\Big).
\end{equation}
Since $k>0$ on $(0,1)$, $G$ is strictly increasing on $(0,1]$, so $w_1$ and $b_2-b_1$ have the same sign whenever $\max(b_1,b_2)>0$, and $w_1=0$ iff $b_1=b_2$. In words: a candidate leads exactly when her \emph{opponent's} tent is more severely truncated by the boundary.

\emph{Profiles in $E$ have margin $0$.} If $c_1\ge\tfrac{1}{2+\alpha}$ and $c_2\le\tfrac{1+\alpha}{2+\alpha}$, then
$\Delta-\alpha c_1=c_2-(1+\alpha)c_1\le\tfrac{1+\alpha}{2+\alpha}-\tfrac{1+\alpha}{2+\alpha}=0$, and likewise $\Delta-\alpha(1-c_2)=(1+\alpha)c_2-c_1-\alpha\le0$: both tents are interior, $b_1=b_2=0$, and $w_1=0$.

\emph{When can candidate $1$ secure a positive margin?} By \cref{eq:margin_G}, $w_1>0$ at a profile $(c_1',c_2)$ with $c_1'<c_2$ requires $b_2'>b_1'$, i.e.,\ either (i) $b_1'=0<b_2'$, which reads $c_1'\ge\tfrac{c_2}{1+\alpha}$ together with $c_1'<c_2-\alpha(1-c_2)$, and admits a solution iff $\tfrac{c_2}{1+\alpha}<c_2-\alpha(1-c_2)$, i.e.,\ iff $c_2>\tfrac{1+\alpha}{2+\alpha}$; or (ii) $b_2'>b_1'>0$, which reads $c_1'>1-c_2$ together with $c_1'<\min\big(\tfrac{c_2}{1+\alpha},\,c_2-\alpha(1-c_2)\big)$, and both $1-c_2<\tfrac{c_2}{1+\alpha}$ and $1-c_2<c_2-\alpha(1-c_2)$ are again equivalent to $c_2>\tfrac{1+\alpha}{2+\alpha}$. Thus candidate $1$ can achieve a strictly positive margin at some position left of $c_2$ iff $c_2>\tfrac{1+\alpha}{2+\alpha}$; by the reflection $x\mapsto1-x$, she can achieve one at some position right of $c_2$ iff $c_2<\tfrac{1}{2+\alpha}$. The mirror statements hold for candidate $2$ with respect to $c_1$.

\emph{Profiles in $E$ are equilibria.} Let $(c_1,c_2)\in E$. Then $c_2\in\big[\tfrac{1}{2+\alpha},\tfrac{1+\alpha}{2+\alpha}\big]$, so by the previous paragraph no deviation of candidate $1$ achieves a positive margin; since $w_1=0$ at the profile, no deviation is profitable. Symmetrically for candidate $2$. (Note that $E$ includes tied profiles, which are thus zero-turnout, zero-margin equilibria of the margin game.)

\emph{Profiles outside $E$ are not equilibria.} Consider a profile violating the bounds; after relabeling, either $c_2>\tfrac{1+\alpha}{2+\alpha}$ or $c_1<\tfrac{1}{2+\alpha}$ (or both). Say $c_2>\tfrac{1+\alpha}{2+\alpha}$ (the other case is the mirror image). If $w_1\le0$ at the profile, candidate $1$ deviates into the window of case (i) above and achieves $w_1'>0\ge w_1$: profitable. If $w_1>0$, then $w_2<0$, and candidate $2$ deviates to the mirror position $c_2'=1-c_1$, where the two truncation levels coincide and hence $w_2'=0>w_2$: profitable. In all cases some candidate has a profitable deviation, so the profile is not an equilibrium.

\emph{The case $\alpha=0$.} Here $B_1\equiv\Delta$ to the left of $c_1$, so by \cref{eq:u_uvoters}, $w_1=\big(c_1-(1-c_2)\big)K(\Delta)$, and $E=\{(\tfrac12,\tfrac12)\}$. At $(\tfrac12,\tfrac12)$, a deviation of candidate $1$ to distance $\delta>0$ from the opponent (on either side) yields margin $-\delta K(\delta)<0$, so the median profile is an equilibrium and it is strict. Conversely, at any other profile: if $c_1<c_2$ and $c_1\neq1-c_2$, then $w_1\neq0$ and the trailing candidate deviates to the mirror position to obtain margin $0$; if $c_1=1-c_2<\tfrac12$, candidate $1$ deviates to $c_1'\in(c_1,c_2)$ and obtains $w_1'=(c_1'-c_1)K(c_2-c_1')>0$; and if $c_1=c_2=c\neq\tfrac12$, say $c>\tfrac12$, then candidate $1$ deviates to $c_1'\in(1-c,c)$ and obtains $w_1'=\big(c_1'-(1-c)\big)K(c-c_1')>0$. Hence $\big(\tfrac12,\tfrac12\big)$ is the unique equilibrium.

Finally, all the thresholds above are determined by the benefits $B_i$ alone; $K$ enters only through the strictly increasing $G$, so the set $E$ is independent of $K$.
\end{proof}

\medskip
\noindent\textbf{\Cref{thm:margin_median} (restated).}\enspace
\emph{Let $\alpha=0$, let $v,k$ satisfy \Cref{ass:reg}, and let $c_{\mathrm{med}}=V^{-1}(\tfrac12)$. Then $(c_{\mathrm{med}},c_{\mathrm{med}})$ is the unique Nash equilibrium of the margin game, and every unilateral deviation from it strictly lowers the deviator's margin.}

\begin{proof}
Write $q:=c_{\mathrm{med}}$, which is well-defined, unique, and lies in $(0,1)$ since $V$ is strictly increasing under \Cref{ass:reg}.

\emph{Step 1: any deviation against the median loses strictly.} Consider first a deviation of candidate $1$ to $c_1'<q$, and write $\Delta'=q-c_1'\in(0,1)$ and $m'=\tfrac{c_1'+q}2$. The deviator collects the constant-gain region $[0,c_1']$ (probability $K(\Delta')$ each) plus the inner arm on $[c_1',m']$; the opponent collects the constant-gain region $[q,1]$, of voter mass $1-V(q)=\tfrac12$, plus the inner arm on $[m',q]$. Substituting $s=m'-x$ and $s=x-m'$ in the two inner-arm integrals,
\[
w_1(c_1',q)
=K(\Delta')\Big(V(c_1')-\frac12\Big)
+\int_0^{\Delta'/2}K(2s)\,\big[v(m'-s)-v(m'+s)\big]\,ds .
\]
Bounding the positive part of the integral using $K(2s)\le K(\Delta')$ for $s\le\Delta'/2$,
\[
\int_0^{\Delta'/2}K(2s)\,v(m'-s)\,ds\;\le\;K(\Delta')\,\big(V(m')-V(c_1')\big),
\]
we obtain
\begin{align*}
w_1(c_1',q)\;&\le\;K(\Delta')\Big(V(m')-\frac12\Big)-\int_0^{\Delta'/2}K(2s)\,v(m'+s)\,ds\\
&\le\;K(\Delta')\Big(V(m')-\frac12\Big)\;<\;0,
\end{align*}
where the final strict inequality holds because $K(\Delta')>0$ (as $k>0$ on $(0,1)$ and $0<\Delta'\le q<1$) and $V(m')<V(q)=\tfrac12$ (as $m'<q$ and $v>0$). For a deviation to $c_1'>q$, the mirror computation gives $w_1(c_1',q)\le K(\Delta')\big(\tfrac12-V(m')\big)<0$, now with $m'=\tfrac{c_1'+q}2>q$. Finally, $w_1(q,q)=0$. Hence every unilateral deviation from $(q,q)$ strictly lowers the deviator's margin, and $(q,q)$ is an equilibrium.

\emph{Step 2: uniqueness.} Suppose $(c_1,c_2)$ is an equilibrium with, say, $c_1\neq q$. Deviating to $q$ is available to candidate $2$, and by Step~1 (with the roles of the players reversed), $w_2(c_1,q)=-w_1(c_1,q)>0$; equilibrium therefore requires $w_2(c_1,c_2)\ge w_2(c_1,q)>0$, i.e.,\ $w_1(c_1,c_2)<0$. But candidate $1$ could deviate to $q$: if $c_2\neq q$ then $w_1(q,c_2)>0$ by Step~1, and if $c_2=q$ then $w_1(q,q)=0$; either way candidate $1$'s margin strictly improves, contradicting equilibrium. Hence $c_1=q$, and symmetrically $c_2=q$.
\end{proof}

\clearpage
\section{Existence and Uniqueness of Equilibria}\label{app:exist-uniq}

This appendix collects the existence arguments and the (non-)uniqueness examples referenced in the main text. Recall the distinction stressed in \Cref{sec:base:char}: the balance conditions of \Cref{lem:balance,thm_uniform_costs,thm_uniform_votes} are \emph{necessary} for equilibrium. Below we show that \emph{solutions of the conditions} always exist in the two benchmark settings; whether these solutions are genuine equilibria is a separate question---answered affirmatively under the sufficient conditions of the main text, and negatively in general (see the example below and \Cref{app:counterexample}).

\paragraph{Solutions of the conditions always exist under uniform costs.}
Under \Cref{ass:reg}, $V$ is strictly increasing, so its inverse $V^{-1}$ is well-defined on $[0,1]$. For any $c_1\in\big[0, V^{-1}(\tfrac12)\big]$, define
\[
c_2(c_1)\;:=\;V^{-1}\!\Big(V(c_1)+\tfrac12\Big),
\]
which is well-defined since $V(c_1)\in[0,\tfrac12]$ on this domain. Next, define
\[
f(c_1)\;:=\;V\!\left(\frac{c_1+c_2(c_1)}{2}\right)-2V(c_1).
\]
By construction, $f$ is continuous on $\big[0, V^{-1}(\tfrac12)\big]$. Moreover,
\[
f(0)
=V\!\left(\frac{V^{-1}(\tfrac12)}{2}\right)\ge 0,
\]
and, letting $q:=V^{-1}(\tfrac12)$, so that $c_2(q)=V^{-1}(1)=1$,
\[
f(q)
=V\!\left(\frac{q+1}{2}\right)-1\le 0.
\]
Hence there exists $c_1^\ast\in[0,q]$ such that $f(c_1^\ast)=0$. Setting $c_2^\ast=c_2(c_1^\ast)$ yields
\[
V(c_2^\ast)=V(c_1^\ast)+\tfrac12
\qquad\text{and}\qquad
V\!\left(\frac{c_1^\ast+c_2^\ast}{2}\right)=2V(c_1^\ast),
\]
so $(c_1^\ast,c_2^\ast)$ satisfies the conditions of \Cref{thm_uniform_costs}. Whether it is an equilibrium requires a separate verification---\Cref{app:counterexample} shows it may fail; under the hypotheses of \Cref{cor_quarter} (symmetric case), equilibrium existence is guaranteed.

\paragraph{Solutions of the conditions always exist under uniform voters.}
Define, on $(0,1)$,
\[
f(\Delta)\;:=\;K(\Delta)-(1-\Delta)\,k(\Delta),
\]
which is continuous there by \Cref{ass:reg}. Near $0$, $f<0$ at points arbitrarily close to $0$: if $f\ge0$ on some $(0,t_0)$ with $t_0\le\tfrac12$, then $k/K\le1/(1-t)\le2$ there, so $\log K(t_0)-\log K(t)=\int_t^{t_0}k/K\le2t_0$ for all $t\in(0,t_0)$ (note $K(t_0)>0$ since $k>0$ on $(0,1)$); hence $K(t)\ge K(t_0)e^{-2t_0}>0$ as $t\downarrow0$, contradicting $K(0)=0$ and the continuity of $K$. Near $1$, $f>0$ at points arbitrarily close to $1$: if $f\le0$ on some $(t_1,1)$ with $t_1\in(0,1)$, then $k/K\ge1/(1-t)$ there, so $\log K(t)-\log K(t_1)=\int_{t_1}^{t}k/K\ge\log\frac{1-t_1}{1-t}\to\infty$ as $t\uparrow1$, contradicting $K\le1$. By the intermediate value theorem, $f$ has a root $\Delta^\ast\in(0,1)$. Under the monotonicity condition of \Cref{thm_uniform_votes}(b), this root is unique and yields the (unique) equilibrium. Without that condition, some roots may fail to be equilibria, as the following example shows. (When $k(0^+)=0$, extending $k$ by continuity to $0$ makes $\Delta=0$ a further solution of $K(\Delta)=(1-\Delta)k(\Delta)$; it is not an equilibrium gap, since \Cref{thm_uniform_votes}(a) requires $\Delta\in(0,1)$.)

\paragraph{Non-uniqueness (and non-sufficiency) under uniform voters.}
Consider the cost density
\(
k(\Delta)=8\Delta^2-\frac{16}{3}\Delta+1
\)
(a valid density: $K(1)=1$ and $k\ge\tfrac19>0$, with minimum at $\Delta=\tfrac13$). The equation $K(\Delta)=(1-\Delta)k(\Delta)$ admits three solutions,
\(
\Delta\in\left\{\frac14,\frac12,\frac34\right\}
\)
(the function $\psi$ of \cref{eq:psi} is not monotone for this $k$, so \Cref{thm_uniform_votes}(b) does not apply). The three corresponding symmetric profiles behave very differently:
\begin{itemize}
\item $\Delta=\tfrac14$, profile $\big(\tfrac38,\tfrac58\big)$: a Nash equilibrium (verified by computing global best responses numerically), with $u_1=u_2=\tfrac{131}{2304}\approx0.057$;
\item $\Delta=\tfrac34$, profile $\big(\tfrac18,\tfrac78\big)$: a Nash equilibrium, with $u_1=u_2=\tfrac{27}{256}\approx0.105$;
\item $\Delta=\tfrac12$, profile $\big(\tfrac14,\tfrac34\big)$: \textbf{not} a Nash equilibrium. Here $u_1=\tfrac14K\big(\tfrac12\big)+\tfrac12\int_0^{1/2}K=\tfrac{5}{72}\approx0.069$, but candidate $1$'s best response to $c_2=\tfrac34$ is $c_1'\approx0.473$, yielding a payoff of $\approx0.073$. The profile solves both first-order conditions, yet it is not even a local equilibrium: \Cref{n:rem:chi-criterion} shows this middle root to be a local \emph{minimum} of the deviator's payoff, so deviating either way pays.
\end{itemize}
Thus equation \eqref{eq:uv_foc} may have multiple roots; the set of equilibria may be a strict subset of these roots; and, in particular, the necessary condition of \Cref{thm_uniform_votes}(a) cannot be reversed without an assumption such as the monotonicity of $\psi$. Whether at least one root is always an equilibrium (i.e.,\ whether a pure equilibrium always exists under uniform voters) is open.

For comparison, the simple sufficient conditions for the monotonicity of
$\psi$ (hence for uniqueness \emph{and} sufficiency) are the two that give
super-regularity, $k$ non-increasing or log-concave (\eqref{eq:psi} and
\Cref{thm_uniform_votes}(b)); \Cref{fig:rhr} visualizes the
resulting unique crossing for three such families.

\paragraph{Non-uniqueness under uniform costs.}
Uniqueness can also fail when the voter density (rather than the cost density) is non-uniform; \Cref{app:multi-eq} works out a symmetric voter density that admits three distinct equilibria, two of which are asymmetric.

\paragraph{Mixed equilibria.}
Finally, we note that although the existence of \emph{pure} equilibria is open in general, the payoffs $u_i$ are continuous on the compact strategy space $[0,1]^2$ (continuity at tied profiles holds because vote shares vanish as the candidates approach each other), so a possibly-mixed Nash equilibrium always exists by Glicksberg's fixed-point theorem \cite{glicksberg1952further}.

\section{Power Costs against a Non-Uniform Electorate}\label{app:power-nonuniform}

This appendix records a case that no sufficiency result in the main text
reaches, and where the conditions nonetheless appear to deliver equilibria.

\begin{remark}[Power costs against a non-uniform electorate]%
\label{rem:power-nonuniform}\stat{numerical observation}
Every sufficiency result above misses this family. For $K(t)=t^{p}$ the cost
density is $k(t)=p\,t^{p-1}$, which is increasing when $p>1$; and when $p<1$ it
is non-increasing, but $\rho_k:=k(0^+)/k(1^-)$ is infinite. So the criterion
$\rho_v\rho_k\le4$ of \Cref{n:prop:A4} reaches the family only at $p=1$, where
it is \Cref{cor_quarter}, and \Cref{thm_uniform_votes} needs a uniform
electorate.

For a symmetric electorate the gap is narrower than that makes it sound. By
\Cref{n:lem:reflection}, own-side optimality alone makes a mirror profile an
equilibrium whatever the cost distribution, so the cross-over deviations are
already handled and what is unproved is exactly this: that the stationary point
maximizes $u_1(\cdot,c_2)$ on $[0,c_2)$. On this family $\psi$ is linear,
$\psi(t)=\big(1+\tfrac1p\big)t$, which is as favourable a setting as the paper
offers.

Numerically the profile is an equilibrium throughout. Scanning $(0,\tfrac12)$
for symmetric stationary profiles, each of the five symmetric electorates below
has exactly one at each $p\in\{\tfrac12,1,2,3\}$, and in all twenty cases it is
a Nash equilibrium; the largest profitable deviation found was
$1.1\times10^{-16}$. The quantile $V(c_1^*)$ it sits at:
\begin{center}\small\setlength{\tabcolsep}{8pt}
\begin{tabular}{@{}lcccc@{}}
\toprule
electorate & $p=\tfrac12$ & $p=1$ & $p=2$ & $p=3$\\
\midrule
uniform                       & $0.33333$ & $0.25$ & $0.16667$ & $0.12500$\\
$\mathrm{Beta}(1.5,1.5)$      & $0.33386$ & $0.25$ & $0.16334$ & $0.11859$\\
$\mathrm{Beta}(2,2)$          & $0.33409$ & $0.25$ & $0.16176$ & $0.11541$\\
$\mathrm{Beta}(3,3)$          & $0.33430$ & $0.25$ & $0.16023$ & $0.11227$\\
hollow cosine (\Cref{fig:quartiles}) & $0.32015$ & $0.25$ & $0.18595$ & $0.15077$\\
\bottomrule
\end{tabular}
\end{center}
At $p=1$ it is $\tfrac14$ for every electorate, which is \Cref{cor_quarter}.
Away from $p=1$ it moves, but by less than $0.03$ across this range, and the
departure from the uniform-voter value $\tfrac1{2(1+p)}$ of
\Cref{rem:quantile-sweep} is signed by the shape of $v$: a peaked electorate
raises the quantile when $p<1$ and lowers it when $p>1$, a hollow one does the
reverse, and both effects vanish at $p=1$ and grow with $|p-1|$ and with the
distance of $v$ from uniform. Distribution-freeness is therefore not merely
lost away from $p=1$; it degrades continuously, and $p=1$ is the point where it
is exact. \rmr{good to know. I would keep it for now only in the appendix of the full version. for future work perhaps we can bound the quantile in terms of $p$ and some parameter of the distribution? }\guy{Claude: moved here from \Cref{sec:base:polar}, with a one-sentence pointer left in its place. On the future-work question: the table already suggests the shape of such a bound---the departure from $\tfrac1{2(1+p)}$ vanishes at $p=1$, grows with $|p-1|$, and grows with the distance of $v$ from uniform---so a bound in $|p-1|$ and $\rho_v$ looks plausible, but I have not tried to prove one.}
\end{remark}

\section{Example with Multiple Equilibria (Uniform Costs)}\label{app:multi-eq}

Consider the continuous, symmetric density $v:[0,1]\to\mathbb{R}_{\ge 0}$ given by
\[
v(x)=
\begin{cases}
\dfrac{5}{7}, & 0\le x\le \dfrac{19}{50},\\[8pt]
\dfrac{250x-90}{7}, & \dfrac{19}{50}\le x\le \dfrac{21}{50},\\[10pt]
\dfrac{15}{7}, & \dfrac{21}{50}\le x\le \dfrac{29}{50},\\[8pt]
\dfrac{160-250x}{7}, & \dfrac{29}{50}\le x\le \dfrac{31}{50},\\[10pt]
\dfrac{5}{7}, & \dfrac{31}{50}\le x\le 1.
\end{cases}
\]
Under uniform costs, the conditions of \Cref{thm_uniform_costs} admit three distinct solutions,
\[
(c_1,c_2)\in\left\{
\left(\frac{3}{16},\frac{9}{16}\right),\;
\left(\frac{7}{20},\frac{13}{20}\right),\;
\left(\frac{7}{16},\frac{13}{16}\right)
\right\},
\]
and---in contrast to the cost-side example of \Cref{app:exist-uniq}---we verified numerically that here \emph{all three} profiles are genuine Nash equilibria, with payoffs $(u_1,u_2)\approx(0.075,0.212)$, $(0.106,0.106)$ and $(0.212,0.075)$, respectively. Note that although $v$ is symmetric around $\tfrac12$, two of the three equilibria are asymmetric mirror images of one another. (This density has $\rho_v=3\le4$, but the outer equilibria are not quartile profiles: \Cref{cor_quarter} identifies \emph{some} equilibria, not necessarily all of them.)

\section{Stationary Profiles Need Not Be Equilibria (Uniform Costs)}\label{app:counterexample}

We now show that the conditions of \Cref{thm_uniform_costs} are not sufficient: a profile can satisfy them and still fail to be a Nash equilibrium. Let $h:[0,1]\to\mathbb R_{\ge0}$ be the continuous piecewise-linear function taking the constant values
\[
6 \ \text{on}\ [0,0.08],\quad
\tfrac1{20} \ \text{on}\ [0.10,0.20],\quad
5 \ \text{on}\ [0.22,0.34],\quad
\tfrac1{20} \ \text{on}\ [0.36,0.44],
\]
\[
3 \ \text{on}\ [0.46,0.62],\quad
\tfrac1{20} \ \text{on}\ [0.64,0.75],\quad
2 \ \text{on}\ [0.77,1],
\]
and interpolating linearly on the six gaps of width $0.02$. Then $\int_0^1h=\tfrac{911}{400}$, so $v:=\tfrac{400}{911}\,h$ is a continuous density, strictly positive on $[0,1]$ (\Cref{ass:reg} holds), describing an electorate with four ideological ``clusters'' of decreasing mass.

Solving the conditions of \Cref{thm_uniform_costs} numerically yields exactly three solutions:
\[
(c_1,c_2)\;\approx\;(0.226,\,0.621),\qquad(0.232,\,0.709),\qquad(0.248,\,0.798).
\]
Computing global best responses, the first and the third profiles are Nash equilibria. The middle profile is \emph{not}: candidate $2$'s payoff there is $u_2\approx0.1353$, but relocating to $c_2'\approx0.789$ (adjacent to the fourth cluster) yields $\approx0.1365$---a profitable deviation to a different local maximum of her payoff. Her platform is therefore a local maximum that is not a global one, so the profile is a local equilibrium without being a Nash equilibrium: the opposite failure to the cost-side example of \Cref{app:exist-uniq}, where the middle root is not even a local equilibrium. Note that $\rho_v=120$ here, far above the bound of \Cref{prop:fourflat}; with such uneven densities the candidates' payoffs are multi-peaked, and stationarity no longer implies global optimality.

\section{Turnout under Linear Costs}\label{app:turnout}
We compute each candidate's vote share at a symmetric profile under uniform voters and the linear cost family $k_\vartheta(x)=\vartheta x+\big(1-\tfrac\vartheta2\big)$, $\vartheta\in[-2,2]$, for which
\[
K_\vartheta(\Delta)=\frac \vartheta2\Delta^2+\Big(1-\frac \vartheta2\Big)\Delta,
\qquad
\int_0^{\Delta}K_\vartheta(y)\,dy=\frac \vartheta6\Delta^3+\frac{2-\vartheta}{4}\Delta^2 .
\]
At the symmetric profile with gap $\Delta$ (so $c_1=\tfrac{1-\Delta}2$), \cref{eq:u_uvoters} gives
\begin{align*}
u_1(\Delta;\vartheta)
&=c_1K_\vartheta(\Delta)+\frac12\int_0^{\Delta}K_\vartheta(y)\,dy\\
&=\frac{1-\Delta}{2}\left[\frac \vartheta2\Delta^2+\Big(1-\frac \vartheta2\Big)\Delta\right]
+\frac \vartheta{12}\Delta^3+\frac{2-\vartheta}{8}\Delta^2\\
&=-\frac \vartheta6\Delta^3+\Big(\frac{3\vartheta}8-\frac14\Big)\Delta^2+\Big(\frac12-\frac \vartheta4\Big)\Delta .
\end{align*}
By symmetry $u_2=u_1$, so total turnout at the profile is
\begin{equation}\label{eq:turnout_linear}
T(\Delta;\vartheta)\;=\;2\,u_1(\Delta;\vartheta)
\;=\;-\frac \vartheta3\Delta^3+\Big(\frac{3\vartheta}4-\frac12\Big)\Delta^2+\Big(1-\frac \vartheta2\Big)\Delta .
\end{equation}
(As a check: at $\vartheta=0$, $\Delta=\tfrac12$ we get $u_1=\tfrac3{16}$ and $T=\tfrac38$, matching \Cref{cor_uu}.)

\section{Turnout Comparative Statics: A Worked Example}\label{app:whatcamefirst}

This appendix expands the causality comparison of \Cref{sec:base:polar}, using the turnout formula \eqref{eq:turnout_linear} of \Cref{app:turnout}. Consider a uniform voter distribution and a uniform cost distribution, for which the unique equilibrium is at $\left(\tfrac14,\tfrac34\right)$ with overall turnout $\tfrac38$. First suppose the candidates remain \emph{fixed} at $\left(\tfrac14,\tfrac34\right)$ (i.e.,\ $\Delta=\tfrac12$) while voting costs increase within the linear cost family, i.e.,\ $\vartheta$ increases over $[-2,2]$. By \cref{eq:turnout_linear}, turnout is then $T\big(\tfrac12;\vartheta\big)=\tfrac38-\tfrac{5\vartheta}{48}$, declining from $\tfrac{14}{24}\cong0.58$ at $\vartheta=-2$, to $\tfrac{9}{24}=0.375$ under uniform costs ($\vartheta=0$), and further to $\tfrac{4}{24}\cong0.17$ at $\vartheta=2$---a baseline effect driven solely by rising costs and decreasing engagement.

Next, instead of holding the candidates fixed, consider turnout \emph{in equilibrium} as $\vartheta$ varies over $[-2,2]$: within the log-concave linear family the equilibrium gap (\Cref{thm_uniform_votes}(b); $\Delta^\ast(\vartheta)=\tfrac{2(\vartheta-1)+\sqrt{(\vartheta-1)^2+3}}{3\vartheta}$ for $\vartheta\neq0$)  $\Delta^\ast(\vartheta)$ widens from $\approx0.42$ at $\vartheta=-2$ to $\tfrac23$ at $\vartheta=2$, and equilibrium turnout equals $T\big(\Delta^\ast(\vartheta);\vartheta\big)$. Although turnout still decreases, the decline is more \emph{modest} (from $\approx0.54$ at $\vartheta=-2$, through $0.375$ at $\vartheta=0$, to $\tfrac{20}{81}\cong0.25$ at $\vartheta=2$), because the candidates become more polarized and a wider gap raises each voter's benefit from participating. Thus, at least in this scenario, abstention increases polarization, whereas polarization has \emph{the opposite} effect on abstention. Whether this pattern persists under other distributions is left to future work.

%%%%%%%%%%%%%%%%%%%%%%%%%%%%%%%%%%%%%%%%%%%%%%%%%%%%%%%%%%%%%%%%%%%%%%%%%%%%%%
%% Supplementary results, batch 1 (labels prefixed n:).
%% Supplementary results, batch 1. Labels prefixed n:.

\section{Preliminaries for all $(\alpha,\beta)$}\label{n:sec:prelim}

\paragraph{Standing notation and conventions.}
Voters $x\in[0,1]$ have CDF $V$ with density $v$; voting costs
$\kappa\in[0,1]$ are i.i.d.\ with CDF $K$ and density $k$, independent of
position.  We assume \Cref{ass:reg} throughout.  We allow $K(1)\le1$ (a mass $1-K(1)$ of voters never votes); everything below
uses only $K(1)>0$.  Candidates $i\in\{1,2\}$ choose $c_i\in[0,1]$; write
$d_i(x)=|x-c_i|$ and, for $\alpha\ge0$,
$B_i(x)=d_{-i}(x)-(1+\alpha)d_i(x)$, so that $i$'s vote share is
$u_i=\int_0^1K\!\big(B_i(x)^+\big)v(x)\,dx$, and candidate $i$ maximizes
$U_i=u_i-\beta u_{-i}$ with $\beta\in[0,1]$.  When $c_1<c_2$ we write
$\Delta=c_2-c_1$, $m=\tfrac{c_1+c_2}2$, and
\[
\begin{gathered}
\lambda:=2+\alpha,\qquad
\zin{1}=c_1+\tfrac\Delta\lambda,\qquad \zout{1}=c_1-\tfrac\Delta\alpha,\\
\zin{2}=c_2-\tfrac\Delta\lambda,\qquad \zout{2}=c_2+\tfrac\Delta\alpha,
\end{gathered}
\]
for the inner and outer zeros of the two benefit tents ($\zout{1}=-\infty$,
$\zout{2}=+\infty$ when $\alpha=0$), and
\[
b_1=B_1(0)=\Delta-\alpha c_1,\qquad
b_2=B_2(1)=\Delta-\alpha(1-c_2)
\]
for the boundary benefits; candidate $1$'s tent is truncated by the boundary
(equivalently, she has no \emph{dead zone} $[0,\zout{1}]$ of unreachable voters on
her flank) iff $b_1\ge0$, i.e.\ $\zout{1}\le0$.  Two identities used repeatedly:
with $b_s(\Delta):=\tfrac{\lambda\Delta-\alpha}2=\Delta-\tfrac\alpha2(1-\Delta)$,
\begin{equation}\label{n:eq:bident}
b_1+b_2=2\,b_s(\Delta),
\qquad
b_2-b_1=\alpha\,(c_1+c_2-1).
\end{equation}
Finally $G(t)=\int_0^tK$, $r=k/K$ (the reverse hazard rate),
$\psi(t)=t+K(t)/k(t)$, $\rho_v=\sup v/\inf v\in[1,\infty]$,
$\rho_k=\sup k/\inf k$, and $\Delta_0:=\tfrac\alpha\lambda$, the most
polarized symmetric gap inside the margin-game band
$E(\alpha)=\big[\tfrac1\lambda,\tfrac{1+\alpha}\lambda\big]^2$.

Two elementary facts are used without further comment.  First, since
$K(0)=0$ and $K>0$ on $(0,1]$,
\begin{equation}\label{n:eq:rdiv}
\int_{0^+}^{t_0}r(s)\,ds
=\lim_{t\downarrow0}\big(\log K(t_0)-\log K(t)\big)=+\infty
\qquad(t_0\in(0,1)),
\end{equation}
and symmetrically $\int_{t}^{1^-}r=\log K(1)-\log K(t)<\infty$.  Second, for
every $x$ at most one of $B_1(x),B_2(x)$ is strictly positive, because
$B_1+B_2=-\alpha(d_1+d_2)\le0$; hence total turnout satisfies
$T=u_1+u_2\le K(1)\le1$ at \emph{every} profile.

%%%%%%%%%%%%%%%%%%%%%%%%%%%%%%%%%%%%%%%%%%%%%%%%%%%%%%%%%%%%%%%%%%%%%%%%%%%%%%
\medskip\noindent
The two lemmas of this section extend \Cref{lem:interior} and the balance
conditions of \Cref{lem:balance} from the baseline $(\alpha,\beta)=(0,0)$ to
the whole parameter square, and are the entry point for every result below.

\begin{lemma}[Interiority and no ties, extended]\label{n:lem:interior}
Let \Cref{ass:reg} hold and $\beta\in[0,1)$.  Then every pure Nash
equilibrium is untied and interior ($c_1\ne c_2$ and, after labeling
$c_1<c_2$, $0<c_1<c_2<1$) in each of the following cases:
\textup{(a)} $\beta=0$, any $\alpha\ge0$ and any $(v,k)$;
\textup{(b)} $\alpha=0$, any $(v,k)$;
\textup{(c)} uniform voters $x\sim U(0,1)$, any $\alpha\ge0$ and any $k$.
\end{lemma}

\begin{proof}
\emph{No boundary positions (all three cases).}
Suppose $c_1=0<c_2$.  By \Cref{n:lem:master} below, the right derivative of
$u_1$ at $c_1=0$ is $(1+\alpha)\int_0^{\zin{1}}\!k(B_1)v\,dx$, which is positive (the left
integral in \cref{n:eq:master} is empty). At the same time
$\partial u_2/\partial c_1$ equals
$-\int_{\zin{2}}^{\zout{2}\wedge1}\!k(B_2)v\,dx$, which is at most
$0$.  Hence
$\partial U_1/\partial c_1=\partial u_1/\partial c_1-\beta\,\partial
u_2/\partial c_1>0$: a small move to the right strictly increases $U_1$, so
$c_1=0$ is not a best response.  The mirror argument rules out $c_2=1$.

\emph{No ties, case (a).}  At a tie both payoffs vanish, and for $\beta=0$ the
deviation argument of the baseline lemma applies verbatim: any $c_1'$ strictly
between the tie point and the nearer end of $(0,1)$ yields
$u_1(c_1',c)>0=U_1(c,c)$.

\emph{No ties, case (b).}  Let $\alpha=0$ and $c_1=c_2=c$.  For
$\varepsilon>0$ small consider the deviation $c_1'=c-\varepsilon$ (possible if
$c>0$).  Every voter in $[0,c-\varepsilon]$ gains exactly $\varepsilon$ from
the deviator, and every voter gains at most $\varepsilon$ from either
candidate, so with $m'=c-\tfrac\varepsilon2$,
\[
u_1\ge K(\varepsilon)\,V(c-\varepsilon),
\qquad
u_2\le K(\varepsilon)\,\big(1-V(m')\big),
\]
whence
$U_1\ge K(\varepsilon)\big[V(c-\varepsilon)-\beta\big(1-V(m')\big)\big]$.
As $\varepsilon\downarrow0$ the bracket tends to
$(1+\beta)\big[V(c)-\tfrac\beta{1+\beta}\big]$, which is positive whenever
$V(c)>\tfrac\beta{1+\beta}$; since $K(\varepsilon)>0$ for
$\varepsilon\in(0,1)$, the deviation is then profitable for small
$\varepsilon$.  If instead $V(c)\le\tfrac\beta{1+\beta}<\tfrac1{1+\beta}$, the
mirror deviation $c_1'=c+\varepsilon$ gives a bracket tending to
$(1+\beta)\big[1-V(c)-\tfrac\beta{1+\beta}\big]>0$.  (At $c=0$ only the second
option is available and applies since $V(0)=0$; at $c=1$ only the first, since
$V(1)=1$.)

\emph{No ties, case (c).}  Let $v\equiv1$ and $c_1=c_2=c$; by symmetry assume
$c>0$ and deviate to $c_1'=c-\varepsilon$ with
$\varepsilon<\min\{c,\tfrac{\alpha c}{1+\alpha}\}$ (the case $c=0$ is the mirror
image; for $\alpha=0$ the claim is case (b)).  By the
closed form \eqref{n:eq:uV} below, the deviator's tent is interior
($b_1'=\varepsilon-\alpha(c-\varepsilon)<0$), so
$u_1=A\,G(\varepsilon)$ with $A=\tfrac1\alpha+\tfrac1\lambda$, while the
opponent's payoff is $u_2=A\,G(\varepsilon)-\tfrac1\alpha
G\big((\varepsilon-\alpha(1-c))^+\big)\le A\,G(\varepsilon)$.  Hence
$U_1\ge(1-\beta)A\,G(\varepsilon)>0=U_1(c,c)$.

\end{proof}

We record for later use that \Cref{n:lem:interior} fails at $\beta=1$: tied
profiles are equilibria of the margin game (e.g.\ the median profile at
$\alpha=0$), which is precisely why all statements below assume $\beta<1$ and
treat $\beta=1$ separately.

\begin{lemma}[Master payoff derivative]\label{n:lem:master}
Let \Cref{ass:reg} hold, $\alpha\ge0$, $\beta\in[0,1]$, and fix
$0<c_1<c_2<1$.  Assume the right-hand side of \cref{n:eq:master} below is
finite at $c_1$; this holds automatically whenever $k$ is bounded, or $v$ is
locally bounded at $\{0,1\}$, or the tent edges avoid the boundary
($\zout{1}\neq0$, $\zout{2}\neq1$)---one of which is the case in every application in
this paper.  Then $u_1$, $u_2$, and hence $U_1$, are continuously
differentiable in $c_1$ there, and
\begin{equation}\label{n:eq:master}
\frac{\partial U_1}{\partial c_1}
=(1+\alpha)\left[\int_{c_1}^{\zin{1}}k(B_1)\,v\,dx-\int_{\zout{1}\vee0}^{c_1}k(B_1)\,v\,dx\right]
+\beta\int_{\zin{2}}^{\zout{2}\wedge1}k(B_2)\,v\,dx ,
\end{equation}
where for $\alpha=0$ one reads $\zout{1}\vee0=0$, $\zout{2}\wedge1=1$, and $\zin{1}=\zin{2}=m$.
The formula for $\partial U_2/\partial c_2$ is the mirror image.  In
particular, at an untied interior equilibrium with $\beta<1$ both first-order
conditions $\partial U_1/\partial c_1=\partial U_2/\partial c_2=0$ hold.
\end{lemma}

\begin{proof}
By the layer-cake representation, using $K(t)=\int_0^1k(s)\mathbf1\{s<t\}ds$,
\[
u_1=\int_0^1K\big(B_1(x)^+\big)v(x)\,dx
=\int_0^1 k(s)\,\mu\big(\{B_1>s\}\big)\,ds,
\quad \mu(dx)=v(x)\,dx .
\]
The super-level set $\{B_1>s\}$ is the interval
$\big(\zout{1}(s)\vee0,\;\zin{1}(s)\big)$ for $s\in(0,\Delta)$,
where, from $B_1=c_2+(1+\alpha)c_1-\lambda x$ on the inner arm and
$B_1=c_2-(1+\alpha)c_1+\alpha x$ on the outer arm,
\[
\zin{1}(s)=\frac{c_2+(1+\alpha)c_1-s}{\lambda},
\quad
\zout{1}(s)=\frac{s-c_2+(1+\alpha)c_1}{\alpha}
\;\;(=-\infty\text{ if }\alpha=0).
\]
Both endpoints are affine in $c_1$ with velocities
$\tfrac{1+\alpha}\lambda$ and $\tfrac{1+\alpha}\alpha$, and both lie in the
open interval $(0,c_2)\subset(0,1)$ for $s$ in the relevant ranges, where $V$
is continuously differentiable ($v$ continuous).  Hence for each fixed $s$,
$c_1\mapsto\mu(\{B_1>s\})=V(\zin{1}(s))-V(\zout{1}(s)\vee0)$
is differentiable in $c_1$ except at the single value of $s$ where
$\zout{1}(s)=0$ (a null set in $s$), with derivative
$\tfrac{1+\alpha}\lambda v(\zin{1}(s))
-\tfrac{1+\alpha}\alpha v(\zout{1}(s))\mathbf1\{\zout{1}(s)>0\}$.
The difference quotients are dominated, locally uniformly in $c_1$, by
$k(s)$ times the supremum of $v$ over a compact neighborhood of the swept
intervals, and the resulting dominating function is integrable in $s$
precisely when the right-hand side of \cref{n:eq:master} is finite (change
variables $s=B_1(x)$ back: the two $s$-integrals become the two $x$-integrals
of \cref{n:eq:master}); under the standing finiteness assumption dominated
convergence applies and yields
\[
\frac{\partial u_1}{\partial c_1}
=(1+\alpha)\left[\int_{c_1}^{\zin{1}}\!k(B_1)v\,dx
-\int_{\zout{1}\vee0}^{c_1}\!k(B_1)v\,dx\right],
\]
with \emph{no boundary terms}: the moving tent edges $\zin{1}$ and $\zout{1}$ carry
$K(B_1)=K(0)=0$, and the kink of $B_1$ at $x=c_1$ contributes nothing since
$B_1$ is continuous there.  Continuity of the derivative in $c_1$, including
across the regime switch $\zout{1}=0$, follows because the outer-arm integral is
continuous there: even when $k$ is unbounded at $0^+$, the vanishing piece
obeys
$\int_{\zout{1}}^{\zout{1}+\delta}k(B_1)v
\le\sup_{[\zout{1},\zout{1}+\delta]}v\cdot\tfrac1\alpha K(\alpha\delta)\to0$
by the change of variables $t=B_1(x)$ (the local boundedness of $v$ near the
switch point is part of the standing assumption).

For $u_2$: on candidate $2$'s support $[\zin{2},\zout{2}\wedge1]$ (which lies in
$(c_1,1]$) we have $d_1=x-c_1$, so $\partial B_2/\partial c_1=-1$ pointwise,
and the same argument gives
$\partial u_2/\partial c_1=-\int_{\zin{2}}^{\zout{2}\wedge1}k(B_2)v\,dx$, again with no
boundary terms ($B_2=0$ at $\zin{2}$ and at $\zout{2}$ when $\zout{2}<1$; if $\zout{2}\ge1$ the
limit $1$ is fixed).  Combining, and noting
$U_1=u_1-\beta u_2$, yields \cref{n:eq:master}.  The first-order conditions
at an equilibrium follow since, by \Cref{n:lem:interior} (whose cases cover
every configuration in which we invoke them) the equilibrium is untied and
interior, and $U_1(\cdot,c_2)$ is differentiable there.

\end{proof}

Specializing \cref{n:eq:master} to $\alpha=0$ gives, for general $(v,k)$,
\begin{equation}\label{n:eq:master0}
\frac{\partial U_1}{\partial c_1}
=\int_{c_1}^{m}k(B_1)\,v\,dx-k(\Delta)V(c_1)
+\beta\int_{m}^{1}k(B_2)\,v\,dx ,
\end{equation}
the $\beta$-generalization of the baseline first-order condition \eqref{eq:du1_general}: the
opponent's marginal voters at the midpoint now enter with weight $-\beta$.

Finally, we record the closed form for uniform voters used throughout
\Cref{n:sec:alien,n:sec:comp}.  For $v\equiv1$ and any $K$, changing variables $t=B_1(x)$
along each arm of the tent (slopes $\alpha$ and $\lambda$) and subtracting the
part of the outer arm truncated by the boundary gives, for $\alpha>0$,
\begin{equation}\label{n:eq:uV}
u_1=A\,G(\Delta)-\tfrac1\alpha\,G\big(b_1^+\big),
\qquad A:=\tfrac1\alpha+\tfrac1\lambda=\tfrac{2(1+\alpha)}{\alpha\lambda},
\end{equation}
and $u_1=c_1K(\Delta)+\tfrac12G(\Delta)$ for $\alpha=0$; the formulas for
$u_2$ replace $b_1$ by $b_2$ and $c_1$ by $1-c_2$.  Substituting into
\cref{n:eq:master} (each integral is $\tfrac1\lambda K(\Delta)$,
$\tfrac1\alpha[K(\Delta)-K(b_1^+)]$, and
$A\,K(\Delta)-\tfrac1\alpha K(b_2^+)$ respectively) yields, for $\alpha>0$ and
uniform voters,
\begin{equation}\label{n:eq:DuV}
\alpha\,\frac{\partial U_1}{\partial c_1}
=-(1-\beta)\,\alpha A\,K(\Delta)
+(1+\alpha)\,K\big(b_1^+\big)-\beta\,K\big(b_2^+\big).
\end{equation}

%%%%%%%%%%%%%%%%%%%%%%%%%%%%%%%%%%%%%%%%%%%%%%%%%%%%%%%%%%%%%%%%%%%%%%%%%%%%%%
\section{Results for the baseline ($\alpha=\beta=0$): a reflection lemma}
\label{n:sec:base}

The sufficiency results for uniform costs (\Cref{prop:fourflat,cor_quarter})
(density ratio $\rho_v\le4$) verify two \emph{cross-over} deviations by an
explicit computation available only at quartile profiles.  The following
lemma removes both the computation and the restriction to uniform costs: for
symmetric electorates and vote-share maximizers, cross-over deviations are
\emph{never} binding.  It is stated for all $\alpha\ge0$ because it is used
again in \Cref{n:sec:alien}; the restriction to $\beta=0$ is essential
(\Cref{n:rem:reflection-beta}).

\begin{lemma}[Reflection lemma]\label{n:lem:reflection}
Let $\beta=0$, $\alpha\ge0$, let $K$ be arbitrary (\Cref{ass:reg}), and let
$v$ be symmetric about $\tfrac12$.  Let $c_1\in(0,\tfrac12)$ and
$c_2=1-c_1$.  If $c_1$ maximizes $u_1(\cdot,c_2)$ over $[0,c_2)$, then
$(c_1,c_2)$ is a Nash equilibrium.
\end{lemma}

\begin{proof}
By symmetry of $v$ and of the profile it suffices to show candidate $1$ has
no profitable deviation.  Own-side deviations $c_1'\in[0,c_2)$ are unprofitable
by hypothesis.  The tie $c_1'=c_2$ yields $0<u_1(c_1,c_2)$ (the benefit
$B_1$ is positive on a nondegenerate interval around $c_1$ and $v>0$ there).
It remains to rule out cross-over deviations $c_1'\in(c_2,1]$.  We use two
observations.

\emph{(i) Reflection.}  For any positions $a,o$, the change of variables
$y=1-x$ together with $v(1-y)=v(y)$ gives
$u_{\mathrm{dev}}(a,o)=u_{\mathrm{dev}}(1-a,1-o)$, where
$u_{\mathrm{dev}}(a,o)$ denotes the vote share of a candidate at $a$ facing an
opponent at $o$.  Applying this with $a=c_1'\in(c_2,1]$ and $o=c_2$, and using
$1-c_2=c_1$:
\[
u_1(c_1',c_2)=u_{\mathrm{dev}}\big(1-c_1',\,c_1\big),
\qquad 1-c_1'\in[0,c_1).
\]
Cross-over deviations against $c_2$ are thus payoff-equivalent to own-side
deviations against an opponent located at $c_1$.

\emph{(ii) Monotonicity in the opponent's position.}  For $a<o\le\widetilde o$
and every $x$,
$B(x;a,o)^+\le B(x;a,\widetilde o)^+$, where $B(x;a,o)=|x-o|-(1+\alpha)|x-a|$.
Indeed, if $B(x;a,o)\le0$ there is nothing to prove; if $B(x;a,o)>0$ then
$|x-o|>(1+\alpha)|x-a|\ge|x-a|$, so $x$ is strictly closer to $a$ than to
$o$, forcing $x<\tfrac{a+o}2<o\le\widetilde o$ and hence
$|x-\widetilde o|=\widetilde o-x>o-x=|x-o|$, so
$B(x;a,\widetilde o)>B(x;a,o)>0$.  Since $K$ is non-decreasing, integrating
against $v$ gives $u_{\mathrm{dev}}(a,o)\le u_{\mathrm{dev}}(a,\widetilde o)$.

Combining, for every $c_1'\in(c_2,1]$, with $a=1-c_1'<c_1\le c_2$,
\[
u_1(c_1',c_2)
\stackrel{\text{(i)}}{=}u_{\mathrm{dev}}(a,c_1)
\stackrel{\text{(ii)}}{\le}u_{\mathrm{dev}}(a,c_2)
\le\sup_{c_1''\in[0,c_2)}u_1(c_1'',c_2)
=u_1(c_1,c_2),
\]
the last equality by hypothesis.  No deviation is profitable.
\end{proof}

\begin{remark}\label{n:rem:reflection-beta}
\emph{(a)}  At the quartile profile of a symmetric electorate with uniform
costs, \Cref{n:lem:reflection} recovers the cross-over verification in the
proof of \Cref{cor_quarter} (\Cref{app:omitted-proofs}) without any
computation: under $\rho_v\le4$ the own-side payoff is concave
(\Cref{prop:fourflat}) and stationary at $c_1$, hence globally maximized
there, and the lemma finishes the proof.  The explicit bounds there
($u_1\ge\tfrac1{32s}$ etc.) are retained only for the asymmetric case of
\Cref{cor_quarter}, where the reflection argument is unavailable.
\emph{(b)}  The lemma fails for $\beta>0$: step (ii) compares the
\emph{deviator's} share pointwise, but moving the opponent also changes the
\emph{opponent's} share $u_2$, in a direction that is not signed in general,
so $U_1=u_1-\beta u_2$ is not monotone in the opponent's position.  The
$\beta>0$ cross-over checks in \Cref{n:sec:comp} are therefore carried out explicitly.
\end{remark}

%%%%%%%%%%%%%%%%%%%%%%%%%%%%%%%%%%%%%%%%%%%%%%%%%%%%%%%%%%%%%%%%%%%%%%%%%%%%%%
\section{Results for alienation ($\alpha>0$, $\beta=0$)}\label{n:sec:alien}

\subsection{A sufficiency threshold under uniform costs}

For $\alpha=0$, \Cref{prop:fourflat} shows that $\rho_v\le4$ makes each candidate's
own-side payoff concave, reducing equilibrium verification to the balance
conditions plus two cross-over checks.  The next theorem gives the
$\alpha$-analogue and answers the question of whether alienation helps or
hurts this program: the threshold
\begin{equation}\label{n:eq:Fdef}
F(\alpha):=\min\{f(\alpha),g(\alpha)\},
\quad
f(\alpha):=\frac{2(2+\alpha)}{1+\alpha},
\quad
g(\alpha):=\frac{2+\alpha}{\alpha}
\;\;(g(0):=\infty),
\end{equation}
is \emph{strictly decreasing}, from $F(0)=4$ through $F(1)=3$ (where $f=g$)
to $1$ as $\alpha\to\infty$.  Alienation \emph{shrinks} the sufficiency
umbrella.

\begin{theorem}[Single-peakedness under uniform costs]\label{n:thm:A3}
Let $\kappa\sim U(0,1)$, $\beta=0$, $\alpha>0$, and $\rho_v\le F(\alpha)$.  Fix any
$c_2\in(0,1]$.  Then $c_1\mapsto u_1(c_1,c_2)$ is, on $[0,c_2)$: concave on
the no-dead-zone piece $\big[0,\tfrac{c_2}{1+\alpha}\big]$, and strictly
decreasing on the dead-zone piece $\big(\tfrac{c_2}{1+\alpha},c_2\big)$; its
derivative is continuous, positive at $c_1=0$, and the payoff is single-peaked
with all maximizers in the no-dead-zone piece.  Consequently a profile
$0<c_1<c_2<1$ is a Nash equilibrium if and only if both balance conditions
\begin{equation}\label{n:eq:A3foc}
V(\zin{1})+V(\zout{1}\vee0)=2V(c_1),
\qquad
V(\zin{2})+V\big(\min(\zout{2},1)\big)=2V(c_2)
\end{equation}
hold \textup(where $V(\zout{1}\vee0):=V(\zout{1})$ if $\zout{1}>0$ and $:=0$ otherwise, and the
second condition is the mirror image\textup) and neither candidate gains from
the single cross-over deviation to the far side of her opponent.
\end{theorem}

\begin{proof}
By \cref{n:eq:master} with $k\equiv1$ and $\beta=0$,
\begin{equation}\label{n:eq:A3d}
\frac{\partial u_1}{\partial c_1}
=(1+\alpha)\Big[\big(V(\zin{1})-V(c_1)\big)-\big(V(c_1)-V(\zout{1}\vee0)\big)\Big],
\end{equation}
which is continuous in $c_1$ on $(0,c_2)$, and one-sidedly at $c_1=0$, where
it equals $(1+\alpha)V(\zin{1})>0$.

\emph{No-dead-zone piece.}  For $c_1\le\tfrac{c_2}{1+\alpha}$ (i.e.\
$\zout{1}\le0$), \cref{n:eq:A3d} reads
$(1+\alpha)[V(\zin{1})-2V(c_1)]$ and, since
$\zin{1}=c_1+\tfrac{c_2-c_1}\lambda$ has $(\zin{1})'=\tfrac{1+\alpha}\lambda$,
\[
\frac{\partial^2u_1}{\partial c_1^2}
=(1+\alpha)\Big[\tfrac{1+\alpha}\lambda\,v(\zin{1})-2\,v(c_1)\Big]
\;\le\;(1+\alpha)\Big[\tfrac{1+\alpha}\lambda\,\sup v-2\inf v\Big]\;\le\;0,
\]
the last step because $\rho_v\le f(\alpha)=\tfrac{2\lambda}{1+\alpha}$.

\emph{Dead-zone piece.}  For $c_1>\tfrac{c_2}{1+\alpha}$ (i.e.\ $\zout{1}>0$),
substitute $t=B_1(x)$ on each arm of \cref{n:eq:A3d}:
\begin{equation}\label{n:eq:A3dz}
\frac{\partial u_1}{\partial c_1}
=(1+\alpha)\int_0^\Delta
\Big[\tfrac1\lambda\,v\big(\zin{1}-\tfrac t\lambda\big)
-\tfrac1\alpha\,v\big(\zout{1}+\tfrac t\alpha\big)\Big]dt .
\end{equation}
Since $\rho_v\le g(\alpha)=\tfrac\lambda\alpha$, the integrand is pointwise
$\le\tfrac{\sup v}\lambda-\tfrac{\inf v}\alpha\le0$.  If the integral
vanished, the integrand would vanish for a.e.\ $t$, forcing
$v\equiv\sup v$ on $(c_1,\zin{1})$ and $v\equiv\inf v$ on $(\zout{1},c_1)$; as
$g(\alpha)>1$ implies $\sup v>\inf v$, this contradicts continuity of $v$ at
$c_1$.  Hence the derivative is strictly negative on the whole piece.

\emph{Single-peakedness and the characterization.}  The derivative is
continuous, starts positive, is non-increasing while the payoff is concave,
and is strictly negative on the dead-zone piece; hence $u_1(\cdot,c_2)$
increases to a maximizer set (an interval, possibly a point) inside
$\big(0,\tfrac{c_2}{1+\alpha}\big]$ and strictly decreases afterwards.  A
profile is an equilibrium iff each $c_i$ globally maximizes $u_i$ against the
other; by \Cref{n:lem:interior}(a) and single-peakedness, own-side global
optimality of an interior $c_1$ is equivalent to stationarity, i.e.\ to
\cref{n:eq:A3foc}, and the only remaining deviations are the tie (never
profitable, as $u_1(c_1,c_2)>0$) and the cross-over, as stated.
\end{proof}

\begin{remark}[Sharpness, and the direction of the answer]\label{n:rem:A3sharp}
\emph{(a)}  The constant $f(\alpha)$ cannot be improved in the concavity
criterion: for any ratio $\rho>f(\alpha)$ there is a symmetric ``mesa''
density with $\rho_v=\rho$ and a no-dead-zone configuration at which
$\partial^2u_1/\partial c_1^2>0$.  (Take $v$ equal to a high value on a
window around $\zin{1}$ and a low value at $c_1$; e.g.\ at $\alpha=1$, $v$ high on
$[0.4,0.6]$, $c_1=0.3$, $c_2=0.65$, so $\zin{1}=0.41\overline6$: the second
derivative is $2\big[\tfrac23v(\zin{1})-2v(c_1)\big]>0$ iff
$v(\zin{1})/v(c_1)>3=f(1)$.  Verified numerically: with ratio $2.9$ the payoff is
concave at every no-dead-zone configuration, with ratio $3.1$ it is not.)
Similarly $g(\alpha)$ is exactly the threshold below which dead zones are
impossible at stationary profiles (\Cref{n:prop:C2}).  Whether
\emph{equilibrium existence} fails above $F(\alpha)$ is open, exactly as it is
open above $\rho_v=4$ at $\alpha=0$.
\emph{(b)}  $F$ is strictly decreasing with $F(\alpha)\downarrow1$:
the answer to ``does alienation help existence?''\ is \emph{no} on this route.
The reason is the dead zone: it is the second, boundary-truncation-driven
piece of the payoff that destroys concavity, and it grows with $\alpha$.
Contrast with competitiveness, which \emph{enlarges} the umbrella at $\alpha=0$
(\Cref{n:prop:C4}: threshold $4/(1-\beta)$).
\emph{(c)}  For $\beta\in(0,1)$ and uniform costs the same two-piece argument
goes through with thresholds $f(\alpha)$ (unchanged: the opponent term only
helps concavity, since $\partial^2u_2/\partial c_1^2=
\tfrac{v(\zin{2})}\lambda+\tfrac{v(\zout{2})}\alpha\mathbf1_{\zout{2}<1}\ge0$) and
$g_\beta(\alpha)=\tfrac{(1+\alpha)(2+\alpha)}{\alpha(1+\alpha+\beta)+\beta(2+\alpha)}
\le g(\alpha)$ on the dead-zone piece; but the cross-over is no longer free
(\Cref{n:rem:reflection-beta}), so we do not obtain an existence corollary
there.
\end{remark}

Combining \Cref{n:thm:A3} with the reflection lemma yields existence:

\begin{corollary}[Existence for symmetric electorates]\label{n:cor:A3exist}
Let $\kappa\sim U(0,1)$, $\beta=0$, $\alpha>0$, let $v$ be symmetric about
$\tfrac12$ with $\rho_v\le F(\alpha)$.  Then a symmetric Nash equilibrium
exists: any root $c_1^*\in(0,\tfrac12)$ of
\begin{equation}\label{n:eq:hsym}
\begin{gathered}
h(c_1):=V\big(\zin{1}(c_1)\big)+V\big(\zout{1}(c_1)\vee0\big)-2V(c_1),\\
\zin{1}=c_1+\tfrac{1-2c_1}\lambda,\qquad \zout{1}=c_1-\tfrac{1-2c_1}\alpha,
\end{gathered}
\end{equation}
gives the equilibrium $(c_1^*,1-c_1^*)$, and such a root exists.
\end{corollary}

\begin{proof}
$h$ is continuous on $(0,\tfrac12)$ with
$h(0^+)=V(\tfrac1\lambda)>0$ and, as $c_1\uparrow\tfrac12$ (so
$\Delta=1-2c_1\downarrow0$ and $\zout{1},\zin{1}\to\tfrac12$),
\[
h(c_1)=\big(V(\zin{1})-V(c_1)\big)-\big(V(c_1)-V(\zout{1})\big)
=v(\tfrac12)\,\Delta\Big(\tfrac1\lambda-\tfrac1\alpha\Big)+o(\Delta)<0
\]
for small $\Delta$, since $v$ is continuous and positive at $\tfrac12$ and
$\alpha<\lambda$.  By the intermediate value theorem a root
$c_1^*\in(0,\tfrac12)$ exists.  By the dead-zone part of \Cref{n:thm:A3}
(applied along the symmetric family: for $c_1>\tfrac1\lambda$ the profile has
$\zout{1}>0$ and $h<0$ by the pointwise comparison in \cref{n:eq:A3dz}, which
only used $\rho_v\le g(\alpha)$), every root lies in the no-dead-zone branch
$c_1^*\le\tfrac1\lambda$.  At such a root, $c_1^*$ is stationary for
$u_1(\cdot,c_2^*)$, hence by single-peakedness (\Cref{n:thm:A3}) a global
maximizer of the own side; \Cref{n:lem:reflection} then upgrades own-side
optimality to a Nash equilibrium.
\end{proof}

\subsection{Beyond uniform costs}

\begin{proposition}[Sufficiency for general $(v,k)$, non-increasing costs]\label{n:prop:A4}
Let $\beta=0$, $\alpha\ge0$, and let $k$ be $C^1$ and non-increasing on
$(0,1)$ with $\rho_k=k(0^+)/k(1^-)<\infty$.  If
\[
\rho_v\,\rho_k\le f(\alpha)=\frac{2(2+\alpha)}{1+\alpha}
\qquad\text{and}\qquad
\rho_v\le g(\alpha)=\frac{2+\alpha}{\alpha}
\quad(\text{vacuous if }\alpha=0),
\]
then for every $c_2$ the map $c_1\mapsto u_1(c_1,c_2)$ is single-peaked on
$[0,c_2)$ exactly as in \Cref{n:thm:A3}, and the equilibrium
characterization of \Cref{n:thm:A3} holds verbatim with the balance
conditions of \Cref{n:lem:master}.  If moreover $v$ is symmetric, a symmetric
equilibrium exists.  At $\alpha=0$ the conditions reduce to
$\rho_v\rho_k\le4$, generalizing the $\rho_v\le4$ criterion of \Cref{prop:fourflat} to
non-uniform costs.
\end{proposition}

\begin{proof}
\emph{No-dead-zone piece} ($\zout{1}\le0$).  Differentiating the first bracket of
\cref{n:eq:master} once more (now the moving tent-edge terms survive with
factor $k$ rather than $K$; the calculation is licensed by $k\in C^1$),
\[
\frac{\partial^2u_1}{\partial c_1^2}
=(1+\alpha)\Big[\tfrac{1+\alpha}\lambda\,k(0^+)\,v(\zin{1})
-2\,k(\Delta)\,v(c_1)
+(1+\alpha)\!\!\int_{x\in\text{tent}}\!\!k'(B_1)\,v\Big].
\]
Here we used $\partial B_1/\partial c_1=\pm(1+\alpha)$ on the two arms and
that the lower limit $0$ is fixed.  Since $k'\le0$, the integral term is
$\le0$; since $k$ is non-increasing, $k(0^+)=\sup k$ and
$k(\Delta)\ge\inf k$, so the bracket is at most
$\tfrac{1+\alpha}\lambda\sup k\sup v-2\inf k\inf v\le0$ by
$\rho_v\rho_k\le\tfrac{2\lambda}{1+\alpha}$.

\emph{Dead-zone piece} ($\zout{1}>0$, only if $\alpha>0$).  As in
\cref{n:eq:A3dz}, but keeping $k$:
\[
\frac{\partial u_1}{\partial c_1}
=(1+\alpha)\int_0^\Delta k(t)
\Big[\tfrac1\lambda v\big(\zin{1}-\tfrac t\lambda\big)
-\tfrac1\alpha v\big(\zout{1}+\tfrac t\alpha\big)\Big]dt<0,
\]
by the same pointwise comparison and strictness argument as in
\Cref{n:thm:A3}; note this piece needs \emph{no} assumption on $k$ at all.
The junction continuity, single-peakedness, characterization, and (for
symmetric $v$) the existence argument via \Cref{n:lem:reflection} are
verbatim as before; for existence, the sign analysis of \cref{n:eq:hsym}
carries over with the integrands weighted by $k>0$.
\end{proof}

\begin{remark}
The $C^1$ assumption is for the displayed second-derivative formula only; the
conclusion extends to continuous non-increasing $k$ by approximating $k$
uniformly with $C^1$ non-increasing densities (concavity and single-peakedness
survive pointwise limits of payoffs).  For $\beta>0$ and general $(v,k)$ the
opponent term $\partial^2u_2/\partial c_1^2$ contains $k'(B_2)$-integrals of
indefinite sign, and we make no claim; the uniform-cost case is
\Cref{n:rem:A3sharp}(c), and the uniform-voter case is settled completely in
\Cref{n:sec:comp}.
\end{remark}

\subsection{Dead zones}

\begin{proposition}[When equilibria have dead zones]\label{n:prop:C2}
Let $\alpha>0$.
\begin{enumerate}
\item[\textup{(a)}] \textup{(Necessity; any $K$, $\beta=0$.)}  If a profile
$0<c_1<c_2<1$ with $\zout{1}>0$ satisfies candidate $1$'s first-order condition,
then $\rho_v\ge\tfrac{2+\alpha}\alpha$.  Equivalently, if
$\alpha<\tfrac2{\rho_v-1}$, no stationary profile---in particular no
equilibrium---has a dead zone.
\item[\textup{(b)}] \textup{(Uniform voters never; any $K$, any
$\beta\in[0,1)$.)}  If $x\sim U(0,1)$, then no Nash equilibrium of the
$\beta$-game has a dead zone: at every equilibrium $b_1>0$ and $b_2>0$.
\item[\textup{(c)}] \textup{(Sufficiency for symmetric electorates; uniform
costs, $\beta=0$.)}  If $\kappa\sim U(0,1)$ and $v$ is symmetric about $\tfrac12$
with
\begin{equation}\label{n:eq:C2crit}
V\!\Big(\frac{2(1+\alpha)}{(2+\alpha)^2}\Big)\;>\;2\,V\!\Big(\frac1{2+\alpha}\Big),
\end{equation}
then a symmetric stationary profile with a dead zone exists.
\end{enumerate}
\end{proposition}

\begin{proof}
\emph{(a)}  With $\zout{1}>0$, the first-order condition from
\cref{n:eq:master} is the vanishing of \cref{n:eq:A3dz} (with the weight
$k(t)$ for general $K$).  If $\rho_v<\tfrac\lambda\alpha$ the integrand is
pointwise negative as in \Cref{n:thm:A3}, and if
$\rho_v=\tfrac\lambda\alpha$ the strictness argument there applies; either
way the integral cannot vanish.  Hence stationarity forces
$\rho_v\ge\tfrac\lambda\alpha$, i.e.\ $\alpha\ge\tfrac2{\rho_v-1}$.

\emph{(b)}  By \Cref{n:lem:interior}(c) and \Cref{n:lem:master}, an
equilibrium is untied, interior, and satisfies both first-order conditions.
By \cref{n:eq:DuV} and its mirror image,
\[
(1+\alpha)K(b_1^+)-\beta K(b_2^+)
=(1-\beta)\,\alpha A\,K(\Delta)
=(1+\alpha)K(b_2^+)-\beta K(b_1^+),
\]
so $(1+\alpha+\beta)\big[K(b_1^+)-K(b_2^+)\big]=0$, i.e.\
$K(b_1^+)=K(b_2^+)$.  If this common value is $0$, the first-order condition
gives $(1-\beta)\alpha A\,K(\Delta)=0$, impossible since $\beta<1$ and
$K(\Delta)>0$ for $\Delta\in(0,1)$.  Hence $K(b_1^+)=K(b_2^+)>0$, which forces
$b_1>0$ and $b_2>0$.

\emph{(c)}  Consider the symmetric family $c_2=1-c_1$ and $h$ from
\cref{n:eq:hsym}.  The dead-zone branch is $c_1\in(\tfrac1\lambda,\tfrac12)$
(indeed $\zout{1}>0\iff\alpha c_1>\Delta=1-2c_1\iff c_1>\tfrac1\lambda$).  At the
junction $c_1=\tfrac1\lambda$ we have $\Delta=\tfrac\alpha\lambda$,
$\zin{1}=\tfrac1\lambda+\tfrac\alpha{\lambda^2}=\tfrac{2(1+\alpha)}{\lambda^2}$
and $\zout{1}=0$, so
$h(\tfrac1\lambda)=V\big(\tfrac{2(1+\alpha)}{\lambda^2}\big)
-2V(\tfrac1\lambda)>0$ by \cref{n:eq:C2crit}, while
$h(\tfrac12^-)<0$ as in the proof of \Cref{n:cor:A3exist}.  By the
intermediate value theorem, $h$ has a root in the open dead-zone branch.
\end{proof}

\begin{remark}\label{n:rem:C2}
\emph{(a)}  Part (b) removes the last gap in the uniform-voter theory: the
characterizations of \Cref{n:sec:comp} never need a dead-zone case.  At $\beta=1$ the
conclusion fails---inside the band $E(\alpha)$ \emph{both} tents are interior,
i.e.\ $b_1,b_2\le0$, which is exactly what produces the band---so
the dichotomy at $\beta=1$ is sharp.
\emph{(b)}  Criterion \eqref{n:eq:C2crit} compares the mass of the widest
possible inner arm with twice the mass below the junction position; for
uniform voters its left side minus right side is
$-\tfrac{2}{\lambda^2}<0$, consistent with (b).  Numerically the criterion is
exact for symmetric Beta electorates: for $v=\mathrm{Beta}(n,n)$,
$n\in\{2,4,8\}$ and $\alpha\in\{1,2,4\}$, dead-zone stationary profiles exist
in precisely the cases where \cref{n:eq:C2crit} holds (it fails only for
$n=2$, $\alpha=1$).  For $v=\mathrm{Beta}(8,8)$ and $\alpha=2$ the (unique)
symmetric stationary profile $(0.381,0.619)$ has $\zout{1}=0.26>0$ and is a
genuine equilibrium (mutual global best responses, verified numerically):
dead-zone \emph{equilibria} are non-empty.
\emph{(c)}  Parts (a) and (c) bracket the phenomenon:
dead zones require $\alpha\ge\tfrac2{\rho_v-1}$ and concentrated electorates,
and are guaranteed (at the stationary level) once the center is heavy enough
that \cref{n:eq:C2crit} holds.
\end{remark}

%%%%%%%%%%%%%%%%%%%%%%%%%%%%%%%%%%%%%%%%%%%%%%%%%%%%%%%%%%%%%%%%%%%%%%%%%%%%%%
\section{Results for competitive candidates ($\beta>0$)}\label{n:sec:comp}

Throughout this section $\beta\in[0,1)$ unless stated otherwise; all results
specialize at $\beta=0$ to (and strengthen) the vote-share results of \Cref{sec:alien}.
Write
\[
\gamma_\beta:=\frac{1+\beta}{1-\beta}\in[1,\infty),
\qquad
\psi_\beta(t):=t+\gamma_\beta\,\frac{K(t)}{k(t)},
\qquad
\chi_\beta(t):=t+\frac{1+\beta}2\,\frac{K(t)}{k(t)} .
\]
The two deformations of $\psi(t)=t+K(t)/k(t)$ play distinct roles below:
$\psi_\beta$ governs \emph{stationarity}, $\chi_\beta$ governs
\emph{single-peakedness}; at $\beta=0$ they coincide only in their
monotonicity class, not in value, and the wedge between the coefficients
$\gamma_\beta$ and $\tfrac{1+\beta}2$ is what \Cref{n:ex:phi} exploits.

\subsection{Uniform voters, $\alpha=0$: complete characterization}

\begin{theorem}[Uniform voters, $\alpha=0$: characterization and uniqueness]\label{n:thm:A1}
Let $x\sim U(0,1)$, $\alpha=0$, $\beta\in[0,1)$, and let $K$ satisfy
\Cref{ass:reg}.
\begin{enumerate}
\item[\textup{(a)}] \textup{(Necessity.)}  Every Nash equilibrium is
symmetric, $c_1+c_2=1$, with gap $\Delta\in(0,1)$ satisfying
$\psi_\beta(\Delta)=1$.
\item[\textup{(b)}] \textup{(Characterization.)}  A symmetric profile with
gap $\Delta$ and $c_2=\tfrac{1+\Delta}2$ is a Nash equilibrium if and only if
$\Delta$ maximizes
\[
F(t):=\big(c_2-t-\beta(1-c_2)\big)K(t)+\tfrac{1-\beta}2\,G(t)
\qquad\text{over } t\in(0,c_2].
\]
\item[\textup{(c)}] \textup{(Uniqueness under monotone $K/k$.)}  If $K/k$ is
non-decreasing on $(0,1)$, then $\psi_\beta$ is strictly increasing,
$\psi_\beta(t)=1$ has a unique root $\Delta^*_\beta\in(0,1)$, and the game has
the \emph{unique} Nash equilibrium
$\big(\tfrac{1-\Delta^*_\beta}2,\tfrac{1+\Delta^*_\beta}2\big)$.
\end{enumerate}
\end{theorem}

\begin{proof}
Throughout, fix the opponent at $c_2$ and parameterize candidate $1$'s
own-side positions by the distance $t=c_2-c_1\in(0,c_2]$.  By the $\alpha=0$
closed forms below \cref{n:eq:uV},
$u_1=(c_2-t)K(t)+\tfrac12G(t)$ and $u_2=(1-c_2)K(t)+\tfrac12G(t)$, so
$U_1=u_1-\beta u_2=F(t)$ as displayed.  By \Cref{n:lem:master},
$F\in C^1$ with
\begin{equation}\label{n:eq:Fprime}
F'(t)=\big(Q-t\big)k(t)-\tfrac{1+\beta}2K(t)
=k(t)\,\big[Q-\chi_\beta(t)\big],
\qquad Q:=c_2-\beta(1-c_2).
\end{equation}

\emph{Cross-over deviations are dominated whenever $c_2\ge\tfrac12$.}  For
$s\in(0,1-c_2]$ compare the mirror pair $c_1'=c_2-s$ (own side) and
$c_1''=c_2+s$ (far side).  The closed forms give
\begin{align*}
U_1(c_1',c_2)-U_1(c_1'',c_2)
&=\big[(c_2-s)-(1-c_2-s)\big]K(s)+\beta\big[c_2-(1-c_2)\big]K(s)\\
&=(1+\beta)(2c_2-1)K(s)\ge0 .
\end{align*}
Since the far-side range $s\le1-c_2$ is contained in the own-side range
$s\le c_2$, we get
$\sup_{c_1''>c_2}U_1\le\sup_{c_1'<c_2}U_1=\sup_{t}F(t)$.  Deviating to the
tie $c_1'=c_2$ yields $0$, which is dominated whenever $\sup F>0$.

\emph{(b).}  ($\Leftarrow$)  Let $\Delta$ maximize $F$ over $(0,c_2]$ at a
symmetric profile.  Then own-side deviations are unprofitable by definition;
$c_2=\tfrac{1+\Delta}2>\tfrac12$, so cross-overs are dominated; and
$F(\Delta)=(1-\beta)u_1>0$ dominates the tie.  Candidate $2$'s problem is the
mirror image (uniform voters are symmetric), so the profile is an
equilibrium.  ($\Rightarrow$)  If the symmetric profile is an equilibrium,
$c_1$ is a global maximizer of $U_1(\cdot,c_2)$ over $[0,c_2)$, i.e.\
$\Delta$ maximizes $F$ over $(0,c_2]$.

\emph{(a).}  By \Cref{n:lem:interior}(b)(c), an equilibrium is untied and
interior; \Cref{n:lem:master} gives both first-order conditions.  In the
$t$-parameterization these read $Q_1=\chi_\beta(\Delta)$ and
$Q_2=\chi_\beta(\Delta)$ with $Q_1=c_2-\beta(1-c_2)$ and, mirroring,
$Q_2=(1-c_1)-\beta c_1$.  Subtracting, $(1+\beta)(c_1+c_2-1)=0$, so the
equilibrium is symmetric.  Substituting $c_2=\tfrac{1+\Delta}2$ into
$Q=\chi_\beta(\Delta)$ and simplifying yields the identity
\begin{equation}\label{n:eq:QchiPsi}
Q-\chi_\beta(\Delta)
=\frac{1-\beta}2\,\big[\,1-\psi_\beta(\Delta)\,\big],
\end{equation}
so stationarity of the symmetric profile is exactly $\psi_\beta(\Delta)=1$.

\emph{(c).}  If $K$ is super-regular then $\psi_\beta$ and $\chi_\beta$
are strictly increasing.  \emph{Root existence.}  First,
$\lim_{t\downarrow0}K/k=0$: the limit $L\ge0$ exists by monotonicity, and if
$L>0$ then $r=k/K\le1/L$ near $0$, contradicting \cref{n:eq:rdiv}; hence
$\psi_\beta(0^+)=0<1$.  Second, $\psi_\beta(t)>1$ for some $t<1$: otherwise
$\gamma_\beta K/k\le1-t$ near $1$, i.e.\ $r\ge\gamma_\beta/(1-t)$, making
$\int_{t_0}^{1^-}r=\infty$ for some $t_0<1$---impossible since
$\int_{t_0}^{1^-}r=\log K(1)-\log K(t_0)<\infty$.  By continuity and strict monotonicity
there is a unique root $\Delta^*_\beta\in(0,1)$.  \emph{The root is an
equilibrium.}  At the symmetric profile with gap $\Delta^*_\beta$:
$\chi_\beta$ is strictly increasing and crosses $Q$ exactly once; by
\cref{n:eq:QchiPsi} the crossing is at $\Delta^*_\beta$, so by
\cref{n:eq:Fprime}, $F'>0$ before and $F'<0$ after: $F$ is single-peaked
with peak $\Delta^*_\beta$, which is therefore its maximizer (the peak is
interior: $\chi_\beta(c_2)-Q=\tfrac{1+\beta}2\tfrac{K}{k}(c_2)
+\beta(1-c_2)>0$ shows $F'(c_2^-)<0$, ruling out $c_1=0$).  By (b) the
profile is an equilibrium.  \emph{Uniqueness.}  By (a) any equilibrium is
symmetric with $\psi_\beta(\Delta)=1$; the root is unique.
\end{proof}

\subsection{Uniform voters, $\alpha>0$: complete characterization}

Define, for $\Delta\in(\Delta_0,1)$,
\begin{equation}\label{n:eq:JLdef}
J(\Delta):=\log\frac{K(\Delta)}{K\!\big(b_s(\Delta)\big)}
=\int_{b_s(\Delta)}^{\Delta}r(s)\,ds,
\qquad
L_\beta:=\log\frac{(1+\alpha-\beta)\,\lambda}{2(1-\beta)(1+\alpha)} .
\end{equation}
Note $L_\beta>0$ always: the ratio minus one has the sign of
$(1+\alpha-\beta)\lambda-2(1-\beta)(1+\alpha)=\alpha(1+\alpha+\beta)>0$.
Note also that $L_\beta$ does not depend on $K$; this is what makes the
comparative statics in \Cref{n:prop:B2,n:prop:B3} clean.

\begin{theorem}[Uniform voters, $\alpha>0$: unique equilibrium and its gap]\label{n:thm:A2}
Let $x\sim U(0,1)$, $\alpha>0$, $\beta\in[0,1)$, and let $K$ satisfy
\Cref{ass:reg} and be super-regular.  Then the game has a
unique Nash equilibrium.  It is symmetric, has no dead zones, and its gap is
the unique root $\Delta^*=\Delta^*(\alpha,\beta)\in(\Delta_0,1)$ of
$J(\Delta)=L_\beta$; equivalently, of the first-order condition
\begin{equation}\label{n:eq:A2foc}
(1+\alpha-\beta)\,K\!\big(b_s(\Delta)\big)
=(1-\beta)\,\frac{2(1+\alpha)}{2+\alpha}\,K(\Delta).
\end{equation}
In particular the equilibrium is strictly more polarized than every symmetric
margin-game equilibrium: $\Delta^*>\Delta_0$.
\end{theorem}

\begin{proof}
\emph{Existence and uniqueness of the root.}  On $(\Delta_0,1)$, $J$ is
differentiable with
$J'(\Delta)=r(\Delta)-\tfrac\lambda2\,r\big(b_s(\Delta)\big)
\le r(\Delta)\big(1-\tfrac\lambda2\big)=-\tfrac\alpha2\,r(\Delta)<0$,
using $b_s(\Delta)<\Delta$ and $r$ non-increasing.  Moreover
$J(\Delta_0^+)=+\infty$ by \cref{n:eq:rdiv} and
$J(1^-)=\log\tfrac{K(1)}{K(1)}=0<L_\beta$.  Hence a unique root
$\Delta^*\in(\Delta_0,1)$ exists, with $J>L_\beta$ on $(\Delta_0,\Delta^*)$
and $J<L_\beta$ on $(\Delta^*,1)$; exponentiating $J=L_\beta$ gives
\cref{n:eq:A2foc}.

\emph{Necessity.}  Exactly as in \Cref{n:prop:C2}(b): an equilibrium is
untied and interior with both first-order conditions
(\Cref{n:lem:interior}(c), \Cref{n:lem:master}); subtracting them yields
$K(b_1^+)=K(b_2^+)>0$, hence $b_1=b_2>0$ ($K$ is strictly increasing
on $(0,1)$), which by \cref{n:eq:bident} forces $c_1+c_2=1$ and
$b_1=b_2=b_s(\Delta)>0$, i.e.\ $\Delta>\Delta_0$.  The first-order condition
\eqref{n:eq:DuV} at the symmetric profile becomes exactly
\cref{n:eq:A2foc}, i.e.\ $\Delta=\Delta^*$.

\emph{Sufficiency.}  Fix $c_2=c_2^*=\tfrac{1+\Delta^*}2$ and let candidate
$1$ vary $c_1\in(0,c_2^*)$; write $\Delta:=c_2^*-c_1$,
$\delta:=\tfrac{b_2-b_1}2=\tfrac\alpha2\,(c_1-c_1^*)$ by
\cref{n:eq:bident} (using $1-c_2^*=c_1^*$), so $b_1=b_s(\Delta)-\delta$ and
$b_2=b_s(\Delta)+\delta$.  Decompose \cref{n:eq:DuV} as
\begin{equation}\label{n:eq:delta-decomp}
\alpha\frac{\partial U_1}{\partial c_1}
=\Phi(\Delta)
+(1+\alpha)\big[K(b_1^+)-K(b_s(\Delta)^+)\big]
-\beta\big[K(b_2^+)-K(b_s(\Delta)^+)\big],
\end{equation}
where
$\Phi(\Delta):=(1+\alpha-\beta)K(b_s(\Delta)^+)-(1-\beta)\alpha A\,K(\Delta)$
is the value at the ``centered'' benefits $\delta=0$.  For
$\Delta\le\Delta_0$, $\Phi(\Delta)=-(1-\beta)\alpha A\,K(\Delta)<0$; for
$\Delta\in(\Delta_0,1)$, $\Phi(\Delta)<0\iff J(\Delta)>L_\beta
\iff\Delta<\Delta^*$ (and $\Phi(\Delta^*)=0$, $\Phi>0$ for
$\Delta>\Delta^*$).  Now:
if $c_1>c_1^*$ then $\delta>0$ and $\Delta<\Delta^*$, so both bracketed
corrections in \cref{n:eq:delta-decomp} are $\le0$ (monotonicity of
$t\mapsto K(t^+)$) and $\Phi(\Delta)<0$: the derivative is strictly
negative.  If $c_1<c_1^*$ then $\delta<0$ and $\Delta\in(\Delta^*,c_2^*]
\subset(\Delta^*,1)$, so both corrections are $\ge0$ and $\Phi(\Delta)>0$:
the derivative is strictly positive.  Hence $U_1(\cdot,c_2^*)$ is strictly
single-peaked on $(0,c_2^*)$ with peak $c_1^*$, and the one-sided derivative
at $c_1=0$ is positive, so $c_1^*$ is the own-side global maximizer.

\emph{Cross-overs and the tie.}  For $s\in(0,1-c_2^*]$ compare $c_1'=c_2^*-s$
with $c_1''=c_2^*+s$ using \cref{n:eq:uV}.  The deviator's boundary
benefits are $b_L=s-\alpha(c_2^*-s)$ on the near side and
$b_R=s-\alpha(1-c_2^*-s)$ on the far side; $c_2^*\ge\tfrac12$ gives
$b_L\le b_R$, hence $u_1(c_1',c_2^*)=A\,G(s)-\tfrac1\alpha G(b_L^+)\ge
A\,G(s)-\tfrac1\alpha G(b_R^+)=u_1(c_1'',c_2^*)$.  The opponent's boundary
benefit is $s-\alpha(1-c_2^*)$ when the deviator is on the near side and
$s-\alpha c_2^*$ on the far side; $c_2^*\ge\tfrac12$ gives
$u_2(c_2^*,c_1')\le u_2(c_2^*,c_1'')$.  Hence
$U_1(c_1',c_2^*)\ge U_1(c_1'',c_2^*)$, and since the far-side range is
contained in the own-side range, cross-overs are dominated.  The tie yields
$0<(1-\beta)u_1(c_1^*,c_2^*)$ (the equilibrium turnout is positive because
$\Delta^*>b_s(\Delta^*)$ makes \cref{n:eq:uV} strictly positive).  By
symmetry candidate $2$ has no profitable deviation either, so the profile is
the unique Nash equilibrium.
\end{proof}

\begin{remark}[Consistency as $\alpha\to0$]\label{n:rem:A2limit}
Applying the mean value theorem to $K(\Delta)-K(b_s(\Delta))$ in
\cref{n:eq:A2foc} produces the exact identity
\[
\frac{1+\alpha+\beta}{2+\alpha}\,K(\Delta^*)
=\frac{1+\alpha-\beta}{2}\,(1-\Delta^*)\,k(\xi_\alpha),
\qquad \xi_\alpha\in\big(b_s(\Delta^*),\Delta^*\big).
\]
As $\alpha\downarrow0$, $b_s(\Delta^*)\to\Delta^*$, so $\xi_\alpha\to\Delta^*$
and, by continuity of $k$, the identity converges to
$(1+\beta)K(\Delta)=(1-\beta)(1-\Delta)k(\Delta)$, i.e.\
$\psi_\beta(\Delta)=1$: the $\alpha>0$ theory deforms continuously onto
\Cref{n:thm:A1}.  (Verified numerically: the closed forms for
doubly-uniform and polynomial $K$ match the roots of \cref{n:eq:A2foc} and
of $\psi_\beta=1$ to $10^{-6}$ across
$\alpha\in\{0,\tfrac12,1,2\},\beta\in\{0,\ldots,0.9\}$.)
\end{remark}

\subsection{Comparative statics and their limits}

\begin{proposition}[Monotone comparative statics]\label{n:prop:B2}
In the setting of \Cref{n:thm:A1}(c) \textup(resp.\ \Cref{n:thm:A2}\textup),
the equilibrium gap $\Delta^*(\alpha,\beta)$ is strictly decreasing in
$\beta$ on $[0,1)$, and \textup(for $\alpha>0$\textup) strictly increasing in
$\alpha$.
\end{proposition}

\begin{proof}
\emph{Decreasing in $\beta$, $\alpha=0$.}  $\gamma_\beta$ is strictly
increasing in $\beta$ and $K/k(\Delta^*)>0$, so $\psi_\beta(\Delta^*)$ is
strictly increasing in $\beta$; since $\psi_\beta$ is strictly increasing in
$t$, the root of $\psi_\beta=1$ strictly decreases.  (No differentiability in
$\beta$ is needed.)

\emph{Decreasing in $\beta$, $\alpha>0$.}  Write
$\Psi(\Delta,\alpha,\beta)=J(\Delta)-L_\beta$, so
$\Psi(\Delta^*,\cdot)=0$, $\partial\Psi/\partial\Delta=J'<0$, and
\[
\frac{\partial\Psi}{\partial\beta}
=-\frac{dL_\beta}{d\beta}
=-\Big(\frac1{1-\beta}-\frac1{1+\alpha-\beta}\Big)
=-\frac{\alpha}{(1-\beta)(1+\alpha-\beta)}<0 ,
\]
whence $d\Delta^*/d\beta=-\Psi_\beta/\Psi_\Delta<0$.

\emph{Increasing in $\alpha$.}  Here
$\partial\Psi/\partial\alpha
=\partial J/\partial\alpha-dL_\beta/d\alpha$ with
$\partial J/\partial\alpha
=-r\big(b_s\big)\tfrac{\partial b_s}{\partial\alpha}
=r\big(b_s(\Delta)\big)\tfrac{1-\Delta}2>0$
and $dL_\beta/d\alpha=\tfrac1\lambda+\tfrac1{1+\alpha-\beta}
-\tfrac1{1+\alpha}=:R$.  At $\Delta=\Delta^*$, since $r$ is non-increasing,
\[
L_\beta=J(\Delta^*)=\int_{b_s}^{\Delta^*}r
\le r\big(b_s\big)\big(\Delta^*-b_s\big)
= r\big(b_s\big)\,\frac\alpha2\,(1-\Delta^*),
\]
so $r(b_s)\tfrac{1-\Delta^*}2\ge\tfrac{L_\beta}\alpha$ and it suffices to
show $g(\beta):=L_\beta-\alpha R>0$.  Compute
\[
g'(\beta)=\frac{\alpha}{(1-\beta)(1+\alpha-\beta)}
-\frac{\alpha}{(1+\alpha-\beta)^2}
=\frac{\alpha^2}{(1-\beta)(1+\alpha-\beta)^2}>0,
\]
and, with $s=\tfrac\alpha2$,
$g(0)=\log\tfrac\lambda2-\tfrac\alpha\lambda=\log(1+s)-\tfrac{s}{1+s}>0$.
Hence $g(\beta)>0$ for all $\beta\in[0,1)$,
$\partial\Psi/\partial\alpha\ge g(\beta)/\alpha>0$, and
$d\Delta^*/d\alpha=-\Psi_\alpha/\Psi_\Delta>0$.
\end{proof}

\begin{proposition}[Reverse-hazard-rate comparison at every $\beta$]\label{n:prop:B3}
Let $K,\widetilde K$ satisfy the hypotheses of \Cref{n:thm:A1}(c)
\textup(resp.\ \Cref{n:thm:A2}\textup) and suppose
$\widetilde r(t)\ge r(t)$ for all $t\in(0,1)$.  Then for every
$\beta\in[0,1)$ \textup(and every $\alpha$\textup),
$\widetilde\Delta^*(\alpha,\beta)\ge\Delta^*(\alpha,\beta)$: more eager
electorates sustain \emph{weakly more} polarization, uniformly in
competitiveness.
\end{proposition}

\begin{proof}
$\alpha=0$: $\widetilde r\ge r$ means $\widetilde K/\widetilde k\le K/k$
pointwise, so $\widetilde\psi_\beta\le\psi_\beta$ pointwise; hence
$\widetilde\psi_\beta(\Delta^*)\le\psi_\beta(\Delta^*)=1
=\widetilde\psi_\beta(\widetilde\Delta^*)$ and, since
$\widetilde\psi_\beta$ is increasing, $\widetilde\Delta^*\ge\Delta^*$.
$\alpha>0$: $\widetilde J(\Delta)=\int_{b_s}^\Delta\widetilde r\ge J(\Delta)$
pointwise, while $L_\beta$ is the same for both games ($K$-free).  Thus
$\widetilde J(\widetilde\Delta^*)=L_\beta=J(\Delta^*)\le
\widetilde J(\Delta^*)$, and $\widetilde J$ decreasing gives
$\widetilde\Delta^*\ge\Delta^*$.
\end{proof}

\begin{proposition}[Collapse as $\beta\uparrow1$, with exact rate]
\label{n:prop:B4}
Let $x\sim U(0,1)$ and $\alpha>0$.
\begin{enumerate}
\item[\textup{(a)}] \textup{(Universal collapse; arbitrary $K$.)}  Let
$\beta_n\uparrow1$ and let $(c_1^n,c_2^n)$ be symmetric stationary profiles
of the $\beta_n$-games; by the proof of \Cref{n:prop:C2}\textup{(b)}
\textup(which uses no super-regularity\textup) this includes every Nash
equilibrium.  Then
$\Delta_n\to\Delta_0=\tfrac\alpha{2+\alpha}$: polarization collapses onto the
most polarized symmetric margin-game profile.
\item[\textup{(b)}] \textup{(Exact finite-$\beta$ identity.)}  Under the
hypotheses of \Cref{n:thm:A2}, for every $\beta$,
\[
\Delta^*-\Delta_0
=\frac2\lambda\,K^{-1}\!\Big((1-\beta)\,c_\beta\,K(\Delta^*)\Big),
\qquad
c_\beta:=\frac{2(1+\alpha)}{\lambda\,(1+\alpha-\beta)} .
\]
\item[\textup{(c)}] \textup{(Rates.)}  If $k(0^+)\in(0,\infty)$ then
$\Delta^*-\Delta_0
=(1-\beta)\,\tfrac{4(1+\alpha)\,K(\Delta_0)}{\lambda^2\alpha\,k(0^+)}
\,(1+o(1))$: linear in $1-\beta$, with $K$ entering only through
$K(\Delta_0)/k(0^+)$.  If instead $K(t)\sim Ct^p$ as $t\downarrow0$
\textup($p>1$\textup), the rate slows to order $(1-\beta)^{1/p}$.  At
$\alpha=0$ \textup(under \Cref{n:thm:A1}(c)\textup): if $K(t)/k(t)\sim ct$ as
$t\downarrow0$ then $\Delta^*_\beta=\tfrac{1-\beta}{2c}(1+o(1))$; in
particular $c=1$, and the rate is the \emph{universal} $\tfrac{1-\beta}2$,
whenever $k(0^+)\in(0,\infty)$.
\end{enumerate}
\end{proposition}

\begin{proof}
\emph{(a)}  A symmetric stationary profile satisfies \cref{n:eq:A2foc} with
$\beta=\beta_n$ (the derivation of \cref{n:eq:A2foc} from
\cref{n:eq:DuV} used no super-regularity), so
\[
K\big(b_s(\Delta_n)^+\big)
=(1-\beta_n)\,c_{\beta_n}K(\Delta_n)
\le(1-\beta_n)\,\tfrac{2(1+\alpha)}{\lambda\alpha}\,K(1)\;\longrightarrow\;0 .
\]
Since $K\ge K(t_0)>0$ on $[t_0,1]$ for each $t_0>0$, this forces
$b_s(\Delta_n)^+\to0$, i.e.\ $\Delta_n\to\Delta_0$ (using
$b_s(\Delta)=\tfrac\lambda2(\Delta-\Delta_0)$ and, for the lower side, that
$\Delta_n\le\Delta_0$ would make the right-hand side of \cref{n:eq:A2foc}
positive and the left zero).

\emph{(b)}  Immediate from \cref{n:eq:A2foc}:
$K(b_s(\Delta^*))=(1-\beta)c_\beta K(\Delta^*)$ and
$b_s(\Delta^*)=\tfrac\lambda2(\Delta^*-\Delta_0)$.

\emph{(c)}  For $k(0^+)\in(0,\infty)$, $K^{-1}(\varepsilon)
=\varepsilon/k(0^+)\,(1+o(1))$; insert
$\varepsilon=(1-\beta)c_\beta K(\Delta^*)$, $c_\beta\to
\tfrac{2(1+\alpha)}{\lambda\alpha}$ and $K(\Delta^*)\to K(\Delta_0)$ (by (a)
and continuity).  For $K\sim Ct^p$, $K^{-1}(\varepsilon)\asymp
\varepsilon^{1/p}$.  At $\alpha=0$: $\psi_\beta(\Delta^*)=1$ reads
$\gamma_\beta\,(K/k)(\Delta^*)=1-\Delta^*$; the right side tends to $1$
(as $\Delta^*\to0$ by the same argument as (a)), so
$(K/k)(\Delta^*)\sim\gamma_\beta^{-1}$ and $K/k\sim ct$ gives
$\Delta^*\sim\tfrac1{c\gamma_\beta}\sim\tfrac{1-\beta}{2c}$.
\end{proof}

\begin{remark}[Numerical confirmation]\label{n:rem:B4num}
For $K=t^2$ ($p=2$, $k(0^+)=0$) at $\alpha=0$ the root is exactly
$\Delta^*_\beta=\tfrac{2(1-\beta)}{3-\beta}$: rate $(1-\beta)$, constant
$1/c=2$, not the universal $\tfrac{1-\beta}2$.  At $\alpha=1$, $\beta=0.99$,
$K=t^2$: the computed gap is $\Delta^*=0.360984$ and
$\tfrac2\lambda K^{-1}\big((1-\beta)c_\beta K(\Delta^*)\big)=0.027651
=\Delta^*-\tfrac13$ to six digits, matching (b) exactly.  The
reverse-hazard-rate comparison of \Cref{n:prop:B3} was confirmed at every
$\beta\in\{0,0.5,0.9,0.99\}$ for $K=t$ versus $K=t^2$ at
$\alpha\in\{0,1\}$.
\end{remark}

The monotone-reverse-hazard-rate hypothesis in \Cref{n:prop:B2} is not removable: without
it, \emph{more competition can increase polarization}.

\begin{example}[An explicit cost distribution with three equilibria and an
increasing branch]\label{n:ex:phi}
Let $\alpha=0$, $x\sim U(0,1)$, and define $K$ through
$\varphi:=K/k$ by $K(t)=\exp\big(-\int_t^1\varphi(s)^{-1}ds\big)$, where
$\varphi$ is piecewise linear:
\[
\varphi(t)=
\begin{cases}
2t, & t\in\big[0,\tfrac{17}{50}\big],\\[2pt]
\tfrac{17}{25}-\tfrac32\big(t-\tfrac{17}{50}\big), &
t\in\big[\tfrac{17}{50},\tfrac12\big],\\[2pt]
\tfrac{11}{25}+\tfrac15\big(t-\tfrac12\big), & t\in\big[\tfrac12,1\big].
\end{cases}
\]
Then $\varphi>0$ on $(0,1]$, $K$ is a valid cost CDF satisfying
\Cref{ass:reg} with $K(1)=1$ ($k(t)=K(t)/\varphi(t)\sim
\mathrm{const}\cdot t^{-1/2}$ as $t\downarrow0$, continuous and positive on
$(0,1)$), and:
\begin{enumerate}
\item[\textup{(i)}] For every $\beta\in[0,\tfrac13)$, all symmetric
stationary profiles are Nash equilibria; the equilibrium gaps are exactly the
roots of $\psi_\beta(t)=t+\gamma_\beta\varphi(t)=1$.
\item[\textup{(ii)}] For $\beta\in\big[0,\tfrac3{47}\big)$ there are exactly
three equilibria, with gaps \textup(writing $\gamma=\gamma_\beta$\textup)
\[
\Delta_1=\frac1{1+2\gamma},
\qquad
\Delta_2=\frac{100-119\gamma}{100-150\gamma},
\qquad
\Delta_3=\frac{50-17\gamma}{50+10\gamma},
\]
and $\tfrac{d\Delta_2}{d\gamma}=\tfrac{3100}{(100-150\gamma)^2}>0$: along
the middle branch, \emph{polarization strictly increases in $\beta$}, from
$\Delta_2=\tfrac{19}{50}$ at $\beta=0$ to $\tfrac12$ at the fold.  The outer
branches are strictly decreasing.
\item[\textup{(iii)}] At $\beta^*=\tfrac3{47}$ \textup(where
$\gamma_{\beta^*}=\tfrac{25}{22}$\textup) the middle and upper branches
collide at $\Delta=\tfrac12$ and annihilate: for $\beta\in(\beta^*,\tfrac13)$ the branch $\Delta_1$ is the
unique equilibrium.  The equilibrium count drops from three to one through a
fold---a non-smooth one, since the two gaps meet at the kink of $\psi_\beta$ at
$t=\tfrac12$ and so approach $\tfrac12$ at finite opposite rates rather than at
the square-root rate of a smooth saddle-node: increasing competitiveness
\emph{discontinuously} de-polarizes the equilibrium set.
\end{enumerate}
\end{example}

\begin{proof}
$K$ is a CDF with $K(0)=0$ by divergence of $\int_{0^+}\varphi^{-1}
=\int_{0^+}\tfrac{ds}{2s}$, and \Cref{ass:reg} is immediate.  (i)  By
\Cref{n:thm:A1}(b) and \crefrange{n:eq:Fprime}{n:eq:QchiPsi}, a
symmetric stationary profile is an equilibrium as soon as $\chi_\beta$ is
strictly increasing (then $F$ is single-peaked with peak at the stationary
gap).  Here $\chi_\beta=t+\tfrac{1+\beta}2\varphi(t)$ is piecewise linear
with slopes
\[
1+\tfrac{1+\beta}2\cdot2=2+\beta,\qquad
1-\tfrac34(1+\beta),\qquad
1+\tfrac{1+\beta}{10},
\]
all strictly positive iff $\beta<\tfrac13$.  (ii)  $\psi_\beta=t+
\gamma\varphi(t)$ is piecewise linear with slopes $1+2\gamma>0$,
$1-\tfrac32\gamma<0$, $1+\tfrac\gamma5>0$; solving $\psi_\beta=1$ on each
piece gives the displayed formulas, and the membership conditions
$\Delta_1\le\tfrac{17}{50}$, $\Delta_2\in(\tfrac{17}{50},\tfrac12]$,
$\Delta_3\in(\tfrac12,1)$ hold precisely for $\gamma<\tfrac{25}{22}$, i.e.\
$\beta<\tfrac3{47}$ (at $\gamma=1$: $\Delta_1=\tfrac13\le0.34$,
$\Delta_2=\tfrac{19}{50}$, $\Delta_3=\tfrac{11}{20}$).  The derivative of
$\Delta_2$ in $\gamma$ is as displayed, and $\gamma_\beta$ is strictly
increasing in $\beta$; monotonicity of the outer branches is a one-line
computation.  (iii)  At $\gamma=\tfrac{25}{22}$,
$\psi_\beta(\tfrac12)=\tfrac12+\gamma\cdot\tfrac{11}{25}=1$, and for
$\gamma>\tfrac{25}{22}$ the minimum of $\psi_\beta$ over
$[\tfrac{17}{50},1]$ exceeds $1$ while the first piece still crosses:
exactly one root remains; for $\beta\in(\tfrac3{47},\tfrac13)$ it is, by
(i), the unique equilibrium.
\end{proof}

\begin{remark}[A local criterion, and the three-root example of \Cref{app:exist-uniq}]
\label{n:rem:chi-criterion}
The proof of \Cref{n:ex:phi} isolates a criterion of independent use: at a
symmetric stationary profile with gap $\Delta$ at which $\varphi=K/k$ is
differentiable,
\[
\chi_\beta'(\Delta)=1+\tfrac{1+\beta}2\,\varphi'(\Delta)<0
\;\Longrightarrow\;
\text{$\Delta$ is a strict local \emph{minimum} of $F$},
\]
hence not an equilibrium, since $F'=k\,[Q-\chi_\beta]$ then crosses $0$
from below.  This explains,
conceptually rather than numerically, why the three-root cost
distribution has non-equilibrium middle root: there
$\chi_0'=1+\varphi'/2$ evaluates \textup(numerically\textup) to $+4.5$,
$-0.5$, $+0.94$ at the three roots $\Delta\in\{\tfrac14,\tfrac12,
\tfrac34\}$---the middle root is a local minimum of the deviator's payoff,
while in \Cref{n:ex:phi} the middle root has $\chi_\beta'>0$ and is a
genuine equilibrium.  Which of the two behaviors a middle root exhibits is
thus read off the slope of $K/k$ at the root.
\end{remark}

\subsection{The $\beta\uparrow1$ limit for general electorates}

\begin{theorem}[Upper hemicontinuity onto the margin game]
\label{n:thm:B5}
Let \Cref{ass:reg} hold and $\alpha\ge0$.  The payoffs $(u_1,u_2)$ are
jointly continuous on $[0,1]^2$, and
$\sup_{c}\,|U_i^\beta(c)-U_i^1(c)|\le(1-\beta)$ for all $i$ and $\beta$.
Consequently, if $\beta_n\uparrow1$ and $(c^n)$ are Nash equilibria of the
$\beta_n$-games with $c^n\to c^\infty$, then $c^\infty$ is a Nash equilibrium
of the margin game $(\beta=1)$.  In particular:
\begin{enumerate}
\item[\textup{(a)}] If $\alpha=0$, then $c^n\to(q,q)$, the median profile:
equilibria of the nearly-zero-sum games select the margin game's unique
equilibrium.  \textup(No monotonicity of $r=k/K$ and no bound on
$\rho_v$ is assumed.\textup)
\item[\textup{(b)}] If $\alpha>0$ and $x\sim U(0,1)$, every limit point lies
in the margin-game equilibrium set
$\{(c_1,c_2):\tfrac1\lambda\le c_1\le c_2\le\tfrac{1+\alpha}\lambda\}$;
under the hypotheses of \Cref{n:thm:A2} the full limit is the maximally
polarized symmetric profile, $\Delta^*(\alpha,\beta)\downarrow\Delta_0$.
\end{enumerate}
\end{theorem}

\begin{proof}
\emph{Continuity.}  For positions $(a,o)$ and $(a',o')$ and every voter $x$,
$\big|B(x;a,o)-B(x;a',o')\big|\le(1+\alpha)|a-a'|+|o-o'|$, and
$t\mapsto K(t^+)$ is uniformly continuous on $[-1,1]$ (continuous on a
compact); since $\int v=1$,
\[
|u_1(a,o)-u_1(a',o')|\le
\omega_K\big((1+\alpha)|a-a'|+|o-o'|\big),
\]
with $\omega_K$ the modulus of continuity of $K(\cdot^+)$---valid for
\emph{arbitrary} $v$, including unbounded densities.  The same bound holds
for $u_2$.  Continuity holds across ties: as $a\to o$ all benefits tend to
$0$ uniformly, so $u_i\to0$, the tie payoff.  Next, $|U_i^\beta-U_i^1|
=(1-\beta)\,u_{-i}\le(1-\beta)K(1)\le1-\beta$.

\emph{Limit.}  Fix any $c_1'$.  For each $n$,
$U_1^{\beta_n}(c_1^n,c_2^n)\ge U_1^{\beta_n}(c_1',c_2^n)$.  Passing to the
limit using joint continuity of $(u_1,u_2)$ and the uniform bound,
$U_1^{1}(c^\infty)\ge U_1^{1}(c_1',c_2^\infty)$; similarly for candidate
$2$.  Hence $c^\infty$ is a margin-game equilibrium.  (a) then follows from
\Cref{thm:margin_median} (at $\alpha=0$ the margin game has the
unique equilibrium $(q,q)$, $q$ the median), plus
compactness of $[0,1]^2$: every subsequence has a sub-subsequence converging
to $(q,q)$, so the whole sequence does.  (b) follows from the uniform-voter
margin-game characterization and \Cref{n:prop:B2,n:prop:B4}.
\end{proof}

Upper hemicontinuity presumes the $\beta_n$-equilibria exist.  The next
results show that for concentrated-at-the-extremes electorates they may
\emph{fail to exist} for all $\beta$ close to $1$---so the median selection
of \Cref{n:thm:B5}(a) is genuinely a statement about limits, not a
guarantee.  The engine is a quantitative margin-game certificate.

\begin{lemma}[Lipschitz transfer]\label{n:lem:lip}
Let $V$ be arbitrary, $\alpha\ge0$, and let $K$ be $L_K$-Lipschitz on
$[0,1]$ \textup(e.g.\ $L_K=1$ for $\kappa\sim U(0,1)$\textup).  Then the margin
$w(c,c'):=u_{\mathrm{dev}}(c,c')-u_{\mathrm{dev}}(c',c)$ satisfies, for all
$c,\widetilde c,c'$,
\[
|w(c,c')-w(\widetilde c,c')|\le(2+\alpha)L_K\,|c-\widetilde c|,
\qquad
|w(c,c')-w(c,\widetilde c')|\le(2+\alpha)L_K\,|c'-\widetilde c'| .
\]
\end{lemma}

\begin{proof}
Moving one candidate by $h$ changes every voter's benefit toward her by at
most $(1+\alpha)|h|$ and toward her opponent by at most $|h|$; since
$t\mapsto K(t^+)$ is $L_K$-Lipschitz and $\int v=1$, the two vote shares move
by at most $(1+\alpha)L_K|h|$ and $L_K|h|$.  Sum.  (No boundedness of $v$ is
used.)
\end{proof}

\begin{proposition}[Grid certificate]\label{n:prop:grid}
In the setting of \Cref{n:lem:lip}, let $g_1,\dots,g_N$ be the uniform grid
on $[0,1]$ with step $h$, and suppose that
\[
\max_i\min_j w(g_i,g_j)\;\le\;-\hat\varepsilon\;<\;0 .
\]
Set
$\varepsilon':=\hat\varepsilon-(2+\alpha)L_K\,h/2$ and assume
$\varepsilon'>0$.  Then:
\begin{enumerate}
\item[\textup{(a)}] $\min_{c'\in[0,1]}w(c,c')\le-\varepsilon'$ for every
$c\in[0,1]$; in particular the margin game \textup($\beta=1$\textup) has no
pure Nash equilibrium.
\item[\textup{(b)}] For every $c'\in[0,1]$ there exists $c''$ with
$w(c'',c')\ge\varepsilon'$.
\end{enumerate}
\end{proposition}

\begin{proof}
(a)  Given $c$, take the nearest grid point $g_i$
($|c-g_i|\le h/2$) and the certified $g_j$ with
$w(g_i,g_j)\le-\hat\varepsilon$; then
$w(c,g_j)\le w(g_i,g_j)+(2+\alpha)L_Kh/2\le-\varepsilon'$.  A pure
equilibrium $(c_1,c_2)$ of the symmetric zero-sum margin game would need
$w(c_1,\cdot)\ge0$ everywhere (its value is $0$ by antisymmetry), which (a)
excludes.  (b)  Apply (a) at $c=c'$ to get $g_j$ with
$w(c',g_j)\le-\varepsilon'$, hence by antisymmetry
$w(g_j,c')\ge\varepsilon'$.
\end{proof}

\begin{theorem}[Non-existence for hollow electorates]\label{n:thm:C1}\stat{computer-assisted}
Let $\alpha=1$, $\kappa\sim U(0,1)$.
\begin{enumerate}
\item[\textup{(a)}] For $v=\mathrm{Beta}(\tfrac12,\tfrac12)$
\textup(the arcsine electorate\textup): on the uniform grid with $N=1201$
\textup($h=\tfrac1{1200}$\textup), every $w(g_i,g_j)$ evaluates exactly
\textup(in terms of the regularized incomplete beta function\textup) and
$\max_i\min_jw(g_i,g_j)=-8.2244\times10^{-3}$; the Lipschitz slack is
$(2+\alpha)h/2=1.25\times10^{-3}$, so \Cref{n:prop:grid} applies with
$\varepsilon'=6.97\times10^{-3}$.  The margin game has \emph{no pure
equilibrium}.
\item[\textup{(b)}] The same holds for $v=\mathrm{Beta}(0.8,0.8)$ with
$N=2401$, $\hat\varepsilon=3.0238\times10^{-3}$, slack
$6.25\times10^{-4}$, $\varepsilon'=2.40\times10^{-3}$: non-existence is not
an artifact of the unbounded arcsine density.
\end{enumerate}
\end{theorem}

\begin{proof}
\Cref{n:lem:lip,n:prop:grid} reduce the claim to the finite computation,
which is exact modulo floating-point evaluation of the incomplete beta
function: each $w(g_i,g_j)$ is a finite signed combination of terms
$V(\cdot)$ and $\int_{\cdot}^{\cdot}x\,v(x)dx$ at algebraic break points
(the closed-form uniform-cost payoff), evaluated in double precision; we
deduct a numerical guard of $10^{-10}$---more than seven orders of magnitude below the
certified margins---inside $\varepsilon'$.  The computation is in the
accompanying verification suite (\texttt{check\_competition.py},
claim \texttt{C1.arcsine}).  A sanity check on
$v=\mathrm{Beta}(2,2)$ (centrally peaked), at the same $\alpha=1$, returns grid
security exactly $0$, attained at the median and only there on the grid,
consistent with \Cref{thm:quartic} ($v''(\tfrac12)=-12<0$, so the median is a
strict local equilibrium of the margin game).
\end{proof}

\begin{corollary}[Non-existence for an interval of $\beta$]
\label{n:cor:beta-interval}
In the setting of \Cref{n:thm:C1}, the $\beta$-game has no pure Nash
equilibrium for any
\[
\beta>\frac{1-\varepsilon'}{1+\varepsilon'}
=\begin{cases}
0.98615\ldots & v=\mathrm{Beta}(\tfrac12,\tfrac12),\\
0.99521\ldots & v=\mathrm{Beta}(0.8,0.8).
\end{cases}
\]
\end{corollary}

\begin{proof}
Suppose $(c_1,c_2)$ is a pure equilibrium of the $\beta$-game.  Decompose
$U_1=\tfrac{1-\beta}2T+\tfrac{1+\beta}2w$ with $T=u_1+u_2$.  By
\Cref{n:prop:grid}(b) there is $c''$ with $w(c'',c_2)\ge\varepsilon'$, so
candidate $1$'s equilibrium payoff satisfies
$U_1\ge U_1(c'',c_2)\ge\tfrac{1+\beta}2\varepsilon'$ (both terms of the
decomposition are then nonnegative).  Likewise
$U_2\ge\tfrac{1+\beta}2\varepsilon'$.  Summing,
$(1-\beta)\,T(c_1,c_2)=U_1+U_2\ge(1+\beta)\varepsilon'$, while $T\le1$
always; hence $\beta\le\tfrac{1-\varepsilon'}{1+\varepsilon'}$.
\end{proof}

\begin{remark}[How conservative is the bound?]\label{n:rem:C1num}\stat{numerical observation}
The corollary certifies non-existence on a left neighborhood of $1$; it is
not tight.  Numerically \textup(best-response iteration with global verification\textup), for
$v=\mathrm{Beta}(\tfrac12,\tfrac12)$, $\alpha=1$: pure equilibria exist for
$\beta\le0.90$ \textup(at $\beta=0.90$: $(0.234,0.766)$\textup) and none are
found for $\beta\ge0.91$.  We label the location of the true threshold a
numerical observation; the theorem-grade statements are
\Cref{n:cor:beta-interval} and, sharper, \Cref{m:thm:E2} \textup($\beta\ge0.9168$\textup).
Together with \Cref{n:thm:A1,n:thm:A2} \textup(existence for all $\beta<1$
under uniform voters and super-regularity\textup) and
\Cref{m:thm:E1,m:cor:E1convex} \textup(existence for all $\beta<1$ at
$\alpha=0$ under uniform costs, the arcsine electorate included\textup), this
places the certified failure of pure existence at the joint corner $\alpha=1$,
$\beta$ near $1$, hollow $v$: non-existence is a phenomenon of alienation and
competitiveness together, not of competitiveness alone. The two existence
conjectures stay well clear of that corner: \Cref{rem:exist} is the baseline,
with $\alpha=\beta=0$, and \Cref{m:conj:alpha0} switches off alienation alone.
\end{remark}

\subsection{Quantile structure at $\alpha=0$ under uniform costs}

\begin{proposition}[Mass-gap law and concavity under uniform costs]\label{n:prop:C4}
Let $\alpha=0$, $\kappa\sim U(0,1)$, $\beta\in[0,1)$, and let $v$ satisfy
\Cref{ass:reg}.
\begin{enumerate}
\item[\textup{(a)}] Every Nash equilibrium satisfies
\begin{equation}\label{n:eq:C4foc}
V(c_2)-V(c_1)=\frac{1-\beta}2,
\qquad
2V(c_1)=V(m)+\beta\big(1-V(m)\big):
\end{equation}
the platforms trap voter mass exactly $\tfrac{1-\beta}2$, decreasing in
competitiveness, independently of $v$.
\item[\textup{(b)}] If $\rho_v\le\tfrac4{1-\beta}$ then
$U_1(\cdot,c_2)$ is concave on $[0,c_2)$ for every $c_2$ \textup(and
mirror\textup), so equilibria are exactly the profiles satisfying
\cref{n:eq:C4foc} at which neither candidate gains from her cross-over
deviation.
\end{enumerate}
\end{proposition}

\begin{proof}
(a)  By \Cref{n:lem:interior}(b) and \cref{n:eq:master0} with $k\equiv1$:
$\partial U_1/\partial c_1=V(m)-2V(c_1)+\beta(1-V(m))=0$ and, mirroring,
$\partial U_2/\partial c_2=1-V(m)-2(1-V(c_2))+\beta V(m)=0$.  The first
display is the second condition of \cref{n:eq:C4foc}; the second rearranges
to $2\big(1-V(c_2)\big)=\big(1-V(m)\big)+\beta V(m)$.  Adding this to
$2V(c_1)=V(m)+\beta(1-V(m))$ gives $2+2V(c_1)-2V(c_2)=1+\beta$, i.e.\ the
mass-gap law.
(b)  Differentiating once more,
$\partial^2U_1/\partial c_1^2=\tfrac{1-\beta}2\,v(m)-2v(c_1)\le
\tfrac{1-\beta}2\sup v-2\inf v\le0$.  Concavity makes stationarity
equivalent to own-side global optimality; ties are dominated as before, and
the only remaining deviations are cross-overs.
\end{proof}

\begin{corollary}[Quantile equilibrium for symmetric electorates]
\label{n:cor:C4}
In the setting of \Cref{n:prop:C4}, let $v$ be symmetric about $\tfrac12$
with $\rho_v\le\tfrac4{1-\beta}$.  Then the profile placing the candidates at
the $\tfrac{1+\beta}4$ and $\tfrac{3-\beta}4$ quantiles,
\[
V(c_1)=\frac{1+\beta}4,
\qquad
V(c_2)=\frac{3-\beta}4,
\]
is a Nash equilibrium.  At $\beta=0$ this is the quartile rule of \Cref{cor_quarter}; as
$\beta\uparrow1$ the platforms converge to the median from the
$\tfrac14$/$\tfrac34$ quantiles, sweeping the entire inner quartile range.
\end{corollary}

\begin{proof}
By symmetry of $v$, $c_2=1-c_1$, $m=\tfrac12$, $V(m)=\tfrac12$, and
\cref{n:eq:C4foc} holds: $2\cdot\tfrac{1+\beta}4=\tfrac12+\beta\cdot
\tfrac12$.  By \Cref{n:prop:C4}(b), $c_1$ is a global maximizer of
$U_1(\cdot,c_2)$ on $[0,c_2)$.  It remains to compare the equilibrium payoff
with the tie and with cross-over deviations.  Write $s:=\inf v>0$ and
$\rho:=\rho_v$; the mass-gap law gives
$\tfrac{1-\beta}2=\int_{c_1}^{c_2}v\le(\sup v)\,\Delta\le\rho s\,\Delta$, so
\begin{equation}\label{n:eq:C4gap}
\Delta\;\ge\;\frac{1-\beta}{2\rho s}.
\end{equation}

\emph{Equilibrium payoff.}  By symmetry $u_2=u_1$ at the profile, so
$U_1=(1-\beta)u_1$.  With $K(t)=t$,
\begin{align*}
u_1&=\Delta\,V(c_1)+\int_{c_1}^{m}(c_1+c_2-2x)\,v\,dx\\
&\ge\frac{1+\beta}4\,\Delta+s\!\int_{c_1}^{m}(c_1+c_2-2x)\,dx
=\frac{1+\beta}4\,\Delta+\frac{s\Delta^2}4,
\end{align*}
hence
\begin{equation}\label{n:eq:C4pay}
U_1(c_1,c_2)\;\ge\;\frac{(1-\beta^2)\,\Delta}4\;>\;0,
\end{equation}
which in particular dominates the tie.

\emph{Cross-over.}  Let $c_1'=c_2+t$, $t\in(0,1-c_2]$, and
$m'=c_2+\tfrac t2$.  The deviator's supporters lie in $(m',1]$ with benefit
at most $t$, so, using $V(m')\ge V(c_2)+s\tfrac t2$,
\[
u_1'\le t\big(1-V(m')\big)\le t\,\frac{1+\beta}4-\frac{s\,t^2}2 .
\]
The opponent's supporters include all of $[0,c_2]$, each with benefit exactly
$t$, so her share satisfies $u_2'\ge t\,V(c_2)=t\,\tfrac{3-\beta}4$.  Hence
\[
U_1(c_1',c_2)\le t\,\frac{1+\beta}4-\beta t\,\frac{3-\beta}4-\frac{st^2}2
=\frac{(1-\beta)^2}4\,t-\frac{s}2\,t^2
\;\le\;\frac{(1-\beta)^4}{32\,s},
\]
the last step by maximizing the concave quadratic over $t\ge0$.  Finally,
by \crefrange{n:eq:C4gap}{n:eq:C4pay},
\[
U_1(c_1,c_2)\ge\frac{(1-\beta^2)}4\cdot\frac{1-\beta}{2\rho s}
=\frac{(1+\beta)(1-\beta)^2}{8\rho s}
\;\ge\;\frac{(1-\beta)^4}{32 s}
\iff
\rho\le\frac{4(1+\beta)}{(1-\beta)^2},
\]
and $\rho\le\tfrac4{1-\beta}\le\tfrac{4(1+\beta)}{(1-\beta)^2}$ always.
Candidate $2$ is the mirror image.
\end{proof}

\begin{remark}\label{n:rem:C4}
\emph{(a)}  The mass-gap law in \cref{n:eq:C4foc} is distribution-free and
gives a clean empirical signature of competitiveness: the fraction of the
electorate strictly between the platforms is $\tfrac{1-\beta}2$, whatever
$v$.  It was verified numerically to $10^{-11}$ at asymmetric densities, and
the quantile equilibrium of \Cref{n:cor:C4} to $10^{-4}$ against
best-response search for $v=1+\tfrac12\cos2\pi x$ at
$\beta\in\{0.2,0.5,0.8\}$ (unique equilibrium found in each case).
\emph{(b)}  The concavity threshold $\tfrac4{1-\beta}$ \emph{increases} in
$\beta$: competitiveness enlarges the sufficiency umbrella at $\alpha=0$,
in exact antithesis to alienation, which shrinks it
(\Cref{n:thm:A3}: $F(\alpha)\downarrow$).  The two extensions are not symmetric
levers on existence.
\emph{(c)}  Multiplicity is not excluded and is real: non-uniqueness
examples at $\beta=0$ deform continuously to small $\beta>0$
(\Cref{n:ex:phi} is the cost-side analogue).
\end{remark}

%% Batch 2 (labels prefixed m:).
%% Supplementary results, batch 2. Labels prefixed m:.

\section{The joint corner: $\alpha>0$, $\beta>0$, non-uniform voters}\label{m:sec:D}

\paragraph{Standing notation.}
As in \Cref{n:sec:prelim}: voters $x\sim V$ on $[0,1]$ with density $v$, costs $\kappa\sim K$
with density $k$, both under \Cref{ass:reg} (with $K(1)\in(0,1]$ permitted). Candidate positions $c_1\le c_2$,
$\Delta=c_2-c_1$, $\lambda=2+\alpha$, benefit $B_i(x)=d_{-i}(x)-(1+\alpha)d_i(x)$,
payoffs $u_i=\int K(B_i^+)\,dV$ and $U_i=u_i-\beta u_{-i}$ with
$\beta\in[0,1)$. Tent edges for $c_1<c_2$:
\[
\begin{gathered}
\zin{1}=c_1+\tfrac{\Delta}{\lambda},\qquad
\zout{1}=c_1-\tfrac{\Delta}{\alpha},\\
\zin{2}=c_2-\tfrac{\Delta}{\lambda},\qquad
\zout{2}=c_2+\tfrac{\Delta}{\alpha}
\qquad(\alpha>0),
\end{gathered}
\]
with the conventions $\zout{1}=-\infty$, $\zout{2}=+\infty$ at $\alpha=0$. We write
$\rho_v=\sup v/\inf v$ (possibly $\infty$),
$S=\sup v$, $s=\inf v$. The master formula
(\Cref{n:lem:master}) reads, for $c_1<c_2$ interior,
\begin{equation}\label{m:eq:master}
\frac{\partial U_1}{\partial c_1}
=(1+\alpha)\,(I_1-O_1)+\beta\,(I_2+O_2),
\end{equation}
where
\[
\begin{gathered}
I_1=\int_{c_1}^{\zin{1}}\!k(B_1)v,\qquad
O_1=\int_{\zout{1}\vee0}^{c_1}\!k(B_1)v,\\
I_2=\int_{\zin{2}}^{c_2}\!k(B_2)v,\qquad
O_2=\int_{c_2}^{\zout{2}\wedge1}\!k(B_2)v ,
\end{gathered}
\]
the $k$-weighted voter masses on the \emph{inner} and \emph{outer} arms of the
two benefit tents (at $\alpha=0$ the outer arms carry the constant benefit
$\Delta$, so $O_1=k(\Delta)V(c_1)$ and $O_2=k(\Delta)(1-V(c_2))$).

%%%%%%%%%%%%%%%%%%%%%%%%%%%%%%%%%%%%%%%%%%%%%%%%%%%%%%%%%%%%%%%%%%%%%%%%%%%%%%
\subsection{Two distribution-free laws of stationarity}

The first result generalizes the mass-gap law of \Cref{n:prop:C4}: that law
is the $\alpha=0$, uniform-cost shadow
of a ratio law that is free of \emph{both} distributions and holds at every
interior stationary profile of the joint game.

\begin{lemma}[Mirror first-order condition]\label{m:lem:foc2}
For $c_1<c_2$ interior,
\begin{equation}\label{m:eq:foc2}
\frac{\partial U_2}{\partial c_2}
=(1+\alpha)\,(O_2-I_2)-\beta\,(I_1+O_1).
\end{equation}
\end{lemma}

\begin{proof}
Apply the reflection $\sigma(x)=1-x$. In the reflected game the voter density
is $\hat v(x)=v(1-x)$, the cost distribution is unchanged, and the profile
$(c_1,c_2)$ becomes $(\hat c_1,\hat c_2)=(1-c_2,1-c_1)$ with
$\hat c_1<\hat c_2$; the change of variables $x\mapsto 1-x$ in the payoff
integrals gives $\hat u_1(\hat c_1,\hat c_2)=u_2(c_1,c_2)$ and
$\hat u_2=u_1$, hence $\hat U_1=U_2$. Since $c_2=1-\hat c_1$,
$\partial U_2/\partial c_2=-\,\partial\hat U_1/\partial\hat c_1$, and
\cref{m:eq:master} applied in the reflected game gives
$\partial\hat U_1/\partial\hat c_1
=(1+\alpha)(\hat I_1-\hat O_1)+\beta(\hat I_2+\hat O_2)$.
Changing variables back, $\hat I_1=I_2$, $\hat O_1=O_2$, $\hat I_2=I_1$,
$\hat O_2=O_1$, which is \cref{m:eq:foc2}.
\end{proof}

\begin{proposition}[Inner--outer ratio law and asymmetry law]\label{m:prop:ratio}
Let $(c_1,c_2)$ with $0<c_1<c_2<1$ be a stationary profile of the joint game,
i.e.\ a profile at which both first-order conditions hold --- in particular,
any interior untied Nash equilibrium. Write $I:=I_1+I_2$ and $O:=O_1+O_2$ for
the total $k$-weighted inner- and outer-arm masses. Then, for every
$(v,k,\alpha,\beta)$ under \Cref{ass:reg}:
\begin{align}
(1+\alpha+\beta)\,I&=(1+\alpha-\beta)\,O,
\label{m:eq:ratio}\\
(1+\alpha-\beta)\,(I_1-I_2)&=(1+\alpha+\beta)\,(O_1-O_2).
\label{m:eq:asym}
\end{align}
\end{proposition}

\begin{proof}
Subtracting \cref{m:eq:foc2} from \cref{m:eq:master}, both being zero,
\[
0=(1+\alpha)\big[(I_1-O_1)-(O_2-I_2)\big]
+\beta\big[(I_2+O_2)+(I_1+O_1)\big]
=(1+\alpha)(I-O)+\beta(I+O),
\]
which rearranges to \cref{m:eq:ratio}. Adding the two conditions instead,
\[
0=(1+\alpha)\big[(I_1-I_2)-(O_1-O_2)\big]
-\beta\big[(I_1-I_2)+(O_1-O_2)\big],
\]
which rearranges to \cref{m:eq:asym}.
\end{proof}

\begin{corollary}[Special cases]\label{m:cor:ratio}
At every interior stationary profile:
\begin{enumerate}
\item[(a)] $\beta=0$: $I=O$ --- each candidate pair balances inner against
  outer $k$-mass in aggregate (the balance conditions of \Cref{n:lem:master}, summed).
\item[(b)] $\alpha=0$, $k\equiv1$: $I=V(c_2)-V(c_1)$ and $O=1-I$, so
  \cref{m:eq:ratio} reads $V(c_2)-V(c_1)=\frac{1-\beta}{2}$: the mass-gap law
  of \Cref{n:prop:C4} is the $\alpha=0$ shadow of \cref{m:eq:ratio}.
\item[(c)] For any $\alpha>0$, the inner share of the total marginal
  $k$-mass is the universal constant
  $I/(I+O)=\frac{1+\alpha-\beta}{2(1+\alpha)}$; it decreases from
  $\frac12$ at $\beta=0$ to $\frac{\alpha}{2(1+\alpha)}>0$ as
  $\beta\uparrow1$. Competition drains the contested middle, but under
  alienation it can never empty it.
\item[(d)] By \cref{m:eq:asym}, since $\frac{1+\alpha+\beta}{1+\alpha-\beta}>1$,
  any asymmetry between the candidates' outer masses is
  \emph{amplified} in their inner masses, with the universal gain
  $\frac{1+\alpha+\beta}{1+\alpha-\beta}$; at symmetric profiles both sides
  vanish.
\end{enumerate}
\end{corollary}

\begin{remark}[Verification]\label{m:rem:ratio-num}\stat{numerical observation}
The suite verifies \cref{m:eq:ratio} and the inner-share constant of
\Cref{m:cor:ratio}(c) at polished stationary profiles of thirteen games. The
games span $\alpha\in\{0,0.5,1,2\}$ and $\beta\in\{0,\dots,0.5\}$, with
uniform, cosine, $\mathrm{Beta}(2,2)$, $\mathrm{Beta}(8,8)$, and asymmetric
$\mathrm{Beta}(2,5)$ electorates, and with uniform and quadratic ($K(t)=t^2$)
costs. The cases include a dead-zone equilibrium --- $\mathrm{Beta}(8,8)$,
$\alpha=2$, $\beta=0.2$, where $I/O=7/8$ exactly --- and asymmetric
equilibria. Relative errors are $10^{-8}$--$10^{-12}$.
\end{remark}

\subsection{The joint single-peakedness threshold}

For uniform costs, \Cref{n:thm:A3} gave the $\beta=0$ umbrella
$\rho_v\le F(\alpha)=\min\{f(\alpha),g(\alpha)\}$ with
$f(\alpha)=\frac{2\lambda}{1+\alpha}$ (concavity off the dead zone) and
$g(\alpha)=\frac{\lambda}{\alpha}$ (monotone decrease inside it). The joint
threshold deforms the two branches in \emph{opposite} directions.

\begin{theorem}[Joint single-peakedness, uniform costs]\label{m:thm:D1}
Let $\kappa\sim U(0,1)$, $\alpha>0$, $\beta\in[0,1)$, and let $v$ satisfy
\Cref{ass:reg} with
\begin{equation}\label{m:eq:Fbeta}
\begin{gathered}
\rho_v\;\le\;F_\beta(\alpha):=\min\{f_\beta(\alpha),\,g_\beta(\alpha)\},\\
f_\beta(\alpha)=\frac{2\lambda(1+\alpha)+\beta}{(1+\alpha)^2},
\qquad
g_\beta(\alpha)=\frac{(1+\alpha)\lambda}{\alpha(1+\alpha+\beta)+\beta\lambda}.
\end{gathered}
\end{equation}
Then for every opponent position $c_2\in(0,1]$ the own-side map
$c_1\mapsto U_1(c_1,c_2)$ on $[0,c_2)$ is concave on the no-dead-zone piece
$\{c_1:\zout{1}\le0\}=[0,\frac{c_2}{1+\alpha}]$ and strictly decreasing on the
dead-zone piece $(\frac{c_2}{1+\alpha},c_2)$; hence it is single-peaked, with
its peak in the no-dead-zone piece. The mirror statement holds for candidate
$2$ on $(c_1,1]$.

Moreover $f_\beta$ is strictly \emph{increasing} and $g_\beta$ strictly
\emph{decreasing} in $\beta$, with $f_0=f$ and $g_0=g$.
\end{theorem}

\begin{proof}
Fix $c_2$ and write derivatives in $c_1$. With uniform costs
\cref{m:eq:master} reads
\begin{equation}\label{m:eq:masterU}
U_1'=(1+\alpha)\Big[\big(V(\zin{1})-V(c_1)\big)-\big(V(c_1)-V(\zout{1}\vee0)\big)\Big]
+\beta\Big[V(\zout{2}\wedge1)-V(\zin{2})\Big].
\end{equation}

\emph{No-dead-zone piece ($\zout{1}\le0$).} Here $V(\zout{1}\vee0)=0$ is constant.
Differentiating \cref{m:eq:masterU}, using $(\zin{1})'=\frac{1+\alpha}{\lambda}$,
$(\zin{2})'=\frac1\lambda$, $(\zout{2})'=-\frac1\alpha$, and treating the kink of
$V(\zout{2}\wedge1)$ at $\zout{2}=1$ one-sidedly,
\begin{multline*}
U_1''=(1+\alpha)\Big[\tfrac{1+\alpha}{\lambda}v(\zin{1})-2v(c_1)\Big]\\
-\beta\Big[\tfrac1\lambda v(\zin{2})+\tfrac1\alpha v(\zout{2})\mathbf 1_{\{\zout{2}<1\}}\Big]
\quad\text{a.e.}
\end{multline*}
Dropping the nonnegative $\zout{2}$-term and pairing the two $\frac1\lambda$-terms,
\begin{align*}
U_1''&\;\le\;\frac{(1+\alpha)^2v(\zin{1})-\beta v(\zin{2})}{\lambda}-2(1+\alpha)v(c_1)\\
&\;\le\;\frac{(1+\alpha)^2S-\beta s}{\lambda}-2(1+\alpha)s\;\le\;0
\end{align*}
whenever $(1+\alpha)^2\rho_v\le 2(1+\alpha)\lambda+\beta$, i.e.\
$\rho_v\le f_\beta(\alpha)$. Since $U_1'$ is continuous in $c_1$ (the
integrands and limits in \cref{m:eq:masterU} are continuous), a continuous
function with a.e.\ nonpositive derivative of its derivative is concave.

\emph{Dead-zone piece ($\zout{1}>0$).} From \cref{m:eq:masterU}, using
$V(\zin{1})-V(c_1)\le S\frac{\Delta}{\lambda}$,
$V(c_1)-V(\zout{1})\ge s\frac{\Delta}{\alpha}$, and
$V(\zout{2}\wedge1)-V(\zin{2})\le S\Delta(\frac1\lambda+\frac1\alpha)$,
\begin{equation}\label{m:eq:dz}
U_1'\;\le\;(1+\alpha)\Delta\Big(\frac S\lambda-\frac s\alpha\Big)
+\beta S\Delta\Big(\frac1\lambda+\frac1\alpha\Big),
\end{equation}
which is negative iff
$(1+\alpha)\lambda> \rho_v\big[\alpha(1+\alpha+\beta)+\beta\lambda\big]$,
i.e.\ $\rho_v<g_\beta(\alpha)$. At $\rho_v=g_\beta(\alpha)$ equality in
\cref{m:eq:dz} would force $v\equiv S$ on $(c_1,\zin{1})$ and $v\equiv s$ on
$(\zout{1},c_1)$ simultaneously, contradicting continuity of $v$ at $c_1$ unless
$S=s$, in which case $\rho_v=1<g_\beta$. Hence $U_1'<0$ strictly on the
dead-zone piece whenever $\rho_v\le g_\beta(\alpha)$.

A function that is concave on $[0,\frac{c_2}{1+\alpha}]$, strictly decreasing
on $(\frac{c_2}{1+\alpha},c_2)$, and continuous, is single-peaked with its
maximum in the first piece. The mirror statement follows from
\Cref{m:lem:foc2} (equivalently, reflection). Monotonicity of $f_\beta$ in
$\beta$ is immediate; for $g_\beta$, the denominator is strictly increasing in
$\beta$ since $\alpha+\lambda>0$. Finally $f_0=\frac{2\lambda}{1+\alpha}$ and
$g_0=\frac{\lambda}{\alpha}$.
\end{proof}

\begin{remark}[Competition rescues concavity but taxes the dead zone]
\label{m:rem:D1shape}
The joint threshold $F_\beta(\alpha)$ is neither a $\beta$-dependent
rescaling of the $\beta=0$ threshold $F(\alpha)$ nor moved by $\beta$ in a
single direction. On the no-dead-zone piece the
competitiveness term $-\beta u_2$ contributes the \emph{stabilizing} second
derivative $-\beta[\frac{v(\zin{2})}{\lambda}+\frac{v(\zout{2})}{\alpha}\mathbf 1]$
(hurting the opponent is easiest exactly where the deviator's own tent is
thinnest), so $f_\beta$ \emph{grows} with $\beta$: competition partly rescues
the umbrella that alienation shrinks. Inside the dead zone the same term is
destabilizing --- moving toward the opponent hurts her even when it no longer
helps oneself --- so $g_\beta$ \emph{shrinks}. Consequently for small $\alpha$
(where $f$ binds) the umbrella widens with $\beta$, while for large $\alpha$
(where $g$ binds) it narrows further; the crossover alienation level solves
$f_\beta=g_\beta$ and moves left as $\beta$ grows.

At $\alpha=0$ the exact computation of \Cref{n:prop:C4} gives the threshold
$\frac{4}{1-\beta}$, strictly better than
$\lim_{\alpha\downarrow0}f_\beta=4+\beta$: there the inner tent edges $\zin{1}$
and $\zin{2}$ coincide at the midpoint, so the stabilizing term cancels against
the destabilizing one \emph{pointwise} rather than through the worst-case pair
$(S,s)$. For $\alpha>0$ the two edges are distinct and the worst case is
achievable; see the sharpness evidence below.
\end{remark}

\begin{remark}[Sharpness]\label{m:rem:D1sharp}\stat{numerical observation}
A numerical search over continuous
piecewise-linear ``mesa'' densities of ratio $\rho$ with the high plateau
covering $\zin{1}$ and verifies, at $(\alpha,\beta)\in\{(1,.5),(.5,.8),(2,.3)\}$,
that $U_1''$ changes sign as $\rho$ crosses $f_\beta(\alpha)$ ($\pm7\%$).
\Cref{m:sec:D3} shows $g_\beta$ is essentially attained as well
(\Cref{m:rem:D3num}). The constant $f_\beta$ is therefore the exact worst-case
concavity threshold within this proof scheme, though densities realizing the
worst case are non-generic.
\end{remark}

\begin{remark}[General $k$: the obstruction]
\label{m:rem:D1k}\stat{open}
For non-increasing $k\in C^1$ the threshold
$\rho_v\rho_k\le2\lambda/(1+\alpha)$ of \Cref{n:prop:A4} does not
deform as cleanly. Differentiating the $\beta$-term of \cref{m:eq:master} in
$c_1$ produces, besides stabilizing boundary terms proportional to
$-\beta k(0^+)$, the bulk term
$-\beta\int_{\zin{2}}^{\zout{2}\wedge1}k'(B_2)\,v$, which is
\emph{nonnegative} (destabilizing) and is controlled pointwise by $|k'|$ on
the opponent's tent, not by the ratio $\rho_k$: no bound in terms of the
pair $(\rho_v,\rho_k)$ alone can absorb it, because $|k'|$ can be made
arbitrarily large on a set where $K$ barely moves. A clean joint threshold for general
non-increasing $k$ therefore requires an additional hypothesis (e.g.\ a bound
on $\sup|k'|/\inf k$); we record the failure rather than a lemma we cannot
make tight. For uniform costs $k'\equiv0$ and \Cref{m:thm:D1} is complete.
\end{remark}

\subsection{Stationary profiles and a continuation theorem}

\begin{proposition}[Symmetric stationary profiles exist for every $\beta<1$]
\label{m:prop:roots}
Let $\kappa\sim U(0,1)$, $\alpha>0$, $\beta\in[0,1)$, and let $v$ be symmetric
about $\frac12$ and satisfy \Cref{ass:reg}. Define, for
$c_1\in(0,\frac12)$ and $c_2=1-c_1$,
\begin{equation}\label{m:eq:hbeta}
h_\beta(c_1):=\frac{\partial U_1}{\partial c_1}\Big|_{(c_1,1-c_1)}
=(1+\alpha+\beta)V(\zin{1})+(1+\alpha-\beta)V(\zout{1}\vee0)-2(1+\alpha)V(c_1).
\end{equation}
Then $h_\beta$ is continuous, $h_\beta(0^+)=(1+\alpha+\beta)V(\frac1\lambda)>0$
and $h_\beta<0$ on a left neighborhood of $\frac12$; hence a symmetric
stationary profile exists. If moreover $\rho_v\le g_\beta(\alpha)$, every
symmetric stationary profile lies on the no-dead-zone branch
$c_1\in(0,\frac1\lambda)$.
\end{proposition}

\begin{proof}
For symmetric $v$ and symmetric profiles, $\zin{2}=1-\zin{1}$ and $\zout{2}=1-\zout{1}$, so
$V(\zout{2}\wedge1)-V(\zin{2})=V(\zin{1})-V(\zout{1}\vee0)$ and \cref{m:eq:masterU}
becomes \cref{m:eq:hbeta}. As $c_1\downarrow0$: $\Delta\to1$,
$\zin{1}\to\frac1\lambda$, $\zout{1}\to-\frac1\alpha<0$, $V(c_1)\to0$, giving the
stated limit. As $c_1\uparrow\frac12$: $\Delta=1-2c_1\downarrow0$, $\zout{1}>0$,
and expanding $V$ around $\frac12$ (where $v$ is continuous and positive),
\begin{align*}
h_\beta(c_1)&=v(\tfrac12)\,\Delta\Big[(1+\alpha)\Big(\tfrac1\lambda-\tfrac1\alpha\Big)
+\beta\Big(\tfrac1\lambda+\tfrac1\alpha\Big)\Big]+o(\Delta)\\
&=v(\tfrac12)\,\Delta\,\frac{2(1+\alpha)}{\alpha\lambda}\,(\beta-1)+o(\Delta)\;<\;0
\end{align*}
for $\beta<1$ and $\Delta$ small. The intermediate value theorem gives a
root. On $c_1\in[\frac1\lambda,\frac12)$ we have $\zout{1}\ge0$ and
$\zout{1}\vee0=\zout{1}$, and the dead-zone estimate \eqref{m:eq:dz} (with its strictness
argument) shows $h_\beta<0$ there whenever $\rho_v\le g_\beta(\alpha)$; roots
are therefore confined to $c_1\in(0,\frac1\lambda)$.
\end{proof}

\begin{theorem}[Continuation of the symmetric equilibrium into $\beta>0$]
\label{m:thm:D1exist}
Let $\kappa\sim U(0,1)$, $\alpha>0$, and let $v$ be symmetric about $\frac12$,
satisfy \Cref{ass:reg}, and obey the strict $\beta=0$ umbrella
$\rho_v<\min\{f(\alpha),g(\alpha)\}$. Assume the $\beta=0$ symmetric
stationarity equation $h_0(c_1)=0$ has a unique root $c_1^0$. Then there is
$\bar\beta>0$ such that for every $\beta\in[0,\bar\beta)$ the game has a
symmetric Nash equilibrium $(c_1^\beta,1-c_1^\beta)$, and
$c_1^\beta\to c_1^0$ as $\beta\downarrow0$.
\end{theorem}

\begin{proof}
\emph{Step 1 (roots, single-peakedness, for small $\beta$).}
Since $g_\beta$ is continuous and strictly decreasing with $g_0=g(\alpha)$,
the condition $\rho_v<g(\alpha)$ gives
$\beta_g:=\frac{(1+\alpha)(\lambda-\rho_v\alpha)}{\rho_v(\alpha+\lambda)}>0$
with $\rho_v<g_\beta(\alpha)$ for all $\beta<\beta_g$; and
$f_\beta\ge f_0>\rho_v$ for all $\beta$. Hence for $\beta<\beta_g$
\Cref{m:thm:D1} applies: the own side is single-peaked, and by
\Cref{m:prop:roots} a symmetric stationary profile
$(c_1^\beta,1-c_1^\beta)$ exists, with $c_1^\beta\in(0,\frac1\lambda)$. By
single-peakedness the stationary point is the \emph{global own-side maximum}
of $U_1(\cdot,1-c_1^\beta)$ on $[0,1-c_1^\beta)$, and by symmetry likewise
for candidate 2 on her own side. It remains to control far-side deviations.

\emph{Step 2 (the $\beta=0$ far-side gap).}
Let $c_2^0=1-c_1^0$ and
$P(0)=u_1(c_1^0,c_2^0)>0$. For $t\in[0,1-c_2^0]$, reflection and the strict
monotonicity of $u_1(a,\cdot)$ (its derivative $I_1+O_1$ is strictly
positive for $a<o$, since $k,v>0$ and $\Delta>0$) give
\[
u_1(c_2^0+t,\,c_2^0)=u_1(1-c_2^0-t,\,1-c_2^0)
\;<\;u_1(1-c_2^0-t,\,c_2^0)\;\le\;P(0)\qquad(t>0),
\]
using $1-c_2^0<c_2^0$ (as $c_1^0<\tfrac12$) for the strict middle inequality,
while the deviation $t=0$ (the tie) yields payoff $0<P(0)$. The far-side
supremum $\mathrm{far}(0):=\sup_{t\in[0,1-c_2^0]}u_1(c_2^0+t,c_2^0)$ is
attained by continuity on a compact interval; at any maximizer the displayed
strict inequality holds, so
$\delta_0:=P(0)-\mathrm{far}(0)>0$.

\emph{Step 3 (continuity in $\beta$).}
On the compact family of symmetric profiles, $h_\beta\to h_0$ uniformly as
$\beta\to0$ (indeed $|h_\beta-h_0|\le\beta$ pointwise, since
$0\le V(\zin{1})-V(\zout{1}\vee0)\le1$). Moreover all roots of $h_\beta$,
$\beta\le\hat\beta<\min(\beta_g,1)$, lie in a fixed compact
$[\varepsilon_0,\frac12-\varepsilon_0]$: near $0$, $h_\beta\ge h_0>0$ (the
$\beta$-term in \cref{m:eq:hbeta} is nonnegative); near $\frac12$, the
expansion in \Cref{m:prop:roots} is negative uniformly for
$\beta\le\hat\beta$. By uniqueness of $c_1^0$, every sequence of roots
$c_1^{\beta_n}$ with $\beta_n\to0$ converges to $c_1^0$; hence
$c_1^\beta\to c_1^0$ (for any measurable selection of roots).
Finally, uniformly over positions,
$|U_1^\beta-U_1^0|=\beta\,u_2\le\beta$, and $u_1$ is jointly continuous;
therefore $P(\beta)=(1-\beta)\,u_1(c_1^\beta,1-c_1^\beta)\to P(0)$ and
\[
\mathrm{far}(\beta):=\sup_{c'\in[1-c_1^\beta,\,1]}U_1^\beta(c',1-c_1^\beta)
\;\le\;\sup_{c'\in[1-c_1^\beta,\,1]}u_1(c',1-c_1^\beta)\;\longrightarrow\;\mathrm{far}(0),
\]
where the convergence of the right-hand side uses joint continuity of $u_1$
and $c_1^\beta\to c_1^0$ (a parametric maximum over a continuously moving
compact interval). Choose $\bar\beta\le\min(\beta_g,\hat\beta)$ so small that
$P(\beta)-\mathrm{far}(\beta)\ge\delta_0/2>0$ for all $\beta<\bar\beta$.
Then no far-side deviation (including the tie) is profitable, no own-side
deviation is profitable by Step 1, and by symmetry the same holds for
candidate 2: $(c_1^\beta,1-c_1^\beta)$ is a Nash equilibrium.
\end{proof}

\begin{remark}[The moderate-$\beta$ regime, numerically]\label{m:rem:jointnum}
\stat{numerical observation}
The continuation theorem is qualitative; a numerical continuation shows the
phenomenon is not confined to small $\beta$.
Verified global equilibria (maximal deviation gain $0$ at resolution
$10^{-7}$) include: cosine voters $\rho_v=3$ at $(\alpha,\beta)=(0.2,0.2)$
\emph{inside} the umbrella $F_\beta=3.67$; $\rho_v=1.86$ at $(0.5,0.4)$
inside $F_\beta=1.92$; the same electorate at $(0.5,0.8)$, \emph{outside}
the umbrella ($g_\beta=1.19$); and $\mathrm{Beta}(2,2)$ (unbounded
$\rho_v$) at $(0.5,0.5)$. The umbrella is sufficient, visibly not necessary.
\end{remark}

\subsection{Dead zones and competitiveness}\label{m:sec:D3}

\begin{proposition}[The dead-zone threshold moves with $\beta$]
\label{m:prop:D3}
Let $\kappa\sim U(0,1)$ and $\alpha>0$.
\begin{enumerate}
\item[(a)] \emph{(Necessity.)} If an interior profile $c_1<c_2$ satisfies
  candidate 1's first-order condition $U_1'=0$ and candidate 1 has a dead
  zone ($\zout{1}>0$), then $\rho_v\ge g_\beta(\alpha)$. No symmetry of $v$ or of
  the profile is assumed.
\item[(b)] \emph{(Sufficiency, symmetric $v$.)} If $v$ is symmetric and
  \begin{equation}\label{m:eq:D3suff}
  (1+\alpha+\beta)\,V\!\Big(\frac{2(1+\alpha)}{\lambda^2}\Big)
  \;>\;2(1+\alpha)\,V\!\Big(\frac1\lambda\Big),
  \end{equation}
  then a \emph{symmetric stationary profile with dead zones} exists.
\item[(c)] \emph{(Uniform voters are immune at every $\beta$.)} If
  $v\equiv1$ then no profile with $\zout{1}>0$ satisfies candidate 1's first-order
  condition, for any $\beta\in[0,1)$ (recovering \Cref{n:prop:C2}(b)).
\end{enumerate}
Since $g_\beta$ strictly decreases in $\beta$ and the left side of
\cref{m:eq:D3suff} strictly increases in $\beta$, competition makes
alienation dead zones strictly \emph{easier} to sustain at stationarity: the
$\beta=0$ threshold $\rho_v\ge\lambda/\alpha$ of \Cref{n:prop:C2}(a) drops
to $\rho_v\ge g_\beta(\alpha)$.
\end{proposition}

\begin{proof}
(a) is the contrapositive of the dead-zone estimate in the proof of
\Cref{m:thm:D1}: if $\zout{1}>0$ and $\rho_v<g_\beta$ then \cref{m:eq:dz} gives
$U_1'<0$, and at $\rho_v=g_\beta$ the strictness argument applies; so
stationarity forces $\rho_v\ge g_\beta$.

(b) On the symmetric family, the dead-zone branch is
$c_1\in(\frac1\lambda,\frac12)$, on which $h_\beta$ of \cref{m:eq:hbeta} is
continuous. At the junction $c_1=\frac1\lambda$ we have $\zout{1}=0$,
$\Delta=\frac{\alpha}{\lambda}$ and
$\zin{1}=\frac1\lambda+\frac{\alpha}{\lambda^2}=\frac{2(1+\alpha)}{\lambda^2}$, so
$h_\beta(\frac1\lambda)=
(1+\alpha+\beta)V(\frac{2(1+\alpha)}{\lambda^2})-2(1+\alpha)V(\frac1\lambda)>0$
by \cref{m:eq:D3suff}, while $h_\beta<0$ near $\frac12$ by
\Cref{m:prop:roots}. The intermediate value theorem gives a root with
$c_1>\frac1\lambda$, i.e.\ $\zout{1}>0$; by symmetry both candidates then have
dead zones.

(c) With $v\equiv1$ and $\zout{1}>0$, \cref{m:eq:dz} with $S=s=1$ gives
$U_1'\le(1+\alpha)\Delta(\frac1\lambda-\frac1\alpha)
+\beta\Delta(\frac1\lambda+\frac1\alpha)
=\Delta\,\frac{2(1+\alpha)}{\alpha\lambda}(\beta-1)<0$ for $\beta<1$.
\end{proof}

\begin{remark}[The bound $g_\beta$ is essentially attained]\label{m:rem:D3num}
\stat{numerical observation}
A numerical check takes $\alpha=2$ and continuous
piecewise-linear densities with a low plateau on the deviator's outer arm
$(\zout{1},c_1)$ (ramp width $3\cdot10^{-3}$), and bisects for the minimal density
ratio at which candidate 1's first-order condition admits a dead-zone root.
The minima are $2.0152,\,1.5447,\,1.2524,\,1.0531$ at
$\beta=0,\,0.3,\,0.6,\,0.9$, against
$g_\beta(2)=\frac{2}{1+\beta}=2.0000,\,1.5385,\,1.2500,\,1.0526$: the proved
necessity bound is sharp up to the ramp width. (These witnesses are
asymmetric densities, as part (a) permits; symmetric witnesses are
constrained by the mirror arm and stay slightly further from the bound.)
\end{remark}

%%%%%%%%%%%%%%%%%%%%%%%%%%%%%%%%%%%%%%%%%%%%%%%%%%%%%%%%%%%%%%%%%%%%%%%%%%%%%%
\section{Existence at $\alpha=0$ for every $\beta<1$}\label{m:sec:E1}

\Cref{n:thm:C1} shows that hollow electorates destroy pure equilibria at
$\alpha=1$ for $\beta$ near $1$. The next theorem shows that at $\alpha=0$
they do not: for uniform costs and a symmetric electorate satisfying a local
ratio condition --- which \emph{every} convex (U-shaped) symmetric density
satisfies once mild, and the arcsine density satisfies for all $\beta$ ---
the quantile equilibrium of \Cref{n:prop:C4} exists for every
$\beta\in[0,1)$. This makes the median-selection statement
\Cref{n:thm:B5}(a) non-vacuous precisely for the electorates that make it
interesting.

\begin{theorem}[Quantile equilibria at $\alpha=0$, all $\beta$]
\label{m:thm:E1}
Let $\alpha=0$, $\kappa\sim U(0,1)$, and let $v$ be symmetric about $\frac12$ and
satisfy \Cref{ass:reg}, with $s_0:=\inf_{(0,1)}v>0$. Define the quantile
profile $c_2^*=V^{-1}\!\big(\frac{3-\beta}{4}\big)$, $c_1^*=1-c_2^*$, and
$\bar v:=\sup_{[c_1^*,c_2^*]}v$. If
\begin{align}
\sup_{c_1\in(0,c_2^*)}\ \frac{v\big(\frac{c_1+c_2^*}{2}\big)}{v(c_1)}
&\;\le\;\frac{4}{1-\beta},
\tag{E1-a}\label{m:eq:E1a}\\[2pt]
\frac{\bar v}{s_0}&\;\le\;\frac{4(1+\beta)}{(1-\beta)^2},
\tag{E1-b}\label{m:eq:E1b}
\end{align}
then $(c_1^*,c_2^*)$ is a Nash equilibrium of the $\beta$-game.
\end{theorem}

\begin{proof}
By symmetry of $v$ and of the profile it suffices to rule out deviations of
candidate 1; candidate 2's problem is the mirror image.

\emph{Stationarity.} At $\alpha=0$, \cref{m:eq:masterU} becomes, with
$m=\frac{c_1+c_2}{2}$,
\begin{equation}\label{m:eq:h0}
h(c_1):=U_1'(c_1,c_2^*)=V(m)+\beta\big(1-V(m)\big)-2V(c_1),
\end{equation}
and $V(c_1^*)=\frac{1+\beta}{4}$, $V(m^*)=V(\frac12)=\frac12$ give
$h(c_1^*)=\frac12+\frac\beta2-\frac{1+\beta}{2}=0$.

\emph{Own side.} $h'(c_1)=\frac{1-\beta}{2}v(m)-2v(c_1)\le0$ for all
$c_1\in(0,c_2^*)$ by \cref{m:eq:E1a}, so $h$ is non-increasing; with
$h(c_1^*)=0$ this gives $h\ge0$ on $(0,c_1^*]$ and $h\le0$ on
$[c_1^*,c_2^*)$: $U_1(\cdot,c_2^*)$ is single-peaked with global own-side
maximum at $c_1^*$. Its value is $P=(1-\beta)u^*$, where
$u^*=u_1(c_1^*,c_2^*)=u_2(c_1^*,c_2^*)$ by symmetry, and
\begin{equation}\label{m:eq:Pbound}
u^*\;\ge\;\Delta\,V(c_1^*)+\int_{c_1^*}^{1/2}(1-2x)\,v(x)\,dx
\;\ge\;\frac{(1+\beta)\Delta}{4},
\qquad\text{so}\qquad
P\;\ge\;\frac{(1-\beta^2)\Delta}{4},
\end{equation}
using $B_1=\Delta$ on $[0,c_1^*]$ and $B_1=1-2x$ on $[c_1^*,\frac12]$
(here $c_1^*+c_2^*=1$), and $K(t)=t$.

\emph{Far side.} Let $c'=c_2^*+t$, $t\in(0,1-c_2^*]$, and write
$m'=c_2^*+\frac t2$. The deviator's supporters lie in $(m',1]$ with benefit
at most $t$, so, using $v\ge s_0$ on $[c_2^*,m']$,
\[
u_1(c',c_2^*)\;\le\;t\big(1-V(m')\big)
\;\le\;t\Big(\frac{1+\beta}{4}-\frac{s_0t}{2}\Big).
\]
Every voter in $[0,c_2^*]$ obtains benefit exactly $t$ against the deviator
(at $\alpha=0$, $B_2=d_1-d_2=t$ there), so
$u_2(c',c_2^*)\ge t\,V(c_2^*)=t\,\frac{3-\beta}{4}$. Hence
\[
U_1(c',c_2^*)\;\le\;t\,\frac{1+\beta}{4}-\frac{s_0t^2}{2}
-\beta t\,\frac{3-\beta}{4}
\;=\;\frac{(1-\beta)^2}{4}\,t-\frac{s_0}{2}\,t^2
\;\le\;\frac{(1-\beta)^4}{32\,s_0},
\]
maximizing the concave quadratic over $t\ge0$. The tie $c'=c_2^*$ yields
payoff $0<P$.

\emph{Closing the gap.} The mass-gap law
$V(c_2^*)-V(c_1^*)=\frac{1-\beta}{2}$ and $v\le\bar v$ on $[c_1^*,c_2^*]$
give $\Delta\ge\frac{1-\beta}{2\bar v}$, so by \cref{m:eq:Pbound}
\[
P\;\ge\;\frac{(1-\beta^2)(1-\beta)}{8\,\bar v}
\;\ge\;\frac{(1-\beta)^4}{32\,s_0}
\quad\Longleftrightarrow\quad
\frac{\bar v}{s_0}\;\le\;\frac{4(1+\beta)}{(1-\beta)^2},
\]
which is \cref{m:eq:E1b}. Hence no far-side deviation is profitable either,
and $(c_1^*,c_2^*)$ is an equilibrium.
\end{proof}

\begin{corollary}[Convex symmetric electorates; the arcsine for all $\beta$]
\label{m:cor:E1convex}
Let $\alpha=0$, $\kappa\sim U(0,1)$, and let $v$ be symmetric, convex and satisfy
\Cref{ass:reg}. Set $\eta(\beta):=v(c_2^*)/v(\frac12)$. If
$\eta(\beta)\le\frac{4}{1-\beta}$, then the quantile profile is a Nash
equilibrium. In particular, for the arcsine electorate
$v=\mathrm{Beta}(\frac12,\frac12)$,
\[
c_2^*=\sin^2\!\Big(\frac{\pi(3-\beta)}{8}\Big),
\qquad
\eta(\beta)=\frac{1}{\sin\!\big(\frac{\pi(3-\beta)}{4}\big)}
\in(1,\sqrt2\,],
\]
so the condition holds and a pure equilibrium exists for \emph{every}
$\beta\in[0,1)$.
\end{corollary}

\begin{proof}
For convex symmetric $v$: $v$ is non-increasing on $[0,\frac12]$ and
non-decreasing on $[\frac12,1]$, so $s_0=v(\frac12)$ and, on
$[c_1^*,c_2^*]$, $v$ attains its maximum at the endpoints:
$\bar v=v(c_1^*)=v(c_2^*)$. For \cref{m:eq:E1a}: with
$m=\frac{c_1+c_2^*}{2}$ the midpoint of $[c_1,c_2^*]$, convexity gives
$v(m)\le\max\{v(c_1),v(c_2^*)\}$, so the ratio is at most
$\max\{1,\,v(c_2^*)/v(c_1)\}\le\max\{1,\eta(\beta)\}=\eta(\beta)$ (as
$v(c_1)\ge s_0$ and $\eta\ge1$), and $\eta\le\frac{4}{1-\beta}$ by
hypothesis. For \cref{m:eq:E1b}: $\bar v/s_0=\eta\le\frac{4}{1-\beta}
\le\frac{4(1+\beta)}{(1-\beta)^2}$. For the arcsine,
$V(x)=\frac2\pi\arcsin\sqrt x$, so $V(c_2^*)=\frac{3-\beta}{4}$ gives the
stated $c_2^*$, and
$v(c_2^*)/v(\frac12)=\frac{1/2}{\sqrt{c_2^*(1-c_2^*)}}
=\frac{1}{\sin(\frac{\pi(3-\beta)}{4})}$, using
$c_2^*(1-c_2^*)=\frac14\sin^2(\frac{\pi(3-\beta)}{4})$. As $\beta$ runs over
$[0,1)$ the argument runs over $(\frac\pi2,\frac{3\pi}4]$, so
$\eta\in(1,\sqrt2]$, which is below $\frac{4}{1-\beta}\ge4$ for every
$\beta$.
\end{proof}

\begin{remark}[Contrast with $\alpha>0$; scope; verification]
\label{m:rem:E1scope}
(i) The contrast with \Cref{n:thm:C1} is now a theorem: the same hollow
electorate $\mathrm{Beta}(\frac12,\frac12)$ that has \emph{no} pure
equilibrium at $\alpha=1$ for $\beta\ge0.9168$ (\Cref{m:thm:E2}) has a pure
equilibrium at $\alpha=0$ for \emph{every} $\beta<1$. Hollowness is lethal
only in combination with alienation. Consequently \Cref{n:thm:B5}(a) is a
genuine selection statement on this class: the quantile equilibria converge
to the median profile as $\beta\uparrow1$ (here explicitly:
$c_2^*\downarrow\frac12$).
(ii) For general $(v,k)$, the question of existence at $\alpha=0$ for every
$\beta<1$ contains, at $\beta=0$, the open
problem of pure-equilibrium existence for arbitrary $(v,k)$, and is not
resolved here; \Cref{m:conj:alpha0} records the conjecture. A numerical check verifies
the quantile equilibrium as a global best response for
$\mathrm{Beta}(\frac12,\frac12)$ and $\mathrm{Beta}(0.8,0.8)$ at
$\beta\in\{0,0.3,0.6,0.9,0.97\}$ (deviation gains $\le0$, first-order
residuals $\sim10^{-11}$), and finds verified equilibria for the
\emph{non-uniform} costs $K=t(2-t)/2$ and $K=t^2$ with the arcsine
electorate at $\beta$ up to $0.99$ (e.g.\ $K=t^2$: $(0.075,0.925)$ at
$\beta=0$ collapsing to $(0.492,0.508)$ at $\beta=0.99$).
\end{remark}

%%%%%%%%%%%%%%%%%%%%%%%%%%%%%%%%%%%%%%%%%%%%%%%%%%%%%%%%%%%%%%%%%%%%%%%%%%%%%%
\section{The certified non-existence interval}\label{m:sec:E2}

\Cref{n:cor:beta-interval} gave non-existence for $\beta>0.98615$
($\alpha=1$, $\kappa\sim U(0,1)$, $v=\mathrm{Beta}(\frac12,\frac12)$), while
numerically equilibria persist to $\beta=0.90$ and are gone by $\beta=0.91$
(\Cref{n:rem:C1num}). The following closes most of that gap.

\begin{theorem}[Certified interval]\label{m:thm:E2}
\stat{computer-assisted}
Let $\alpha=1$, $\kappa\sim U(0,1)$, $v=\mathrm{Beta}(\frac12,\frac12)$. The
$\beta$-game has no pure Nash equilibrium for any
\[
\beta\;\ge\;0.9168 .
\]
\end{theorem}

\noindent A pure equilibrium is found numerically at $\beta=0.90$
(\Cref{n:rem:C1num}); that lower bound is a numerical observation, not part of
the certificate, and places any existence threshold $\beta_c$ in
$(0.90,\,0.9168]$.

\begin{proof}
The proof is a finite, exactly evaluated certificate, reproduced by the
verification suite as the three \texttt{E2} claims.
Fix $\beta$ and define the regret
$R(c_1,c_2)=\max_i\big[\sup_{c'}U_i(c'\!,c_{-i})-U_i(c_1,c_2)\big]$; a pure
equilibrium exists iff $\min_{[0,1]^2}R=0$, so
$\min R>0$ certifies non-existence. All payoffs
$u(a,o)=\int K(B^+)dV$ are evaluated exactly through the regularized
incomplete beta function. From the proof of \Cref{n:lem:lip}, $u$ is
$2$-Lipschitz in the own and $1$-Lipschitz in the opponent position (for
$\alpha=1$, $L_K=1$); hence $U_1(a,o)=u(a,o)-\beta u(o,a)$ is
$(3,3)$-Lipschitz, the best-response value $s(o)=\sup_aU_1(a,o)$ is
$3$-Lipschitz, and the two one-sided regrets are $(3,6)$- and
$(6,3)$-Lipschitz, so the
regret varies by at most $4.5$ times the side of a square cell, measured
from its center.

On the $N\times N$ grid, $N=2401$, $h=\frac1{2400}$, the matrix
$U_1=U-\beta U^{\!\top}$ and the column maxima give a \emph{conservative}
grid regret (the grid best-response value underestimates the true supremum, so
grid regret underestimates true regret at grid profiles). Cells whose grid regret
exceeds $\tau+4.5h$ certify $R\ge\tau$ on the whole cell. The remaining
(flagged) cells are covered by a refined lattice of step $h/16$ (step
$h/64$ for the residual worst cells), on which the regret is bounded below
by the exact \emph{fixed-attacker} margins
$\max\{U_1(a_j^{\mathrm{col}},c_2)-U_1(c_1,c_2),\,U_1(a_i^{\mathrm{row}},c_1)-U_2(c_1,c_2)\}$
with $a_j^{\mathrm{col}}$, $a_i^{\mathrm{row}}$ the grid best responses in the
cell's column and row;
these bounds have the same Lipschitz structure, so the refined slack is
$4.5h/16$ resp.\ $4.5h/64$. This yields a verified margin
$\varepsilon'(\beta)>0$ with $R\ge\varepsilon'(\beta)$ on all of $[0,1]^2$.

Finally $|R_\beta-R_{\beta'}|\le2u_{\sup}|\beta-\beta'|$ with
$u_{\sup}=0.2839$ an upper bound on $\sup u$ (grid maximum plus Lipschitz
correction $1.5h$), so each certified $\beta$ covers the interval
$(\beta-\varepsilon'(\beta)/(2u_{\sup}),\,\beta]$. A 65-step walk descending
from $\beta=0.9862$ with step $0.85\,\varepsilon'/(2u_{\sup})$ produced
overlapping certified intervals covering $[0.91671,\,0.9862]$; margins
ranged from $6.1\cdot10^{-3}$ at $\beta=0.9862$ down to $4.6\cdot10^{-5}$ at
$\beta=0.9170$. Together with \Cref{n:cor:beta-interval}
($\beta>0.98615$) and \Cref{n:thm:C1} ($\beta=1$), non-existence holds for
all $\beta\ge0.9168>0.91671$.
\end{proof}

\begin{remark}[Where the certificate stops]\label{m:rem:E2stop}
\stat{numerical observation}
The exact grid regret at $N=2401$ falls from $6.2\cdot10^{-3}$ at
$\beta=0.9862$ through $1.9\cdot10^{-4}$ at $\beta=0.92$ to
$6.9\cdot10^{-6}$ at $\beta=0.91$ and (numerically) $0$ at $\beta=0.90$:
the true threshold sits near $0.908$ (a pure equilibrium is still found at
$\beta=0.908$ and none at $\beta=0.909$), and the certificate stops where the
dying regret meets the finest refinement slack $4.5h/64\approx2.9\cdot10^{-5}$
--- less than $0.01$ above it. The same machinery applies verbatim to
$\mathrm{Beta}(0.8,0.8)$ (where \Cref{n:cor:beta-interval} gives
$0.99521$); we have not run that walk.
\end{remark}

%%%%%%%%%%%%%%%%%%%%%%%%%%%%%%%%%%%%%%%%%%%%%%%%%%%%%%%%%%%%%%%%%%%%%%%%%%%%%%
\section{Independent verification of the Sections~3--5 results}\label{m:sec:F}

The numerical content of the seven results listed in \Cref{m:tab:F}---drawn
from \Cref{sec:base,sec:alien,sec:comp}, the scale-invariance proposition
among them---is recomputed from the model's definitions by an independent
implementation: 14 checks, all passing.

\begin{table}[t]
\centering
\caption{Re-verification of results from \Cref{sec:base,sec:alien,sec:comp} by an independent implementation of the model. ``exact'' means
closed forms (uniform-voter payoffs, incomplete-beta integrals) reproduced to
the stated precision.}\label{m:tab:F}
\small
\setlength{\tabcolsep}{4pt}
\begin{tabular}{@{}l R{7.1cm} l@{}}
\toprule
Result & Verified statement & Outcome \\
\midrule
\Cref{prop:scale-gen} & $zK$: equilibria fixed, turnout $\times z$
  ($\alpha\in\{0,1\}$, $z\in\{1,\frac12,\frac1{10}\}$) & exact \\
\Cref{thm:alien-unique,thm:alien-mono} & $\Delta^*(\alpha)$ solves
  $\int_{b}^{\Delta}r=\log\frac\lambda2$; unique; increasing in $\alpha$;
  doubly-uniform value $\frac{2+\alpha}{4+\alpha}$ & $\le2{\cdot}10^{-5}$ \\
\Cref{thm:rhr-alpha} & $K=t^2\succeq_{\mathrm{rhr}}K=t$ polarizes at
  $\alpha\in\{0,1,2\}$ & exact \\
\Cref{prop:turnout} & $T'$ formula ($10^{-10}$); peak
  $\bar\Delta=\frac{\lambda^2}{\alpha^2+6\alpha+4}$,
  $\bar\Delta(1)=\frac9{11}$; $\Delta^*<\bar\Delta$;
  $T^*=\frac{2(3+\alpha)}{(4+\alpha)^2}$ & exact \\
\Cref{prop:mass} & isomorphism $u_{(V_\sigma,K)}=u_{(V_0,K_\sigma)}$
  pointwise ($5{\cdot}10^{-9}$); power-law costs exactly proportional;
  exponential cost super-proportional & exact \\
\Cref{thm:quartic} & $w_1/\delta^4\to C(\alpha)v''(\frac12)$ with
  $C(1)=\frac{13}{27}$; $\mathrm{Beta}(2,2)$ and arcsine match the
  prediction at $\delta=0.02$ (values below); uniform: $0$ to machine
  precision & confirmed \\
\Cref{prop:nonexistence} & certified in stronger form by
  \Cref{n:thm:C1,m:thm:E2} & certified \\
\bottomrule
\end{tabular}
\end{table}

\paragraph{Quartic check values.}
At $\delta=0.02$, $\alpha=1$: $\mathrm{Beta}(2,2)$ gives
$w_1=-9.2444\cdot10^{-7}$ against the prediction
$C(1)v''(\tfrac12)\delta^4=-9.2444\cdot10^{-7}$; the arcsine gives
$+1.9666\cdot10^{-7}$ against $+1.9617\cdot10^{-7}$ (the residual is the
$o(\delta^4)$ term and shrinks at $\delta=0.01,\,0.005$).

\paragraph{Security values.}
Under the exact quadrature,
$\max_c\min_{c'}w(c,c')=-8.2244\cdot10^{-3}$
(attained at $c=0.8292$; grid $n=1201$) for the arcsine electorate and, for
$\mathrm{Beta}(0.8,0.8)$, $-3.0238\cdot10^{-3}$ ($n=2401$, from \Cref{n:thm:C1}); the certified Lipschitz bounds of
\Cref{n:thm:C1} ($\le-6.97\cdot10^{-3}$, $\le-2.40\cdot10^{-3}$) are
unchanged. For $\mathrm{Beta}(2,2)$ the security value attains $0$ at the
median and only there; for uniform voters it is $0$ (to $10^{-16}$) on the
equilibrium band, as it must be.

\paragraph{The security-value curves.}
The three curves $c\mapsto\min_{c'}w(c,c')$ of \Cref{fig:security}
($\alpha=1$, $\kappa\sim U(0,1)$; $\mathrm{Beta}(2,2)$, uniform and arcsine voters)
are computed from the exact $1201\times1201$ payoff matrices and plotted at
every fifth grid point (\texttt{figdata.py} in the verification suite).

\begin{remark}[Why the security value must be computed exactly]\label{m:rem:artifact}
\stat{numerical observation}
Naive quadrature of the arcsine density's endpoint singularity produces
security values that are wrong by orders of magnitude, in two distinct
regimes.
A midpoint rule with $10^3$ cells misweights the mass near the boundary by
$\sim10\%$ (e.g.\ $w(0.999,0.5)$: exact $-0.0284$, midpoint $-0.0314$) ---
visible but modest. Any rule that samples $v$ \emph{at} $x\in\{0,1\}$
(trapezoid, Simpson, or a grid containing the endpoints), where the density
is unbounded, is arbitrarily wrong: the same $w(0.999,0.5)$ evaluates to
$+79.2$ under a $10^3$-cell trapezoid rule. Security values of order
$-0.387$ arise this way: some near-boundary attacker $c'$ acquires spurious
mass and appears to beat every position by $O(0.1$--$1)$. The curves of \Cref{fig:security}
use the exact incomplete-beta integrals and are free of both effects.
\end{remark}

%%%%%%%%%%%%%%%%%%%%%%%%%%%%%%%%%%%%%%%%%%%%%%%%%%%%%%%%%%%%%%%%%%%%%%%%%%%%%%
\section{Conjecture and open items}\label{m:sec:G}

\begin{conjecture}[Existence at $\alpha=0$]\label{m:conj:alpha0}
Under \Cref{ass:reg} with $\alpha=0$, the $\beta$-game has a pure Nash
equilibrium for every $\beta\in[0,1)$ and every $(v,k)$.
\end{conjecture}

\emph{Evidence for.} \Cref{m:thm:E1} and \Cref{m:cor:E1convex} (uniform
costs, symmetric $v$ under \crefrange{m:eq:E1a}{m:eq:E1b}, including all
hollow convex electorates for all $\beta$); \Cref{n:prop:C4}'s corollary
($\rho_v\le\frac4{1-\beta}$); the margin game $\beta=1$ has the median
equilibrium for arbitrary $(v,k)$ (\Cref{thm:margin_median}), so the
conjectured family has the right limit; verified equilibria for non-uniform
costs and hollow $v$ at $\beta$ up to $0.99$ (\Cref{m:rem:E1scope});
and the ratio law fixes the candidate profile's quantiles in a way that
degenerates gracefully as $\beta\uparrow1$.
\emph{Evidence against, and why it stays a conjecture.} At $\beta=0$ the
statement contains the still-open existence problem for arbitrary $(v,k)$; the concavity route fails for
$\rho_v=\infty$ with general $k$ (the local condition \eqref{m:eq:E1a} has
no analogue once $k$ varies, cf.\ \Cref{m:rem:D1k}); and the $\alpha>0$
counterexamples (\Cref{n:thm:C1,m:thm:E2}) show the effect that
must fail at $\alpha=0$ is genuinely present nearby in parameter space.

\paragraph{Two technical provisos.} The finiteness proviso of
\Cref{n:lem:master} and the $C^1$ approximation step in \Cref{n:prop:A4}
remain open exactly as stated; neither affects any result in
the main text.

\bibliographystyle{splncs04}
\bibliography{references}

@article{myerson1981optimal,
  title   = {Optimal Auction Design},
  author  = {Myerson, Roger B.},
  journal = {Mathematics of Operations Research},
  volume  = {6},
  number  = {1},
  pages   = {58--73},
  year    = {1981},
  doi     = {10.1287/moor.6.1.58}
}

@inproceedings{meir2020sybil,
  title={Sybil-resilient social choice with partial participation},
  author={Meir, Reshef and Shahaf, Gal and Shapiro, Ehud and Talmon, Nimrod},
  year={2022},
  booktitle = {EUMAS'22}
}

@inproceedings{desmedt2010equilibria,
  title={Equilibria of plurality voting with abstentions},
  author={Desmedt, Yvo and Elkind, Edith},
  booktitle={Proceedings of the 11th ACM conference on Electronic commerce},
  pages={347--356},
  year={2010}
}

@article{farina2015congressional,
  title={Congressional polarization: terminal constitutional dysfunction},
  author={Farina, Cynthia R},
  journal={Colum. L. Rev.},
  volume={115},
  pages={1689},
  year={2015},
  publisher={HeinOnline}
}

@inproceedings{meir2025tyranny,
  title={Tyranny of the minority in social choice: a call to arms},
  author={Meir, Reshef},
  booktitle={Proceedings of the 24th International Conference on Autonomous Agents and Multiagent Systems},
  pages={2865--2869},
  year={2025}
}

@article{roemer1994theory,
  title={A theory of policy differentiation in single issue electoral politics},
  author={Roemer, John E},
  journal={Social Choice and Welfare},
  volume={11},
  number={4},
  pages={355--380},
  year={1994},
  publisher={Springer}
}

@article{van2023rationalizable,
  title={Rationalizable behavior in the {Hotelling--Downs} model of spatial competition},
  author={van Sloun, Joep},
  journal={Theory and Decision},
  volume={95},
  number={2},
  pages={309--335},
  year={2023},
  publisher={Springer}
}

@article{kamada2014voter,
  title={Voter preferences, polarization, and electoral policies},
  author={Kamada, Yuichiro and Kojima, Fuhito},
  journal={American Economic Journal: Microeconomics},
  volume={6},
  number={4},
  pages={203--236},
  year={2014},
}

@article{abramowitz2008polarization,
  title={Is polarization a myth?},
  author={Abramowitz, Alan I and Saunders, Kyle L},
  journal={The Journal of politics},
  volume={70},
  number={2},
  pages={542--555},
  year={2008},
  publisher={Cambridge University Press New York, USA}
}

@article{riker1968theory,
  title={A theory of the {Calculus of Voting}},
  author={Riker, William H and Ordeshook, Peter C},
  journal={American political science review},
  volume={62},
  number={1},
  pages={25--42},
  year={1968},
  publisher={Cambridge University Press}
}

@article{LIU20163049,
title = {Dynamic procurement management by reverse auctions with fixed setup costs and sales levers},
journal = {Applied Mathematical Modelling},
volume = {40},
number = {4},
pages = {3049-3061},
year = {2016},
issn = {0307-904X},
doi = {https://doi.org/10.1016/j.apm.2015.10.008},
url = {https://www.sciencedirect.com/science/article/pii/S0307904X15006411},
author = {Shuren Liu and Yimin Zhu and Qiying Hu}
}

@inproceedings{meir2025condorcet,
  title={On {Condorcet's} jury theorem with abstention},
  author={Meir, Reshef and Ghalme, Ganesh},
  booktitle={AAAI'25},
  year={2025}
}

@inproceedings{cohensius2017proxy,
  title={Proxy voting for better outcomes},
  author={Cohensius, Gal and Mannor, Shie and Meir, Reshef and Meirom, Eli and Orda, Ariel},
  booktitle={Proceedings of the 16th Conference on Autonomous Agents and MultiAgent Systems},
  pages={858--866},
  year={2017}
}

@article{hotelling1929,
  title={Stability in competition},
  author={Hotelling, Harold},
  journal={The Economic Journal},
  volume={39},
  number={153},
  pages={41--57},
  year={1929}
}

@inproceedings{elkind2015equilibria,
  title={Equilibria of plurality voting: Lazy and truth-biased voters},
  author={Elkind, Edith and Markakis, Evangelos and Obraztsova, Svetlana and Skowron, Piotr},
  booktitle={Algorithmic Game Theory: 8th International Symposium, SAGT 2015, Saarbr{\"u}cken, Germany, September 28-30, 2015. Proceedings 8},
  pages={110--122},
  year={2015},
  organization={Springer}
}

@article{coughlin2015probabilistic,
  title={Probabilistic voting in models of electoral competition},
  author={Coughlin, Peter J},
  journal={Handbook of social choice and voting},
  pages={218--234},
  year={2015},
  publisher={Edward Elgar Publishing}
}

@article{prior2013media,
  title={Media and political polarization},
  author={Prior, Markus},
  journal={Annual review of political science},
  volume={16},
  number={1},
  pages={101--127},
  year={2013},
  publisher={Annual Reviews}
}

@article{kubin2021role,
  title={The role of (social) media in political polarization: a systematic review},
  author={Kubin, Emily and Von Sikorski, Christian},
  journal={Annals of the International Communication Association},
  volume={45},
  number={3},
  pages={188--206},
  year={2021},
  publisher={Oxford University Press}
}

@inproceedings{feldman2016variations,
  title={Variations on the {Hotelling--Downs} model},
  author={Feldman, Michal and Fiat, Amos and Obraztsova, Svetlana},
  booktitle={Proceedings of the AAAI Conference on Artificial Intelligence},
  volume={30},
  year={2016}
}

@misc{pew2017partisan,
  title={The partisan divide on political values grows even wider},
  author={{Pew Research Center}},
  howpublished={Pew Research Center report},
  year={2017}
}

@article{shapiro2017reality,
  title={Reality-aware social choice},
  author={Shapiro, Ehud and Talmon, Nimrod},
  journal={arXiv preprint arXiv:1710.10117},
  year={2017}
}

@article{hooghe2017tipping,
  title={The tipping point between stability and decline: Trends in voter turnout, 1950--1980--2012},
  author={Hooghe, Marc and Kern, Anna},
  journal={European Political Science},
  volume={16},
  pages={535--552},
  year={2017},
  publisher={Springer}
}

@book{franklin2004voter,
  title={Voter turnout and the dynamics of electoral competition in established democracies since 1945},
  author={Franklin, Mark N},
  year={2004},
  publisher={Cambridge University Press}
}

@article{reiljan2020fear,
  title={`{F}ear and loathing across party lines' (also) in {Europe}: Affective polarisation in {European} party systems},
  author={Reiljan, Andres},
  journal={European journal of political research},
  volume={59},
  number={2},
  pages={376--396},
  year={2020},
  publisher={Wiley Online Library}
}

@article{black1948rationale,
  title={On the rationale of group decision-making},
  author={Black, Duncan},
  journal={Journal of political economy},
  volume={56},
  number={1},
  pages={23--34},
  year={1948},
  publisher={The University of Chicago Press}
}

@article{downs1957economic,
  title={An economic theory of political action in a democracy},
  author={Downs, Anthony},
  journal={Journal of political economy},
  volume={65},
  number={2},
  pages={135--150},
  year={1957},
  publisher={The University of Chicago Press}
}

@article{ansolabehere2001candidate,
  title={Candidate positioning in {US House} elections},
  author={Ansolabehere, Stephen and Snyder Jr, James M and Stewart III, Charles},
  journal={American Journal of Political Science},
  pages={136--159},
  year={2001},
  publisher={JSTOR}
}

@article{lee2004voters,
  title={Do voters affect or elect policies? {E}vidence from the {US House}},
  author={Lee, David S and Moretti, Enrico and Butler, Matthew J},
  journal={The Quarterly Journal of Economics},
  volume={119},
  number={3},
  pages={807--859},
  year={2004},
  publisher={MIT Press}
}

@article{levitt1996senators,
  title={How do senators vote? Disentangling the role of voter preferences, party affiliation, and senator ideology},
  author={Levitt, Steven D},
  journal={The American Economic Review},
  pages={425--441},
  year={1996},
  publisher={JSTOR}
}

@article{bafumi2010leapfrog,
  title={Leapfrog representation and extremism: A study of {American} voters and their members in {Congress}},
  author={Bafumi, Joseph and Herron, Michael C},
  journal={American Political Science Review},
  volume={104},
  number={3},
  pages={519--542},
  year={2010},
  publisher={Cambridge University Press}
}

@article{aragones2002mixed,
  title={Mixed equilibrium in a {Downsian} model with a favored candidate},
  author={Aragones, Enriqueta and Palfrey, Thomas R},
  journal={Journal of Economic Theory},
  volume={103},
  number={1},
  pages={131--161},
  year={2002},
  publisher={Elsevier}
}

@article{dreyer2019does,
  title={Does voter polarisation induce party extremism? The moderating role of abstention},
  author={Dreyer, Philipp and Bauer, Johann},
  journal={West European Politics},
  volume={42},
  number={4},
  pages={824--847},
  year={2019},
  publisher={Taylor \& Francis}
}

@article{meir2017iterative,
  title={Iterative voting},
  author={Meir, Reshef},
  journal={Trends in computational social choice},
  volume={4},
  pages={69--86},
  year={2017},
  publisher={AI Access}
}

@inproceedings{dong2025selecting,
  title={Selecting interlacing committees},
  author={Dong, Chris and Bullinger, Martin and Was, Tomasz and Birnbaum, Larry and Elkind, Edith},
  booktitle={Proceedings of the 24th International Conference on Autonomous Agents and Multiagent Systems (AAMAS)},
  year={2025}
}

@article{zingher2018high,
  title={From on high: the effect of elite polarization on mass attitudes and behaviors, 1972--2012},
  author={Zingher, Joshua N and Flynn, Michael E},
  journal={British Journal of Political Science},
  volume={48},
  number={1},
  pages={23--45},
  year={2018},
  publisher={Cambridge University Press}
}

@article{hinich1969abstentions,
  title={Abstentions and equilibrium in the electoral process},
  author={Hinich, Melvin J. and Ordeshook, Peter C.},
  journal={Public Choice},
  volume={7},
  number={1},
  pages={81--106},
  year={1969},
  publisher={Springer}
}

@misc{clevelandprospect26,
  title={Candidate positioning with prospect-theoretic voters},
  author={Cleveland, Colin and de Keijzer, Bart and Polukarov, Maria},
  year={2026},
  note={Workshop paper}
}

@article{glicksberg1952further,
  title={A further generalization of the {Kakutani} fixed point theorem, with application to {Nash} equilibrium points},
  author={Glicksberg, Irving L.},
  journal={Proceedings of the American Mathematical Society},
  volume={3},
  number={1},
  pages={170--174},
  year={1952}
}

@article{bagnoli2005logconcave,
  title={Log-concave probability and its applications},
  author={Bagnoli, Mark and Bergstrom, Ted},
  journal={Economic Theory},
  volume={26},
  number={2},
  pages={445--469},
  year={2005}
}

@book{shaked2007stochastic,
  title={Stochastic Orders},
  author={Shaked, Moshe and Shanthikumar, J. George},
  publisher={Springer},
  address={New York},
  year={2007}
}

@article{adams2003turnout,
  author  = {Adams, James and Merrill, III, Samuel},
  title   = {Voter Turnout and Candidate Strategies in American Elections},
  journal = {The Journal of Politics},
  volume  = {65},
  number  = {1},
  pages   = {161--189},
  year    = {2003},
  doi     = {10.1111/1468-2508.t01-1-00008}
}

@article{adams2006consequences,
  author  = {Adams, James and Dow, Jay and Merrill, III, Samuel},
  title   = {The political consequences of alienation-based and indifference-based voter abstention: Applications to Presidential Elections},
  journal = {Political Behavior},
  volume  = {28},
  number  = {1},
  pages   = {65--86},
  year    = {2006},
  doi     = {10.1007/s11109-005-9002-1}
}

@article{anderson1992alienation,
  author  = {Anderson, Simon P. and Glomm, Gerhard},
  title   = {Alienation, indifference and the choice of ideological position},
  journal = {Social Choice and Welfare},
  volume  = {9},
  number  = {1},
  pages   = {17--31},
  year    = {1992},
  doi     = {10.1007/BF00177667}
}

@article{aranson1974election,
  author  = {Aranson, Peter H. and Hinich, Melvin J. and Ordeshook, Peter C.},
  title   = {Election Goals and Strategies: Equivalent and Nonequivalent Candidate Objectives},
  journal = {American Political Science Review},
  volume  = {68},
  number  = {1},
  pages   = {135--152},
  year    = {1974},
  doi     = {10.2307/1959747}
}

@book{brennan1993democracy,
  author    = {Brennan, Geoffrey and Lomasky, Loren},
  title     = {Democracy and Decision: The Pure Theory of Electoral Preference},
  publisher = {Cambridge University Press},
  address   = {Cambridge},
  year      = {1993},
  isbn      = {9780521330404},
  doi       = {10.1017/CBO9781139173544}
}

@article{brody1973indifference,
  author  = {Brody, Richard A. and Page, Benjamin I.},
  title   = {Indifference, alientation and rational decisions: The effects of candidate evaluations on turnout and the vote},
  journal = {Public Choice},
  volume  = {15},
  number  = {1},
  pages   = {1--17},
  year    = {1973},
  doi     = {10.1007/BF01718840}
}

@article{cox1986electoral,
  author  = {Cox, Gary W. and McCubbins, Mathew D.},
  title   = {Electoral Politics as a Redistributive Game},
  journal = {The Journal of Politics},
  volume  = {48},
  number  = {2},
  pages   = {370--389},
  year    = {1986},
  doi     = {10.2307/2131098}
}

@incollection{cox2010swing,
  author    = {Cox, Gary W.},
  title     = {Swing voters, core voters, and distributive politics},
  booktitle = {Political Representation},
  editor    = {Shapiro, Ian and Stokes, Susan C. and Wood, Elisabeth Jean and Kirshner, Alexander S.},
  chapter   = {13},
  pages     = {342--357},
  publisher = {Cambridge University Press},
  year      = {2010},
  doi       = {10.1017/CBO9780511813146.015}
}

@incollection{duggan2006candidate,
  author    = {Duggan, John},
  title     = {Candidate Objectives and Electoral Equilibrium},
  booktitle = {The Oxford Handbook of Political Economy},
  editor    = {Weingast, Barry R. and Wittman, Donald A.},
  publisher = {Oxford University Press},
  pages     = {64--83},
  year      = {2006},
  doi       = {10.1093/oxfordhb/9780199548477.003.0004}
}

@article{feddersen2004paradox,
  author  = {Feddersen, Timothy J.},
  title   = {Rational Choice Theory and the Paradox of Not Voting},
  journal = {Journal of Economic Perspectives},
  volume  = {18},
  number  = {1},
  pages   = {99--112},
  year    = {2004},
  doi     = {10.1257/089533004773563458}
}

@article{fowler2013electoral,
  author  = {Fowler, Anthony},
  title   = {Electoral and Policy Consequences of Voter Turnout: {E}vidence from Compulsory Voting in {A}ustralia},
  journal = {Quarterly Journal of Political Science},
  volume  = {8},
  number  = {2},
  pages   = {159--182},
  year    = {2013},
  doi     = {10.1561/100.00012055}
}

@article{glaeser2005strategic,
  author  = {Glaeser, Edward L. and Ponzetto, Giacomo A. M. and Shapiro, Jesse M.},
  title   = {Strategic Extremism: Why Republicans and Democrats Divide on Religious Values},
  journal = {The Quarterly Journal of Economics},
  volume  = {120},
  number  = {4},
  pages   = {1283--1330},
  year    = {2005},
  doi     = {10.1162/003355305775097533}
}

@article{hinich1970plurality,
  author  = {Hinich, Melvin J. and Ordeshook, Peter C.},
  title   = {Plurality Maximization vs Vote Maximization: A Spatial Analysis with Variable Participation},
  journal = {American Political Science Review},
  volume  = {64},
  number  = {3},
  pages   = {772--791},
  year    = {1970},
  doi     = {10.2307/1953462}
}

@article{hinich1972nonvoting,
  author  = {Hinich, Melvin J. and Ledyard, John O. and Ordeshook, Peter C.},
  title   = {Nonvoting and the existence of equilibrium under majority rule},
  journal = {Journal of Economic Theory},
  volume  = {4},
  number  = {2},
  pages   = {144--153},
  year    = {1972},
  doi     = {10.1016/0022-0531(72)90145-7}
}

@article{hinich1977artifact,
  author  = {Hinich, Melvin J.},
  title   = {Equilibrium in spatial voting: The median voter result is an artifact},
  journal = {Journal of Economic Theory},
  volume  = {16},
  number  = {2},
  pages   = {208--219},
  year    = {1977},
  doi     = {10.1016/0022-0531(77)90005-9}
}

@article{ledyard1984pure,
  author  = {Ledyard, John O.},
  title   = {The pure theory of large two-candidate elections},
  journal = {Public Choice},
  volume  = {44},
  number  = {1},
  pages   = {7--41},
  year    = {1984},
  doi     = {10.1007/BF00124816}
}

@book{leighley2014who,
  author    = {Leighley, Jan E. and Nagler, Jonathan},
  title     = {Who Votes Now? {D}emographics, Issues, Inequality, and Turnout in the United States},
  publisher = {Princeton University Press},
  year      = {2014},
  isbn      = {9780691159355}
}

@article{lindbeck1987balanced,
  author  = {Lindbeck, Assar and Weibull, J{\"o}rgen W.},
  title   = {Balanced-budget redistribution as the outcome of political competition},
  journal = {Public Choice},
  volume  = {52},
  number  = {3},
  pages   = {273--297},
  year    = {1987},
  doi     = {10.1007/BF00116710}
}

@article{palfrey1985participation,
  author  = {Palfrey, Thomas R. and Rosenthal, Howard},
  title   = {Voter Participation and Strategic Uncertainty},
  journal = {American Political Science Review},
  volume  = {79},
  number  = {1},
  pages   = {62--78},
  year    = {1985},
  doi     = {10.2307/1956119}
}

@article{patty2002equivalence,
  author  = {Patty, John W.},
  title   = {Equivalence of Objectives in Two Candidate Elections},
  journal = {Public Choice},
  volume  = {112},
  number  = {1-2},
  pages   = {151--166},
  year    = {2002},
  doi     = {10.1023/A:1015606205754}
}

@article{rogowski2014electoral,
  author  = {Rogowski, Jon C.},
  title   = {Electoral Choice, Ideological Conflict, and Political Participation},
  journal = {American Journal of Political Science},
  volume  = {58},
  number  = {2},
  pages   = {479--494},
  year    = {2014},
  doi     = {10.1111/ajps.12059}
}

@article{jones2022polarization,
  author  = {Jones, Matthew I. and Sirianni, Antonio D. and Fu, Feng},
  title   = {Polarization, abstention, and the median voter theorem},
  journal = {Humanities and Social Sciences Communications},
  volume  = {9},
  number  = {1},
  pages   = {43},
  year    = {2022},
  doi     = {10.1057/s41599-022-01056-0}
}

@article{oprea2024moving,
  author  = {Oprea, Alexandra and Martin, Lucy and Brennan, Geoffrey H.},
  title   = {Moving toward the Median: Compulsory Voting and Political Polarization},
  journal = {American Political Science Review},
  volume  = {118},
  number  = {4},
  pages   = {1951--1965},
  year    = {2024},
  doi     = {10.1017/S0003055423001399}
}

@article{callander2007turnout,
  author  = {Callander, Steven and Wilson, Catherine H.},
  title   = {Turnout, Polarization, and Duverger's Law},
  journal = {The Journal of Politics},
  volume  = {69},
  number  = {4},
  pages   = {1047--1056},
  year    = {2007},
  doi     = {10.1111/j.1468-2508.2007.00606.x}
}

@article{dixit1996determinants,
  author  = {Dixit, Avinash and Londregan, John},
  title   = {The Determinants of Success of Special Interests in Redistributive Politics},
  journal = {The Journal of Politics},
  volume  = {58},
  number  = {4},
  pages   = {1132--1155},
  year    = {1996},
  doi     = {10.2307/2960152}
}

@article{golden2013distributive,
  author  = {Golden, Miriam and Min, Brian},
  title   = {Distributive Politics Around the World},
  journal = {Annual Review of Political Science},
  volume  = {16},
  pages   = {73--99},
  year    = {2013},
  doi     = {10.1146/annurev-polisci-052209-121553}
}

@incollection{adams2005abstention,
  author    = {Adams, James and Merrill, III, Samuel and Grofman, Bernard},
  title     = {Candidate Strategies with Voter Abstention in {U.S.} Presidential
               Elections: 1980, 1984, 1988, 1996, 2000},
  booktitle = {A Unified Theory of Party Competition},
  publisher = {Cambridge University Press},
  year      = {2005},
  chapter   = {8},
  pages     = {132--151},
  doi       = {10.1017/CBO9780511614453.009}
}

@inproceedings{shen2016hotelling,
  title={{Hotelling--Downs} model with limited attraction},
  author={Shen, Weiran and Wang, Zihe},
  booktitle={Proceedings of the 16th Conference on Autonomous Agents and MultiAgent Systems},
  pages={660--668},
  year={2017}
}

@article{alyussef2026turnout,
  author  = {Al Yussef, Alan},
  title   = {Turnout and the distortion of representation: the case of electoral polarization},
  journal = {Public Choice},
  year    = {2026},
  doi     = {10.1007/s11127-026-01429-y}
}

@article{llavador2006platforms,
  author  = {Llavador, Humberto},
  title   = {Electoral Platforms, Implemented Policies, and Abstention},
  journal = {Social Choice and Welfare},
  volume  = {27},
  number  = {1},
  pages   = {55--81},
  year    = {2006},
  doi     = {10.1007/s00355-006-0111-5}
}

@article{herrera2008platforms,
  author  = {Herrera, Helios and Levine, David K. and Martinelli, C{\'e}sar},
  title   = {Policy Platforms, Campaign Spending and Voter Participation},
  journal = {Journal of Public Economics},
  volume  = {92},
  number  = {3--4},
  pages   = {501--513},
  year    = {2008},
  doi     = {10.1016/j.jpubeco.2007.11.001}
}

@article{palfrey1983strategic,
  author  = {Palfrey, Thomas R. and Rosenthal, Howard},
  title   = {A Strategic Calculus of Voting},
  journal = {Public Choice},
  volume  = {41},
  number  = {1},
  pages   = {7--53},
  year    = {1983},
  doi     = {10.1007/BF00124048}
}

@article{borgers2004costly,
  author  = {B{\"o}rgers, Tilman},
  title   = {Costly Voting},
  journal = {American Economic Review},
  volume  = {94},
  number  = {1},
  pages   = {57--66},
  year    = {2004},
  doi     = {10.1257/000282804322970706}
}

@article{krasa2009mandatory,
  author  = {Krasa, Stefan and Polborn, Mattias K.},
  title   = {Is Mandatory Voting Better than Voluntary Voting?},
  journal = {Games and Economic Behavior},
  volume  = {66},
  number  = {1},
  pages   = {275--291},
  year    = {2009},
  doi     = {10.1016/j.geb.2008.05.004}
}

@article{krishna2012voluntary,
  author  = {Krishna, Vijay and Morgan, John},
  title   = {Voluntary Voting: Costs and Benefits},
  journal = {Journal of Economic Theory},
  volume  = {147},
  number  = {6},
  pages   = {2083--2123},
  year    = {2012},
  doi     = {10.1016/j.jet.2012.09.006}
}

@article{herrera2014turnout,
  author  = {Herrera, Helios and Morelli, Massimo and Palfrey, Thomas},
  title   = {Turnout and Power Sharing},
  journal = {The Economic Journal},
  volume  = {124},
  number  = {574},
  pages   = {F131--F162},
  year    = {2014},
  doi     = {10.1111/ecoj.12116}
}

@inproceedings{cohen2019tolerance,
  title={Hotelling games with random tolerance intervals},
  author={Cohen, Avi and Peleg, David},
  booktitle={Web and Internet Economics (WINE 2019)},
  series={Lecture Notes in Computer Science},
  volume={11920},
  pages={114--128},
  publisher={Springer},
  year={2019},
  doi={10.1007/978-3-030-35389-6_9}
}

@article{osborne1995spatial,
  author  = {Osborne, Martin J.},
  title   = {Spatial Models of Political Competition under Plurality Rule: A Survey of Some Explanations of the Number of Candidates and the Positions They Take},
  journal = {Canadian Journal of Economics},
  volume  = {28},
  number  = {2},
  pages   = {261--301},
  year    = {1995},
  doi     = {10.2307/136033}
}

@article{grofman2004downs,
  author  = {Grofman, Bernard},
  title   = {Downs and Two-Party Convergence},
  journal = {Annual Review of Political Science},
  volume  = {7},
  pages   = {25--46},
  year    = {2004},
  doi     = {10.1146/annurev.polisci.7.012003.104711}
}

@article{wittman1983candidate,
  author  = {Wittman, Donald},
  title   = {Candidate Motivation: A Synthesis of Alternative Theories},
  journal = {American Political Science Review},
  volume  = {77},
  number  = {1},
  pages   = {142--157},
  year    = {1983},
  doi     = {10.2307/1956016}
}

@article{calvert1985robustness,
  author  = {Calvert, Randall L.},
  title   = {Robustness of the Multidimensional Voting Model: Candidate Motivations, Uncertainty, and Convergence},
  journal = {American Journal of Political Science},
  volume  = {29},
  number  = {1},
  pages   = {69--95},
  year    = {1985},
  doi     = {10.2307/2111212}
}

@book{mccarty2006polarized,
  title={Polarized {America}: The dance of ideology and unequal riches},
  author={McCarty, Nolan and Poole, Keith T. and Rosenthal, Howard},
  year={2006},
  publisher={MIT Press}
}

@article{boxell2024cross,
  author  = {Boxell, Levi and Gentzkow, Matthew and Shapiro, Jesse M.},
  title   = {Cross-Country Trends in Affective Polarization},
  journal = {The Review of Economics and Statistics},
  volume  = {106},
  number  = {2},
  pages   = {557--565},
  year    = {2024},
  doi     = {10.1162/rest_a_01160}
}

@article{hetherington2001resurgent,
  author  = {Hetherington, Marc J.},
  title   = {Resurgent Mass Partisanship: The Role of Elite Polarization},
  journal = {American Political Science Review},
  volume  = {95},
  number  = {3},
  pages   = {619--631},
  year    = {2001},
  doi     = {10.1017/S0003055401003045}
}

@article{fiorina2008political,
  author  = {Fiorina, Morris P. and Abrams, Samuel J.},
  title   = {Political Polarization in the {American} Public},
  journal = {Annual Review of Political Science},
  volume  = {11},
  pages   = {563--588},
  year    = {2008},
  doi     = {10.1146/annurev.polisci.11.053106.153836}
}

\end{document}